\documentclass[11pt,a4paper,twoside,openright]{memoir}

\checkandfixthelayout

\setsecnumdepth{subsection}
\usepackage{tikz}
\usetikzlibrary{}
\usepgfmodule{oo}

\usepackage[spanish,german,catalan,italian,english]{babel}
\usepackage[utf8]{inputenc}
\usepackage{amsthm}
\usepackage{amssymb}
\usepackage{amsmath}
\usepackage{bbold}
\usepackage{bbm}
\usepackage{nonfloat}
\usepackage{color}
\usepackage{xcolor}
\usepackage[pdftex]{hyperref}
\usepackage{etoolbox}
\newcommand{\bra}[2]{\mathinner{\langle #2|}_{#1}}
\newcommand{\ket}[2]{\mathinner{|#2\rangle}_{\hspace{-0.1em} #1}}

\newcommand{\ketbra}[3]{\mathinner{|#2\rangle\langle #3|}_{#1}}

\usepackage{dsfont}
\usepackage{mathdots}
\usepackage{mathtools}
\usepackage{enumerate}
\usepackage{csquotes}
\usepackage{stmaryrd}
\usepackage[cal=boondox]{mathalfa}
\usepackage{graphicx}
\usepackage{stackengine}
\usepackage{scalerel}
\usepackage{array}
\usepackage{makecell}
\newcolumntype{x}[1]{>{\centering\arraybackslash}p{#1}}
\usepackage{tikz}
\usepackage{tcolorbox}
\usepackage{etoolbox}
\usepackage{epigraph}
\usepackage[mathscr]{eucal}
\usepackage{capt-of}
\usepackage[percent]{overpic}
\usepackage{xspace}

\newcommand{\etal}{\textit{et al. }}

\DeclareGraphicsExtensions{.jpeg, .png, .pdf}

\newtheorem{thm}{Theorem}
\newtheorem*{thm*}{Theorem}
\newtheorem{prop}[thm]{Proposition}
\newtheorem*{prop*}{Proposition}
\newtheorem{lemma}[thm]{Lemma}
\newtheorem*{lemma*}{Lemma}
\newtheorem{cor}[thm]{Corollary}
\newtheorem*{cor*}{Corollary}
\newtheorem{cj}{Conjecture}
\newtheorem*{cj*}{Conjecture}

\newtheorem*{Def*}{Definition}

\counterwithin{thm}{chapter}
\counterwithin{cj}{chapter}

\theoremstyle{definition}
\newtheorem{rem}[thm]{Remark}

\newtheorem{ex}[thm]{Example}

\newcommand{\cI}{\mathcal{I}}
\newcommand{\cR}{\mathcal{R}}
\newcommand{\cC}{\mathcal{C}}
\newcommand{\cS}{\mathcal{S}}
\newcommand{\cT}{\mathcal{T}}

\newcommand{\cH}{\mathcal{H}}

\newcommand{\eps}{\varepsilon}
\newcommand{\DM}{D_{M}}  
\newcommand{\NN}{\mathbb{N}}
\newcommand{\setS}[1]{\mathscr S \left( #1 \right)}
\newcommand{\iid}{i.i.d.\xspace}

\newcommand{\bb}{\begin{equation}}
\newcommand{\bbb}{\begin{equation*}}
\newcommand{\ee}{\end{equation}}
\newcommand{\eee}{\end{equation*}}
\newcommand{\Tr}{\text{Tr}\,}
\newcommand{\tr}{\text{Tr}}
\newcommand{\rk}{\text{rk}\,}
\newcommand*{\coloneqq}{\mathrel{\vcenter{\baselineskip0.5ex \lineskiplimit0pt \hbox{\scriptsize.}\hbox{\scriptsize.}}} =}

\newcommand{\texteq}[1]{\stackrel{\mathclap{\scriptsize \mbox{#1}}}{=}}
\newcommand{\textleq}[1]{\stackrel{\mathclap{\scriptsize \mbox{#1}}}{\leq}}

\newcommand{\textgeq}[1]{\stackrel{\mathclap{\scriptsize \mbox{#1}}}{\geq}}

\newcommand{\id}{{\mathds{1}}}

\newcommand{\lmatrix}{\left(\begin{smallmatrix}}
\newcommand{\rmatrix}{\end{smallmatrix}\right)}
\newcommand{\norm}[1]{\left\lVert#1\right\rVert}

\newcommand{\hyptest}[2]{\begin{tcolorbox}[nofloat, colback=uab100!5!white,colframe=uab100!75!black,title=#2] #1 \end{tcolorbox}  }
\newcommand{\bonus}[2]{\begin{tcolorbox}[colback=uab100!5!white,colframe=uab100!75!black,title=#2] #1 \end{tcolorbox}  }
\DeclareMathOperator{\supp}{supp}
\DeclareMathOperator{\poly}{poly}

\tcbset{colback=green!5,colframe=orange!50!red,
fonttitle=\bfseries, float=htb}

\newtcolorbox[auto counter]{example}[3][]
{float*=ht,title=Observation ~\thetcbcounter: #2,label= ex:#3 ,#1}

\renewenvironment{thebibliography}[1]{
  \begin{oldthebibliography}{#1}
    \setlength{\itemsep}{1ex}
    \setlength{\parskip}{0em}
}
{
  \end{oldthebibliography}
}

\patchcmd{\bibsetup}{\interlinepenalty=5000}{\interlinepenalty=10000}{}{}

\stackMath

\makeatletter
\newcommand*\rel@kern[1]{\kern#1\dimexpr\macc@kerna}
\newcommand*\widebar[1]{%
  \begingroup
  \def\mathaccent##1##2{%
    \rel@kern{0.8}%
    \overline{\rel@kern{-0.8}\macc@nucleus\rel@kern{0.2}}%
    \rel@kern{-0.2}%
  }%
  \macc@depth\@ne
  \let\math@bgroup\@empty \let\math@egroup\macc@set@skewchar
  \mathsurround\z@ \frozen@everymath{\mathgroup\macc@group\relax}%
  \macc@set@skewchar\relax
  \let\mathaccentV\macc@nested@a
  \macc@nested@a\relax111{#1}%
  \endgroup
}

 \usepackage{lipsum}
\usepackage{tcolorbox}

\definecolor{uab100}{RGB}{240,69,35}    

\begin{document}

\frontmatter

\clearpage
\thispagestyle{empty}
\thispagestyle{empty}

\begin{center}
\setlength{\parskip}{0pt}
{\large\textbf{Universitat Aut\`onoma de Barcelona}\par}
\vfill
{\huge \bfseries From asymptotic hypothesis testing to entropy inequalities \par}
\vfill
{\LARGE by \par}
\smallskip
{\LARGE Christoph Hirche \par}
\smallskip
{\large under supervision of \par}
\smallskip
{\large Prof. John Calsamiglia \par}
\vfill
{\large A thesis submitted in partial fulfillment for the \par}
{\large degree of Doctor of Philosophy \par}
{\large in Physics \par}
\bigskip
\bigskip
{\large in \par}
{\large Unitat de F\'isica Te\`orica: Informaci\'o i Fen\`omens Qu\`antics \par}
{\large Departament de F\'isica \par}
{\large Facultat de Ci\`encies \par}
\bigskip
\bigskip
\bigskip
{\Large Bellaterra, March, 2018 \par}
\begin{figure}[ht]
  \centering
  \mbox{\hspace{1.2ex} \includegraphics[height=4cm, width=5cm, keepaspectratio]{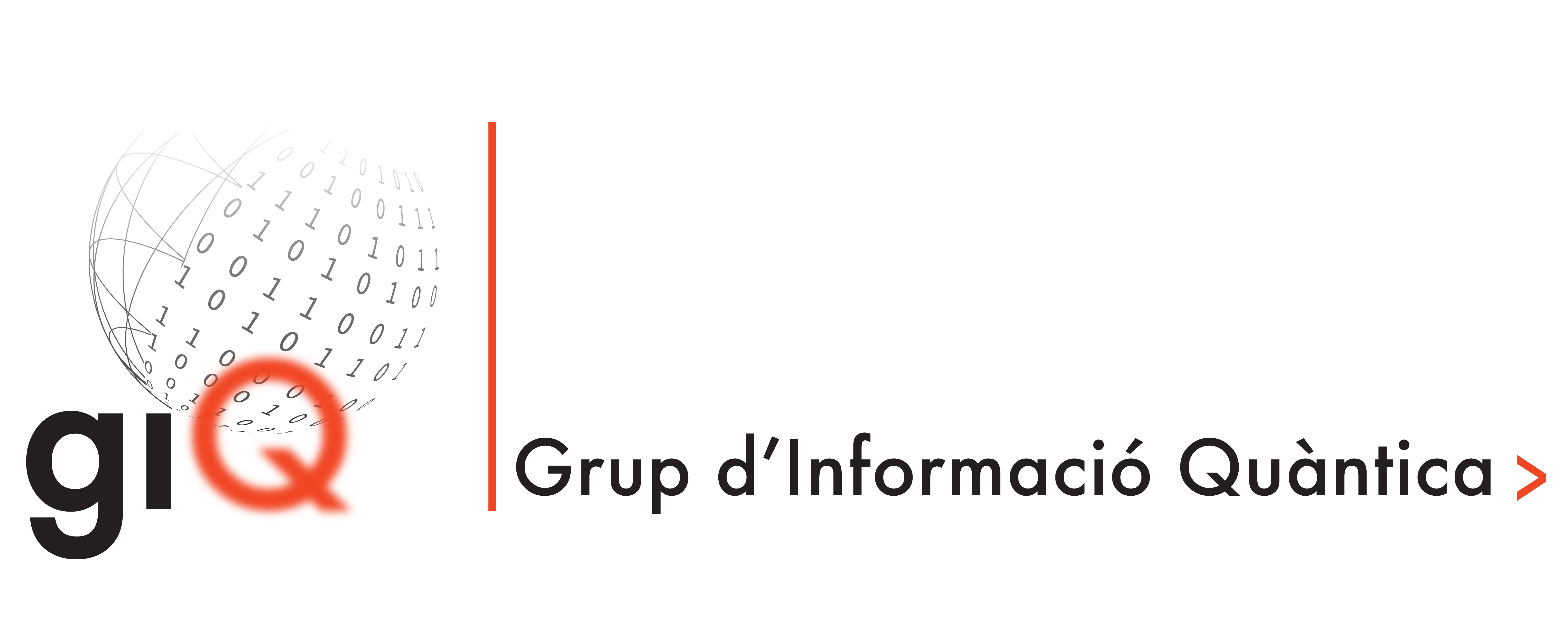}}
\end{figure}
%
%
%
%
%
%
%
%
%
%
%
%
%
%
%
%
%
\end{center}

\null
\newpage

\thispagestyle{empty}
\begin{vplace}[0.7]
\setlength{\epigraphwidth}{12cm}
\epigraph{\includegraphics[width=\textwidth]{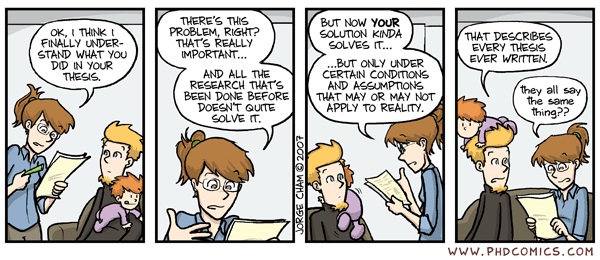}
}{"Piled Higher and Deeper" by Jorge Cham \\
www.phdcomics.com}
\end{vplace}

\cleardoublepage

\begin{abstract}
This thesis addresses the interplay between asymptotic hypothesis testing and entropy inequalities in quantum information theory. 
In the first part of the thesis we focus on hypothesis testing. Here, we consider two main settings; one can either fix quantum states while optimizing over possible measurements or fix a measurement and evaluate its capability to discriminate quantum states by optimizing over such states. With regard to the former setting, we prove a general result on the optimal error rate in asymmetric composite hypothesis testing, which leads to a composite quantum Stein's Lemma. We also discuss how this gives an operational interpretation to several quantities of interest, such as the relative entropy of coherence, and how to transfer the result to symmetric hypothesis testing. For the latter, we give the optimal asymptotic error rates in several symmetric and asymmetric settings, as well as discuss properties and examples of these rates.  

In the second part, the focus is shifted to entropy inequalities. We start with recoverability inequalities, which have gained much attention recently. As it turns out, they are closely related to the first part of the thesis. Using tools which we developed to prove the composite Stein's Lemma, we further prove a strengthened lower bound on the conditional quantum mutual information in terms of a regularized relative entropy featuring an explicit and universal recovery map. Next, we show two a priori different approaches to give an operational interpretation to the relative entropy of recovery via composite hypothesis testing. Then, we discuss and extend some recent counterexamples, which show that the non-regularized relative entropy of recovery is not a lower bound on the conditional quantum mutual information; additionally we provide more counterexamples where some of the involved systems are classical, showing that also in this restricted setting the same bound does not hold. Ultimately we employ the connection between hypothesis testing and recoverability to show that the regularization in our composite Stein's Lemma is indeed needed. 

We then turn to a seemingly different type of entropy inequalities called bounds on information combining, which are concerned with the conditional entropy of the sum of random variables with associated side information. Using a particular recoverability inequality, we show a non-trivial lower bound and additionally conjecture optimal lower and upper bounds. Furthermore, we discuss implications of our bounds to the finite blocklength behavior of Polar codes to attain optimal communication capacities in quantum channels.

Finally, we discuss \mbox{R\'enyi-$2$} entropy inequalities for Gaussian states on infinite dimensional systems, by exploiting their formulation as log-det inequalities to find recoverability related bounds on several interesting quantities. We apply this to Gaussian steerability and entanglement measures, proving their monogamy and several other features. 
\end{abstract}

\newpage

\begin{otherlanguage}{spanish}
\begin{abstract}
Esta tesis trata sobre la relación entre el contraste de hipótesis cuánticas y las desigualdades entrópicas en la teoría cuántica de la información. 
En la primera parte de la tesis nos centramos en el contraste de hipótesis. Aquí, consideramos dos configuraciones principales, o bien fijar los estados cuánticos y optimizar sobre las posibles medidas, o bien fijar una medida y evaluar su capacidad de discriminación de estados cuánticos optimizando sobre estos últimos. En la primera configuración, demostramos un resultado general en la tasa de error óptima en el contraste de hipótesis compuestas asimétricas, que lleva a un Lema de Stein cuántico compuesto. También discutimos como esto da una interpretación operacional a varias cantidades de interés, como la entropía relativa de la coherencia, y cómo transferir este resultado al contraste  de hipótesis en el régimen asintotico. En la segunda, damos la tasa de error asintótica óptima en varias configuraciones, tanto simétricas como asimétricas, y también discutimos las propiedades y algunos ejemplos de estas tasas. 

En la segunda parte, nos centramos en las desigualdades entrópicas. Empezamos con las desigualdades de recuperabilidad, que han recibido mucha atención recientemente. Como vemos, están estrechamente relacionadas con la primera parte de la tesis. Utilizando las herramientas desarrolladas para demostrar el Lema de Stein compuesto, demostramos un límite inferior para la información mutua condicionada en términos de una entropía relativa regularizada que presenta un mapa de recuperación universal explícito. A continuación, mostramos dos enfoques a priori diferentes para dar una interpretación operacional a la entropía relativa de recuperación a través del contraste de hipótesis compuestas. Luego discutimos y ampliamos algunos contraejemplos recientes afirmando que la entropía relativa de recuperación no regularizada no es una cota inferior a la información mutua cuántica condicionada. Además, aportamos más contraejemplos donde algunos de los sistemas cuánticos involucrados son, en realidad, clásicos, viendo que incluso en esta configuración restringida la cota inferior no es correcta. En última instancia, empleamos la conexión entre la contraste de hipótesis y la recuperabilidad para mostrar que la regularización en nuestro Lema de Stein compuesto es, de hecho, necesaria.  
Luego nos centramos en un tipo aparentemente diferente de desigualdades entrópicas, llamadas cotas a la combinación de información, relacionadas con la entropía condicional de la suma de variables aleatorias con información lateral asociada. Usando una desigualdad de recuperabilidad particular, mostramos una cota inferior no trivial y, adicionalmente, conjeturamos cotas óptimas tanto inferiores como superiores. Además, discutimos las implicaciones de nuestras cotas en el comportamiento de la longitud de bloque finita en códigos polares, utilizados en la comunicación clásica sobre canales cuánticos.

Finalmente, discutimos las desigualdades de la entropía de Rényi-2 para estados Gaussianos en sistemas de  dimension infinita, haciendo uso de su formulación como desigualdades log-det para encontrar límites relacionados con la recuperabilidad en varias cantidades de interés. Esto último lo aplicamos a medidas Gaussianas de entrelazamiento y ``steerability'', lo que demuestra su monogamia entre otras características. 
\end{abstract}
\end{otherlanguage}

\newpage

\begin{otherlanguage}{catalan}
\begin{abstract}
Aquesta tesi tracta sobre la relació entre el contrast d'hipòtesis quàntiques i les desigualtats entròpiques en la teoria quàntica de la informació.
A la primera part de la tesi ens centrem en el contrast d'hipòtesis. Aquí, considerem dues configuracions principals, o bé fixar els estats quàntics i optimitzar sobre les possibles mesures, o bé fixar una mesura i avaluar la seva capacitat de discriminació d'estats quàntics optimitzant sobre aquests últims. A la primera configuració,  demostrem un resultat general a la taxa d'error òptima en el contrast d'hipòtesis compostes asimètriques, que porta a un Lema de Stein quàntic compost. També discutim com això dóna una interpretació operacional a diverses quantitats d'interès, com l'entropia relativa de la coherència, i com transferir aquest resultat al contrast d'hipòtesis en el règim asimptòtic. A la segona, donem la taxa d'error asimptòtica òptima en diverses configuracions, tant simètriques com asimètriques, i també discutim les propietats i alguns exemples d'aquestes taxes.

A la segona part, ens centrem en les desigualtats entròpiques. Comencem amb les desigualtats de recuperabilitat, que han rebut molta atenció recentment. Com veiem, estan estretament relacionades amb la primera part de la tesi. Utilitzant les eines desenvolupades per demostrar el Lema de Stein compost,  demostrem una fita inferior per a la informació mútua condicionada en termes d'una entropia relativa regularitzada que presenta un mapa de recuperació universal explícit. A continuació, mostrem dos enfocs a priori diferents per donar una interpretació operacional a l'entropia relativa de recuperació a través del contrast d'hipòtesis compostes. Després  discutim i  ampliem alguns contraexemples recents afirmant que l'entropia relativa de recuperació no regularitzada no és una bona fita inferior a la informació mútua quàntica condicionada. A més, aportem més contraexemples on alguns dels sistemes quàntics involucrats són, en realitat, clàssics, veient que fins i tot en aquesta configuració restringida la fita inferior no és correcta. Al final, fem servir la connexió entre la contrast d'hipòtesis i la recuperabilitat per mostrar que la regularització al nostre Lema de Stein compost és, de fet, necessària.
Després ens centrem en un tipus aparentment diferent de desigualtats entrópicas, anomenades fites a la combinació d'informació, relacionades amb l'entropia condicional de la suma de variables aleatòries amb informació lateral associada. Usant una desigualtat de recuperabilitat particular,  mostrem una cota inferior no trivial i, addicionalment, conjecturem cotes òptimes tant inferiors com superiors. A més, discutim les implicacions de les nostres cotes en el comportament de la longitud de bloc finita en codis polars, emprats en comunicació clàssica en canals quànitcs.

Finalment, discutim les desigualtats de l'entropia de Rényi-2 per a estats Gaussians en sistemes de dimensió infinita, fent ús de la seva formulació com desigualtats log-det per trobar límits relacionats amb la recuperabilitat en diverses quantitats d'interès. Això últim ho apliquem a mesures gaussianes d'entrellaçament i `` steerability '',  demostrant així la seva monogàmia entre altres característiques.

\end{abstract}
\end{otherlanguage}

\cleardoublepage

\chapter*{Acknowledgements}

First of all I would like to thank my supervisor John Calsamiglia for advising me over the last three years, giving me great freedom to work on whatever I was interested in and helping me circumvent all the bureaucratic hurdles we encountered. 

I am happy to defend my thesis in front of a committee of distinguished experts and I thank Nilanjana Datta, Ramon Muñoz-Tapia, Frédéric Dupuis, Anna Sanpera and Alexander Müller-Hermes for taking over that task. 

I am grateful to John Calsamiglia, Vindhiya Prakash and Andreu Riera for proofreading all or parts of my thesis. 

My time in the field has been extremely enjoyable and this is foremost due to the large number of friends and colleagues I had the pleasure to work with: Gerardo Adesso, Emilio Bagan, Mario Berta, Fernando G. S. L. Brand\~ao, John Calsamiglia, Tom Cooney, Andrew J. Ferris, Masahito Hayashi, Masato Koashi, Ludovico Lami, Ciara Morgan, Yoshifumi Nakata, Jonathan P. Olson, David Poulin, David Reeb, Kaushik P. Seshadreesan, John Watrous, Mark M. Wilde and Andreas Winter. 
A particular mention goes to David Reeb for sharing the enthusiasm to work on that one entropy inequality for so many years. 

The path to this thesis lead me to many places, meeting even more great people. Starting with my first quantum information course at the University of Bristol given by Noah Linden and Sandu Popescu. Even more influential at that time was a lecture I probably wasn't even supposed to be in: A graduate level course by Andreas Winter on quantum Shannon theory. 
Following that, I came back to Hannover with the idea to work on this field and convinced to start my Bachelor thesis with Reinhard F. Werner. I am thankful to him for guiding me through my remaining two years in Hannover. 
Another lucky coincidence was that, at about the same time, Ciara Morgan joined the group in Hannover, she was working on just the topics I was interested in and took over the adviser role for my Bachelor and Master thesis, doing a great job at introducing me to the field and the people working in it. I also learned a great deal from Mark M. Wilde, with whom I enjoyed my first  international collaboration at that time. 

The next big step lead me to Barcelona, starting my PhD in GIQ, where I couldn't have wished for a better welcome, thanks to the group, foremost its former members Cecilia Lancien, Mohammad Mehboudi, Krishnakumar Sabapathy, Giannicola Scarpa and everyone else who has over time belonged to the infamous Graciosos group. 

I am also grateful to Fernando G. S. L. Brand\~ao for having me as a visitor for almost four month in his group at Caltech, where I had a great time working with him and Mario Berta. 

Finally, it would not have been the same without all the great friends I met at conferences and while visiting several groups, from interesting discussions to the much needed beer after a long day, of whom I can unfortunately only name very few here: Felix Leditzky (and the group in Boulder I enjoyed visiting last year), Christian Majenz, Frédéric Dupuis, Anna-Lena Hashagen, Daniel Stilck Franca (and many more from the groups in Munich), Rene Schwonnek (and pretty much everyone else from QIG in Hannover) and all the friends in Cambridge, Copenhagen, London, Pasadena, Zürich and many other groups. 

Last but not least, I thank all my friends and family back in Hannover, to where I always enjoy coming back, and of course my \includegraphics[height=\fontcharht\font`\B]{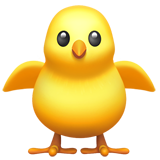}.

\cleardoublepage

\tableofcontents
 
\mainmatter

\cleardoublepage

\mainmatter

\chapter{Introduction} \label{chapter1}

\section{Introduction} \label{sec intro}

In the recent past, research in quantum information theory has been fast approaching practically implementable scenarios, both in theory and practice. With the prospect of having functional mid-scale quantum computers \cite{P18}, quantum communication \cite{satellite} or even a quantum internet \cite{internet} in the foreseeable future, we might soon be able to go beyond the limitations of classical computer science. 
Nevertheless (or rather, for exactly that reason) one of the most important goals of the field remains to find the ultimate bounds on the capabilities obtained by using quantum mechanics. To judge the potential the new technologies would bring within our reach, we need to know what the boundaries are, to which we can push these possibilities. 

In order to do this, we also need to strengthen the set of tools available to us. Two of the most used tools in quantum information theory that lie unarguably at its heart, are hypothesis testing and entropy inequalities; both of which have many important applications and have provided us with the capability to explore the boundaries of information theory.

This work is focused on extending the framework of these two core tools. In particular, showing their close connection and how a better understanding of either one of them can facilitate the investigation of the other. 

The first part of this thesis is focused on asymptotic hypothesis testing.
Hypothesis testing originates in statistical mathematics and is concerned with the question of which hypothesis, from a given set of possible options, is true. A commonly used example is that of a courtroom trial. Generally the defendant is assumed to be not-guilty until proven otherwise. Therefore we call ``non-guilty'' the null hypothesis and ``guilty'' the alternative hypothesis. In this scenario, there are two possible errors to make; one could convict an innocent person (Type 1 error) or acquit a person that committed the crime (Type 2 error). The most natural question is now how well can we minimize both errors at the same time, this is called symmetric hypothesis testing. On the other hand in many scenarios like the one above, we wish to make sure that a certain error is extremely unlikely (here, convicting an innocent person). So, in this case we would fix the probability of a Type 1 error to be very small and try to find the lowest Type 2 error possible under that constraint. This task we call asymmetric hypothesis testing. 

The underlying task of confirming a certain hypothesis is a very fundamental problem and as such also finds application in quantum information theory. Here, our hypotheses are simply the assumptions that we are in possession of a certain quantum state and we would like to verify which state it is. We decide on this by applying a measurement to the state in question. Unfortunately, in most scenarios the errors cannot be made arbitrarily small, due to the inherent uncertainty in quantum mechanics (how a judge can not decide perfectly well if missing a crucial piece of information). However, lower errors can be achieved when we have access to many copies of the given state. The fundamental case of having access to an infinite amount of copies is called asymptotic hypothesis testing. A more precise introduction to the main results in asymptotic hypothesis testing is given in Chapter \ref{hypoTesting}. 

In general, when we are interested in the optimal bounds on hypothesis testing we can consider two different types of optimization. We can optimize over possible measurements for fixed states, which in our courtroom example would be equivalent to picking an optimal strategy for the judge to come to a conclusion in a specific case. Alternatively, we can optimize over the quantum states for a fixed measurement, equivalent to being given a judge with limited capability and asking how sure is the decision in the easiest court case. The former corresponds to the well known setting of quantum state discrimination. The latter describes the capability of a quantum measurement to discriminate states. In the setting where we allow for many arbitrarily chosen states it gives the ultimate capability of the device, which we call the discrimination power of a quantum detector. 
In Chapter \ref{discPower} we give reasonably simple expressions for the discrimination power in several symmetric and asymmetric hypothesis testing scenarios, allowing for arbitrarily chosen states, including non-\iid, entangled and adaptively picked quantum states, thus determining the ultimate discrimination power of the device.

While the relatively simple quantum state discrimination scenario described above is well explored, one often has to deal with more complicated settings. A particular complication that might arise is when we can't determine exactly what the state assigned to our hypothesis is, but rather we only have the information that it must belong to a certain set of states. This is called composite hypothesis testing. While certain special cases had been investigated in the literature, in Chapter \ref{CompHypo} we provide a very general composite quantum Stein's Lemma for arbitrary convex combinations of tensor power states from a freely chosen set. Here, the setting of asymmetric hypothesis testing leading to Stein's Lemma is particularly interesting as we can show that the regularized formula we get in the composite setting is actually optimal, in the sense that a simpler version with optimization over only a single copy of the state is not true. This we are able to prove by connecting the hypothesis testing result to recent results in the field of entropy inequalities and recoverability, establishing a close relation between the two areas of research. 

Thus, entropy inequalities and recoverability will be the topic of the second part of this thesis. 
While many different types of entropy inequalities exist and have proven useful in quantum information theory, the subfield of recoverability inequalities has attracted particular interest in recent years. Motivated by a conjecture in~\cite{Winterconj} the first breakthrough was achieved by Fawzi and Renner~\cite{FR14}, where they show that the conditional quantum mutual information can be lower bounded by a positive term given by the fidelity between the original state and a recovered version of the same state. Here, the latter refers to applying a quantum channel, the so called recovery map, to the state after losing (tracing out) a subsystem. Soon after, Brandao \etal ~\cite{BHOS15} showed that the bound can be further strengthened to either involve a regularized relative entropy or an unregularized measured relative entropy, again each between the same states as in the fidelity bound. In subsequent works, the bounds received further improvement showing properties such as universality (independence of the recovery map of a certain subsystem of the state) and making the channel explicit. An interesting question that remained open was whether the regularized relative entropy bound could be simplified to one that is optimized only on a single copy of the quantum system, until recently Fawzi and Fawzi~\cite{FF17} provided a counter example that shows that such a simplification is indeed not possible. 

In chapter \ref{recoverability} will discuss these recoverability bounds in more detail, with particular focus on the entropic bounds and the connection to hypothesis testing. First we show how the mentioned counterexample also leads to the described result that in the composite quantum Stein's Lemma it can not be sufficient to optimize over a single system. 
Next we investigate the operational interpretation of the regularized relative entropy of recovery and show that such can indeed be given as the optimal rate in certain composite discrimination scenarios. Finally we show how tools developed in the last part, for proving the composite Stein's Lemma, can be used to give a novel recovery lower bound on conditional quantum mutual information based on the regularized relative entropy featuring an explicit and universal recovery map. 

Next, we will turn to a somewhat different type of entropy inequalities which aim to bound the conditional entropy of the sum of random variables when some side information is available. When the side information is classical, those inequalities are well known as \textit{bounds on information combining}. Here, we are interested in the more intricate scenario where the side information is quantum. As we will see, neither does the direct generalization of the classical case hold here nor can any of the classical proof techniques be applied. While this makes it generally very hard to prove any bounds, the afore discussed recoverability bounds come to our rescue. Using the fidelity type lower bound on the conditional quantum mutual information, we prove a non-trivial lower bound on the information gain under combination. This bound, while not optimal, captures some of the important features which the classical bounds possess. Aside from the recoverability bounds, the main techniques we use are duality of classical-quantum channels and novel bounds on the concavity of the von Neumann entropy in terms of the fidelity. Apart from the proven bound, we also conjecture what we believe to be the optimal bounds in the quantum case and provide ample numerical evidence along with some analytical arguments. 
The chapter closes by showing applications of our proven and conjectured lower bounds for investigating a particular class of quantum codes called polar codes. 

Most of the topics presented so far are focused on finite dimensional systems. The final Chapter \ref{Covariance} of the thesis will instead consider entropy inequalities in infinite dimensions. In particular we will look at the practically relevant set of gaussian states. For this class we argue that in many cases it can be useful to consider a different type of entropy, namely the \mbox{R\'enyi-$2$} entropy. Quantities in terms of the \mbox{R\'enyi-$2$} entropies can be written in terms of the logarithm of the determinant of the covariance matrix corresponding to a gaussian state. This allows us to take a novel matrix analytic approach to entropy inequalities. Most of the chapter will consider general (not necessarily quantum) covariance matrices and their properties, with a closer look into questions concerned with recoverability and bounds on entropic terms. In the final section we will then turn to quantum covariance matrices and apply our previous results to the classification of correlations, in particular entanglement and steerability.  We investigate gaussian \mbox{R\'enyi-$2$} version of well known quantities such as the entanglement of formation and the squashed entanglement and explore their properties, ultimately culminating in the insight that both quantities turn out to be equal. This gives a strong contrast to the standard von Neumann case. 

In summary, the thesis is organized as follows. In the remainder of this chapter we will introduce some notation and often used quantities. Then we begin the first part on asymptotic hypothesis testing with an introduction to the most relevant settings in Chapter~\ref{hypoTesting}. In Chapter \ref{CompHypo}, we will discuss results on composite hypothesis testing and in Chapter \ref{discPower}, we provide optimal bounds on the discrimination power of quantum detectors. Then, we will move to the second part of the thesis, which is recoverability and entropy inequalities. Here, in Chapter \ref{recoverability}, we first discuss recoverability and its connection to hypothesis testing. Using recoverability inequalities we then prove bounds on information combining in Chapter \ref{infoCombining}. In Chapter \ref{Covariance} we discuss log-det inequalities in infinite dimension. We then end with some final thoughts and open problems in Chapter~\ref{theEnd}.

The various chapters are based on the following papers and preprints. 
\begin{description}
\item[Chapter~\ref{CompHypo}] \cite{BBH17} M. Berta, F. G. S. L. Brandao, C. Hirche, \textit{On Composite Quantum Hypothesis Testing}, preprint (2017), arXiv 1709.07268. 
\item[Chapter~\ref{discPower}] \cite{HHBC17} C. Hirche, M. Hayashi, E. Bagan, J. Calsamiglia, \textit{Discrimination Power of a Quantum Detector}, Phys. Rev. Lett. 118, 160502, (2017), arxiv 1610.07644. 
\item[Chapter~\ref{recoverability}]  \cite{BBH17} and \cite{CHMOSWW16} T. Cooney, C. Hirche, C. Morgan, J. P. Olson, K. P. Seshadreesan, J. Watrous, M. M. Wilde, \textit{Operational meaning of quantum measures of recovery}, Phys. Rev. A 94, 022310, (2016), arXiv 1509.07127. 
\item[Chapter~\ref{infoCombining}] \cite{HR17} C. Hirche, D. Reeb, \textit{Bounds on Information Combining with Quantum Side Information}, preprint (2017), arXiv 1706.09752. 
\item[Chapter~\ref{Covariance}] \cite{LHAW16} L. Lami, C. Hirche, G. Adesso, A. Winter, \textit{Schur complement inequalities for covariance matrices and monogamy of quantum correlations} Phys. Rev. Lett. 117, 220502 (2016), arXiv:1607.05285 \\
and \cite{LHAW17} L. Lami, C. Hirche, G. Adesso, A. Winter, \textit{From log-determinant inequalities to Gaussian entanglement via recoverability theory} IEEE Trans. Inf. Theory, 63, 11, 7553-7568 (2017), arXiv:1703.06149. 
\end{description}
Finally, during my PhD I contributed to a number of papers that are not directly covered in this thesis, which are listed here for completeness. 
\begin{itemize}
\item \cite{NHKW17} Y. Nakata, C. Hirche, M. Koashi, A. Winter, \textit{Efficient Quantum Pseudorandomness with Nearly Time-Independent Hamiltonian Dynamics}, Phys. Rev. X 7, 021006, (2017), arXiv 1609.07021. 
\item \cite{NHMW15} Y. Nakata, C. Hirche, C. Morgan, A. Winter, \textit{Decoupling with random diagonal unitaries}, Quantum 1, 18, (2017), arXiv 1509.05155. 
\item \cite{NHMW17} Y. Nakata, C. Hirche, C. Morgan, A. Winter, \textit{Unitary 2-designs from random X- and Z-diagonal unitaries}, J. Math. Phys. 58, 052203, (2017), arXiv 1502.07514. 
\item \cite{FHP17} A. J. Ferris, C. Hirche, D. Poulin, \textit{Convolutional Polar Codes}, preprint (2017), arXiv:1704.00715. 
\end{itemize}

\section{Some Preliminaries}\label{Pre}

In this section we will introduce the necessary notation along with some definitions and simple observations. 
For most of this thesis all inner product spaces $\cH$ are finite-dimensional, except when stated otherwise, and $S(\cH)$ denotes the set of positive semi-definite linear operators on $\cH$ of trace one, namely quantum states. Generally, for two positive semidefinite operators we use the Loewner order $S\geq T$, meaning that $S-T$ is positive semidefinite. 
The systems in question will often be modeled by random variables, where we denote classical random variables usually by $X,Y,Z$ and quantum ones based on quantum states by $A,B,C$. Often we associate random variables with inputs and outputs of channels. Here, $X$ usually models a classical input, while $Y$ and $Z$ are classical outputs and $B$ denotes quantum outputs. The channels themselves will usually be denoted with calligraphic letters $\mathcal W$. 

All the matrix inverses in this work are understood as generalized inverses.

Most of this thesis will be based on the investigation of entropic quantities. The most commonly used one is the von Neumann entropy which, for a quantum state $\rho_A$ on a quantum system $A$, which is defined by
\begin{equation}
H(A) = -\tr \rho \log\rho.
\end{equation}
This, reduces to the Shannon entropy in the case of classical states (those which are diagonal in the computational basis). 
Throughout this work, we leave the base of the logarithm unspecified, unless stated otherwise, so that the resulting statements are valid in any base (like binary, or natural); our figures, however, use the natural logarithm. The particular case of the Shannon entropy for a binary probability distribution is called the binary entropy and denoted as $h_2(p)=-p\log{p} - (1-p)\log{(1-p)}$. In the following we will often use the inverse of this function
\begin{equation}
h_2^{-1}:[0,\log2]\to[0,1/2]. 
\end{equation}
Starting from the underlying von Neumann entropy we can define many more entropic quantities of interest, such as the conditional entropy of a bipartite quantum state $\rho_{AB}$ on a quantum system $AB$, which is defined by 
\begin{equation}
H(A|B) = H(AB) - H(B). 
\end{equation}
Whenever the conditioning system is classical, i.e. the state is of the form $\sum_{i=1}^{d} p_i \rho_i^{A}\otimes\ketbra{Y}{i}{i}$, we can state the following important property,
\begin{equation}
H(A|Y) = \sum_{y} p(y) H(A|Y=y).
\end{equation}
This obviously holds also for the Shannon entropy, but most importantly, we cannot write down such a decomposition when the conditioning system is quantum.

Often we want to stress the duality of classical-quantum states to classical-quantum channels. In this case, for a given channel $\mathcal W$ with input modeled by a random variable $X$ and the output by $B$, we write equivalently 
\begin{equation}
H(X|B) = H(\mathcal W). \label{Hchannel}
\end{equation}
Usually we assume here the uniform distribution over input values $X$. An additional useful entropic quantity is the mutual information defined as 
\begin{equation}
I(A:B) = H(A) + H(B) - H(AB).
\end{equation} 
Again for a channel $\mathcal W$ with classical input, we write 
\begin{equation}
I(\mathcal W)= I(X:B).
\end{equation}
When fixing $X$ to be a binary random variable distributed according to the uniform probability distribution, this is also the \emph{symmetric capacity} of that channel and furthermore we have
\begin{equation}
I(\mathcal W) = I(X:B) = \log 2 - H(X|B)  = \log 2 - H(\mathcal W).
\end{equation}
A particular class of channels for which uniformly distributed inputs indeed optimize the capacity are \emph{symmetric channels}. We call a binary channel symmetric if the two output states are related by a unitary transformation. 

A different entropic quantity is the quantum relative entropy, which is defined in~\cite{U62} as
\begin{equation}
D(\rho\|\sigma):=\begin{cases} \tr\big[\rho\left(\log\rho-\log\sigma\right)\big] \quad & \supp(\rho)\subseteq\supp(\sigma)\\ \infty & \text{otherwise}\end{cases}.
\end{equation}
This, in the case of classical probability distributions, is the Kullback-Leibler divergence. The relative entropy is of particular importance as it can serve as a parent quantity for many of the entropies defined so far. 
As an intermediate step between the classical and quantum relative entropy we will need the measured relative entropy defined as~\cite{donald86,HP91}
\begin{align}\label{eq:measured}
D_{M}(\rho\|\sigma):=\sup_{(\mathcal{X},M)}D\left(
\sum_{x\in\mathcal{X}}\tr\left[M_x\rho\right]|x\rangle\langle x|\middle\|\sum_{x\in\mathcal{X}}\tr\left[M_x\rho\right]|x\rangle\langle x|\right)\,,
\end{align}
where the optimization is over finite sets $\mathcal{X}$ and POVMs $M$ on $\cH$ with $\tr\left[M_x\rho\right]$ being a measure on $\mathcal{X}$ for any $x \in \mathcal{X}$. Each POVM element $M_{x}\geq 0$ corresponds to an outcome $x$ and  $\sum_{x}M_{x}=\id$. It is known that we can restrict the a priori unbounded supremum to rank-one projective measurements~\cite[Thm.~2]{BFT15}.

Furthermore we will use the fidelity \cite{U73} defined as
\begin{equation}
F(\rho,\sigma) = \| \sqrt{\rho}\sqrt{\sigma} \|_1 = \tr\left[\sqrt{\sqrt{\rho}\,\sigma\sqrt{\rho}}\right],
\end{equation}
and the Chernoff quantity~\cite{ACMBMAV07, NS06} 
\begin{equation}
\phi(s|\rho \| \sigma):=\log\tr{\rho^s\sigma^{1-s}},
\end{equation}
which, for classical probability distributions simply becomes
\begin{equation}
\phi(s|P\| Q ):=\log\sum_k P_k^s Q_k{}^{\!\!1-s}. \label{clCB}
\end{equation}

\part{Asymptotic Hypothesis testing}

\chapter{The Testing Settings}\label{hypoTesting}

Hypothesis testing is arguably one of the most fundamental primitives in quantum information theory. As such it has found many applications, for example, in quantum channel coding for providing an alternative proof of the classical capacity of a quantum channel~\cite{HN03,WR12PRL}, in quantum reading~\cite{Pirandola2011}, quantum illumination~\cite{Lloyd2008,Tan2008,Wilde17} or for giving an operational interpretation to abstract quantities~\cite{BP10,HT14,CHMOSWW16}. 

The underlying problem is to differentiate between the assigned hypotheses. Here, the so called null, $H_0$, and alternative, $H_1$, hypotheses refer respectively to two possible states,~$\rho$, $\sigma$, of a quantum system~$\cH_S$. In~quantum hypothesis testing one is confronted with the task of deciding which  hypothesis holds by performing a measurement on the quantum system $\cH_S$. 

Beginning from this basic problem we ask questions about the optimal behavior of hypothesis testing given certain resources. In the quantum case the most natural settings are either when we have two given states and we look for optimal discrimination, thus we optimize over measurements, or when we have a given measurement and we want to assess its power to discriminate, where we then optimize over quantum states. 

In the next two sections we will discuss each of these settings in more detail.


\section{Optimizing over measurements}\label{discrimination}

A particular hypothesis testing setting is that of quantum state discrimination where quantum states are assigned to each of the hypotheses and we aim to determine which state is actually given, by optimizing over all possible measurements. The distinct scenarios of interest differ in the priority given to different types of error or in how many copies of a system are given to aid the discrimination. 

For the asymptotic regime we investigate the setting where the goal is to discriminate between two $n$-party quantum states $\rho_n$ and $\sigma_n$ living on the $n$-fold tensor product of some finite-dimensional inner product space $\cH^{\otimes n}$. That is, we are optimizing over all two-outcome positive operator valued measures (POVMs) with $\{M_n,(1-M_n)\}$ and associate $M_n$ with accepting $\rho_n$ as well as $\left(1-M_n\right)$ with accepting $\sigma_n$. This naturally gives rise to the two possible errors
\begin{align}\label{eq:errorsI+II}
&\alpha_n(M_n):=\tr\big[\rho_n(1-M_n)\big]\;&\text{Type 1 error,} \\ 
&\beta_n(M_n):= \tr\big[\sigma_nM_n\big]\;&\text{Type 2 error.}
\end{align}
Depending on the problem at hand, one may need to know either~$\beta_n$ for a maximum allowed value of~$\alpha_n$ or the average error probability $p_{\textrm err}= p \alpha_n+(1-p) \beta_n$, where $p$ is some probability giving priors for $H_0$ and $H_1$.
The first option is called asymmetric hypothesis testing and in its most direct setting we minimize the type 2 error as
\begin{align}
\beta_n(\eps):=\inf_{0\leq M_n\leq1}\big\{\beta_n(M_n)\big|\alpha_n(M_n)\leq\eps\big\},
\end{align}
while we require the type 1 error not to exceed a small constant $\eps\in(0,1)$. We are then interested in finding the optimal asymptotic error exponent (whenever the limit exists)
\begin{align}
\zeta(\eps):=\lim_{n\to\infty}-\frac{\log\beta_n(\eps)}{n}
\end{align}
and correspondingly 
\begin{align}
\zeta(0):=\lim_{\eps\to0}\zeta(\eps).
\end{align}
A well studied discrimination setting is that between fixed independent and identical (\iid) states $\rho^{\otimes n}$ and $\sigma^{\otimes n}$ where the error exponent $\zeta_{SL}(\eps)$ is determined by the \textbf{quantum Stein's lemma}~\cite{HP91,ON00,Audenaert2008} in terms of the quantum relative entropy. $\forall\eps\in(0,1)$ we have
\begin{align}\label{eq:brandao_old}
\zeta_{SL}(\eps)= D(\rho\|\sigma).
\end{align} 
Alternatively one might require a more strict bound on the type 1 error, for instance that it decays exponential in $n$ at a rate $r$. Then the goal becomes to minimize 
\begin{align}
\beta_n(r):=\inf_{0\leq M_n\leq1}\big\{\beta_n(M_n)\big|\alpha_n(M_n)\leq e^{-rn}\big\}.
\end{align}
This leads to the \textbf{quantum Hoeffdings bound}~\cite{H07a,OH04,N06} which is given as 
\begin{eqnarray}
\zeta_{\textrm HB}\!=\sup_{0\le s\le 1}\!\!{-s r - \phi( s|\rho\|\sigma)\over 1-s}. 
\label{other}
\end{eqnarray}
The second option is known as  symmetric hypothesis testing  and leads to \textit{minimum error state discrimination}~\cite{H1976}, where  $p_{\textrm err}$ is minimized over all POVMs~$\mathscr M = \{M_n,(1-M_n)\}$, here for $p=\frac{1}{2}$,
\begin{align}
p^*_{\textrm err} &= \inf_{0\leq M_n\leq1}\frac{1}{2} \left( 1 + \tr[M_n( \sigma -  \rho)]\right) \label{p_err} \\
&= \frac{1}{2}\left( 1 + \frac{1}{2}\norm{\sigma -  \rho}_1 \right), 
\end{align}
where, $\norm{\cdot}_1$ is the trace norm. This, in the asymptotic limit, leads to the \textbf{quantum Chernoff bound}~\cite{ACMBMAV07, NS06}:
\begin{equation}
\zeta_{\textrm CB}=-\min_{0\leq s\leq 1}\phi(s | \rho \| \sigma).
\label{Cbound}
\end{equation}
Note that the quantum Chernoff bound holds for any choice of $p$ and is independent of its value. 
It has been conjectured \cite{calsamiglia_local_2010}, with the support of some numerical evidence, that a collective non-separable measurement is required to attain this bound. 

There are several ways of generalizing the problem discussed. In multiple hypothesis testing the problem discussed is generalized to more than two hypothesis, whereas in composite hypothesis testing the knowledge about the states is limited only to membership in a certain set.
Multiple hypothesis testing is of particular interest in the symmetric setting where a breakthrough has been recently achieved by Li \cite{Li16} showing that the optimal error rate when discriminating between $n$ hypothesis $\{ \sigma_i \}_i^n$ is given by the multiple Chernoff bound 
\begin{align}
\zeta_{\textrm mCB} = - \max_{i,j} \min_{0\leq s\leq 1}\phi(s | \rho \| \sigma).
\end{align}
On the other hand, the case of composite hypothesis testing has only been solved in certain special cases. We give a general solution in Chapter \ref{CompHypo}. 

In all settings the errors and rates for classical hypothesis testing can be recovered from those above by simply taking the matrices $\rho$ and  $\sigma$ to be diagonal with entries given by two probability distributions $P=\{P_k\}$, $Q=\{Q_k\}$, associated to~$H_0$ and~$H_1$ respectively.


\section{Optimizing over states}\label{OptOverStates}

In experiments, we often encounter situations where measurement devices (e.g., Stern Gerlach apparatus, heterodyne detectors, photon counters, fluorescence spectrometers)  are a given. A natural question is then to ask about the ability or power of those devices to perform certain quantum information-processing tasks. 
The informational power of a measurement has been addressed in several ways~\cite{OCMB11}, such as via the ``intrinsic data" it provides~\cite{w03} or the capacity of the quantum-classical channel it defines~\cite{OCMB11,DDS11,D15,h12a,BRW14}, or via some associated entropic quantities~\cite{s14,s15,s16,DBO14}. 
 
Here we aim to explore how well a quantum measurement device can discriminate two hypotheses. 
This problem is dual to that of exploring how well two given quantum states can be discriminated~\cite{H1976} and is of practical interest since preparing appropriate probe states is often easier than tailoring optimal measurements for a given state~pair. 


The basic task is again to discriminate two states $\rho$ and $\sigma$ using a given measurement device. We are interested in the scenario where the device can be used $n$ number of times. The given measurement device is the only means of extracting classical data from the quantum system. However, for better performance, one is free to apply any trace preserving quantum operation to the system prior to the measurement.  Likewise, we view data processing also as a free operation. We are then aiming to find what is the minimum 
error probability of discriminating between $\rho$ and $\sigma$, with that fixed measurement. 

We can go a step further and minimize the error probability over all state pairs. This characterizes an intrinsic limitation on the discrimination performance of the measurement device since, in general, a device cannot perfectly discriminate two hypothesis, not even when they are given by orthogonal states. 
This characterization is  of practical relevance since it sets the ultimate limit on the successful identification of two arbitrary states when one is bound to a given type of measurement apparatus.

Let us formalize the problem by first exploring the version where the measurement can be used once. We wish to assess the discrimination power of a device given by a specific POVM, ${\mathscr E}=\{E_k\}_{k=1}^m$. Let us assume that the positive operators $E_k$ (generically non-orthogonal) act on a finite $d$-dimensional Hilbert space, $\cH_d$, of the quantum system $S$. For simplicity and to ease up the notation we will assume a POVM with a finite number of outcomes. However the results hold for any POVM, including those with continuous outcomes.
First, using free operations, we need to produce a valid POVM, ${\mathcal F} = \{F_0, F_1\}$, out of~${\mathscr E}$, to discriminate two states $\rho$ and $\sigma$. This can be achieved~\cite{OCMB11} by grouping (post-processing) the measurement outcomes,~$\{1,2,\dots,m\}$, into two disjoint sets $a$, $\bar a$,  and defining 
\begin{align}
E^{a}:=\sum_{k\in a} E_k\quad\text{and}\quad E^{\bar a}:=\sum_{k\in \bar a} E_k=\id-E^a.
\end{align} 
Then, ${\mathscr F}=\{E_{\mathscr{M}}^a,E_{\mathscr{M}}^{\bar a}\}$, where $E^{a}_{\mathscr{M}}=\mathscr{M}^{\dagger}(E^{a})$ (likewise for $E_{\mathscr{M}}^{\bar a}$), for a suitable trace preserving quantum operation    $\mathscr{M}$  (pre-processing). 
 The error probabilities thus read $\alpha=
\tr(E_{\mathscr{M}}^{\bar a} \rho)$ and~$\beta=
\tr(E_{\mathscr{M}}^{ a} \sigma)$.

For now, we will focus on the symmetric hypothesis testing setting. In this single-shot scenario, we can now quantify the discrimination power of~$\mathscr E$ by the minimum average error probability. It~can be written as
\begin{equation}
p^*_{\textrm err} = \min_a\min_{(\rho,\sigma)}\frac{1}{2} \big\{1 + \tr[E^{a}( \sigma -  \rho)]\big\},
\label{p_err}
\end{equation}
where the minimization is over all partitions $\{a,\bar a\}$ of the outcome set (over all post-processing operations) and over all state pairs $(\rho,\sigma)$, so $\mathscr{M}$ can be dropped in the minimization.
One can readily check~\cite{OCMB11} that the minimum single-shot error probability is given by the spread of $E^{a}$, 
\begin{align}
p^*_{\textrm err} =1/2- \min_a (\lambda^{a}_{\mathrm{max}}-\lambda^{a}_{\mathrm{min}})/2.
\end{align} 
This value is attained when~$\rho$ and $\sigma$ are the eigenstates of~$E^{a}$ corresponding to its maximum and minimum eigenvalue, respectively.  As mentioned above, this problem and quantum hypothesis testing are ``dual", in the sense that state pairs and measurements swap roles.

The single-shot scenario above is too restrictive since one can easily envision discrimination settings where the measurement $\mathscr E$ is performed~$n$ number of times, which might lead to a lower error when using more complicated inputs such as entangled states. In the most general setting, a system consisting of $n$ copies of $S$ is prepared (by, say, Alice) in one of the states of the pair~$(\rho^n,\sigma^n)$, corresponding respectively to hypotheses~$H_0$ and~$H_1$. Here, $\rho^n, \sigma^n\in\setS{\cH_d^{\otimes n}}$ can be fully general, not just of the form $\rho^{\otimes n}$, $\sigma^{\otimes n}$. The measurer's (say, Bob's) goal is to tell which hypothesis is true by  performing~$n$ measurements, all of them given by the POVM~$\mathscr E$. 
Free operations include pre-processing of  $(\rho^n,\sigma^n)$ and post-processing of the classical data gathered after each measurement. 
%
As in Equation~(\ref{p_err}), when minimizing over state pairs, it is enough to choose the discriminating POVM as \mbox{${\mathscr F}=\{E^a,E^{\bar a}=\id-E^a\}$}, where $E^a$ now has the form
\begin{equation}
E^{a}= \sum_{{\textbf k}^n \in a}E_{{\textbf k}^n}:= \sum_{{\textbf k}^n \in a} \bigotimes_{i=1}^n E_{k_i}.
\label{Ea n}
\end{equation}
Here ${\textbf k}^r:=\{k_1,k_2,\dots,k_r\}$ denotes a sequence of outcomes of length $r$ (${\textbf k}^0:=\emptyset$), so ${\textbf k}^n$ is obtained after completing all measurements.  The two disjoint sets~$a$ and~$\bar a$ now contain all the sequences  assigned to the hypotheses~$H_0$ and~$H_1$ respectively.
\mbox{Type 1} and type 2 error probabilities are $\alpha_{n} = \tr(E^{\bar a} \rho^{n})$ and $\beta_{n} = \tr(E^{ a} \sigma^{n})$, respectively, and the error probability for symmetric hypothesis testing can be written as~$p_{\textrm err}=\min_{a}(\alpha_n+\beta_n)/2$. Note that we take the priors $p$ and $(1-p)$ assigned to the hypothesis to be equally $\frac{1}{2}$, as in the state discrimination setting, for simplicity. 

It is not hard to see that the errors fall off exponentially with $n$ \cite{cover_2006,hayashi_quantum_2016}. It is then natural to quantify the discrimination power of $\mathscr E$ by the optimal asymptotic exponential rate of $p_{\textrm err}$, which is defined as
\begin{equation}
\xi_{\textrm CB} =- \min_{(\rho^n,\sigma^n)} \;\lim_{n\rightarrow\infty} \frac{1}{n} \log{p_{\textrm err}} .
\label{xiCB}
\end{equation}

Although $p_{\textrm err}$ can still be written as the spread of the optimal grouping, the number of groupings grows super-exponentially with $n$. Moreover, very little is known about the spectrum of operator sums such as those in Eq.~(\ref{Ea n}) and their eigenvectors (i.e.,~$\rho^n$ and~$\sigma^n$). 

So far we have focused on the problem of symmetric hypothesis testing. Similar to the state discrimination setting we can also look at several asymmetric hypothesis testing scenarios. Here it is done for the settings dual to Stein's and Hoeffding's bound. 
The corresponding asymptotic rates are defined as
\begin{eqnarray}
\xi_{\textrm SL}\!&=&-\!\!\lim_{n\to\infty}\min_{(\rho^n,\sigma^n)}{\log \beta_n\over n}\ \mbox{\textrm subject to $\alpha_n\le\epsilon$},\label{xiSL}\\
\xi_{\textrm HB}\!&=&-\!\!\lim_{n\to\infty}\min_{(\rho^n,\sigma^n)}{\log \beta_n\over n}\ \mbox{\textrm subject to $\alpha_n\le{\textrm e}^{-nr}$}. \label{xiHB}
\end{eqnarray}

Simple expressions for the asymptotic rates in all of these settings, symmetric and asymmetric, have previously been unknown and we will provide them in Chapter~\ref{discPower}.

\chapter{Composite Hypothesis testing}\label{CompHypo}

In this chapter we will come back to the hypothesis testing scenario of state discrimination. In Chapter~\ref{hypoTesting} we introduced the basic setting. Nevertheless in many applications, we aim to solve more general discrimination problems. A prominent example of these are composite hypotheses -- here we attempt to discriminate between different sets of states. In this case a particularly interesting setting is that of asymmetric hypothesis testing leading to the quantum Stein's bounds. Here, the case of composite \iid null hypotheses $\rho^{\otimes n}$ with $\rho\in\cS$ and fixed alternative hypothesis was previously investigated in~\cite{hayashi2002,BDKSSS05} leading to the natural error exponent $\forall\eps\in(0,1)$
\begin{align}
\zeta_{\cS,\sigma}(\eps)=\inf_{\rho\in \cS} D(\rho\|\sigma)\,.
\end{align}
On the other hand the problem of composite alternative hypotheses seems to be more involved, unless the states in the alternative set commute. In case the set of alternative hypotheses $\sigma_n\in\cT_n$ for $n\in\NN$ fulfills certain axioms motivated by the framework of resource theories, it was shown in~\cite{BP10} that the error exponent $\zeta_{\rho,\cT}(\eps)$ can be written in terms of the regularized relative entropy distance.
\begin{thm}[Theorem 1 in~\cite{BP10}]\label{CompBrandao}
For any family of sets $\{\cT_n\}_{n\in\NN}$, with $\cT_n \subseteq \setS{\cH^{\otimes n}}$ satisfying the following conditions:
\begin{enumerate}
\item Each $\cT_n$ is convex and closed.
\item Each $\cT_n$ contains $\sigma^{\otimes n}$, for a full rank state $\sigma\in\setS{\cH}$.
\item If $\rho\in\cT_{n+1}$, then $\tr_k(\rho)\in\cT_n$, for every $k\in\{ 1.\dots, n+1\}$. 
\item If $\rho\in\cT_{n}$ and $\nu\in\cT_{m}$, then $\rho\otimes\nu\in\cT_{n+m}$.
\item If $\rho\in\cT_{n}$, then $P_\pi\rho P_\pi\in\cT_{n}$ for every $\pi\in S_n$.
\end{enumerate}
we get that $\forall\eps\in(0,1)$, the following error rate is achievable for asymmetric hypothesis testing between $\rho^{\otimes n}$ and a sequence of states $\sigma_n\in\cT_n$,
\begin{align}
\zeta_{\rho,\cT}(\eps)=\lim_{n\to\infty}\frac{1}{n}\inf_{\sigma_n\in\cT_n} D\left(\rho^{\otimes n}\|\sigma_n\right)\,.
\end{align}
\end{thm}
This regularization is in general needed as known from the case of the relative entropy of entanglement~\cite{Vollbrecht01}. This might not be surprising since the set of alternative hypotheses is not required to be \iid in general.
In what follows we will often come back to this theorem for comparison with our results.  

For our main result we consider the setting where null and alternative hypotheses are both composite and given by convex combinations of $n$-fold tensor powers of states from given sets $\rho\in\cS$ and $\sigma\in\cT$ (see Sect.~\ref{CompStein} for the precise definition). We show that the corresponding asymptotic error exponent $\zeta_{\cS,\cT}(0)$ can be written as
\begin{align}\label{eq:main_first}
\zeta_{\cS,\cT}(0)=\lim_{n\to\infty}\frac{1}{n}\inf_{\substack{\rho\in\cS\\\mu\in\cT}}D\left(\rho^{\otimes n}\|\int\sigma^{\otimes n}\;\mathrm{d}\mu(\sigma)\right),
\end{align}
where in a slight abuse of notation we use $\mu\in\cT$ meaning normalized measures on the set $\cT$. We note that even in the case of a fixed null hypothesis $\cS=\{\rho\}$, our setting is not a special case of the previous results~\cite{BP10}, as our sets of alternative hypotheses are not closed under tensor product -- $\sigma_m\in\cT_m,\;\sigma_n\in\cT_n\nRightarrow\sigma_m\otimes\sigma_n\in\cT_{mn}$ -- which is one of the properties required for the result of~\cite{BP10}. Moreover, we show that the regularization in Equation~\eqref{eq:main_first} is needed, i.e.~in contrast to the classical case~\cite{LM02,BHLP14} in general
\begin{align}\label{eq:reg_needed}
\zeta_{\cS,\cT}(0)\neq\inf_{\substack{\rho\in\cS\\\sigma\in\cT}}D(\rho\|\sigma)\,.
\end{align}
Nevertheless, there exist non-commutative cases in which the regularization is not needed and we discuss several such examples. In particular, we give a novel operational interpretation of the relative entropy of coherence in terms of hypothesis testing. The proofs of our results are transparent in the sense that we start from the composite Stein's lemma for classical probability distributions and then lift the result to the non-commutative setting by only using elementary properties of entropic measures.

\section{A composite quantum Stein's Lemma}\label{CompStein}

For $n\in\mathbb{N}$ we attempt the following discrimination problem.
\hyptest{\begin{description}
\item[Null hypothesis:] the convex sets of iid states  \\ $\cS_n:=\left\{\int\rho^{\otimes n}\;\mathrm{d}\nu(\rho)\middle|\rho\in\cS\right\}$ with $\cS\subseteq S(\cH)$
\item[Alternative hypothesis:] the convex sets of iid states \\ $\cT_n:=\left\{\int\sigma^{\otimes n}\;\mathrm{d}\mu(\sigma)\middle|\sigma\in\cT\right\}$ with $\cT\subseteq S(\cH)$
\end{description}}{}

For $\eps\in(0,1)$ the goal is the quantification of the optimal asymptotic error exponent for composite asymmetric hypothesis testing. As we will see, the following limits exist
\begin{align}
&\zeta_{\cS,\cT}^n(\eps):=-\frac{1}{n}\log\inf_{0\leq M_n\leq1}\left\{\sup_{\mu\in\cT}\tr\left[M_n\sigma_n(\mu)\right]\middle|\sup_{\nu\in\cS}\tr\left[(1-M_n)\rho_n(\nu)\right]\leq\eps\right\}\\
&\zeta_{\cS,\cT}(\eps):=\lim_{n\to\infty}\zeta_{\cS,\cT}^n(\eps)\quad\mathrm{and}\quad\zeta_{\cS,\cT}(0):=\lim_{\eps\to0}\zeta_{\cS,\cT}(\eps)\,,
\end{align}
where we set
\begin{align}
\text{$\rho_n(\nu):=\int\rho^{\otimes n}\mathrm{d}\nu(\rho)$ and $\sigma_n(\mu):=\int\sigma^{\otimes n}\mathrm{d}\mu(\sigma)$}
\end{align}
for the sake of notational simplicity, and $\mu\in\cS$ and $\nu\in\cT$ stand for measures over $\cS$ and $\cT$, respectively. The following is the main result of this section.

\begin{thm}\label{thm:main}
For the discrimination problem as above, we have
\begin{align}\label{eq:main}
\zeta_{\cS,\cT}(0)=\lim_{n\to\infty}\frac{1}{n}\inf_{\substack{\rho\in\cS\\\mu\in\cT}}D\left(\rho^{\otimes n}\|\sigma_n(\mu)\right)\,.
\end{align}
\end{thm}


In principle Theorem~\ref{thm:main} leaves the possibility that Equation~\ref{eq:main} might actually be identical to the relative entropy optimized on a single copy of the quantum system open. This would simplify the error rate above significantly. Unfortunately, we can show that this is indeed not possible. The proof is based on the close relation between hypothesis testing and recoverability entropy inequalities. Therefore we postpone giving the details to the second part of the thesis where we will discuss these topics. The proof can be found in Section~\ref{RegNeeded}. 

\begin{rem}\label{remSup}
Before we start with the proof, let us take a look at the case where above equation diverges. Consider the case where $\supp(\rho)\nsubseteq\supp(\sigma)$ for all $\rho\in\cS,\sigma\in\cT$. In this setting one can always find an appropriate measurement such that the two sets can be distinguished perfectly in a finite number of steps. Therefore when the left hand side diverges, this agrees with the right hand side which diverges as well due to the definition of the relative entropy. Hence, throughout the following argument, we assume that there exist $\rho\in\cS,\sigma\in\cT$ such that $\supp(\rho)\subseteq\supp(\sigma)$.
\end{rem}

We first prove the $\leq$ bound, i.e.~the converse direction, which follows from the following proposition.

\begin{prop}\label{lem:converse_lemma}
For $\rho\in\cS$, $\mu\in\cT$, and $\eps\in(0,1)$ we have
\begin{align}
-\frac{1}{n}\log\inf_{0\leq M_n\leq1}&\Big\{\tr\left[M_n\sigma_n(\mu)\right]\Big|\tr\left[(1-M_n)\rho^{\otimes n}\right]\leq\eps\Big\} \nonumber\\
&\leq\frac{1}{n}\cdot\frac{D\left(\rho^{\otimes n}\middle\|\sigma_n(\mu)\right)+\log2}{1-\eps}\,.
\end{align}
\end{prop}


\begin{proof}
We follow the original converse proof of the quantum Stein's lemma~\cite{HP91} for the states $\rho^{\otimes n}$ and $\sigma_n(\mu)$. By the monotonicity of the quantum relative entropy~\cite{Lindblad1975} under POVMs $\{M_n,(1-M_n)\}$ we have
\begin{align}
D\left(\rho^{\otimes n}\middle\|\sigma_n(\mu)\right) &\geq\alpha_n(M_n)\log\frac{\alpha_n(M_n)}{1-\beta_n(M_n)}+(1-\alpha_n(M_n))\log\frac{1-\alpha_n(M_n)}{\beta_n(M_n)} \nonumber\\ 
&\geq-\log2-(1-\alpha_n(M_n))\log\beta_n(M_n)\,, \label{rearrang}
\end{align}
where we used the notation from Equation~\eqref{eq:errorsI+II}. The claim then follows by a simple rearrangement of Equation \ref{rearrang}.
\end{proof}

By taking the appropriate infima as well as the limits $n\to\infty$ and $\eps\to0$ in Prop.~\ref{lem:converse_lemma} we find
\begin{align}\label{eq:converse}
\zeta_{\cS,\cT}(0)\leq\liminf_{n\to\infty}\frac{1}{n}\inf_{\substack{\rho\in\cS\\\mu\in\cT}}D\left(\rho^{\otimes n}\|\sigma_n(\mu)\right)\,.
\end{align}
For the $\geq$ bound, i.e.~the achievability direction, we show the following statement.

\begin{prop}\label{lem:achievability}
For the discrimination problem as above with $n\in\mathbb{N}$ and $\eps\in(0,1)$, we have
\begin{align}
\zeta_{\cS,\cT}^n(\eps)\geq\frac{1}{n}\inf_{\substack{\rho\in\cS\\\mu\in\cT}}D\left(\rho^{\otimes n}\|\sigma_n(\mu)\right)-\frac{\log\poly(n)}{n}\,,
\end{align}
where $\poly(n)$ stands for terms of order at most polynomial in $n$.
\end{prop}

The basic idea for the proof of Prop.~\ref{lem:achievability} is to start from the corresponding composite Stein's lemma for classical probability distributions and lift the result to the non-commutative setting by solely using properties of quantum entropy. We now prove Prop.~\ref{lem:achievability} in several steps and start with an achievability bound in terms of the measured relative entropy.

\begin{lemma}\label{MRelEnt}
For the discrimination problem as above with $n\in\mathbb{N}$ and $\eps\in(0,1)$, we have
\begin{align}
\zeta_{\cS,\cT}^n(\eps)\geq\frac{1}{n}\inf_{\substack{\nu\in\cS\\\mu\in\cT}}\DM\left(\rho_n(\nu)\middle\|\sigma_n(\mu)\right)\,.
\end{align}
\end{lemma}

\begin{proof}
Analogous to Remark \ref{remSup} it it sufficient to consider the case where there exist $\rho\in\cS,\sigma\in\cT$ such that $\supp(\rho)\subseteq\supp(\sigma)$, otherwise both sides of above equation become infinite by definition. 

For sets of classical probability distributions $P\in\cS$ and $Q\in\cT$ we know from the corresponding commutative result~\cite{LM02,BHLP14} that for $\eps\in(0,1)$
\begin{align}\label{eq:classical}
\zeta_{\cS,\cT}(\eps)=\inf_{\substack{P\in\cS\\Q\in\cT}}D(P\|Q)\,.
\end{align}
Now, the strategy is to first measure the quantum states and then invoke the classical achievability result in Equation~\eqref{eq:classical} for the resulting probability distributions. For that we fix $n\in\mathbb{N}$ and a POVM $\mathcal{M}_n$ on $\cH^{\otimes n}$. For testing the probability distributions $P_n:=\mathcal{M}_n\left(\rho^{\otimes n}\right)$ vs.~$Q_n:=\mathcal{M}_n\left(\sigma^{\otimes n}\right)$ we get an achievability bound
\begin{align}
\zeta_{\cS,\cT}^n(\eps) &\geq\frac{1}{n}\inf_{\substack{\rho\in\cS\\\sigma\in\cT}}D\left(\mathcal{M}_n\left(\rho^{\otimes n}\right)\middle\|\mathcal{M}_n\left(\sigma^{\otimes n}\right)\right) \\ &\geq\frac{1}{n}\inf_{\substack{\nu\in\cS\\\mu\in\cT}}D\left(\mathcal{M}_n\left(\rho_n(\nu)\right)\middle\|\mathcal{M}_n\left(\sigma_n(\mu)\right)\right)\,,
\end{align}
where the second inequality follows since the infimum is taken over a larger set. The claim then follows from applying a minimax theorem for the measured relative entropy (see Appendix, Lemma~\ref{applySion})
\begin{align}
\sup_{\mathcal{M}_n}\inf_{\substack{\nu\in\cS\\\mu\in\cT}}D\left(\mathcal{M}_n\left(\rho_n(\nu)\right)\middle\|\mathcal{M}_n\left(\sigma_n(\mu)\right)\right)=\inf_{\substack{\nu\in\cS\\\mu\in\cT}}\DM\left(\rho_n(\nu)\middle\|\sigma_n(\mu)\right)\,.
\end{align}
\end{proof}

Next, we argue that the measured relative entropy can in fact be replaced by the quantum relative entropy by only paying an asymptotically vanishing penalty term. For this we need the following lemma which can be seen as a generalization of the original technical argument in the proof of quantum Stein's lemma~\cite{HP91}.

\begin{lemma}\label{pinching}
Let $\rho_n,\sigma_n\in S\left(\cH^{\otimes n}\right)$ with $\sigma_n$ permutation invariant. Then, we have
\begin{align}
D\left(\rho_n\middle\|\sigma_n\right)-\log\poly(n)\leq \DM\left(\rho_n\middle\|\sigma_n\right)\leq D\left(\rho_n\middle\|\sigma_n\right)\,.
\end{align}
\end{lemma}

\begin{proof}
Again we can restrict ourselves to the case where there exist $\rho\in\cS,\sigma\in\cT$ such that $\supp(\rho)\subseteq\supp(\sigma)$, since, otherwise all relative entropy terms evaluate to infinity by definition. 

The second inequality follows directly from the definition of the measured relative entropy in Equation~\eqref{eq:measured} together with the fact that the quantum relative entropy is monotone under completely positive trace preserving maps~\cite{Lindblad1975}. We now prove the first inequality with the help of asymptotic spectral pinching~\cite{hayashi2002}. The pinching map with respect to $\omega\in S(\cH)$ is defined as
\begin{align*}
\mathcal{P}_\omega(\cdot):=\sum_{\lambda\in\mathrm{spec}(\omega)} P_\lambda (\cdot)P_\lambda\;\text{with the spectral decomposition $\omega=\sum_{\lambda\in\mathrm{spec}(\omega)}\lambda P_\lambda$.}
\end{align*}
Crucially, we have the pinching operator inequality $\mathcal{P}_\omega[X]\geq\frac{X}{|\mathrm{spec}(\omega)|}$~\cite{hayashi2002}. From this we can deduce that (see, e.g., \cite[Lemma~4.4]{tomamichel2015quantum})
\begin{align}
D\left(\rho_n\middle\|\sigma_n\right)-\log\left|\mathrm{spec}\left(\sigma_n\right)\right|\leq D\left(\mathcal{P}_{\sigma_n}\left(\rho_n\right)\middle\|\sigma_n\right)=\DM\left(\rho_n\middle\|\sigma_n\right)\,,
\end{align}
where the equality follows since $\mathcal{P}_{\sigma_n}\left(\rho_n\right)$ and $\sigma_n$ are diagonal in the same basis. It remains to show that $\left|\mathrm{spec}\left(\sigma_n\right)\right|\leq\poly(n)$. However, since $\sigma_n$ is permutation invariant, the Schur-Weyl duality shows (see, e.g., \cite[Sect.~5]{harrow_phd}) that in the Schur basis
\begin{align}
\sigma_n=\bigoplus_{\lambda\in\Lambda_n}\sigma_{Q_\lambda}\otimes1_{P_\lambda}\, ,
\end{align}
with $|\Lambda_n|\leq\poly(n)$ and $\text{dim} \left[\sigma_{Q_\lambda}\right]\leq\poly(n)$.
This implies the claim of the lemma.
\end{proof}

By combining Lemma~\ref{MRelEnt} together with Lemma~\ref{pinching} we immediately find that
\begin{align}\label{eq:intermediate}
\zeta_{\cS,\cT}^n(\eps)\geq\frac{1}{n}\inf_{\substack{\rho\in\cS\\\mu\in\cT}}D\left(\rho_n(\nu)\|\sigma_n(\mu)\right)-\frac{\log\poly(n)}{n}\,.
\end{align}
Hence, it remains to argue that the infimum over states $\rho_n(\nu)$ can, without loss of generality, be restricted to \iid states $\rho^{\otimes n}$ with $\rho\in\cS$.

\begin{lemma}\label{MinOut}
For the same definitions as before and some $\omega_n\in S\left(\cH^{\otimes n}\right)$, we have
\begin{align}
\frac{1}{n}\inf_{\nu\in\cS}D\left(\rho_n(\nu)\middle\|\omega_n\right)\geq\frac{1}{n}\inf_{\rho\in\cS}D\left(\rho^{\otimes n}\middle\|\omega_n\right)-\frac{\log\poly(n)}{n}\,.
\end{align}
\end{lemma}

\begin{proof}
We observe the following chain of arguments for $\nu\in\cS$
\begin{align}
\frac{1}{n} D\left(\rho_n(\nu)\middle\|\omega_n\right)&= \frac{1}{n} D\left(\sum_{i=1}^Np_i\rho_i^{\otimes n}\middle\|\omega_n\right)\notag\\
&=-\frac{1}{n}H\left(\sum_{i=1}^Np_i\rho_i^{\otimes n}\right)-\frac{1}{n}\sum_{i=1}^Np_i\tr\left[\rho_i^{\otimes n}\log{\omega_n}\right]\notag\\
&\geq-\frac{1}{n}\sum_{i=1}^N p_i H\left(\rho_i^{\otimes n}\right)-\frac{\log{\poly{(n)}}}{n}-\frac{1}{n}\sum_{i=1}^Np_i\tr\left[\rho_i^{\otimes n}\log{\omega_n}\right]\notag\\
&\geq\min_{\rho_i}\frac{1}{n}D\left(\rho^{\otimes n}_i\middle\|\omega_n\right)-\frac{\log{\poly{(n)}}}{n}\notag\\
&\geq\inf_{\rho\in\cS}\frac{1}{n}D\left(\rho^{\otimes n}\middle\|\omega_n\right)-\frac{\log{\poly{(n)}}}{n}\,,
\end{align}
where the first equality holds by an application of Carath\'eodory's theorem with $N\leq\poly(n)$ (Appendix, Lemma~\ref{carat}), and the first inequality by a quasi-convexity property of the von Neumann entropy (Appendix, Lemma~\ref{caratentropy}). (All other steps are elementary.) Since the above argument holds for all $\nu\in\cS$ the claim follows.
\end{proof}

Combining Lemma~\ref{MinOut} with Equation~\eqref{eq:intermediate} leads to Proposition~\ref{lem:achievability} and then taking the limits $n\to\infty$ and $\eps\to0$ we find
\begin{align}
\zeta_{\cS,\cT}(0)\geq\limsup_{n\to\infty}\frac{1}{n}\inf_{\substack{\rho\in\cS\\\mu\in\cT}}D\left(\rho^{\otimes n}\|\sigma_n(\mu)\right)\,.
\end{align}
Together with the converse from Equation~\eqref{eq:converse} we get
\begin{align*}
\liminf_{n\to\infty}\frac{1}{n}\inf_{\substack{\rho\in\cS\\\mu\in\cT}}D\left(\rho^{\otimes n}\|\sigma_n(\mu)\right)\geq\zeta_{\cS,\cT}(0)\geq\limsup_{n\to\infty}\frac{1}{n}\inf_{\substack{\rho\in\cS\\\mu\in\cT}}D\left(\rho^{\otimes n}\|\sigma_n(\mu)\right)\; 
\end{align*}
\begin{align}
\Rightarrow\zeta_{\cS,\cT}(0)=\lim_{n\to\infty}\frac{1}{n}\inf_{\substack{\rho\in\cS\\\mu\in\cT}}D\left(\rho^{\otimes n}\|\sigma_n(\mu)\right)\,,
\end{align}
which finishes the proof of Thm.~\ref{thm:main}. \qed


\section{Examples and Extensions}\label{sec:examples}

Here we discuss several concrete examples of composite discrimination problems -- those we present here all have a single-letter solution.
First we give discrimination problems that have the relative entropy of coherence as optimal error rate in the composite Stein's setting. Then we go to problems where the mutual information turns out to be optimal. Here, aside from the case based on our composite Stein's Lemma in the last section, we also give a setting which goes beyond this setting but still yields the mutual information as an optimal rate. 
Later in Chapter \ref{recoverability} we will also discuss one scenario where the rate cannot be written as a single-letter formula, giving an operational interpretation to the regularized relative entropy of recovery. 


\subsection{Relative entropy of coherence}\label{sec:coherence}

Following the literature around~\cite{BCP14}, the set of states diagonal in a fixed basis $\{|c\rangle\}$ is called incoherent and is denoted by $\cC\subseteq \setS{\cH}$. For clarity, we sometimes denote by $\cC_n$ the set of incoherent $n$-party states. The relative entropy of coherence of $\rho\in \setS{\cH}$ is defined as
\begin{align}
D_{\cC}(\rho):=\inf_{\sigma\in\cC} D(\rho\|\sigma)\,.
\end{align}
Using the result from Sect.~\ref{CompStein} we can characterize the following discrimination problem.
\hyptest{\begin{description}
\item[Null hypothesis:] the fixed state $\rho^{\otimes n}$
\item[Alternative hypothesis:] the convex sets of \iid incoherent states \\ $\bar{\cC}_n:=\left\{\int\sigma^{\otimes n}\;\mathrm{d}\mu(\sigma)\middle|\sigma\in\cC\right\}$
\end{description}}{}

Namely, as a special case of Theorem~\ref{thm:main} we immediately find
\begin{align}\label{eq:coherence}
\zeta_{\rho,\bar{\cC}_n}(0)=\lim_{n\to\infty}\frac{1}{n}\inf_{\mu\in\cC}D\left(\rho^{\otimes n}\|\int\sigma^{\otimes n}\;\mathrm{d}\mu(\sigma)\right)=D_{\cC}(\rho)\,,
\end{align}
where the last equality follows from Lemma~\ref{relEntropyGasym} in the Appendix. In fact there is even a single-letter solution for the following less restricted discrimination problem.
\hyptest{\begin{description}
\item[Null hypothesis:] the fixed state $\rho^{\otimes n}$
\item[Alternative hypothesis:] the convex set of incoherent states $\sigma_n\in\cC_n$
\end{description}}{}

It is straightforward to check that this hypothesis testing problem fits the general framework of~\cite{BP10} and therefore Theorem~\ref{CompBrandao} leads to
\begin{align}\label{eq:coherence_brandao}
\zeta_{\rho,\cC}(0)=\lim_{n\to\infty}\frac{1}{n}\inf_{\sigma_n\in\cC_n}D\left(\rho^{\otimes n}\|\sigma_n\right)=D_{\cC}(\rho)\,,
\end{align}
where the last step again follows from Lemma~\ref{relEntropyGasym}. 
We have therefore two a priori different hypothesis testing scenarios which generally would lead to two different error rates. It is only due to the fact that both satisfy the conditions of Lemma~\ref{relEntropyGasym}, that they turn out to be equal and therefore both give an operational interpretation to the relative entropy of coherence. We remark that our results also easily extend to the relative entropy of frameness~\cite{GMS09}.

In the following we give a simple self-contained proof of Equation~\eqref{eq:coherence_brandao} that is different from the proof in~\cite{BP10} and follows ideas from~\cite{Audenaert2008,HT14,TH15}. The goal is the quantification of the optimal asymptotic error exponent (as we will see the following limit exists)
\begin{align}
\zeta_{\rho,\cC}^n(\eps):=-\frac{1}{n}\log\inf_{\substack{0\leq M_n\leq1\\ \tr\left[M_n\rho^{\otimes n}\right]\geq1-\eps}}\sup_{\sigma_n\in\cC_n}\tr\left[M_n\sigma_n\right]\quad \\ 
\mathrm{with}\quad\zeta_{\rho,\cC}(\eps):=\lim_{n\to\infty}\zeta_{\rho,\cC}^n(\eps)\quad\mathrm{and}\quad\zeta_{\rho,\cC}(0):=\lim_{\eps\to0}\zeta_{\rho,\cC}(\eps)\,.
\end{align}

\begin{prop}\label{thm:convex_achievability}
For the discrimination problem as above we have 
\begin{align} 
\zeta_{\rho,\cC}(0)=D_{\cC}(\rho).
\end{align}
\end{prop}

The converse direction $\leq$ follows exactly as in Lemma~\ref{lem:converse_lemma}, together with Lemma~\ref{relEntropyGasym} to make the expression single-letter. For the achievability direction $\geq$ we make use of a general family of quantum R\'enyi entropies: the Petz divergences~\cite{petz_statbook}. For $\rho,\sigma\in \setS{\cH}$ and $s\in(0,1)\cup(1,\infty)$ they are defined as
\begin{align}
D_s\left(\rho\middle\|\sigma\right):=\frac{1}{s-1}\log\tr\left[\rho^s\sigma^{1-s}\right]\,,
\end{align}
whenever either $s<1$ and $\rho$ is not orthogonal to $\sigma$ in the Hilbert-Schmidt inner product or $s>1$ and the support of $\rho$ is contained in the support of $\sigma$. (Otherwise we set $D_s(\rho\|\sigma):=\infty$.) The corresponding R\'enyi relative entropies of coherence are given by~\cite{Chitambar16}
\begin{align}\label{eq:coherence_petz}
D_{s,\cC}(\rho):=\inf_{\sigma\in\cC}D_s(\rho\|\sigma)
\end{align} 
with the additivity property 
\begin{align}\label{eq:coherence_additivity}
D_{s,\cC}\left(\rho^{\otimes n}\right)=n D_{s,\cC}(\rho)\,.
\end{align} 
Proposition~\ref{thm:convex_achievability} follows by taking the limits $n\to\infty$, $s\to1$, and $\eps\to0$ in the following lemma. (This is independent of what the support of $\rho$ is since the set $\cC$ includes full rank states.)

\begin{lemma}\label{thm:general_convex}
For the discrimination problems as above with $n\in\mathbb{N}$ and $\eps\in(0,1)$ we have for $s\in(0,1)$ that
\begin{align}\label{eqn:general_convex}
\zeta_{\cC}^n(\eps)\geq D_{s,\cC}(\rho)-\frac{1}{n} \frac{s}{1-s}\log\frac{1}{\eps}\,.
\end{align}
\end{lemma}

\begin{proof}
It is straightforward to check with Sion's minimax theorem (Lemma~\ref{Sion}) that
\begin{align}\label{eq:sion_applied}
\inf_{\substack{0\leq M_n\leq1\\ \tr\left[M_n\rho^{\otimes n}\right]\geq1-\eps}}\sup_{\sigma_n\in\cC_n}\tr\left[M_n\sigma_n\right]=\sup_{\sigma_n\in\cC_n}\inf_{\substack{0\leq M_n\leq1\\ \tr\left[M_n\rho^{\otimes n}\right]\geq1-\eps}}\tr\left[M_n\sigma_n\right]\,.
\end{align}
Now, for $\lambda_n\in\mathbb{R}$ with $n\in\mathbb{N}$ we choose $M_n(\lambda_n):=\left\{\rho^{\otimes n}-2^{\lambda_n}\sigma_n\right\}_+$ where $\{\cdot\}_+$ denotes the projector on the eigenspace of the positive spectrum. We have $0\leq M_n(\lambda_n)\leq 1$ and by Audenaert's inequality, Lemma~\ref{lem:audenaert} in the Appendix, with $s\in(0,1)$ we get
\begin{align}\label{eq:audenaert_one}
\tr\left[(1-M_n(\lambda_n))\rho^{\otimes n}\right]\leq2^{(1-s)\lambda_n}\tr\left[\left(\rho^{\otimes n}\right)^s\sigma_n^{1-s}\right]=2^{(1-s)\left(\lambda_n-D_s\left(\rho^{\otimes n}\middle\|\sigma_n\right)\right)}\,.
\end{align}
Moreover, again Audenaert's inequality for $s\in(0,1)$ implies that
\begin{align}\label{eq:audenaert_two}
\tr\left[M_n(\lambda_n)\sigma_n\right]\leq2^{-s\lambda_n}\tr\left[\left(\rho^{\otimes n}\right)^s\sigma_n^{1-s}\right]=2^{-s\lambda_n-(1-s)D_s\left(\rho^{\otimes n}\middle\|\sigma_n\right)}\,.
\end{align}
Hence, choosing
\begin{align}
\text{$\lambda_n:=D_s\left(\rho^{\otimes n}\middle\|\sigma_n\right)+\log\eps^{\frac{1}{1-s}}$ with $M_n:=M_n(\lambda_n)$},
\end{align}
together with Equation~\eqref{eq:audenaert_one}, leads to $\tr\left[M_n\rho^{\otimes n}\right]\geq1-\eps$. Finally, Equation~\eqref{eq:sion_applied} together with Equation~\eqref{eq:audenaert_two} and the additivity property from Equation~\eqref{eq:coherence_additivity} leads to the claim of the lemma.
\end{proof}

A more refined analysis along the lines of above calculation also allows to determine the corresponding Hoeffding bound as well as the strong converse exponent (cf.~\cite{Audenaert2008,HT14}). The former gives an operational interpretation to the R\'enyi relative entropy of coherence $D_{s,\cC}(\rho)$, whereas the latter gives an operational interpretation to the sandwiched R\'enyi relative entropies of coherence~\cite{Chitambar16}
\begin{align}
\tilde{D}_{s,\cC}(\rho):=\inf_{\sigma\in\cC}\tilde{D}_s(\rho\|\sigma)\end{align} 
with the sandwiched R\'enyi relative entropies 
\begin{align}
\tilde{D}_s(\rho\|\sigma):=\frac{1}{s-1}\log\tr\left[\left(\sigma^{\frac{1-s}{2s}}\rho\sigma^{\frac{1-s}{2s}}\right)^s\right]
\end{align}
whenever either $s<1$ and $\rho$ is not orthogonal to $\sigma$ in Hilbert-Schmidt inner product or $s>1$ and the support of $\rho$ is contained in the support of $\sigma$~\cite{muller2013quantum,WWY14}. (Otherwise we set $D_s(\rho\|\sigma):=\infty$.) The crucial insight for the proof is again the additivity property $\tilde{D}_{s,\cC}\left(\rho^{\otimes n}\right)=n\cdot\tilde{D}_{s,\cC}(\rho)$ that was already shown in~\cite{Chitambar16}.


\subsection{Quantum mutual information}\label{ssec:QMI}

We will now discuss some discrimination problems that lead to an optimal error rate given by the quantum mutual information.
Using our main result from Section~\ref{CompStein} we find a solution to the following discrimination problem.
\hyptest{
\begin{description}
\item[Null hypothesis:] the state $\rho_{AB}^{\otimes n}$
\item[Alternative hypothesis:] the convex set of iid states \\ 
$\bar{\cT}_{A^n:B^n}:=\left\{\rho_A^{\otimes n}\otimes\int\sigma_B^{\otimes n}\;\mathrm{d}\mu(\sigma)\middle|\sigma_B\in S(\cH_B)\right\}$
\end{description}}{}

Namely, we have
\begin{align}
\bar{\zeta}_{\rho,\bar{\cT}}(0)=\lim_{n\to\infty}\frac{1}{n}\inf_{\mu\in\bar{\cT}}D\left(\rho^{\otimes n}_{AB}\middle\|\rho_A^{\otimes n}\otimes\int\sigma_B^{\otimes n}\;\mathrm{d}\mu(\sigma)\right)=I(A:B)_\rho\,.
\end{align}
Here the last equality follows from the easily checked identity
\begin{align}\label{eq:mutualinfo_identity}
I(A:B)_\rho=\inf_{\sigma\in S(\cH)}D(\rho_{AB}\|\sigma_A\otimes\sigma_B)\,.
\end{align}
More general composite discrimination problems leading to the quantum mutual information were solved in~\cite{HT14}. From Equation \ref{eq:mutualinfo_identity} one might expect that also the general problem of discriminating $\rho_{AB}^{\otimes n}$ against arbitrary $\sigma_{A^n}\otimes\sigma_{B^n}$ leads to the quantum mutual information; it is not known whether this holds indeed. We approach this question by further extending previous results to the following problem (cf.~the classical work~\cite{TH15}).
\hyptest{\begin{description}
\item[Null hypothesis:] the state $\rho_{AB}^{\otimes n}$
\item[Alternative hypothesis:] the set of states \\ 
$\cT_{A^n:B^n}:=\left\{\sigma_{A^n}\otimes\sigma_{B^n}\in S\left(\cH_{AB}^{\otimes n}\right)\middle|\sigma_{A^n}\vee\sigma_{B^n}\;\text{permutation inv.}\right\}$.
\end{description}}{}
The goal is again the quantification of the optimal asymptotic error exponent (as we will see the following limit exists)
\begin{align}
&\zeta_{\rho,\cT}^n(\eps):=-\frac{1}{n}\cdot\log\inf_{\substack{0\leq M_n\leq1\\ \tr\left[M_n\rho^{\otimes n}\right]\geq1-\eps}}\sup_{\sigma\otimes\sigma\in\cT_n}\tr\left[M_{A^nB^n}\sigma_{A^n}\otimes\sigma_{B^n}\right]\label{eq:error_mutual}\\
&\mathrm{with}\quad\zeta_{\rho,\cT}(\eps):=\lim_{n\to\infty}\zeta_{\rho,\cT}^n(\eps)\quad\mathrm{and}\quad\zeta_{\rho,\cT}(0):=\lim_{\eps\to0}\zeta_{\rho,\cT}(\eps)\,.
\end{align}
Note that the sets $\cT_{A^nB^B}$ are not convex and hence the minimax technique used in Section~\ref{sec:coherence} does not work here. However, following the ideas in~\cite{HT14,TH15} we can exploit the permutation invariance and use de Finetti reductions of the form~\cite{Hayashi2009,christandl09} to find the following.

\begin{prop}\label{prop:quantum_mutual}
For the discrimination problem as above we have 
\begin{align}
\zeta_{\rho,\cT}(0)=I(A:B)_\rho.
\end{align}
\end{prop}

The converse direction $\leq$ follows exactly as in Lemma~\ref{lem:converse_lemma}, together with Equation~\eqref{eq:mutualinfo_identity} to make the expression single-letter. The achievability direction $\geq$ follows from the following lemma by taking the limits $n\to\infty$, $s\to1$, $\eps\to0$ and then applying Equation~\eqref{eq:mutualinfo_identity}.

\begin{lemma}
For the discrimination problem as above with $n\in\mathbb{N}$ and $\eps\in(0,1)$ we have for $s\in(0,1)$ that
\begin{align}
\zeta_{\rho,\cT}^n(\eps)\geq\inf_{\sigma\in S(\cH)}D_s\left(\rho_{AB}\middle\|\sigma_A\otimes\sigma_B\right)-\frac{1}{n}\cdot\frac{s}{1-s}\log\frac{1}{\eps}-\frac{\log\poly(n)}{n}\,.
\end{align}
\end{lemma}

\begin{proof}
We choose
\begin{align}
M_{A^nB^n}(\lambda_n):=\left\{\rho_{AB}^{\otimes n}-2^{\lambda_n}\omega_{A^n}\otimes\omega_{B^n}\right\}_+\quad \nonumber\\ 
\mathrm{with}\quad\omega_{A^n}:={n+|A|^2-1\choose n}^{-1}\cdot\tr_{\tilde{A}^n}\left[P^{\mathrm{Sym}}_{A^n\tilde{A}^n}\right]\,,
\end{align}
where $P^{\mathrm{Sym}}_{A^n\tilde{A}^n}$ denotes the projector onto the symmetric subspace of $\cH_A^{\otimes n}\otimes\cH_{\tilde{A}}^{\otimes n}$ with $|A|=|\tilde{A}|$ (denoting the dimension of $\cH_A$ by $|A|$), and similarly for $B^n$. Since $\omega_{A^n}\otimes\omega_{B^n}$ is permutation invariant we get together with Audenaert's inequality (Appendix Lemma~\ref{lem:audenaert}) that
\begin{align}
\tr\left[(1-M_{A^nB^n}(\lambda_n))\rho^{\otimes n}_{AB}\right]&\leq2^{(1-s)\lambda_n}\tr\left[\left(\rho^{\otimes n}_{AB}\right)^s\left(\omega_{A^n}\otimes\omega_{B^n}\right)^{1-s}\right] \nonumber \\ &\leq2^{(1-s)\left(\lambda_n-\inf_{\sigma\otimes\sigma\in\cT_n}D_s\left(\rho^{\otimes n}_{AB}\middle\|\sigma_{A^n}\otimes\sigma_{B^n}\right)\right)}\,.
\end{align}
Let's assume for the reminder of the proof that $\sigma_{A^n}$ is the permutation invariant state (the proof in the other case works identically). 
Now, we have by Schur-Weyl duality that $\sigma_{A^n}\leq{n+|A|^2-1\choose n}\cdot\omega_{A^n}$ for all permutation invariant $\sigma_{A^n}$ (see, e.g., \cite[Lemma~1]{HT14}). The idea is to make also $\sigma_{B^n}$ permutation invariant, by using the fact that the measurement itself is permutation invariant. Then we can again use Audenaert's inequality (Appendix Lemma~\ref{lem:audenaert}) and we find
\begin{align}
&\tr\left[M_{A^nB^n}(\lambda_n)\left(\sigma_{A^n}\otimes\sigma_{B^n}\right)\right] \\ 
&=\tr\left[M_{A^nB^n}(\lambda_n)\left(\sigma_{A^n}\otimes\left(\sum_{\pi\in S_n}U_{B^n}(\pi)\sigma_{B^n}U_{B^n}^\dagger(\pi)\right)\right)\right]\notag\\
&\leq\underbrace{{n+|A|^2-1\choose n}{n+|B|^2-1\choose n}}_{=:\;p(n)\;\leq\;\poly(n)}\cdot\tr\left[M_{A^nB^n}(\lambda_n)\left(\omega_{A^n}\otimes\omega_{B^n}\right)\right]\notag\\
&\leq p(n)\cdot2^{-s\lambda_n}\tr\left[\left(\rho^{\otimes n}_{AB}\right)^s\left(\omega_{A^n}\otimes\omega_{B^n}\right)^{1-s}\right]\notag\\
&\leq p(n)\cdot2^{-s\lambda_n-(1-s)\inf_{\sigma\otimes\sigma\in\cT_n}D_s\left(\rho^{\otimes n}_{AB}\middle\|\sigma_{A^n}\otimes\sigma_{B^n}\right)}\,,\label{eq:mutual_error-last}
\end{align}
where $S_n$ denotes the symmetric group. \\
We now choose
\begin{align}
\text{$\lambda_n:=\inf_{\sigma\otimes\sigma\in\cT_n}D_s\left(\rho^{\otimes n}_{AB}\middle\|\sigma_{A^n}\otimes\sigma_{B^n}\right)+\log\eps^{\frac{1}{1-s}}$ with $M_{A^nB^n}:=M_{A^nB^n}(\lambda_n)$,}
\end{align}
from which we get $\tr\left[M_{A^nB^n}\rho_{AB}^{\otimes n}\right]\geq1-\eps$ and together with Equation~\eqref{eq:error_mutual} and Equation~\eqref{eq:mutual_error-last} that
\begin{align}
\zeta_{\rho,\cT}^n(\eps)\geq\inf_{\sigma\otimes\sigma\in\cT_n}D_s\left(\rho_{AB}^{\otimes n}\middle\|\sigma_{A^n}\otimes\sigma_{B^n}\right)-\frac{1}{n}\cdot\frac{s}{1-s}\log\frac{1}{\eps}-\frac{\log p(n)}{n}\,.
\end{align}
To deduce the claim it is now sufficient to argue that the R\'enyi quantum mutual information\footnote{This definition is slightly different from the R\'enyi quantum mutual information discussed in~\cite{HT14}.}
\begin{align}
I_s(A:B)_\rho:=\inf_{\sigma\otimes\sigma\in S(\cH)}D_s\left(\rho_{AB}\middle\|\sigma_A\otimes\sigma_B\right)
\end{align}
is additive on tensor product states. This, however, follows exactly as in the classical case~\cite[App.~A-C]{TH15} from the (quantum) Sibson identity~\cite[Lemma~3]{SW13a}
\begin{align}\label{sibson}
D_s\left(\rho_{AB}\middle\|\sigma_A\otimes\sigma_B\right)=D_s\left(\rho_{AB}\middle\|\sigma_A\otimes\bar{\sigma}_B\right)+D_s\left(\bar{\sigma}_B\middle\|\sigma_B\right)\quad\\ 
\mathrm{with}\quad\bar{\sigma}_B:=\frac{\left(\tr_A\left[\rho_{AB}^s\sigma_A^{1-s}\right]\right)^\frac{1}{s}}{\tr\left[\left(\tr_A\left[\rho_{AB}^s\sigma_A^{1-s}\right]\right)^\frac{1}{s}\right]}\,.\nonumber
\end{align}
\end{proof}

A more refined analysis of the above calculation, along the work of~\cite{HT14}, also allows to determine the Hoeffding bound for the product testing discrimination problem as above. However, for the strong converse exponent, we are missing the additivity of the sandwiched R\'enyi quantum mutual information
\begin{align}
\tilde{I}_s(A:B)_\rho:=\inf_{\sigma\otimes\sigma\in S(\cH)}\tilde{D}_s\left(\rho_{AB}\middle\|\sigma_A\otimes\sigma_B\right)
\end{align}
on tensor product states.


\section{The symmetric case}

So far we have focused on asymmetric hypothesis testing in the setting leading to a composite Stein's Lemma. 
Another closely related problem is that of composite symmetric hypothesis testing where, it is well known that in the case of fixed iid states $\rho^{\otimes n}$ vs.~$\sigma^{\otimes n}$, the optimal error exponent is given by the quantum Chernoff bound~\cite{ACMBMAV07,NS09} (see also Section~\ref{discrimination})
\begin{align}
C(\rho,\sigma)=\sup_{0\leq s\leq1}-\log\tr\left[\rho^s\sigma^{1-s}\right]\,.
\end{align}
However, the discrimination problem of testing convex combinations of iid states $\rho^{\otimes n}$ with $\rho\in\cS$ against convex combinations of iid states $\sigma^{\otimes n}$ with $\sigma\in\cT$ is still unsolved and it was conjectured~\cite{AM14} that, as in the commutative case, we have
\begin{align}\label{eq:conjCB}
C_{\cS,\cT}=\inf_{\substack{\rho\in\cS\\\sigma\in\cT}}C(\rho,\sigma)\,.
\end{align}
The most recent progress~\cite{AM14} states that in the case of a fixed null hypothesis $\cS=\{\rho\}$ the rate in Equation~\eqref{eq:conjCB} is achievable up to a factor of two 
\begin{align}
C_{\cS,\cT}\geq \frac{1}{2} \inf_{\substack{\rho\in\cS\\\sigma\in\cT}}C(\rho,\sigma)\,.
\end{align}
A very related problem that allows for an exact single-letter solution is that of multiple state discrimination, with more than two hypothesis (see Section~\ref{discrimination} and  also~\cite{Li16}). We note that extending the proof of the fixed state iid setting one can show that the following rate is achievable in the composite setting (assuming that the limit exists)
\begin{align}\label{eq:symmetric_regularized}
C_{\cS,\cT} \geq \sup_{0\leq s\leq 1} \lim_{n\to\infty}\frac{1}{n}\inf_{\substack{\nu\in\cS\\\mu\in\cT}}-\log\tr\left[\left(\int\rho^{\otimes n}\;\mathrm{d}\nu(\rho)\right)^s\left(\int\sigma^{\otimes n}\;\mathrm{d}\mu(\sigma)\right)^{1-s}\right]\,. 
\end{align}
However, our results about composite asymmetric hypothesis testing raise the question whether it is indeed possible to simplify Equation~\eqref{eq:symmetric_regularized} to the conjecture in Equation~\eqref{eq:conjCB} or whether the regularized version is already optimal. 
In Section~\ref{RegNeeded} we show that optimization over a single system cannot suffice in the case of the composite Stein's Lemma by connecting the asymmetric hypothesis setting to problems on recoverability inequalities and therefore to a recently found counterexample. Generally, it is very difficult to prove that regularizations are necessary and finding a similar connection for symmetric hypothesis testing might be very useful. Finding such a connection remains an interesting open problem.


\chapter{Discrimination power of a quantum detector}\label{discPower}


In this chapter, we will discuss the optimal rates giving the discrimination power of a quantum measurement, as it was defined in Section~\ref{OptOverStates}. Aiming for the ultimate rate, we will allow for asymptotically many uses~$n$ of the measurement device and optimize over all possible input strategies, this can include entangled input states as well as adaptively chosen ones. 
We will prove that, in the regime of an asymptotically large $n$, pairs of entangled states provide no advantage over \iid states of the form $\rho^{\otimes n}$. This is in sheer contrast with the dual problem of state discrimination where the measurement is optimized for fixed \iid states; there, we have strong numerical evidence \cite{calsamiglia_local_2010} that  collective non-separable measurements are required to attain the corresponding optimal exponential rate of the error probability, given by the quantum Chernoff bound~\cite{ACMBMAV07,CMMAB08} (see also Chapter \ref{hypoTesting}). Furthermore, we will also show that while adaptively chosen inputs provide an advantage for a finite $n$, for $n$ going to infinity adaptive strategies do not help. In the proof, we approach state discrimination as a communication problem and allow for adaptive protocols. We argue that these adaptive strategies are general enough to include all those strategies that use entangled input states. Finally we use a result of classical channel discrimination to show that asymptotically \iid states are optimal. 
The optimal rates and their proof can be found in the next section, Section \ref{optRates}.
In Section \ref{SomeProp} we will discuss some properties of the rates, in particular their behavior under mixing of the POVMs. 
Next, we will discuss the difficulties of finding the hypothesis testing errors in the finite $n$ case in Section \ref{finiteNumber}. This will include some concrete examples where adaptive strategies do outperform \iid ones. Finally we will calculate the discrimination power for some example POVMs, namely covariant measurements and the noisy Stern-Gerlach measurement, in Section \ref{DPexamples}.

\section{The optimal rates for discrimination power}\label{optRates}

We will now turn to giving the optimal error rates promised earlier. As described in Section \ref{OptOverStates} the direct approach of calculating the errors involved becomes very quickly infeasible (because of the many possible combinations when using a measurement $n$ times, see also the example in Section~\ref{finiteNumber}). To evaluate the error rates defined in Equations~\eqref{xiCB}, \eqref{xiSL} and \eqref{xiHB} we will thus follow an alternative route.
For most of this section, we will focus on symmetric hypothesis testing, but later we will also give the rates in the asymmetric case and we remark that the proof follows exactly the same arguments as the one presented here.

We are now ready to state our main result. 
\begin{thm}\label{DualCBtheorem}
The optimal exponential rate defined in Equation~\eqref{xiCB} is given by the classical Chernoff Bound:
\begin{equation}
\xi_{\textrm CB}=-\min_{(\rho,\sigma)}\;\min_{0\leq s\leq 1}\phi(s | P\|{\overline P}), 
\label{xi=}
\end{equation}
where $P_k=\tr(E_k \rho)$ and ${\overline P}_k =\tr(E_k \sigma)$ are the outcome probability distributions of  (a single use of)  the POVM~${\mathscr E}=\{E_k\}_{k=1}^m$. This rate can be attained using i.i.d. states, $\rho^{\otimes n}$ and $\sigma^{\otimes n}$.
\end{thm}

The main ingredient of the proof of Theorem~\ref{DualCBtheorem} is to show that our problem is a particular case of {\em classical} channel discrimination. This will allow us to complete the proof using a result by Hayashi~\cite{H08} on the asymptotics of classical channel discrimination with adaptive strategies.

To this end, we momentarily broaden the scope of our original problem. First, we view hypothesis testing as a communication protocol where Alice (the state preparer) sends one of two possible messages, $H_0$,~$H_1$, to~Bob (the measurer) using suitable states in~$\setS{\cH_d^{\otimes n}}$. Bob is allowed to perform~$n$ measurements with his detector to identify with minimum error which of the messages Alice sent. 
Second, in this communication context it is natural to allow classical feedback from Bob to Alice after each measurement. This enables an adaptive protocol (see Figure~\ref{Fig:adap1}) in which
Alice sends one state at a time to Bob's detector and waits for him to provide feedback on the obtained outcome. Alice uses this information to prepare the succeeding state in a way that minimizes the identification error. 
Such protocols are widely used in quantum information theory~\cite{DFLS16, BSST99,PL16,takeoka2016optimal,GM2000,hayashi_comparison_2011}, particularly in quantum channel discrimination~\cite{CDP08,HHLW09,CMW14,PL16,takeoka2016optimal}, where we know that adaptive strategies can improve the performance in the finite repetition case as well as asymptotically (in the sense that there exist channels that can be perfectly discriminated with a finite number of copies adaptively, but require an infinite number with a non-adaptive strategy~\cite{HHLW09}), while for classical channels an improvement is given for finite repetitions but not asymptotically~\cite{H08}. 
%
%
The next Lemma follows from the structure of~$E_{{\textbf k}^n}$ in Equation~\eqref{Ea n}. 
\begin{lemma}\label{lemma}
For {\em any} (possibly entangled) \mbox{$\rho^n\in\setS{\cH_d^{\otimes n}}$} (analogously for $\sigma^n$) there is an adaptive protocol that gives the same outcome probability distribution when applying $n$ copies of a fixed measurement. 
\end{lemma}
\begin{proof}
To prove Lemma~\ref{lemma}, we define $\rho'_\emptyset:=\tr_{[n] \setminus 1}(\rho^n)$, where we denote by $[n]\!\!\setminus\!\! s$ the set $\{1,2,\dots,s-1,s+1,\dots n\}$, $s=1,2,\dots n$. Then,~$\rho^n$ and  $\rho'_\emptyset$ give the same probability distribution to the outcomes of Bob's first measurement: 
\begin{align}
P({k_1}|\rho^n):=\tr[(E_{k_1}\otimes\id)\rho^n]=\tr(E_{k_1}\rho'_\emptyset).
\end{align} 
With Bob's feedback (the value of $k_1$), Alice can next prepare the second (unnormalized) state as $\rho'_{k_1}:=\tr_{[n]\setminus 2} [(E_{k_1}\otimes\id)\rho^n]$. So, $\rho^n$ and $\rho'_{k_1}$ give the same outcome probabilities up to Bob's second  measurements: 
\begin{align}
P({{\textbf k}^2}|\rho^n):=\tr [(E_{{\textbf k}^2}\otimes\id) \rho^n]=\tr(E_{k_2}\rho'_{k_1}).
\end{align} 
Note that the probabilities of previous outcomes are implicit in the normalization of $\rho'_{k_1}$.
We readily see that if Alice's preparation at an arbitrary step $s$ is
\begin{equation}
\rho'_{{\textbf k}^{s-1}}:=\tr_{[n]\setminus s}\left[\left( E_{{\textbf k}^{s-1}}\otimes\id\right) \rho^n\right],
\label{rho'}
\end{equation}
where we used the convention $E_{{\textbf k}^0}=E_{\emptyset}:=\id$, then 
\begin{align}
P({{\textbf k}^{s}}|\rho^n)=\tr [(E_{{\textbf k}^{s}}\otimes\id)\rho^n]
=
\tr (E_{k_s}\rho'_{{\textbf k}^{s-1}}),
\end{align}
for $s=1,2,\dots,n$ (obviously, the analogous relation holds for $\sigma^n$, $\sigma'_{{\textbf k}^{s-1}}$). This completes the proof of the lemma.
\end{proof}
 
Adaptive protocols are thus more general than those in which $\rho^n$ is entangled, so the optimal protocol can be chosen to be adaptive with no loss of generality.

\begin{figure}
\includegraphics[width=26em]{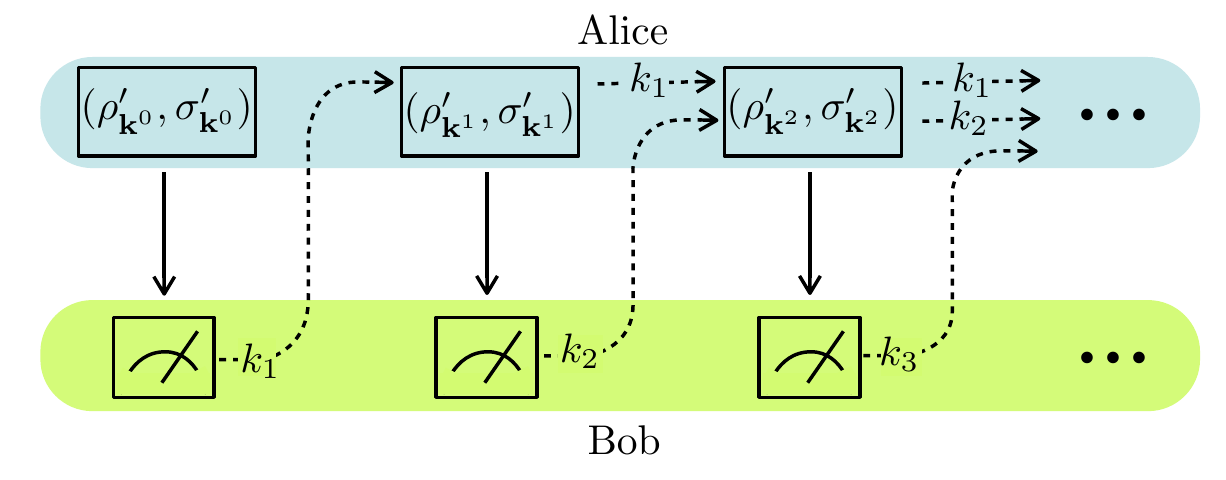}
\vspace{0em}
\caption{Adaptive protocol. At each step (left to right), Alice sends to Bob (solid arrows) the state in Equation~(\ref{rho'}), which she has  prepared using Bob's feedback (dashed arrows).}\label{Fig:adap1}
\end{figure}

Next, to prove Theorem~\ref{DualCBtheorem}, we show that the adaptive communication protocols introduced above, the optimal one in particular, 
can be cast as discrimination of two classical channels. 
To  this end, we choose the classical (continuous) input alphabet as ${\mathcal X}=\setS{\cH_d}\times\setS{\cH_d}$,
where each letter $x=(\rho,\sigma)\in\mathcal{X}$  is a classical description of the pair of states, e.g.,  the two Bloch vectors of $\rho$ and~$\sigma$.
The output alphabet  $\mathcal Y$  is naturally given by the outcome labels of our fixed measurement (i.e., the POVM $\mathscr E$): ${\mathcal Y}=\{1,2,\dots m\}$. 

We can then associate the null and alternate hypothesis $H_0$ and $H_1$ to the classical channels $W_x(k):= \tr (E_k \rho)$ and ${\overline W}_x(k):= \tr (E_k \sigma)$, respectively, where $(\rho,\sigma)=x\in{\mathcal X}$ and  $k\in{\mathcal Y}$.  These channels reproduce the same conditional probabilities, 
 $P_{k}$ and $\bar P_{k}$, that arise in our original problem. Hence the (single-shot) optimal state discrimination is formally equivalent to the optimal channel discrimination obtained by minimizing over the inputs~$x\in\mathcal X$.

This analogy holds also for our general multiple-shot problem. The adaptive protocol defined by the states $\rho'_{{\textbf k}^{s-1}},\, \sigma'_{{\textbf k}^{s-1}}\in\setS{\cH_d}$, $s=1,2,\dots, n$, 
translates into an adaptive channel discrimination strategy 
with~$n$ uses of either $W$ or~$\overline W$, where at each step $s$ we feed the channel with an input letter $x_{{\textbf k}^{s-1}}\in \mathcal X$, conditional on the previous outcomes~${\textbf k}^{s-1}=\{k_1,k_2,\dots,k_{s-1}\}$,  $k_{i}\in \mathcal Y$.
%

We can now invoke the main result in~\cite{H08}, which states that for the problem of asymptotic classical channel discrimination no adaptive strategy can outperform the best non-adaptive or even fixed strategy.  More precisely, it states that the optimal error rate can be attained  by the simple sequence where all the letters are equal,~$x_1=x_2=\dots=x_n$. We hence conclude that the optimal error rates for our original problem can be achieved by i.i.d. state pairs, $(\rho^{\otimes n},\sigma^{\otimes n})$. 
This holds for the Chernoff bound, Hoeffding's bound and for Stein's Lemma (see Equations~\eqref{other} and~\eqref{other0} below).

Computing the exponent rate 
in Equation~(\ref{xiCB}) is now identical to computing the analogous rate for the classical hypothesis testing problem of discriminating between the probability distribution  $P_k=\tr(E_k \rho)$ and ${\overline P}_k =\tr(E_k \sigma)$ after $n$ samplings,  which is given by the classical Chernoff Bound.
This completes the proof of Theorem~\ref{DualCBtheorem}.

In general, it might be very difficult to find the states that optimize the dual Chernoff bound. Limiting the set of states to optimize over, could significantly simplify the computation of the error exponent. A natural conjecture would be that the optimal states are orthogonal to each other, simply because one could expect that perfectly distinguishable states can be more easily distinguished by any POVM. This intuition might depend on the distinguishability measure. For instance, for the Bhattacharyya distance (which coincides with \eqref{clCB} if we fix $s=1/2$, instead of minimizing over it) one can find examples were the  signal states that optimize this distance are not orthogonal.  As far as the discrimination power is concerned, all the gathered evidence so far, which includes analytical results for some particularly symmetric POVMs and numerical results for generic qubit POVMs, seems to indicate that orthogonal signals optimize the  dual Chernoff quantity. In Section \ref{DPexamples} it shown explicitly that in the case of covariant qubit POVM's orthogonal states are optimal for the dual Chernoff bound by direct calculation. Here we also calculate the error exponent for this POVM to be $\frac{\pi}{4}$.

Though in this letter we have focused on the problem dual to symmetric hypothesis testing, which led us to Theorem~\ref{DualCBtheorem}, but the very same arguments concerning the optimality of i.i.d. state pairs apply to the dual Stein's lemma and Hoeffding's bound as well. It follows from our analysis that they can be computed simply as
\begin{align}
\xi_{\textrm SL}\!&=\max_{(\rho,\sigma\!)} D(P\|{\overline P}), \label{other0}\\
\xi_{\textrm HB}\!&=\max_{(\rho,\sigma\!)}\,\sup_{0\le s\le 1}\!\!{-s r - \phi( s|P\|{\overline P})\over 1-s}.
\label{other}
\end{align}

We will now turn to the next section where we investigate some particular properties of the discrimination power. 


\section{Discrimination power under mixing of POVMs}\label{SomeProp}

In this section we derive some properties of the optimal error exponents. 
To this end, we investigate the function
\begin{align}
C_{\mathscr E}&=\min_{(\rho,\sigma)}\min_{0\leq s\leq 1} \sum_{i=1}^{m} \left[\tr(\rho E_i)\right]^s \left[\tr(\sigma E_i)\right]^{1-s},
\label{C_e}
\end{align}
for a given $m$-element POVM, ${\mathscr E}=\{E_i\}_{i=1}^m$. We will denote the corresponding error exponents in the three settings discussed in the last section by $\xi^{\mathscr E}_{\textrm CB}$, $\xi^{\mathscr E}_{\textrm SL}$ and $\xi^{\mathscr E}_{\textrm HB}$ to make the dependence on the measurement explicit.   
Let us investigate the behavior of the discrimination power under mixing of POVMs.
\begin{lemma}
Let ${\mathscr E}$ be a POVM with $m$ elements $E_i$ and ${\mathscr G}$ be a POVM with $n$ elements $G_i$. Define a mixed POVM with $m+n$ elements, $\widehat{\mathscr E}=\{\widehat E_i\}_{i=1}^{m+n}$, through
\begin{equation}
\widehat{ E}_i=\left\{
\begin{array}{rcl}p E_i&\mbox{if}& 1\le i\le m,\\[.5em]
(1-p)G_{i-m} &\mbox{if}& m< i\le m+n .
\end{array}
\right.
\end{equation}
Then, we can upper bound $\xi_{\textrm CB}^{\widehat {\mathscr E}}$ by
\begin{equation}
\xi_{\textrm CB}^{\widehat {\mathscr E}} \leq p \xi_{\textrm CB}^{\mathscr E} + (1-p) \xi_{\textrm CB}^{\mathscr G}
\end{equation}
and lower bound it by 
\begin{equation}
\xi_{\textrm CB}^{\widehat {\mathscr E}} \geq - \log \min \Big\{ p C_{\mathscr E} + (1-p) , p + (1-p) C_{\mathscr G} \Big\}.
\label{lw CB}
\end{equation}
%
Furthermore we can state the following relations for the dual Stein's bound and dual Hoeffding's bound
\begin{align}
\xi_{\textrm SL}^{\widehat {\mathscr E}} \leq p\, \xi_{\textrm SL}^{\mathscr E} + (1-p) \xi_{\textrm SL}^{\mathscr G}, \\[.5em]
\xi_{\textrm SL}^{\widehat {\mathscr E}} \geq \max \left\{ p\, \xi_{\textrm SL}^{\mathscr E} ,  (1-p) \xi_{\textrm SL}^{\mathscr G} \right\}, \\[.5em]
\xi_{\textrm HB}^{\widehat {\mathscr E}} \leq p\, \xi_{\textrm HB}^{\mathscr E} + (1-p) \xi_{\textrm HB}^{\mathscr G}.
\end{align}
\end{lemma}

\begin{proof}
Define ${\mathscr E}$, ${\mathscr G}$ and $\widehat {\mathscr E}$ as above. Let us first give a lower bound for $C_{\widehat{\mathscr E}}$.
\begin{alignat}{2}
C_{\widehat {\mathscr E}} &= \min_{(\rho,\sigma)}\min_{0\leq s\leq 1} &&\sum_{i=1}^{m+n} \left[\tr(\rho \widehat E_i)\right]^s \left[ \tr(\sigma \widehat E_i)\right]^{1-s} \nonumber\\[.5em]
&\geq \min_{(\rho,\sigma)}\min_{0\leq s\leq 1}&& \sum_{i=1}^{m} \left[\tr(\rho \widehat E_i)\right]^s \left[ \tr(\sigma \widehat E_i)\right]^{1-s}  \\ 
&&&+ \min_{(\rho,\sigma)}\min_{0\leq s\leq 1} \sum_{i=m+1}^{m+n} \left[\tr(\rho \widehat E_i)\right]^s \left[\tr(\sigma \widehat E_i)\right]^{1-s}  \nonumber\\[.5em]
&= \min_{(\rho,\sigma)}\min_{0\leq s\leq 1}&& \sum_{i=1}^{m} p\left[ \tr(\rho E_i)\right]^s \left[\tr(\sigma  E_i)\right]^{1-s} \\ 
&&&+ \min_{(\rho,\sigma)}\min_{0\leq s\leq 1} \sum_{i=1}^{n} (1-p) \left[\tr(\rho G_i)\right]^s\left[ \tr(\sigma G_i)\right]^{1-s}  \nonumber\\[.5em]
&= p C_{\mathscr E} + (1-&&p) C_{\mathscr G}.
\end{alignat}
We continue by giving an upper bound to $C_{\widehat{\mathscr E}}$. 
Let $s^*$, $\rho^*$ and $\sigma^*$ be respectively the value of $s$ and the states~$\rho$ and $\sigma$ that attain the minimum value on the right hand side of Equation~\eqref{C_e} for the POVM $\mathscr E$. Then, 
\begin{align}
C_{\widehat {\mathscr E}}&= \min_{(\rho,\sigma)}\min_{0\leq s\leq 1} \sum_{i=1}^{m+n} \left[\tr(\rho \widehat E_i)\right]^s \left[\tr(\sigma \widehat E_i)\right]^{1-s} \\
&\leq p\, C_{\mathscr E} + \sum_{i=m+1}^{m+n}\left[ \tr(\rho^* \widehat E_i)\right]^{s^*} \left[\tr(\sigma^* \widehat E_i)\right]^{1-s^*}.
\end{align}
We prove one of the lower bounds by bounding the second term as 
\begin{align}
\sum_{i=m+1}^{m+n}&\left[ \tr(\rho^* \widehat E_i)\right]^{s^*} \left[\tr(\sigma^* \widehat E_i)\right]^{1-s^*} \nonumber \\
&\leq  (1-p) \left[ \sum_{i=1}^{n} \tr(\rho^* G_i)\right]^{s^*} \left[ \sum_{i=1}^{n} \tr(\sigma^* G_i) \right]^{1-s^*} 
= 1-p,
\end{align}
where the inequality follows from the definition of $\widehat {\mathscr E}$ and H\"{o}lder's inequality. An analogous bound follows by choosing  $s^*$, $\rho^*$ and $\sigma^*$ to be  the value of $s$ and the states $\rho$ and $\sigma$ that attain the minimum value on the right hand side of Equation~(\ref{C_e}) for the POVM $\mathscr G$.  Hence, 
$C_{\widehat {\mathscr E}} \le \min\left\{p\, C_{\mathscr E}+(1-p),p+(1-p)C_{\mathscr G}\right\}$, and Equation~(\ref{lw CB}) follows.
The inequalities for Stein's and Hoeffding's bounds are proven similarly, additionally using the fact that the relative entropy is lower bounded by zero and  that the logarithm is concave. 
\end{proof}

The above shows that by mixing a pair of POVMs one can never increase the discrimination power of the best POVM of the pair. 

Furthermore, it is easy to see that applying additional CPTP maps before performing the measurement does not increase the discrimination power either. 
This follows directly from the fact that the image of a CPTP map is always at most the input state space itself. 


\section{Example for finite number of measurements}\label{finiteNumber}

So far,  special emphasis has been placed on the asymptotics of the problem at hand. 
It is illustrative to examine with a few examples the difficulties arising for finite $n$, where some of the asymptotic results do not hold. Let us focus on two-element POVMs, ${\mathscr E}=\{E_1,E_2=\id-E_1\}$. In this case, $E_1$ and $E_2$ commute and can be diagonalized simultaneously. In the multiple-shot  scenario, the groupings $\{E^{a},E^{\bar a}\}$ will also be diagonal in the very same local basis that diagonalize $E_1$ and $E_2$ and thus each state of the optimal pair, $(\rho^n,\sigma^n)$, in necessarily a product state of elements of that basis.  In this case, however, one can show that i.i.d. states are not necessarily optimal. 
Here we give a concrete example for~$n=3$ where the optimal states are~$\rho^3=|001\rangle\langle001|$ and  $\sigma^3=|110\rangle\langle110|$, rather than~$|000\rangle\langle000|$ and $|111\rangle\langle111|$. Furthermore, we also show that there exists an adaptive protocol with yet a smaller error rate, 
thus outperforming the optimal non-adaptive protocol for $n=3$. 

Let us consider the simple example where the POVM is
\begin{equation} \label{POVMexam}
{\mathscr E} = \left\{ E_1 = \left(
\begin{array}{cc}
0.4&0\\
0&0.2
\end{array}\right), E_2 = 1-E_1 = \left(
\begin{array}{cc}
0.6&0\\
0&0.8
\end{array}\right)\right\},
\end{equation}
and $n=3$ (the measurement defined by $\mathscr E$ is performed 3 times). It follows from the diagonal form of $E_1$ and $E_2$ that the optimal input states are tensor products of pure states and also diagonal in the given basis. From the symmetry of the problem it should be clear that there are only two possible ways to achieve the optimal error rate: 
(i)~use the pair $(\rho_0^{\otimes 3},\rho_1^{\otimes 3})$, or (ii)~use $(\rho_0^{\otimes 2}\otimes \rho_1,\rho_1^{\otimes 2}\otimes \rho_0 )$, where 
\begin{equation}\label{rho 0 rho 1}
\rho_0 =
\begin{pmatrix}
1&0\\
0&0
\end{pmatrix},
\qquad 
\rho_1 = 
\begin{pmatrix}
0&0\\
0&1
\end{pmatrix}.
\end{equation}
We next compute the error probability in both cases. For (i) it can be checked that the optimization over groupings gives
\begin{equation} 
F^{\textrm (i)}_1 = E_2 \otimes E_2 \otimes E_2,\quad  F^{\textrm (i)}_0 =\id-F^{\textrm (i)}_1,
\end{equation}
so $F^{\textrm (i)}_0$ is the sum of the remaining seven tensor products. 
One can easily check that the error probability is
\begin{equation}
p_{\textrm err}^{\textrm(i)} = 0.352. 
\end{equation}
For (ii), the optimization over groupings  gives now
\begin{align} 
F_0^{\textrm(ii)} = E_1 \otimes E_1 \otimes E_1 + E_1 \otimes E_1 \otimes E_2
+ E_1 \otimes E_2 \otimes E_2  + E_2 \otimes E_1 \otimes E_2 ,
\end{align}
and $F_1^{\textrm(ii)}  = 1 - F_0^{\textrm(ii)} $, being the sum of the remaining four products.
This gives
\begin{equation}\label{err1eq}
p_{\textrm err}^{\textrm(ii)} = 0.344. 
\end{equation}
We see that (ii) is optimal. The optimal state pair in this example for finite $n$ is not of the form $(\rho^{\otimes n},\sigma^{\otimes n})$, in contrast with what we found in the asymptotic limit of large $n$. 

Furthermore, we next show that the error rate given in Equation~\eqref{err1eq} can be lowered by an adaptive protocol as follows. We choose the first two input pairs to be, as in the previous examples, $(\rho_0^{\otimes 2},\rho_1^{\otimes 2})$.
 If the first and second measurement return $1$ (i.e., if  ${\textbf k}^2=\{1,1\}$),  the preparation of the third state pair is $(\rho_0,\rho_1)$, if not, we swap the preparations (i.e., the third pair is $(\rho_1,\rho_0)$, as in~(ii)).
One can check that the optimal grouping is 
\begin{align} 
F^{\textrm ad}_0 = &E_1 \otimes E_1 \otimes E_1 + E_1 \otimes E_1 \otimes E_2 
+ E_2 \otimes E_2 \otimes E_1 \nonumber\\ &+ E_1 \otimes E_2 \otimes E_2  + E_2 \otimes E_1 \otimes E_2 ,
\end{align}
and $F^{\textrm ad}_1 = 1 - F^{\textrm ad}_0$.
This results in an error probability of
\begin{equation}\label{err1eq2}
p_{\textrm err}^{\textrm ad} = 0.336. 
\end{equation}

\begin{figure}[t!]
\centering
\includegraphics[scale=0.75]{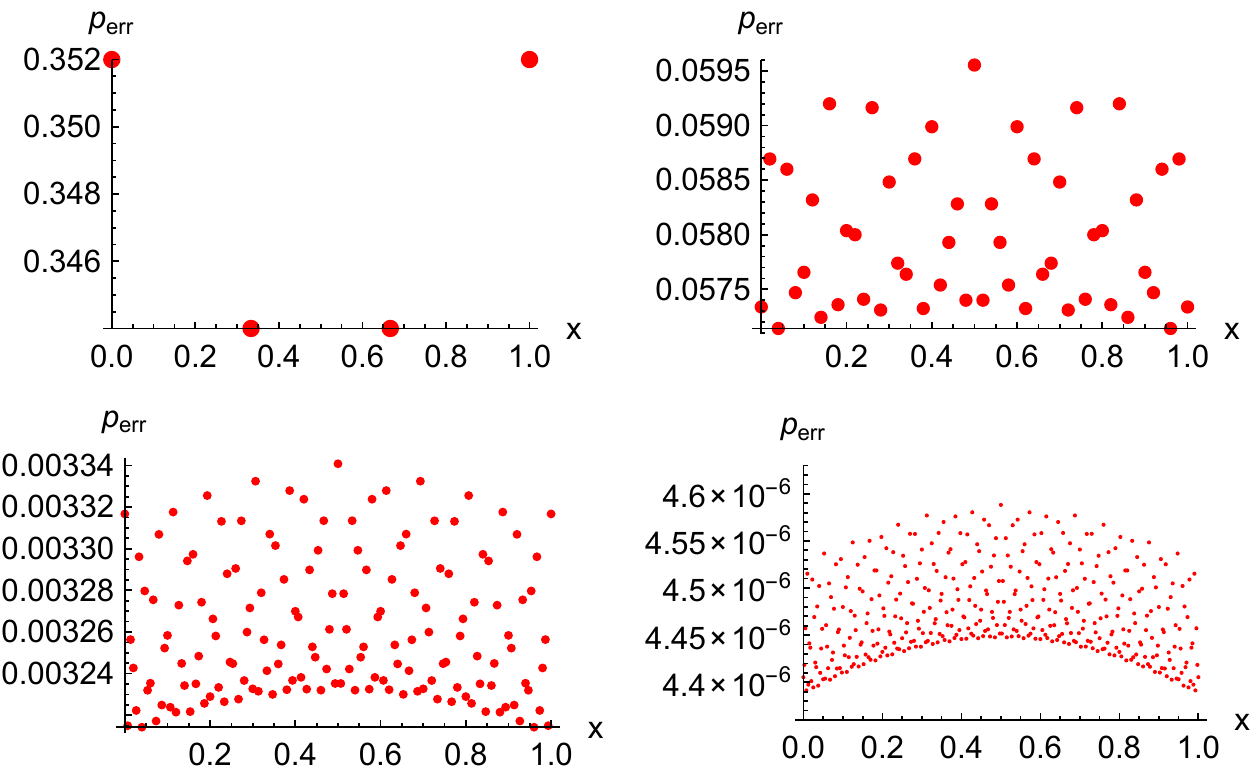}
\caption{Error probability vs. $x$ for state pairs of the form $(\rho_0^{\otimes xn}\otimes\rho_1^{\otimes(1-x)n},\rho_1^{\otimes xn}\otimes\rho_0^{\otimes(1-x)n})$. The values of $n$ are~$3$, $50$, $150$, and $400$ (from top left to lower right), with the POVM given as in Equation~\eqref{POVMexam}.  \label{fig:POVMnm} }
\end{figure}

Equation~\eqref{err1eq} provides an explicit example that non \iid states can outperform \iid ones for finite number of measurements.
Nevertheless, we proved in Section~\ref{optRates} that asymptotically the choice $(\rho_0^{\otimes n},\rho_1^{\otimes n} )$ is optimal. Figure~\ref{fig:POVMnm} illustrates that this is indeed the case by showing plots of the error probability for all possible state pairs at a given choice of $n$. Since  permutations among subsystems do not affect the error probability, the state pairs can be taken to be without loss of generality of the following form:
\begin{equation}\label{rho x}
(\rho^n,\sigma^n)=\left(\rho_0^{\otimes xn}\otimes\rho_1^{\otimes(1-x)n},\rho_1^{\otimes xn}\otimes\rho_0^{\otimes(1-x)n}\right).
\end{equation}
As $n$ increases, we see that the minimum error probability becomes a convex function of $x$.


\section{Examples}\label{DPexamples}

In this section we illustrate our results by computing the discrimination power of the qubit covariant POVM, ${\mathscr E}=\{\id+\mbox{\boldmath $n$}\cdot \mbox{\boldmath$\sigma$}\}_{\mbox{\boldmath\scriptsize$n$}\in{\mathbb S}^2}$, where~$\mbox{\boldmath$\sigma$}$ is the vector of Pauli matrices and ${\mathbb S}^2$ is the unit 2-sphere. The result is $\xi_{\textrm CB}=-\log(\pi/4)$ which can be compared to $\zeta_{\textrm CB}=-(1/2)\log(1-r^2)$, corresponding to a noisy Stern-Gerlach of purity $r$. We see that  ${\mathscr E}$ has the same discrimination power as  a Stern-Gerlach with purity $r \approx 0.62$.

\subsection{Covariant measurements}\label{DPexamplesCovMeas}

In this section we prove that orthogonal states are optimal for the qubit covariant POVM,
\begin{equation}
{\mathscr E}=\left\{E_{\scriptsize\vec n}:= \id+\vec n\cdot\vec\sigma\right\}_{{\scriptsize\vec n}\in{\mathbb S}^2},
\end{equation}
and we compute its discrimination power according to the dual of the Chernoff bound.
For the covariant measurement $\mathscr E$, we have
\begin{equation}
C_{\mathscr E}=\min_{0\le s\le 1}C_s,
\end{equation}
 where
\begin{equation}
C_s=\min_{(\rho,\sigma)}\int_{{\mathbb S}^2} dn\left[\tr\left(\rho E_{\scriptsize\vec n}\right)\right]^s\left[\tr\left(\sigma E_{\scriptsize\vec n}\right)\right]^{1-s}=\min_{\mbox{\scriptsize$(\vec m_1$},\mbox{\scriptsize$\vec m_2)$}} C,
\end{equation}
and
\begin{equation}
C:=\int   dn  \left( {1 +\cos\widehat{\vec n\,\vec m}_1}  \right)^{  1-s}   \left( {1 +\cos\widehat{\vec n\,\vec m}_2} \right)^{  s} .
\end{equation}
Here $\vec m_1$ and $\vec m_2$ are the Bloch vectors of $\sigma$ and $\rho$ respectively.
Choosing with no loss of generality
\begin{equation}
\vec m_1=\hat z,\quad 
\vec m_2= \hat z \sin\alpha + \hat x \cos\alpha,
\end{equation}
where $\hat z$ and $\hat x$ are the unit vectors pointing in the $z$ and $x$ direction
($\vec m_1$ and $\vec m_2$ on the $xz$-plane), we have
\begin{equation}
C=\int_0^\pi \sin\theta \int_0^{2\pi}{d\phi\over4\pi} \left(1+\cos\theta \right)^{1-s}
\left(1+\cos\theta\cos\alpha+\cos\phi\sin\theta\sin\alpha \right)^{s}. \label{covCa}
\end{equation}
where $\theta$ and $\phi$ are the polar and azimuthal angles of the unit vector $\vec n$.
After substantially rewriting Equation \eqref{covCa} one can prove that $dC/d\alpha\le 0$, and it only vanishes at $\alpha=0,\pi$. We provide the detailed calculation in Appendix~\ref{appendixCBcm}.  It follows that $C$ has a maximum at $\alpha=0$ and a minimum at $\alpha=\pi$.

Substituting $\alpha=0$ in the definition of $C$, we can write
\begin{eqnarray}
C_s&=&\int_{-1}^{1} du \left(1+u\right)^{1-s}{\sin(s\pi)\over2\pi}\int_0^\infty dx\, x^{s-1}{1-u\over1+x-u} \nonumber\\
&=&{1\over2}\int_{-1}^{1} du \left(1+u\right)^{1-s}(1-u)^s\nonumber\\
&=&2\int_0^1dt\, t^{1-s}(1-t)^s=2B(2-s,1+s),
\end{eqnarray}
where we have used the relation
\begin{equation}
a^s={\sin(s\pi)\over\pi}\int_0^\infty dx{a x^{s-1}\over a+x}\, , 
\end{equation} we have changed variables as $u\to t=(1+u)/2$, and used the definition of the Euler Beta function $B(a,b)$. Additionally, we know that
\begin{equation}
B(2-s,1+s)={\Gamma(2-s)\Gamma(1+s)\over \Gamma(3)}
={s(1-s)\over2}\Gamma(1-s)\Gamma(s)
={s(1-s)\pi\over2\sin(s\pi)} ,
\end{equation}
where $\Gamma$ is the Gamma function. 
We finally have the result that the Chernoff bound is
$C_{\mathscr E}={\pi/4}$,
as $C_s$ has its minimum at $s=1/2$. It follows that the error exponent corresponding to the dual of the Chernoff bound is
\begin{equation}
\xi_{\textrm CB}=-\log(\pi/4).
\end{equation}


\subsection{Noisy Stern-Gerlach}

The noisy Stern-Gerlach measurement of purity $r$ is defined by
\begin{equation}
{\mathscr E}=\left\{{\id+ r \sigma_z\over2},{\id-r \sigma_z\over2}\right\}.
\end{equation}
Then, if $\vec m_1$, $\vec m_2$ are the Bloch vectors of $\sigma$ and $\rho$ respectively, one has
\begin{equation}
C={1\over2}\left(1+r \cos\theta_1\right)^{1-s}\left(1+r \cos\theta_2\right)^s+{1\over2}\left(1-r \cos\theta_1\right)^{1-s}\left(1-r\cos\theta_2\right)^s.
\end{equation}
One can easily check that $C$ attains its minima over the pair $(\theta_1,\theta_2)$ at $(0,\pi)$ and $(\pi,0)$, thus
\begin{equation}
C_s=\min_{\mbox{\scriptsize$(\vec m_1$},\mbox{\scriptsize$\vec m_2)$}}C=
{1\over2}\left(1\pm r \right)^{1-s}\left(1\mp r \right)^s+{1\over2}\left(1\mp r\right)^{1-s}\left(1\pm r\right)^s.
\end{equation}
The minimum over $s$ is at $s=1/2$, so we obtain
\begin{equation}
C_{\mathscr E}=\sqrt{1-r^2},\qquad \zeta_{\textrm CB}=-{1\over2}\log(1-r^2).
\end{equation}
Note that $C_{\mathscr E}$ vanishes and therefore $\xi_{\textrm CB}$ is going to $\infty$ for a noiseless Stern-Gerlach apparatus, i.e.,  for $r=1$,  as $E_1$ and $E_{2}$, become orthogonal projectors and thus can be used for perfect discrimination.

Both of the examples in this section are optimized by an orthogonal state pair. As mentioned previously it remains an interesting open problem to prove whether this is indeed the case for all measurements. 

We are ready now to move to the second part of this thesis where we will discuss entropy inequalities with a particular focus on recoverability and information combining. Nevertheless we will not leave the topic of hypothesis testing behind as it will provide us with several insights into the problems of the next chapter, showing the close connection between both topics.

\part{Recoverability and entropy inequalities}

\chapter{Recoverability}\label{recoverability}

In this chapter we will start our investigation of entropy inequalities with a focus on recoverability. There are many facets of quantum science in which the notion of quantum state recovery is deeply embedded.  This is particularly true for quantum error correction~\cite{Gai08,LB13} and quantum key distribution~\cite{SBCDLP09}, where the primary goal is that of recovery.  In the former, the task is to reconstruct a quantum state where some part of the state has undergone noise or loss; in the latter, the task is to keep a message secure against an eavesdropper attempting a similar reconstruction.  In either case, the success or failure of a protocol often hinges on whether the particular state in question is recoverable at all, or if the state is beyond repair. 

A particularly important class of states are those that constitute a Markov chain.  A classical Markov chain can be understood as a memory-less random process, i.e., a process in which the state transition probability depends only on the current state and not on past states.  If random variables $X$, $Y$, and $Z$ form a classical Markov chain, denoted as $X\rightarrow Y \rightarrow Z$, then the classical conditional mutual information 
\begin{equation}
I(X;Z|Y)=0.
\end{equation}
Classical Markov chains model an impressive number of natural processes in physics and many other sciences~\cite{Norris97}.

In an attempt at understanding a quantum generalization of these ideas the authors of \cite{HJPW04}\ defined a quantum Markov chain, in analogy with the classical case mentioned above, to be a tripartite state $\rho_{ABC}$\ for which the conditional quantum mutual information (CQMI) $I(A;B|C)_{\rho}$\ is equal to zero. However, later work in \cite{ILW08} (see also \cite{E15}) realized that these notions 
made sense only in the exact case by demonstrating that large deviations from a quantum Markov state as defined in \cite{HJPW04} can sometimes lead to only small increases of the CQMI. 

Meanwhile, it has been known for some time that an equivalent description for the exact case $I(A;B|C)_{\rho} = 0$ exists in terms of recoverability. The work of Petz
\cite{Petz1986,Petz1988}\ implies that there exists a recovery channel
$\mathcal{R}_{C\rightarrow AC}$ such that $\rho_{ABC}=\mathcal{R}_{C\rightarrow AC}(\rho_{BC})$ if and only if $I(A;B|C)_{\rho} = 0$. In fact, it can be shown that the channel can always be chosen as the so called Petz recovery channel: 
\begin{equation}
\cR_{C\to AC}(\cdot):=\rho_{AC}^{\frac{1}{2}}\left(\rho_C^{-\frac{1}{2}}(\cdot)\rho_C^{-\frac{1}{2}}\right)\rho_{AC}^{\frac{1}{2}}. \label{defPetzMap}
\end{equation}
This is in perfect agreement with the exact
classical case mentioned above:\ for a state satisfying $I(A;B|C)_{\rho} = 0$, one could lose the system $A$ and recover it back
from $C$ alone. In other words, all correlations between systems $A$ and $B$ are mediated through system $C$ for quantum Markov chain states. Recoverability in this sense is thus intimately connected to Markovianity and represents a method for handling the approximate case, different from that given in \cite{HJPW04}.

To measure non-Markovianity in the approximate case, the general approach outlined in \cite{SW14} was to quantify the ``distance'' from $\rho_{ABC}$ to its closest recovered version. The main measure on which \cite{SW14} focused was the \textit{fidelity of recovery}, defined as
\begin{equation}
F(A;B|C)_{\rho}\equiv\sup_{\mathcal{R}_{C\rightarrow AC}}F(\rho_{ABC},\mathcal{R}_{C\rightarrow AC}(\rho_{BC})). \label{eq:fidrec}
\end{equation}
The optimization in \eqref{eq:fidrec} is with respect to quantum channels $\mathcal{R}_{C\rightarrow AC}$\ acting on the system $C$ and producing an output on the systems $A$ and $C$. A related measure, defined in \cite[Remark 6]{SW14}, is the \textit{relative entropy of recovery}:
\begin{equation}
D(A;B|C)_{\rho} \equiv \inf_{\mathcal{R}_{C\rightarrow AC}}D(\rho_{ABC}\Vert\mathcal{R}_{C\rightarrow AC}(\rho_{BC})). \label{REORdef}
\end{equation}

From the main result of \cite{FR14}, which established that
\begin{equation}
I(A;B|C)_{\rho} \geq - 2 \log F(A;B|C)_{\rho}, \label{FRlowerbound}
\end{equation}
it is now understood that the CQMI itself is a measure of non-Markovianity as well. Before \cite{FR14}, an operational interpretation for the CQMI had already been given in \cite{DY08, YD09} as twice the optimal rate of quantum communication needed for a sender to transfer one share of a tripartite state to a receiver (generally shared entanglement is required for this task). Here, the decoder at the receiving end of this protocol has the role of a recovery channel, an interpretation later used in \cite{BHOS15}. 
  
Defining the regularized relative entropy of recovery as \cite{BHOS15}
\begin{equation}
D^{\infty}(A;B|C)_{\rho} \equiv  \lim_{n\rightarrow\infty}\frac{1}{n}D(A^n;B^n|C^n)_{\rho^{\otimes n}}
\end{equation}
and the measured relative entropy of recovery
\begin{equation}
D_M(A;B|C)_{\rho} \equiv \inf_{\mathcal{R}_{C\rightarrow AC}}D_M(\rho_{ABC}\Vert\mathcal{R}_{C\rightarrow AC}(\rho_{BC})).
\end{equation}
It was previously shown in \cite{BHOS15} that
\begin{equation}\label{recIneqChain}
I(A;B|C)_\rho\geq D^\infty(A;B|C)_\rho \geq D_M(A;B|C)_{\rho} \geq  - 2 \log F(A;B|C)_{\rho}.
\end{equation}
All the lower bounds on the CQMI are clearly well motivated measures of recovery and non-Markovianity, but hitherto they have been lacking concrete operational interpretations. We will address this point  in a later section with particular focus on the regularized relative entropy of recovery. As a remark, at the end of the chapter we will also briefly discuss an operational interpretation for the fidelity of recovery in complexity theory. 

A wave of recent work \cite{BCY11,Winterconj,K13conj,Z14,BSW14,SBW14,LW14,BLW14,DW15,BT15,SOR15,W15,DW15a,STH16,JRSWW15} on this topic has added to the results from \cite{FR14}, solidifying what appears to be the right notion of quantum Markovianity. 
An important focus of these subsequent improvements has been the structure of the optimal recovery map. A natural conjecture would be that the Petz recovery map, as defined in Equation \eqref{defPetzMap}, can be used instead of the optimization, since that is true in the case where the CQMI is zero. In this exact form it still remains an open problem whether that is actually true for any of the inequalities in Equation \eqref{recIneqChain}. Nevertheless, significant progress has been made. 
Two desirable features for the lower bounds are that the recovery map should be explicit (meaning that no optimization is needed)  and that it should be universal (when recovering system $A$ from system $C$, the map should not depend on system $B$). A particular bound including both of these features has been recently given in~\cite[Thm.~4.1]{SBT16} and is based on the measured relative entropy
\begin{align}\label{eq:cqmi_meas__lowerbound}
I(A:B|C)_\rho\geq  D_M\left(\rho_{ABC}\middle\|\int\beta_0(t)\;\mathrm{d}t\left(\cI_A\otimes\cR^{[t]}_{C\to BC}(\rho_{AC})\right)\right),
\end{align}
for a particular universal probability distribution $\beta_0(t)$ and the rotated Petz recovery maps $\cR^{[t]}_{C\to BC}$ (for the precise definitions see Theorem~\ref{thm:cqmi} in the next section and for a comparison to other bounds see Corollary \ref{cor:cqmi}).

Another natural conjecture that was recently disproved was that we might be able to avoid the regularization in the first lower bound in Equation \eqref{recIneqChain} and instead use the relative entropy of recovery itself as defined in Equation \eqref{REORdef}. This has been recently proven false by a counterexample in~\cite{FF17}. Later in this chapter, we will review the given counterexamples and provide additional ones which show that even in a special case where some of the systems are classical, the conjectured bound does not hold. This will also have direct implications for the topic of Chapter~\ref{infoCombining}. 

In Section \ref{ImprovedLowerBound} we show how to transfer these properties to the regularized relative entropy bound. Then, in Section~\ref{RecOpInt} we will discuss operational interpretations of the recoverability quantities. In Section~\ref{CounterExample} we will have a closer look at the previously mentioned counterexamples. Finally in Section~\ref{RegNeeded} we will show that the composite hypothesis testing problem from the first part of this thesis indeed requires a regularization by connecting it to recoverability.


\section{An improved Markov type lower bound}\label{ImprovedLowerBound}

In this section we apply the techniques developed in the first part of the thesis to strengthen the previously best known quantum relative entropy Markov type lower bound on the conditional quantum mutual information $I(A:B|C)_\rho$ (see Equation~\eqref{eq:cqmi_meas__lowerbound} and~\cite{SBT16}). We find that
\begin{align}\label{eq:cqmi_lowerbound}
I(A:B|C)_\rho\geq\lim_{n\to\infty}\frac{1}{n}D\left(\rho_{ABC}^{\otimes n}\middle\|\int\beta_0(t)\;\mathrm{d}t\left(\cI_A\otimes\cR^{[t]}_{C\to BC}(\rho_{AC})\right)^{\otimes n}\right)
\end{align}
for a particular universal probability distribution $\beta_0(t)$ and the rotated Petz recovery maps $\cR^{[t]}_{C\to BC}$. In contrast to the previously known bounds in terms of quantum relative entropy distance~\cite{STH16,BHOS15}, the recovery map in Equation~\eqref{eq:cqmi_lowerbound} takes a specific form depending only on the reduced state on the systems  $BC$. Note that the regularization in Equation~\eqref{eq:cqmi_lowerbound} cannot go away, using the relative entropy distance, as recently shown in~\cite{FF17}. 

We will now give a formal statement and proof of the lower bound on the conditional quantum mutual information from Equation~\eqref{eq:cqmi_lowerbound}. Then we give a detailed overview on how all known Markov type lower bounds on the conditional quantum mutual information compare and we present the argument that Equation~\eqref{eq:cqmi_lowerbound} represents the last possible strengthening.

\begin{thm}\label{thm:cqmi}
For $\rho_{ABC}\in \setS{\cH_{ABC}}$ we have
\begin{align}\label{eq:cqmi}
I(A:B|C)_\rho\geq\limsup_{n\to\infty}\frac{1}{n}D\left(\rho_{ABC}^{\otimes n}\middle\|\int\beta_0(t)\left(\cI_A\otimes\cR^{[t]}_{C\to BC}(\rho_{AC})\right)^{\otimes n}\mathrm{d}t\right)\,,
\end{align}
where $\beta_0(t):=\frac{\pi}{2}\left(\cosh(\pi t)+1\right)^{-1}$ \\
and $\cR^{[t]}_{C\to BC}(\cdot):=\rho_{BC}^{\frac{1+it}{2}}\left(\rho_C^{\frac{-1-it}{2}}(\cdot)\rho_C^{\frac{-1+it}{2}}\right)\rho_{BC}^{\frac{1-it}{2}}$.
\end{thm}

\begin{proof}
We start from the lower bound~\cite[Thm.~4.1]{SBT16} (see Equation \eqref{eq:cqmi_meas__lowerbound}) applied to $\rho_{ABC}^{\otimes n}$ (with the support conditions taken care of as in the corresponding proof)
\begin{align}
\text{$I(A:B|C)_\rho\geq\frac{1}{n}D_{M}\left(\rho_{ABC}^{\otimes n}\middle\|\sigma_{A^nB^nC^n}\right)$}
\end{align}
with
\begin{align}
\text{$\sigma_{A^nB^nC^n}:=\int\beta_0(t)\left(\sigma_{ABC}^{[t]}\right)^{\otimes n}\mathrm{d}t$ and $\sigma_{ABC}^{[t]}:=\left(\cI_A\otimes R^{[t]}_{C\to BC}\right)(\rho_{AC})$,}
\end{align}
where we have used that the conditional quantum mutual information is additive on tensor product states. Now, we simply observe that $\sigma_{A^nB^nC^n}$ is permutation invariant and hence the claim can be deduced from Lemma~\ref{pinching} together with the limit $n\to\infty$.
\end{proof}

In the proof, we have mainly used a Lemma previously proven in order to help us investigate asymptotic hypothesis testing, which might serve as a first hint that both topics are closely related at least on a technical basis. 
Now, together with previous work, we find the following corollary that encompasses all known Markov type lower bounds on the conditional quantum mutual information.

\begin{cor}\label{cor:cqmi}
For $\rho_{ABC}\in \setS{\cH_{ABC}}$ the conditional quantum mutual information $I(A:B|C)_\rho$ is lower bounded by the three incomparable bounds
\begin{align}
-\int\beta_0(t)\log F\left( \rho_{ABC}, \sigma_{ABC}^{[t]}\right)^2\;\mathrm{d}t, \nonumber\\ 
D_{\mathcal{M}}\left(\rho_{ABC}\middle\|\int\beta_0(t)\sigma_{ABC}^{[t]}\;\mathrm{d}t\right),\; \label{incomp}\\ 
\limsup_{n\to\infty}\frac{1}{n}D\left(\rho_{ABC}^{\otimes n}\middle\|\int\beta_0(t)\left(\sigma_{ABC}^{[t]}\right)^{\otimes n}\mathrm{d}t\right)\,. \nonumber
\end{align}
In contrast to the second and third bound, the first lower bound is not tight in the commutative case but has the advantage that the average over $\beta_0(t)$ stands outside of the distance measure used. All the lower bounds are typically strict -- whereas in the commutative case the second and third bound both become equalities.
\end{cor}

\begin{proof}
The first bound was shown in~\cite[Sect.~3]{JRSWW15}, the second one in~\cite[Thm.~4.1]{SBT16}, and the third one is Theorem~\ref{thm:cqmi}. To see that the bounds are incomparable, notice that the distribution $\beta_0(t)$ cannot be taken outside the relative entropy measure in the second and the third bound since the quantum Stein's lemma would then lead to a contradiction to a recent counterexample from~\cite[Sect.~5]{FF17} (for more details see also Section~\ref{CounterExample}). 
The fact that the lower bounds are typically strict can be seen from numerical work (see, e.g., \cite{BHOS15}).
\end{proof}

It seems that the only remaining conjectured strengthening that is not known to be wrong is the first lower bound in Equation \eqref{incomp} in terms of the non-rotated Petz map~\cite[Sect.~8]{BSW14}
\begin{align}
I(A:B|C)_\rho\geq-\log F\left(\rho_{ABC},\sigma_{ABC}^{[0]}\right)^2\,. \label{conjFidelity}
\end{align}
We refer to~\cite{Lemm17} for the latest progress in that direction. All the same arguments as in the proof of Theorem~\ref{thm:cqmi}, can also be applied to lift the strengthened monotonicity of the relative entropy from~\cite[Cor.~4.2]{SBT16}. For $\rho\in \setS{\cH}$, $\sigma$ a positive semi-definite operator on $\cH$, and $\mathcal{N}$ a completely positive trace preserving map on the same space, this leads to
\begin{align}
D(\rho\|\sigma)-D(\mathcal{N}(\rho)\|\mathcal{N}(\sigma))\geq\limsup_{n\to\infty}\frac{1}{n}D\left(\rho^{\otimes n}\middle\|\int\beta_0(t)\left(\cR^{[t]}_{\sigma,\mathcal{N}}(\rho)\right)^{\otimes n}\mathrm{d}t\right),
\end{align}
where $\cR^{[t]}_{\sigma,\mathcal{N}}(\cdot):=\sigma^{\frac{1+it}{2}}\mathcal{N}^{\dagger}\left(\mathcal{N}(\sigma)^{\frac{-1-it}{2}}(\cdot)\mathcal{N}(\sigma)^{\frac{-1+it}{2}}\right)\sigma^{\frac{1-it}{2}}$. Together with~\cite[Sect.~3]{JRSWW15} and~\cite[Cor.~4.2]{SBT16} we then again have three incomparable lower bounds as in Corollary~\ref{cor:cqmi}.


\section{Recoverability quantities from hypothesis testing $\qquad$}\label{RecOpInt}

In this section, we will discuss how hypothesis testing gives an operational interpretation to the relative entropy of recovery. 
It follows from the concerns in recovery applications that one may have to systematically decide whether a given tripartite quantum state is recoverable or not. In this paper, we discuss two concrete scenarios in which this is the case. Both involve many copies of the state $\rho_{ABC}$---for both settings, the goal is to decide whether a given tripartite state is recoverable.

We will give two different discrimination problems and we will show, that for both problems the regularized relative entropy of recovery gives the optimal error rate. In the first part, we discriminate against states retrieved via a global recovery map and we show that it fulfills the requirements of a result by Brandao et al. (see Theorem~\ref{CompBrandao}) which leads to the desired result. In the second part we discriminate against convex combination of locally recovered states, showing the result via our composite Stein's Lemma in Section~\ref{CompStein}. 

\paragraph{Global recovery map}

Suppose either the state $\rho_{ABC}^{\otimes n}$ or the state $\mathcal{R}_{C^n \to A^n C^n}(\rho_{BC}^{\otimes n})$, where $\mathcal{R}_{C^n \to A^n C^n}$ is some arbitrary collective recovery channel acting on all $n$ of the $C$ systems, is prepared. The goal is then to determine which state has been prepared by performing a collective measurement on all of the systems $A^n B^n C^n$. This gives us a hypothesis testing scenario, for which we prove that
$D^{\infty}(A;B|C)_{\rho}$ is equal to  
the optimal exponent for the Type II error if the Type I error is constrained to be no larger than a constant $\varepsilon \in (0,1)$, that is the Stein's Lemma kind of setting introduced in Chapter \ref{hypoTesting}. Thus, our result establishes a concrete operational interpretation of the regularized relative entropy of recovery in this hypothesis testing experiment.

Since one of the states is not fixed, but allows for an arbitrary recovery map to be applied before measuring, this is an instance of a general composite hypothesis testing problem of discriminating between a state $\rho^{\otimes n}$ and a set $\mathcal{S}^{\left( n\right)  }$ of states, where in our case:
\hyptest{\begin{description}
\item[Null hypothesis:] the fixed state $\rho^{\otimes n}    =\rho_{ABC}^{\otimes n}$
\item[Alternative hypothesis:] the set \\ $\mathcal{S}^{\left(  n\right)  }  =\left\{  \mathcal{R}_{C^{n}\rightarrow A^{n}C^{n}}(  \rho_{BC}^{\otimes n})  :\mathcal{R}\in\text{CPTP}\right\}$.
\end{description}}{} 
with CPTP denoting the set of quantum channels from $C^n$ to $A^n C^n$.
To handle this composite discrimination setting we use Theorem~\ref{CompBrandao} in the beginning of Chapter~\ref{CompHypo} (originally from \cite{BP10}).  For the readers convenience we again state the necessary conditions on the alternative hypothesis here:

\begin{enumerate}
\item Convexity -- $\mathcal{S}^{\left(  n\right)  }$ is convex and closed for all $n$.

\item Full Rank -- There exists a full rank state $\sigma$ such that each $\mathcal{S}^{\left(  n\right)  }$ contains $\sigma^{\otimes n}$.

\item Reduction -- For each $\sigma\in\mathcal{S}^{\left(  n\right)  }$, Tr$_{n}\{\sigma\}  \in\mathcal{S}^{\left(  n-1\right)  }$.

\item Concatenation -- If $\sigma_{n}\in\mathcal{S}^{\left(  n\right)  }$ and $\sigma_{m}\in\mathcal{S}^{\left(  m\right)  }$, then $\sigma_{n}\otimes\sigma_{m} \in\mathcal{S}^{\left(  n+m\right)  }$.

\item Permutation invariance -- $\mathcal{S}^{\left(  n\right)  }$ is closed under permutations.
\end{enumerate}

We now verify that the set $\mathcal{S}^{\left(  n\right)  }$ as defined above, satisfies the necessary properties.

\textbf{Convexity}. Let $\mathcal{R}_{C^{n}\rightarrow A^{n}C^{n}}^{1}(  \rho _{BC}^{\otimes n})  ,\mathcal{R}_{C^{n}\rightarrow A^{n}C^{n}}%
^{2}(  \rho_{BC}^{\otimes n})  \in\mathcal{S}^{\left(  n\right)  } $. Then for all $\lambda\in\left[  0,1\right]  $, we have
\begin{equation*}
\lambda\mathcal{R}_{C^{n}\rightarrow A^{n}C^{n}}^{1}(  \rho_{BC}^{\otimes n})  +\left(  1-\lambda\right)  \mathcal{R}_{C^{n}\rightarrow A^{n}%
C^{n}}^{2}(  \rho_{BC}^{\otimes n})  \in\mathcal{S}^{\left( n\right)  }\, ,
\end{equation*}
because $\lambda\mathcal{R}_{C^{n}\rightarrow A^{n}C^{n}}^{1}+\left(1-\lambda\right)  \mathcal{R}_{C^{n}\rightarrow A^{n}C^{n}}^{2}$ is a quantum channel if
$\mathcal{R}_{C^{n}\rightarrow A^{n}C^{n}}^{1}$ and $\mathcal{R}_{C^{n}\rightarrow A^{n}C^{n}}
^{2}$ are. Furthermore, the set of all CPTP\ maps is closed. \\

\textbf{Full Rank}. Without loss of generality, we can assume that $\rho_{B}$ is a full rank state. A particular recovery channel is one which traces out system $C$ and replaces it with the maximally mixed state on $AC$. Taking $n$ copies of such a state gives a full-rank state in $\mathcal{S}^{\left(  n\right)  }$. \\

\textbf{Reduction}. Let $\mathcal{R}_{C^{n}\rightarrow A^{n}C^{n}}(  \rho_{BC}^{\otimes n})  \in\mathcal{S}^{\left(  n\right)  }$. Consider that%
\begin{align}
\tr_{A_{n}B_{n}C_{n}}\{  \mathcal{R}_{C^{n}\rightarrow A^{n}C^{n}%
}(  \rho_{BC}^{\otimes n})  \}  =  \text{Tr}_{A_{n}C_{n}%
}\{  \mathcal{R}_{C^{n}\rightarrow A^{n}C^{n}}(  \rho_{BC}^{\otimes n-1}\otimes\rho_{C})  \}  .
\end{align}
This state is in $\mathcal{S}^{\left(  n\right)  }$ because the recovery channel for $\rho_{BC}^{\otimes n-1}$ could consist of tensoring in $\rho_{C}$, applying $\mathcal{R}_{C^{n}\rightarrow A^{n}C^{n}}$, and tracing out systems $A_{n}C_{n}$. \\

\textbf{Concatenation}. Let $\mathcal{R}_{C^{n}\rightarrow A^{n}C^{n}}^{1}(  \rho _{BC}^{\otimes n})  \in\mathcal{S}^{\left(  n\right)  }$ and $\mathcal{R}_{C^{m}\rightarrow A^{m}C^{m}}^{2}(  \rho_{BC}^{\otimes m})  \in\mathcal{S}^{\left(  m\right)  }$. Then
\begin{equation}
\mathcal{R}%
_{C^{n}\rightarrow A^{n}C^{n}}^{1}(  \rho_{BC}^{\otimes n}) \otimes\mathcal{R}_{C^{m}\rightarrow A^{m}C^{m}}^{2}(  \rho_{BC}^{\otimes m})  \in\mathcal{S}^{\left(  n+m\right)  },
\end{equation}
 because
\begin{multline}
\mathcal{R}_{C^{n}\rightarrow A^{n}C^{n}}^{1}(  \rho_{BC}^{\otimes n})  \otimes\mathcal{R}_{C^{m}\rightarrow A^{m}C^{m}}^{2}(
\rho_{BC}^{\otimes m})  = \\ \left(  \mathcal{R}_{C^{n}\rightarrow A^{n}C^{n}}^{1}\otimes\mathcal{R}_{C^{m}\rightarrow A^{m}C^{m}}^{2}\right)
(  \rho_{BC}^{\otimes n+m})  ,
\end{multline}
so that the recovery channel consists of the parallel concatenation of
$\mathcal{R}_{C^{n}\rightarrow A^{n}C^{n}}^{1}$ and $\mathcal{R}%
_{C^{m}\rightarrow A^{m}C^{m}}^{2}$. \\

\textbf{Permutation invariance}. Here, we need to show that for $\sigma\in\mathcal{S}^{\left(  n\right) }$, we have $\pi\sigma\pi^{\dag}\in\mathcal{S}^{\left(  n\right)  }$ for
all permutations $\pi$ of the $n$ systems. Let $\mathcal{R}_{C^{n}\rightarrow A^{n}C^{n}}(  \rho_{BC}^{\otimes n})  \in\mathcal{S}^{\left(
n\right)  }$. Then
\begin{align}
&  \pi_{A^{n}B^{n}C^{n}}\mathcal{R}_{C^{n}\rightarrow A^{n}C^{n}}( \rho_{BC}^{\otimes n})  \left(  \pi_{A^{n}B^{n}C^{n}}\right)  ^{\dag
}\nonumber\\
&  =\left(  \pi_{A^{n}}\otimes\pi_{B^{n}}\otimes\pi_{C^{n}}\right) \mathcal{R}
(  \rho_{BC}^{\otimes n})
\left(  \pi_{A^{n}}\otimes\pi_{B^{n}}\otimes\pi_{C^{n}}\right)  ^{\dag}\nonumber\\
&  =\left(  \pi_{A^{n}}\otimes\pi_{C^{n}}\right)  \mathcal{R}
(  \pi_{B^{n}}\rho_{BC}^{\otimes n}\pi_{B^{n}%
}^{\dag})  \left(  \pi_{A^{n}}\otimes\pi_{C^{n}}\right)  ^{\dag}\nonumber\\
&  =\left(  \pi_{A^{n}}\otimes\pi_{C^{n}}\right)  \left[  \mathcal{R}%
(  \pi_{C^{n}}^{\dag}\rho_{BC}^{\otimes n}\pi_{C^{n}})  \right]  \left(  \pi_{A^{n}}\otimes\pi_{C^{n}}\right)
^{\dag}  \nonumber\\
& \in\mathcal{S}^{\left(  n\right)  },
\end{align}
where the second equality follows because the permutation of the $B$ systems commutes with the recovery channel. The third equality follows because $\rho_{BC}^{\otimes n}$ is a permutation invariant state, and the last line because a potential recovery consists of applying the permutation $\pi_{C^{n}}^{\dag}$, followed by $\mathcal{R}_{C^{n}%
\rightarrow A^{n}C^{n}}$, followed by the permutation $\pi_{A^{n}}\otimes \pi_{C^{n}}$. \\

By employing Theorem~\ref{CompBrandao} and the above observations, we can conclude that
\begin{equation}
\zeta_{\rho,\mathcal{S}^{\left( n\right) }}(0)=D^{\infty}(  A;B|C)_{\rho}, \label{eq:key-identity}%
\end{equation}
for all $\varepsilon\in\left(  0,1\right)  $.
As claimed, this gives an operational interpretation of $D^{\infty}(  A;B|C)  _{\rho}$ as the optimal Type II error exponent in a composite asymmetric state discrimination setting.

\bonus{

\begin{overpic}[clip, scale=0.23]{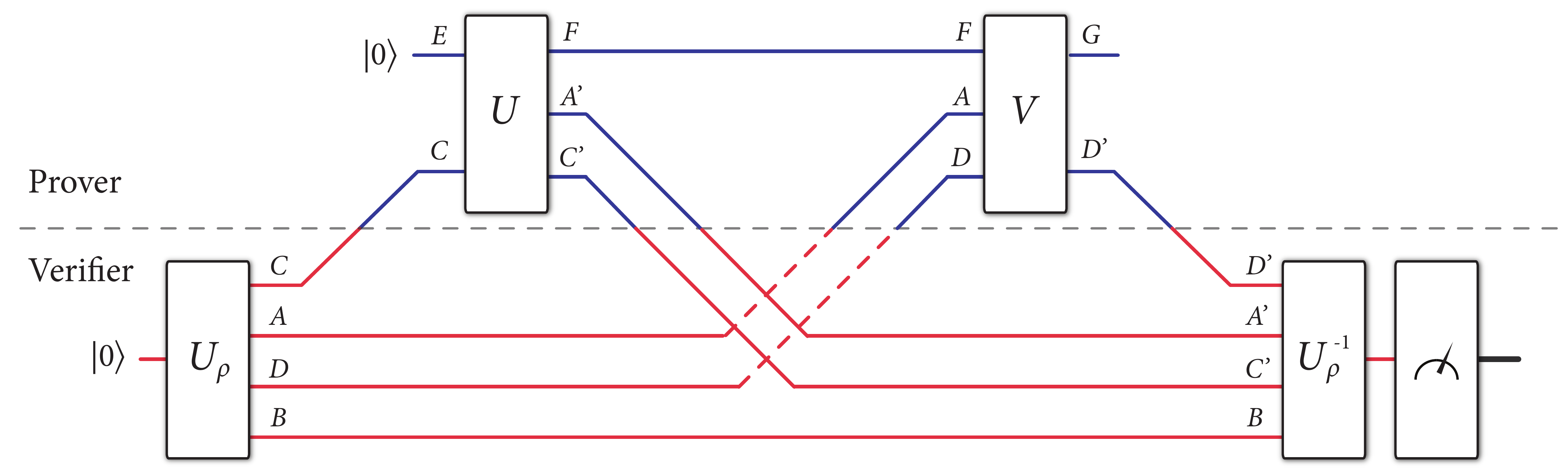}   
\end{overpic}
\captionof{figure}{\label{fig:for}The quantum interactive proof system showing that \textsc{FoR} is in \textsc{QIP}.}

In \cite{CHMOSWW16} we also discussed an operational interpretation for the fidelity of recovery. In that scenario we ask: given a description of a quantum circuit that prepares a state $\rho_{ABC}$, what is the maximum probability with which someone could be convinced that the state is recoverable? Also, how difficult is the task of deciding if the state meets some criteria of recoverability when $A$ is lost?  We address these questions by defining the associated decision problem, called \textsc{FoR} for ``fidelity of recovery.'' Using ideas from quantum complexity theory \cite{W09,VW15}, we show that the fidelity of recovery is equal to the maximum probability with which a verifier can be convinced that $\rho_{ABC}$ is recoverable from $\rho_{BC}$ by acting on system $C$ alone. The quantum interactive proof system to establish this operational meaning for the fidelity of recovery
is depicted in Figure~\ref{fig:for} and
follows intuitively from the duality property of the fidelity of recovery, which was originally established in \cite{SW14}.
It also proves that \textsc{FoR} is contained in the complexity class \textsc{QIP} \cite{W09,VW15}.

However, the proof system in Figure~\ref{fig:for} requires the exchange of four messages  between
the verifier and 
the prover. From a computational complexity theoretic perspective, it is desirable to reduce the number of messages exchanged. In fact, this is certainly possible because a general procedure, which reduces any quantum interactive proof system to an equivalent one which has only three messages exchanged, is already known \cite{QIPAmp00}. We contribute a different proof system for \textsc{FoR} which requires
the exchange of only two messages  between the verifier and the prover. The main idea is that the verifier can force the prover to perform his actions in superposition, and 
the result is that the \textsc{FoR} decision problem is in \textsc{QIP}(2).
 We also argue that \textsc{FoR} is hard for \textsc{QSZK} \cite{W02,W09zkqa}, by building on earlier work in \cite{v011a003}.
Note that both \textsc{QSZK} and \textsc{QIP}(2) contain problems believed to be difficult to solve by a quantum computer. 
For details we refer to \cite{CHMOSWW16}. 
}{Operational interpretation for the Fidelity of Recovery}


\paragraph{Local recovery maps}\label{sec:hypothesis_recovery}

In contrast, using our result from Section~\ref{CompStein} we can also cover the following discrimination problem.
\hyptest{\begin{description}
\item[Null hypothesis:] the fixed state $\rho_{ABC}^{\otimes n}$
\item[Alternative hypothesis:] the convex sets of iid states \\ $\bar{\cR}^n:=\left\{\int\left((\cI_A\otimes\cR_{C\to BC})(\rho_{AC})\right)^{\otimes n}\;\mathrm{d}\mu(\cR)\right\}$ with $\cR_{C\to BC}\in\text{CPTP}$.
\end{description}}{}

This gives us the asymptotic error rate 
\begin{align}
\zeta_{\rho,\bar{\cR}^n}(0)=\lim_{n\to\infty} \frac{1}{n} \inf_{\mu\in\cR} D\left(\rho_{ABC}^{\otimes n}\middle\|\int\big((\cI_A\otimes\cR_{C\to BC})(\rho_{AC})\big)^{\otimes n}\;\mathrm{d}\mu(\cR)\right)\,. 
\end{align}
Interestingly, we can show that both rates are identical.

\begin{prop}
For the discrimination problems in this section we have $\zeta_{\rho,\mathcal{S}^{\left( n\right) }}(0)=\zeta_{\rho,\bar{\cR}^n}(0)$.
\end{prop}

\begin{proof}
The support of $\rho_{ABC}$ lies in the support of at least one state in $\bar{\cR}$ iff it does so for $\mathcal{S}$. If this is not the case, $\zeta_{\rho,\mathcal{S}^{\left( n\right) }}(0)$ and $\zeta_{\rho,\bar{\cR}^n}(0)$ evaluate to infinity. Therefore we can restrict the proof to the former case. 

By definition we have $\zeta_{\rho,\mathcal{S}^{\left( n\right) }}(0)\leq\zeta_{\rho,\bar{\cR}^n}(0)$ and for the other direction we use a de Finetti reduction for quantum channels from~\cite[Lem.~8]{BHOS15} (first derived in~\cite{FR14}). Namely, for $\omega_{C^n}\in \setS{\cH_C^{\otimes n}}$ and permutation invariant $\cR_{C^n\to B^nC^n}$ we have
\begin{align}
\cR_{C^n\to B^nC^n}\left(\omega_{C^n}\right)\leq\poly(n)\cdot\int\left(\cR_{C\to BC}\right)^{\otimes n}\left(\omega_{C^n}\right)\mathrm{d}\nu(\cR)
\end{align}
for some measure $\nu(\cR)$ over the completely positive and trace preserving maps on $C\to BC$. As explained in the proof of~\cite[Prop.~9]{BHOS15}, the joint convexity of the quantum relative entropy together with the operator monotonicity of the logarithm then imply that
\begin{align}
&D\left(\rho_{ABC}^{\otimes n}\middle\|\cR_{C^n\to B^nC^n}\left(\rho_{AC}^{\otimes n}\right)\right)\qquad \nonumber\\ 
&\geq D\left(\rho_{ABC}^{\otimes n}\middle\|\int\big((\cI_A\otimes\cR_{C\to BC})(\rho_{AC})\big)^{\otimes n}\;\mathrm{d}\nu(\cR)\right)-\log\poly(n)\,.
\end{align}
By inspection this leads to $\zeta_{\rho,\mathcal{S}^{\left( n\right) }}(0)\geq\zeta_{\rho,\bar{\cR}^n}(0)$ and hence implies the claim.
\end{proof}
We have thus given two a priori different discrimination problems, leading to the same exponential error rate, which is given by the relative entropy of recovery and therefore giving it an operational interpretation via hypothesis testing.


\section{Counterexamples to the relative entropy conjecture}\label{CounterExample}

In the introduction of this chapter, we discussed several lower bounds on the conditional quantum mutual information, including one based on the regularized relative entropy of recovery (see e.g. Equation \eqref{recIneqChain}). A particularly interesting question that was left open for a long time was whether a similar bound could hold without regularizing the relative entropy, until recently counterexamples were provided in \cite{FF17}. We will now briefly review that counterexample: there exists $\theta\in\left[0,\pi/2\right]$ such that
\begin{align}
&I(A:C|B)_\rho\ngeq\inf_{\cR}D\left(\rho_{ABC}\middle\|(\cI_A\otimes\cR_{B\to BC})(\rho_{AC})\right)\;\label{eq:example_fawzi} \end{align}
for the pure state $\rho_{ABC}=|\rho\rangle\langle\rho|_{ABC}$ with 
\begin{align}
|\rho\rangle_{ABC}=&\frac{1}{\sqrt{2}}|0\rangle_A\otimes|0\rangle_B\otimes|0\rangle_C \nonumber\\
&+\frac{1}{\sqrt{2}}\left(\cos(\theta)|0\rangle_A\otimes|1\rangle_C+\sin(\theta)|1\rangle_A\otimes|0\rangle_C\right)\otimes|1\rangle_B.
\end{align}  
While the above example is already surprisingly simple, one might wonder whether the conjecture can also be proven wrong in even simpler settings. We know that when $\rho$ is chosen to be a completely classical state, Equation \eqref{eq:example_fawzi} always holds with equality. Therefore it would be interesting to look at an intermediate case and that is exactly what we will do in the following paragraph.

\begin{figure}[t!]
\centering
\begin{overpic}[trim=1cm 7cm 1cm 6.5cm, clip, scale=0.5]{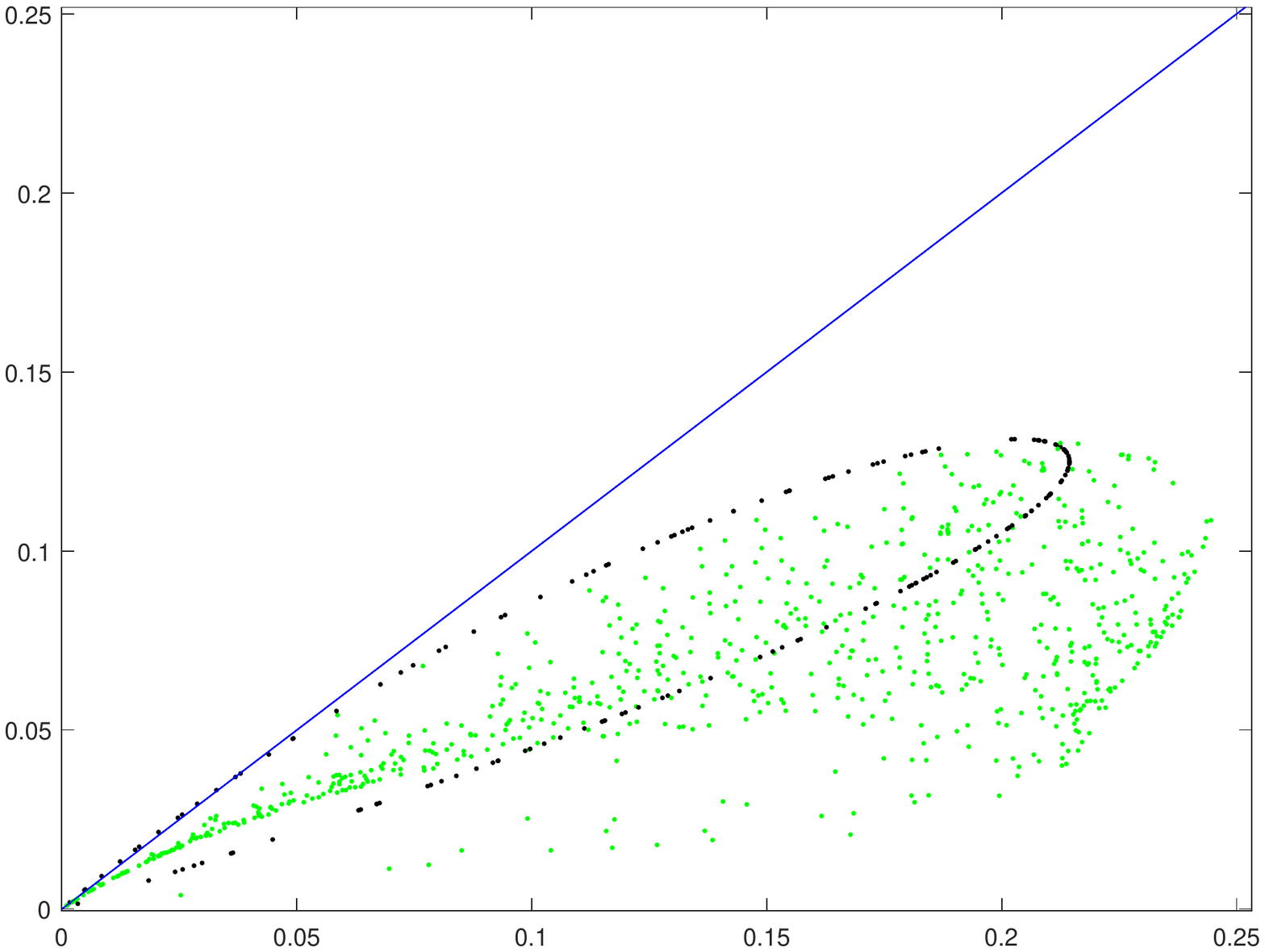}   
\put (-10,81) {$D(A:C|B_1B_2)$}
\put (88,0) {$I(A:C|B_1B_2)$}
\end{overpic}
\caption{
In this figure we plot the conditional quantum mutual information against the relative entropy of recovery for states of the form in Equation~\eqref{ccqTest}. For the green data points the states $\sigma_0$ and $\sigma_1$ are randomly drawn mixed states, while for the black ones we only draw from pure states. The blue line shows where both quantities would be equal. Although the violation is rather small, it can be seen that certain data points are above the blue line and therefore violate the conjectured relative entropy lower bound in the case of classical $A$ and $C$ systems. \label{IvsD}}
\end{figure}

We consider the case where $A$ and $C$ are classical systems and $B$ is quantum. Consider the following state
\begin{align}
\tau_{ACB} &=\frac{1}{4}{\ketbra{}{0}{0}}_A\otimes{\ketbra{}{0}{0}}_C\otimes\sigma_0^{B_1}\otimes\sigma_0^{B_2}+\frac{1}{4}{\ketbra{}{1}{1}}_A\otimes{\ketbra{}{0}{0}}_C\otimes\sigma_1^{B_1}\otimes\sigma_0^{B_2}\nonumber\\
&\quad+\frac{1}{4}{\ketbra{}{1}{1}}_A\otimes{\ketbra{}{1}{1}}_C\otimes\sigma_0^{B_1}\otimes\sigma_1^{B_2}+\frac{1}{4}{\ketbra{}{0}{0}}_A\otimes{\ketbra{}{1}{1}}_C\otimes\sigma_1^{B_1}\otimes\sigma_1^{B_2}. \label{ccqTest}
\end{align}
Note that this state is identical to the state in Equation~\eqref{InfoCombCNOT}, which plays an important role in Chapter~\ref{infoCombining}. We will use the same numerical tools as~\cite{FF17} to investigate this state, namely a recent semidefinite approximation of the matrix logarithm that was put forward in~\cite{FSP17}. A general overview over the the numerics can be seen in Figure~\ref{IvsD}. 

\begin{figure}[t!]
\centering
\begin{overpic}[trim=1cm 9cm 1cm 9cm, clip, scale=0.5]{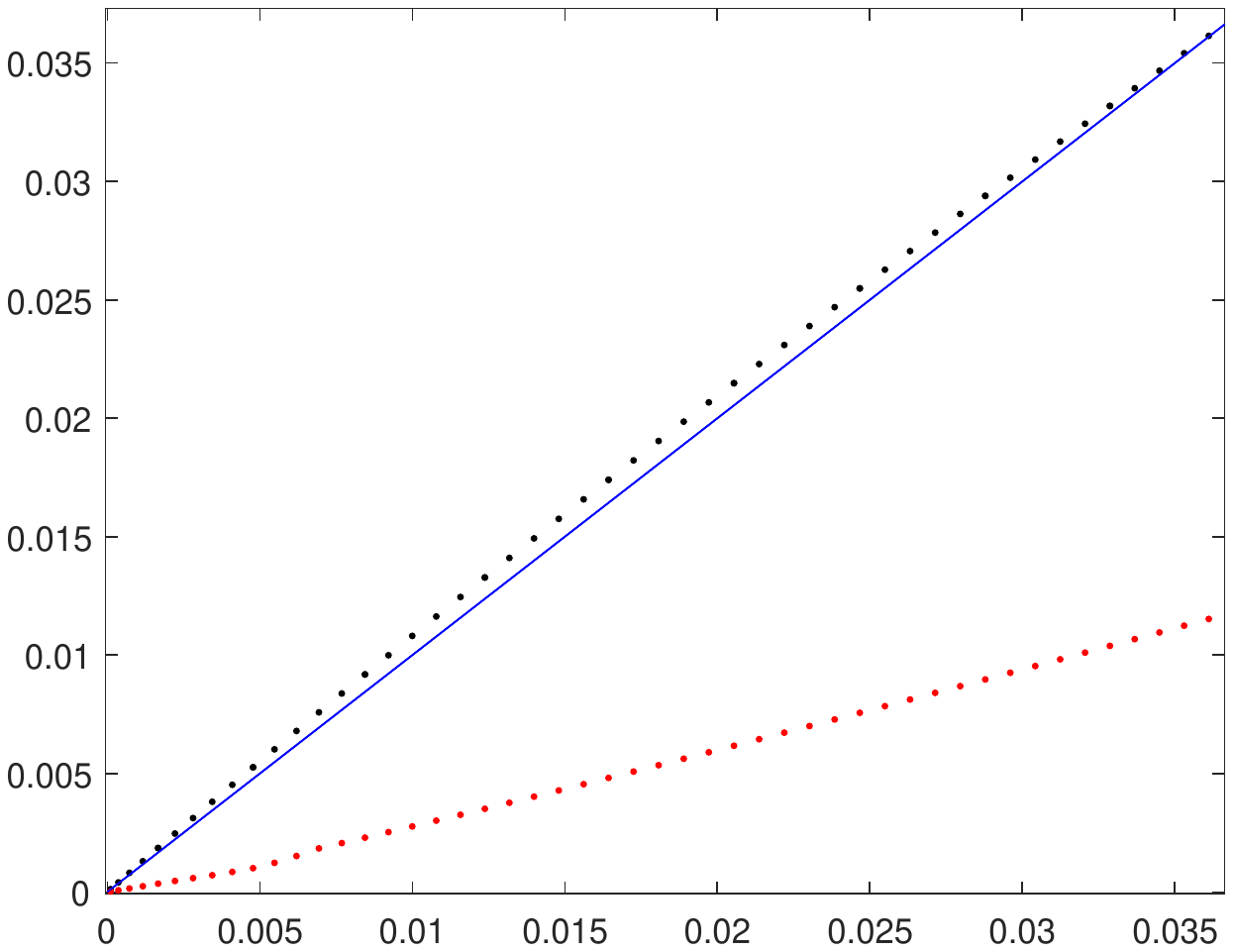}   
\put (-10,60) {$D(A:C|B_1B_2)$ and \textcolor{red}{$D_M(A:C|B_1B_2)$}}
\put (88,0) {$I(A:C|B_1B_2)$}
\end{overpic}
\caption{
In this figure we show a close up on the region of Figure~\ref{IvsD} where the conjecture is violated. The black data points correspond to plotting the relative entropy of recovery against the conditional quantum mutual information for the states defined in Equation~\ref{violationStates}. For comparison we also plot the measured relative entropy of recovery for the same states in red. \label{IvsDclose}}
\end{figure}

Since the violations of the conjecture are more difficult to find than in the general case, we provide a concrete example of a state that has higher relative entropy of recovery than  conditional quantum mutual information. We therefore write $\sigma_0$ and $\sigma_1$ in their corresponding Bloch representation
\begin{align}
\sigma_i = \frac{1}{2}\left( \id + \vec{a_i}\vec{\sigma}\right) = \frac{1}{2}\begin{pmatrix} 1+w_i & u_i-iv_i \\ u_i+iv_i & 1-w_i\end{pmatrix},
\end{align}
where $\vec{a_i} = \left(u_i,v_i,w_i\right)$ is the Bloch vector parameterizing the state $\sigma_i$ and $\vec{\sigma}$ is the vector of Pauli matrices. 
Now, consider the family of pairs of states with 
\begin{align}
\vec{a_0} = \left( -0.9, v_0, 0.4 \right)\;\text{and}\; \vec{a_1} = \left( u_0 + 0.01 x, v_1, w_0 + 0.0001 \right), \label{violationStates}
\end{align}
with $v_i = \sqrt{1 - u^2_i - v^2_i}$ to ensure purity of the state. By careful numerical evaluation of these states we find that at least for all $x\in [1,9]$ the conjecture is violated. In Figure~\ref{IvsDclose} we show the corresponding numerical values for different values of $x$.


\section{Regularization needed a for general composite Stein's lemma}\label{RegNeeded}

We are now in the position to easily show, as promised earlier, that the regularization in the composite Steins lemma in Theorem~\ref{thm:main} is actually necessary. 
Towards that goal we use our bound on the conditional quantum mutual information in Theorem~\ref{thm:cqmi} (see also~\cite{BDKSSS05}). That is, we give a proof for Equation~\eqref{eq:reg_needed}. From Theorem~\ref{thm:cqmi} we have, (alternatively we could use the implicitly stated bound from~\cite[Eq.~38]{BHOS15}.)
\begin{align}
I(A:B|C)_\rho&\geq\limsup_{n\to\infty}\frac{1}{n}D\left(\rho_{ABC}^{\otimes n}\middle\|\int\beta_0(t)\left(\cI_A\otimes\cR^{[t]}_{C\to BC}(\rho_{AC})\right)^{\otimes n}\mathrm{d}t\right)\\
&\geq\limsup_{n\to\infty}\frac{1}{n}\inf_{\mu\in\cR}D\left(\rho_{ABC}^{\otimes n}\middle\|\int\left(\cR_{C\to BC}\left(\rho_{AC}\right)\right)^{\otimes n}\mathrm{d}\mu(\cR)\right)\,.
\end{align}
From the second composite discrimination problem described in Section~\ref{sec:hypothesis_recovery} we see that the latter quantity is equal to the optimal asymptotic error exponent $\bar{\zeta}_{A:B|C}(0)$ for testing $\rho_{ABC}^{\otimes n}$ against $\int\left((\cI_A\otimes\cR_{B\to BC})(\rho_{AC})\right)^{\otimes n}\;\mathrm{d}\mu(\cR)$. Now, if the regularization in the formula for $\bar{\zeta}_{A:B|C}(0)$ would actually not be needed, this would imply that
\begin{align}
I(A:B|C)_\rho\geq\inf_{\cR}D\left(\rho_{ABC}\|(\cI_A\otimes\cR_{B\to BC})(\rho_{AC})\right)\,,
\end{align}
However, this is in contradiction with the counterexample from~\cite[Sect.~5]{FF17} as discussed in Section~\ref{CounterExample}. Hence, we conclude that the regularization for composite convex iid testing is needed in general. \qed


\chapter{Bounds on information combining with quantum side information}\label{infoCombining}

\pgfdeclarelayer{background}
\pgfdeclarelayer{firstbackground}
\pgfdeclarelayer{secondbackground}
\pgfsetlayers{secondbackground,firstbackground,background,main}

\pgfooclass{stamp}{ 
    \method stamp() { 
    }
    \method apply(#1,#2,#3) { 
	\draw (#1+2,#2) -- (#1,#2) -- (#1,#2+1) -- (#1+2,#2+1) -- cycle;
          \node[font=\tiny] at (#1+1,#2+0.5) {#3};
    }
    \method box(#1,#2,#3,#4,#5) { 
	\filldraw[fill=#5] (#1+#3,#2) -- (#1,#2) -- (#1,#2+#4) -- (#1+#3,#2+#4) -- cycle;
   }
 \method cnot(#1,#2,#3,#4) { 
	\draw (#1,#2) -- (#1-#3,#2) -- (#1-#3,#2-#4) -- (#1,#2-#4)-- (#1-#3-2,#2-#4);
	\draw (#1-#3,#2) circle (0.2);
	\draw (#1-#3,#2+0.2) -- (#1-#3,#2) -- (#1-#3-0.25,#2);
	\draw (#1-#3-0.25,#2) -- (#1-#3-2,#2);
   }
}
\pgfoonew \mystamp=new stamp()

\pgfooclass{block}{ 
    \method block() { 
    }
    \method basic(#1,#2,#3) { 
	\mystamp.apply(#1,#2,#3)
	\draw (#1-1,#2+0.5) -- (#1,#2+0.5);
	\draw (#1+2,#2+0.5) -- (#1+3,#2+0.5);
    }
}

\pgfoonew \myblock=new block()

In this chapter, we will move to the seemingly different topic of information combining, while keeping our focus on entropy inequalities. Nevertheless, we will be able to utilize the topics discussed in the previous chapters to aid us with the problems we encounter here; in particular, the recoverability lower bounds on the conditional quantum mutual information will play an important role in what follows. Let us start by introducing the setting of this chapter. 
Many of the tasks in classical and quantum information theory are concerned with the evolution of random variables and their corresponding entropies under certain ``combining operations''. 
A particularly elementary example is the addition of two independent classical random variables. In this case the entropy can be easily computed since we know that the addition of two random variables has a probability distribution which corresponds to the convolution of the probability distributions of the individual random variables. 
The picture changes when we have random variables \emph{with} side information. Now, we are interested in the entropy of the sum conditioned on all the available side information. Evaluating this is substantially more difficult, already in the case of classical side information. 

The field of \textit{bounds on information combining} is concerned with finding optimal entropic bounds on the conditional entropy in ``information combining'' scenarios such as this. 
A particular basic setting is that of binary random variables for which an optimal lower bound, the well known \textit{Mrs. Gerber's Lemma}, was given by Wyner and Ziv in \cite{WZ73}. This bound immediately found many applications (see e.g. \cite{GKbook}). 

Following these results, additional approaches to the problem have been found which also led to an \emph{upper} bound on the conditional entropy of the combined random variables. One method of proof and several additional applications can be found in \cite{RU08}, along with the optimal upper bound. 

However, we are interested in above setting, but with \emph{quantum} -- rather than classical -- side information. Unfortunately, it turns out that none of the classical proof techniques apply in this quantum setting, since conditioning on quantum side information does not generally correspond to a convex combination over unconditional situations. We will review the classical proofs in Section~\ref{clInfoComb}. In the following we are concerned with investigating the optimal entropic bounds under quantum side information and report partial progress along with some conjectures. 

An alternative way of looking at the problem is by associating the random variables along with the side information to channels, where the random variable models the input of the channel leading to a known output given by the side information. This analogy is especially useful when investigating coding problems for classical channels. Recently, Arikan~\cite{A09} introduced the first example of constructive capacity achieving codes with efficient encoding and decoding, called polar codes. The elementary idea of polar codes is to combine two channels by a CNOT gate at their (classical) input, which means that the input of the second channel gets added to the input of the first channel. This adds noise on the first input, but provides assistance when decoding the second channel. 
To evaluate the performance of these codes, the Mrs.\ Gerber's Lemma provides an essential tool to tracking the evolution of the entropy through the coding steps (see e.g.\ \cite{AT14,GX15}, which we will build on below).
Following their introduction in the classical setting, polar codes have been generalized to classical-quantum channels~\cite{WG13}. In this chapter we show that finding good bounds on information combining with quantum side information can therefore also be very useful for proving important properties of classical-quantum polar codes.

As the main result of this chapter, we provide a lower bound on the conditional entropy of added random variables with quantum side information in Section~\ref{lower-bound}, using novel bounds on the concavity of the von Neumann entropy (see the details in Appendix~\ref{BoundsOnConc}), the improvements of strong subadditivity by Fawzi and Renner~\cite{FR14} discussed in Chapter~\ref{recoverability}, and results on channel duality by Renes \etal~\cite{RSH14, R17}. Furthermore, we will provide conjectures on the optimal inequalities (upper and lower bounds) in the quantum case. Finally, we discuss applications of our technical results to other problems in information theory and coding; in particular, we show how to use our results to prove sub-exponential convergence of classical-quantum polar codes to capacity, and that polarization takes place even for non-stationary classical-quantum channels. 
But before we start, we briefly discuss the relation of our problem to the well known entropy power inequalities. 

\paragraph{Entropy power inequalities}\label{EPI}
Bounds on information combining are generalizations of a family of entropic inequalities that are called \emph{entropy power inequalities} (for historic reasons). The first and paradigmatic of these inequalities was suggested by Shannon in the second part of his original paper on information theory~\cite{S48}, stating that
\begin{align}\label{shannon-epi}
 e^{2h(X_1)/n}+e^{2h(X_2)/n}\leq e^{2h(X_1+X_2)/n},
\end{align}
where $X_1$ and $X_2$ are random variables with values in ${\mathbb R}^n$ and $h(X) \coloneqq -\int d^n x\, p_A(x) \ln p_A(x)$ denotes the differential entropy (each of the three terms in Equation \eqref{shannon-epi} is called the \emph{entropy power} of the respective random variable $X_1$, $X_2$, and $X_1+X_2$); rigorous proofs followed later~\cite{S59}. We will give a detailed introduction to the special case of entropies of gaussian random variables in Chapter \ref{Covariance}. Clearly, the inequality in \eqref{shannon-epi} gives a lower bound on the entropy $h(X_1+X_2)$ of the sum $X_1+X_2$ given the individual entropies $h(X_1),h(X_2)$, and it is easy to see that the bound is tight (namely, for Gaussian $X_1,X_2$).

Similar lower bounds on the entropy of a sum of two (or more) random variables with values in a group $(G,+)$ have also been termed entropy power inequalities, see e.g.\ \cite{SW90}. For the simplest group $G={\mathbb Z}_2$, the optimal lower bound follows from a famous theorem in information theory, called \emph{Mrs.\ Gerber's Lemma} \cite{WZ73}, which we will describe below in more detail. For the group $G={\mathbb Z}$ of integers, entropy power inequalities in the form of lower bounds on the entropy have emerged~\cite{HAT14, Tao10} after a combinatorial version of the question had been investigated in the field of arithmetic (in particular, additive) combinatorics for a long time.

Most of the above entropy power inequality-like lower bounds remain valid when classical side information $Y_i$ is available for each of the random variables $X_i$, so that for example the entropic terms in \eqref{shannon-epi} are replaced by $h(X_1|Y_1)$, $h(X_2|Y_2)$, and $h(X_1+X_2|Y_1Y_2)$, respectively. This is due to typical convexity properties of these lower bounds along with a representation of the conditional Shannon entropy as a convex combination of unconditional entropies (see our description of the classical conditional Mrs.\ Gerber's Lemma in Section~\ref{clInfoComb}).

Entropy power inequalities have recently been investigated in the quantum setting \cite{KS14, PMG14,ADO15}, with the action of addition replaced by some quantum combining operation, such as a beamsplitter operation on continuous-variable states or a partial swap. These inequalities also hold under conditioning on \emph{classical} side information.
 
However, when the side information is \emph{quantum} in nature, i.e.\ each $(X_i,Y_i)$ is a classical-quantum state \cite{book2000mikeandike} (for the classical entropy power inequalities) or a fully quantum state (for the quantum entropy power inequalities), the proofs do not go though in the same way anymore. Actually, as we will see in the remainder of this chapter, the inequalities that hold under classical side information can sometimes be violated in the presence of quantum side information.

The only lower bounds available under quantum side information so far can be found in~\cite{K15, dPT17}, where for (Gaussian) quantum states an entropic lower bound was proven for the beamsplitter interaction. No general results for all classical-quantum states have been obtained so far.

In light of these developments, our contribution can be seen as the natural entry point into investigating the influence of \emph{quantum} side information in entropy power inequalities and information combining: For the ``information part'' we concentrate on the simplest scenario, namely \emph{classical} random variables $X_i$ that are \emph{binary-valued}, i.e.\ valued in the simplest non-trivial group $({\mathbb Z}_2,+)$. For the side information $Y_i$, however, we allow any general quantum system and states. Our question, therefore, highlights the added difficulties coming from the quantum nature of side information.

We will now start the main part of this chapter by reviewing the classical bounds on information combining. 

\section{Bounds on information combining in classical information theory}\label{clInfoComb}

For the sake of better understanding the goals and general problems that come with the task of finding inequalities in the quantum setting, we will first have a closer look at the optimal classical inequalities. For a classical binary random variable $X_i$, we can associate a probability distribution $p_i$, for which then $H(X_i) = h_2(p_i)$. Now, it is well known that when we sum two random variables, the corresponding probability distribution is the convolution of the original probability distributions. The binary convolution is defined for binary probability distributions $\{a,1-a\}$ and $\{b,1-b\}$ as $a\ast b := a(1-b) + (1-a)b$. It easily follows that
\begin{align*}
H(X_1+X_2) &= h_2( p_1 \ast p_2) \\ 
&= h_2(h_2^{-1}(H(X_1))\ast h_2^{-1}(H(X_2))).
\end{align*}

In classical information theory the topic of bounds on information combining describes a number of results concerned with what happens, in particular to the entropy of the involved objects, when random variables get combined. This is especially interesting when we have side information for these random variables, due to the analogy with channel problems.
The name of this field goes back to~\cite{LHHH05} where such bounds were used for repetition codes. Later on, many more results were found, also under the name \textit{Extremes of information combining}~\cite{SSZ05}, for MAP decoding and LDPC codes. 

Examples of particular importance are the combinations at the variable and check nodes in belief propagation~\cite{RU08}, and the transformation to \textit{better} and \textit{worse} channels in polar coding~\cite{A09}. 
In the first setting we are concerned with the entropy of the sum $X_1+X_2$ given the side information $Y_1Y_2$, which corresponds to check nodes in belief propagation and the worse channel in polar coding (see Figure~\ref{Fig:Basic}). In the channel picture this can be seen as channel combination
\begin{equation}\label{cl:minus}
(\mathcal W_1\boxast \mathcal W_2)(y_1y_2|u_1) = \frac{1}{2}\sum_{u_2} \mathcal W_1(y_1 | u_1 \oplus u_2)\mathcal W_2(y_2|u_2),
\end{equation} 
and is therefore given by (compare Equation \ref{Hchannel})
\begin{equation}
H(X_1+X_2|Y_1Y_2) = H(\mathcal W_1\boxast \mathcal W_2).
\end{equation}
In the second setting, we are interested in the entropy evolution at a variable node with output states given by
\begin{equation}\label{cl:plus}
(\mathcal W_1\varoast \mathcal W_2) (y_1y_2|u_2)  = \mathcal W_1(y_1 | u_2)\mathcal W_2(y_2|u_2).
\end{equation} 
It turns out that for symmetric channels the combined channel can be reversibly transformed (see e.g.~\cite{R16bp}) into a channel with the output states
\begin{equation}\label{cl:plus2}
u_2\rightarrow \frac{1}{2} \mathcal W(y_1 | u_1 \oplus u_2)\mathcal W(y_2|u_2),
\end{equation}  
which is equivalent to decoding the second input to two channels combined by a CNOT gate given the side information $Y_1Y_2$ but additionally $X_1+X_2$. This again is equal to the generation of a \textit{better} channel when studying polar codes (see again Figure~\ref{Fig:Basic}).
Therefore we are interested in the entropy
\begin{equation}\label{eq:varo-minus}
H(X_2 | X_1+X_2,Y_1Y_2) = H(\mathcal W_1\varoast \mathcal W_2).
\end{equation}

Lower and upper bounds on both of these quantities have many applications in classical information theory, for example in coding theory giving exact bounds on EXIT charts~\cite{RU08} and, of course, the investigation of polar codes~\cite{AT14,GX15}. \\
In classical information theory, the optimal bounds are well known as follows:
\begin{align}
h_2(h_2^{-1}(H_1)\ast h_2^{-1}(H_2)) &\leq H(X_1+X_2|Y_1Y_2) \nonumber\\
&\leq \log 2 - \frac{(\log 2 - H_1)(\log 2 - H_2)}{\log 2}, \label{minus-bounds}
\end{align}
where the lower bound is called the Mrs. Gerber's Lemma, and
\begin{align}
\frac{H_1 H_2}{\log 2} &\leq H(X_2 | X_1+X_2,Y_1Y_2) \nonumber\\
&\leq H_1 + H_2 - h_2(h_2^{-1}(H_1)\ast h_2^{-1}(H_2)),  \label{plus-bounds}
\end{align}
with $H_1 = H(X_1|Y_1)$ and $H_2 = H(X_2|Y_2)$. 

In many situations, it is more intuitive to look at the special case where the two underlying entropies are equal $\left( H = H_1 = H_2\right)$. In that case we can state the following inequalities
\begin{align}
0.799\, \frac{H(\log2-H)}{\log2} + H &\leq h_2(h_2^{-1}(H_1)\ast h_2^{-1}(H_2)) \label{HHlowerbound}\\
&\leq H(X_1+X_2|Y_1Y_2) \\
&\leq  \frac{H(\log2-H)}{\log2} + H, \label{HHupperbound}
\end{align}
where the first is an additional convenient lower bound from~\cite{GX15} and the other two follow from Equation~\ref{minus-bounds}. Note that this special case is also of practical interest, for example for polar codes when two identical channels are combined.  

Additionally, a well known fact is that 
\begin{equation}\label{additiveentropies}
H(X_1+X_2|Y_1Y_2) + H(X_2 | X_1+X_2,Y_1Y_2) = H(X_1|Y_1) + H(X_2|Y_2),
\end{equation}
which we can equivalently write as
\begin{equation}\label{additiveentropies2}
H(\mathcal W_1\boxast \mathcal W_2) + H(\mathcal W_1\varoast \mathcal W_2) = H(\mathcal W_1) + H(\mathcal W_2).
\end{equation}
From this it follows that it is sufficient to prove the inequalities for either Equation~\eqref{minus-bounds} or ~\eqref{plus-bounds}. We will therefore mostly focus on the setting leading to Equation~\eqref{minus-bounds}. 

Moreover, it is even known for which channels equality is achieved in the above equations (see e.g. \cite{RU08}). For the lower bound in Equation \eqref{minus-bounds} this is the binary symmetric channel (BSC) and for the upper bound it is the binary erasure channel (BEC).
Therefore these channels are sometimes called the most and least informative channels. 

Later in this work we will be particularly interested in the lower bound in \eqref{minus-bounds} (and equivalently the upper bounds in \eqref{plus-bounds}). We will review the proofs  of these inequalities in the remainder of this section, with a particular focus on showing difficulties when translating these inequalities to the quantum setting.

\begin{figure}
\begin{minipage}[ht]{0.35\textwidth}
\begin{tikzpicture}[scale=0.5]
	\myblock.basic(0, 0, $\mathcal W_1$)
	\myblock.basic(0, 2, $\mathcal W_2$)
	\mystamp.cnot(0,2.5,1,2)
	\node[font=\tiny] at (-3.5,0.5) {$X_2$};
	\node[font=\tiny] at (-3.5,2.5) {$X_1$};
        \node[font=\tiny] at (3.5,2.5) {$B_1$};
        \node[font=\tiny] at (3.5,0.5) {$B_2$};
\end{tikzpicture}
\end{minipage} \\
\begin{minipage}[ht]{0.43\textwidth}
\begin{tikzpicture}[scale=0.5]
	\myblock.basic(0, 0, $\mathcal W_1$)
	\myblock.basic(0, 2, $\mathcal W_2$)
	\mystamp.cnot(0,2.5,1,2)
    \begin{pgfonlayer}{secondbackground}
	\mystamp.box(-3.75,-0.5,6.25,4,blue!3);
    \end{pgfonlayer}
        \node[font=\tiny] at (-3,4) {$\mathcal W_1\boxast \mathcal W_2$};
	\node[font=\tiny] at (-3.25,0.5) {$X_2$};
	\node[font=\tiny] at (-4.5,2.5) {$X_1$};
	\draw (-4,2.5) -- (-2,2.5);
        \node[font=\tiny] at (3.5,2.5) {$B_1$};
        \node[font=\tiny] at (3.5,0.5) {$B_2$};
\end{tikzpicture}
\end{minipage}
\begin{minipage}[ht]{0.43\textwidth}
\begin{tikzpicture}[scale=0.5]
	\myblock.basic(0, 0, $\mathcal W_1$)
	\myblock.basic(0, 2, $\mathcal W_2$)
	\mystamp.cnot(0,2.5,1,2)
    \begin{pgfonlayer}{secondbackground}
	\mystamp.box(-3.75,-0.5,6.25,4,blue!3);
    \end{pgfonlayer}
        \node[font=\tiny] at (-3,4) {$\mathcal W_1\varoast \mathcal W_2$};
	\node[font=\tiny] at (-4.5,0.5) {$X_2$};
	\node[font=\tiny] at (-3.25,2.5) {$X_1$};
	\draw (-4,0.5) -- (-2,0.5);
	\draw (-0.5,2.5) -- (0.2,4.5) -- (2.7,4.5);
        \node[font=\tiny] at (3.5,2.5) {$B_1$};
        \node[font=\tiny] at (3.5,0.5) {$B_2$};
        \node[font=\tiny] at (3.8,4.5) {$X_1+X_2$};
\end{tikzpicture}
\end{minipage}
\caption{\label{Fig:Basic}A useful figure to understand the concept of information combining is to look at two channels $\mathcal W_1$ and $\mathcal W_2$ which get combined by a CNOT gate, as in the figure in the top diagram. From this we can, in analogy to polar coding, generate two types of channels depicted at the bottom, which are given by $\mathcal W_1\boxast \mathcal W_2$ and $\mathcal W_1\varoast \mathcal W_2$. Both are directly related since the overall entropy is conserved under combining channels in this way (see Eq.\ \eqref{additiveentropies}).  }
\end{figure}
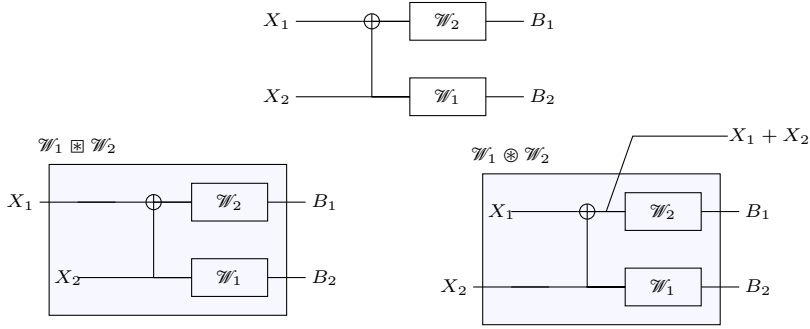

\paragraph{Proof techniques for the classical bounds}
In this paragraph, we will review the classical Mrs.\ Gerber's Lemma~\cite{WZ73} and a corresponding upper bound for combining of classical information (see Equation~\eqref{minus-bounds}), in order to contrast these results and proofs with our later results, where conditioning on \emph{quantum} side information is allowed. (As observed earlier, due to Equation~\eqref{additiveentropies}, this is equivalent to considering the bounds in Equation~\eqref{plus-bounds}.) The following proof sketches illustrate that the classical proofs, which crucially use the fact that the conditional Shannon entropy is affine under conditioning, cannot be easily extended to the case of quantum side information.

\begin{lemma}[Mrs. Gerber's Lemma]
Let $(X_1, Y_1)$ and $(X_2,Y_2)$ be independent pairs of classical random variables, with $X_1$ and $X_2$ being binary. Then:
\begin{equation}
H(X_1 + X_2 | Y_1Y_2) \geq h_2(h_2^{-1}(H(X_1 | Y_1))\ast h_2^{-1}(H(X_2 | Y_2))).
\end{equation}
\end{lemma}

An important ingredient in the proof of Wyner and Ziv~\cite{WZ73} is the observation that the function
\begin{equation}
g_c (H_1,H_2) := h_2(h_2^{-1}(H_1)\ast h_2^{-1}(H_2))
\end{equation}
is convex in $H_1\in[0,\log2]$ for each fixed $H_2\in[0,\log2]$, and, by symmetry, convex in $H_2$ for each fixed $H_1$.  These convexity properties can also be understood as a special case of the convexity of the so-called information bottleneck function, see~\cite{tishbyIBM, WW1975, GNT03}. We recently showed the convexity of its quantum generalization, the quantum information bottleneck function, in~\cite{datta2018convexity}. 

The convexity of $g_c (H_1,H_2)$, together with the representation of the conditional Shannon entropy as an average over unconditioned Shannon entropies, gives a proof of the lemma as follows:
\begin{align}
H&(X_1 + X_2 | Y_1Y_2) \nonumber\\
&= \sum_{y_1,y_2} p(Y_1=y_1) p(Y_2=y_2) H(X_1 + X_2 | Y_1=y_1 Y_2=y_2) \label{ccond}\\
&= \sum_{y_1,y_2} p(Y_1=y_1) p(Y_2=y_2) h_2(h_2^{-1}(H(X_1 | Y_1=y_1))\ast h_2^{-1}(H(X_2 | Y_2=y_2))) \\
&\geq \sum_{y_1} p(Y_1=y_1) h_2(h_2^{-1}(H(X_1 | Y_1=y_1))\ast h_2^{-1}(H(X_2 | Y_2))) \\
&\geq h_2(h_2^{-1}(H(X_1 | Y_1))\ast h_2^{-1}(H(X_2 | Y_2))).
\end{align}
Note that the way in which conditioning is handled by the equality \eqref{ccond} plays a crucial role in the proof. Unfortunately, this equality does generally not hold for the conditional entropy with quantum side information, i.e.\ when $Y_1$, $Y_2$ are quantum systems; in this case it is not even clear what the correct generalization of the right-hand-side of \eqref{ccond} may be. Understanding conditioning on quantum systems is an important but apparently difficult question in quantum information theory, as is illustrated by the much higher difficulty in proving the strong subadditivity property for quantum entropy~\cite{Lieb2002} compared to Shannon entropy. Better understanding of conditioning on quantum side information would not only help for bounds on information combining but for many other open problems as well, like the related question of conditional entropy power inequalities (see the beginning of this chapter) or even quantum cryptography~\cite{DFR16}.

In the proof for the upper bound in Equation~\eqref{minus-bounds} we encounter a very similar problem in handling quantum conditional information. The important inequality for the upper bound is the fact that the function $g_c (H_1,H_2)$ defined above can be bounded by an expression that is affine in both $H_1$ and $H_2$ separately:
\begin{equation}
g_c (H_1,H_2) \leq \log 2 - \frac{(\log 2 - H_1)(\log 2 - H_2)}{\log 2}.
\end{equation}
This follows immediately from the convexity of $g_c$ in $H_1$ and the fact that the inequality holds with equality for each fixed $H_2$ at the two endpoints $H_1\in\{0,\log2\}$, see e.g.\ \cite{LHHH05}. From here, the proof of the classical inequality proceeds in a similar fashion as for the lower bound, using again the expression of the conditional Shannon entropy:
\begin{alignat}{2}
&H(X_1 + &&X_2 | Y_1Y_2) \nonumber\\
&= \sum_{y_1,y_2}&& p(Y_1=y_1) p(Y_2=y_2) h_2(h_2^{-1}(H(X_1 | Y_1=y_1))\ast h_2^{-1}(H(X_2 | Y_2=y_2))) \\
&\leq \sum_{y_1,y_2}&& p(Y_1=y_1) p(Y_2=y_2)\,\dots \nonumber\\ 
&&&\dots\left[ \log 2 - \frac{(\log 2 - H(X_1 | Y_1=y_1))(\log 2 - H(X_2 | Y_2=y_2))}{\log 2} \right] \\
& =  \log 2 &&- \frac{(\log 2 - H(X_1 | Y_1))(\log 2 - H(X_2 | Y_2))}{\log 2}.
\end{alignat}
These two proofs show why finding similar inequalities in the quantum case might be very difficult. Nevertheless we will start in the next section by considering the case of quantum side information. 

\section{Information combining with quantum side information}\label{qcombining}

In this section we introduce the generalized scenario of information combining with quantum side information. 
The main ingredients are generalizations of the channel combinations in Equations~\eqref{cl:minus} and~\eqref{cl:plus} to the case of quantum outputs. Now we are combining two classical-quantum channels, with uniformly distributed binary inputs $\{0,1\}$. Again we will look at both, variable and check nodes under belief propagation and better and worse channels in polar coding.  Since the inputs are classical we can investigate the same combination procedure via CNOT gates. Belief propagation for quantum channels has been recently introduced in~\cite{R16bp}, for polar coding the resulting channels can be seen as special case of those in~\cite{WG13}.
 
The generalization of Equation~\eqref{cl:minus}, where we look at a check node or equivalently try to decode the input of the first channel while not knowing that of the second becomes a channel with output states
\begin{equation}\label{qu:minus}
\mathcal W_1\boxast \mathcal W_2 : u_1 \rightarrow \frac{1}{2}\sum_{u_2} \rho_{u_1 \oplus u_2}^{B_1}\otimes\rho_{u_2}^{B_2}.
\end{equation} 
Similarly the generalization of Equation~\eqref{cl:plus} for a variable node is given by
\begin{equation}\label{qu:plus}
\mathcal W_1\varoast \mathcal W_2: u_2 \rightarrow \rho_{u_2}^{B_1}\otimes\rho_{u_2}^{B_2},
\end{equation} 
which for symmetric channels, by a similar argument then for the classical case, is equivalent up to unitaries to the polar coding setting where we try to decode the second bit while assuming the first bit to be known. This becomes a channel with output states
\begin{equation}\label{qu:plus2}
u_2 \rightarrow  \frac{1}{2}\sum_{u_1}\ketbra{U_1}{u_1}{u_1}\otimes \rho_{u_1 \oplus u_2}^{B_1}\otimes\rho_{u_2}^{B_2},
\end{equation} 
where the additional classical register $U_1$ is used to make the input of the first channel available to the decoder. 

Our goal now is to find bounds on the conditional entropy of those combined channels
\begin{equation}\label{qu:boxast}
H(X_1+X_2|B_1B_2) = H(\mathcal W_1\boxast \mathcal W_2),
\end{equation}
and
\begin{equation}\label{qu:varoast}
H(X_2 | X_1+X_2,B_1B_2) = H(\mathcal W_1\varoast \mathcal W_2),
\end{equation} 
 in terms of the entropies of the original channels, analog to the bounds on information combining in the classical case. 
An important relation between these two entropies can be directly translated to the setting with quantum side information~\cite{WG13}
\begin{equation}\label{qchain}
H(X_1+X_2|B_1B_2) + H(X_2 | X_1+X_2,B_1B_2) = H(X_1|B_1) + H(X_2|B_2). 
\end{equation}
From here it follows that, as in the classical case, proving bounds on the entropy in Equation~\eqref{qu:boxast} automatically also gives bounds on the one in Equation~\eqref{qu:varoast}. 

In the remainder of this section, we will introduce the concept of channel duality and discuss its application to channel combining, which will help us find better bounds on above quantities.

\paragraph{Duality of classical and classical-quantum channels}\label{duality}
The essential idea is to embed a classical channel into a classical-quantum channel, take its complementary channel and apply it to inputs in the conjugate basis. In the way we introduce it here it has been first used in \cite{WR12b} to extend classical polar codes to quantum channels and then has been refined in \cite{RSH14} to investigate properties of polar codes for classical channels. A comprehensive overview with some new applications has recently been given in~\cite{R17}.
We explain the procedure here by applying it to a general binary classical channel $\mathcal W$ with transition probabilities $\mathcal W(y|x)$. 
The first step is to embed the channels into a quantum state 
\begin{equation}
\varphi_x = \sum_{y\in Y} \mathcal W(y|x) \ketbra{}{y}{y}
\end{equation}
and then choose a purification of this state with
\begin{equation}
\ket{}{\varphi_x} = \sum_{y\in Y} \sqrt{\mathcal W(y|x)}\ket{}{y}\ket{}{y}.
\end{equation}
Now we can define our classical quantum channel by an isometry acting as follows
\begin{equation}
U \ket{}{x} = \ket{}{\varphi_x}\ket{}{x}.\label{Udual}
\end{equation}
The dual channel is now defined by the isometry acting on states of the form $\ket{}{\tilde x} = \frac{1}{\sqrt{2}}\sum_z (-1)^{xz} \ket{}{z}$, 
\begin{align}
U \ket{}{\tilde x} &= \frac{1}{\sqrt{2}}\sum_{z\in\{0,1\}} (-1)^{xz} \ket{}{\varphi_z} \ket{}{z} \\
&= \frac{1}{\sqrt{2}}\sum_{\mathclap{\substack{y\in Y \\ z\in\{0,1\}}}} (-1)^{xz} \sqrt{\mathcal W(y|z)} \ket{}{y}\ket{}{y}\ket{}{z}.
\end{align}
Finally the output states are given by tracing out the initial output system 
\begin{equation}\label{Eq:gen}
\sigma_x = \frac{1}{2} \sum_{\mathclap{\substack{y\in Y \\ z,z'\in\{0,1\}}}} (-1)^{x(z+z')} \sqrt{\mathcal W(y|z) \mathcal W(y|z')} \ket{}{y}\ketbra{}{z}{y}\bra{}{z'}.
\end{equation}
We denote the channel dual to $\mathcal W$ as $\mathcal W^{\bot}$.
Note that we can equivalently define the duality via the \textit{channel state} given by
\begin{align}
\ket{}{\Psi_{\mathcal W}} &= \frac{1}{2} \sum_{z\in\{ 0,1\}} \ket{}{z} \ket{}{\varphi_z} \ket{}{z} \\
&= \frac{1}{2} \sum_{x\in\{ 0,1\}} \ket{}{\tilde x} \ket{}{\sigma_x},
\end{align}
where $\ket{}{\sigma_x} = \frac{1}{2} \sum_{z\in\{ 0,1\}} (-1)^{xz}\ket{}{\varphi_z} \ket{}{z}$ is a purification of $\sigma_x$. From $\Psi_{\mathcal W}$ the output states of the channel and its dual can both easily be recovered. 
In the same manner we can define dual channels for arbitrary classical-quantum channels following the steps above starting from Equation~\eqref{Udual} with the $\ket{}{\varphi_z}$ being purifications of the output states of the given channel. \\

This now allows us to calculate the duals of specific channels and also for combinations of channels. We state one result in the following Lemma, which is Theorem 1 in \cite{R17}. 
\begin{lemma}\label{lem:boxvaro}
Let $\mathcal W_1$ and $\mathcal W_2$ be two binary input cq-channels, then the following holds
\begin{align}
\mathcal W_1^{\bot}\boxast \mathcal W_2^{\bot} &= (\mathcal W_1 \varoast \mathcal W_2)^{\bot} \\
\mathcal W_1^{\bot}\varoast \mathcal W_2^{\bot} &= (\mathcal W_1 \boxast \mathcal W_2)^{\bot}.
\end{align}
\end{lemma}

We want to combine above Lemma~\ref{lem:boxvaro} with an observation made in \cite{RB08,WR12a}, which states that for any $\mathcal W$
\begin{equation}\label{II1}
I(\mathcal W) + I(\mathcal W^{\bot}) = \log 2, 
\end{equation} 
which leads us to 
\begin{equation}\label{Hbv}
H( \mathcal W_1 \varoast \mathcal W_2) = \log 2 - H(\mathcal W_1^{\bot}\boxast \mathcal W_2^{\bot}).
\end{equation}
 Note that in general $(\mathcal W^{\bot})^{\bot} \neq \mathcal W$~\cite{R17}, although this relation becomes an equality if $\mathcal W$ is symmetric, but in either case from Equation~\eqref{II1} we can directly conclude that 
\begin{equation}
H((\mathcal W^{\bot})^{\bot})=H(\mathcal W).
\end{equation}

From the above arguments we can directly make an important observation. 
Namely, let $\mathcal W_j$ be the channels corresponding to the states $\rho^{X_jB_j}$ ($j=1,2$), which in particular means $H(\mathcal W_j)=H(X_j|B_j)=H_j$. Then we have the following chain of equalities:
\begin{align}
H(X_1+X_2|B_1B_2)&=H(\mathcal W_1\boxast \mathcal W_2)\nonumber\\[2pt]
&=H(\mathcal W_1)+H(\mathcal W_2)-H(\mathcal W_1\varoast \mathcal W_2)\nonumber\\[2pt]
&=H_1+H_2-H((\mathcal W_1^{\bot}\boxast \mathcal W_2^{\bot})^{\bot})\nonumber\\[2pt]
&=H_1+H_2-\log2+H(\mathcal W_1^{\bot}\boxast \mathcal W_2^{\bot})\label{eqntermafterperpequation}
\end{align}
where the first line is by definition of $\boxast$, the second line the chain rule for mutual information (conservation of entropy), the third line follows from Lemma \ref{lem:boxvaro}, and the fourth line follows from Equation \eqref{II1}.

In particular this can be rewritten, using Equation~\eqref{Hbv}, as
\begin{equation}
H(\mathcal W_1\varoast \mathcal W_2)  - \left( H(\mathcal W_1)+H(\mathcal W_2) \right)/2 =  H(\mathcal W_1^{\bot}\varoast \mathcal W_2^{\bot}) - \left( H(\mathcal W_1^{\bot})+H(\mathcal W_2^{\bot}) \right)/2. 
\end{equation}
This is especially interesting, because it follows directly that due to the additional uncertainty relation given by Equation~\eqref{II1}, the lower bound in the quantum setting has an additional symmetry w.r.t. the transformation $H_i\mapsto\log2-H_i$, which the classical bound does not have. Therefore one can also easily see that there must exist states with quantum side information that violate the classical bound. \\

Finally we will give two particular examples of duals to classical channels (already provided in \cite{RSH14}), which state that the dual of every binary symmetric channel is a channel with pure state outputs and that the dual of a binary erasure channel (BEC) is again a BEC. 
\begin{ex}{Binary symmetric channel (Example 3.8 in~\cite{RSH14}).}
Let $\mathcal W$ be the classical BSC($p$). For every $p$, the output states of the dual channel are of the form
\begin{equation}
\sigma_x = \ketbra{}{\theta_x}{\theta_x},  \label{BSCdual}
\end{equation}
with $\ket{}{\theta_x} = Z^x \left( p\ket{}{0} + (1-p)\ket{}{1}\right)$, where $Z$ is the Pauli-Z matrix.
\end{ex}
\begin{ex}{Binary erasure channel (Example 3.7 in~\cite{RSH14}).}
Let $\mathcal W$ be the classical BEC($p$). For every $p$, the dual channel is again a binary erasure channel, now with erasure probability $1-p$. 
\end{ex}
Recall that the BSC and the BEC correspond exactly to the channels which achieve the classical lower and upper bounds with equality. These examples will thus become useful again when discussing our conjectured optimal bound. 
 
\section{Nontrivial bound for special case of Mrs.\ Gerber's Lemma with quantum side information}\label{lower-bound}

For general \emph{quantum} side information, we prove nontrival lower bounds akin to the classical Mrs.\ Gerber's Lemma, albeit only for the special case when the a priori probabilities are uniform, i.e.\ $p(X_1=0)=p(X_2=0)=1/2$. This case is relevant for several applications, as we show in later sections. A conjecture of the optimal bound, also covering the case of nonuniform probabilities, is made in Section \ref{main}.

\begin{thm}[Mrs.\ Gerber's Lemma with quantum side information for uniform probabilities]\label{QuantumMrsGerberTheorem2DifferentStates}
Let $\rho^{X_1B_1}$ and $\rho^{X_2B_2}$ be \emph{independent} and possibly different classical-quantum states carrying uniform a priori classical probabilites on the binary variables $X_1$, $X_2$, i.e.\
\begin{equation}
\rho^{X_jB_j}=\frac{1}{2}{\ketbra{}{0}{0}}_{X_j}\otimes\sigma_0^{B_j}+\frac{1}{2}{\ketbra{}{1}{1}}_{X_j}\otimes\sigma_1^{B_j},
\end{equation}
where $\sigma_i^{B_j}\in\setS{\cH_{d_j}}$ are quantum states on a $d_j$-dimensional Hilbert space ($i,j=1,2$). We denote their conditional entropies by $H_1 = H(X_1|B_1)$ and $H_2 = H(X_2|B_2)$, respectively. Then the following entropy inequality holds:
\begin{equation}
\begin{split}\label{firstFGmainresult}
&H(X_1+X_2|B_1B_2) \geq\max \Big\{\dots  \nonumber\\
& H_1-2\log\cos\left[\frac{1}{2}\arccos[(1-2h_2^{-1}(\log2-H_1))(e^{H_2}-1)]-\frac{1}{2}\arccos[e^{H_2}-1]\right]\,,\\
& H_2-2\log\cos\left[\frac{1}{2}\arccos[(1-2h_2^{-1}(\log2-H_2))(e^{H_1}-1)]-\frac{1}{2}\arccos[e^{H_1}-1]\right]\,,\\
& H_2-2\log\cos\left[\frac{1}{2}\arccos[(1-2h_2^{-1}(H_1))(2e^{-H_2}-1)]-\frac{1}{2}\arccos[2e^{-H_2}-1]\right]\,,\\
& H_1-2\log\cos\left[\frac{1}{2}\arccos[(1-2h_2^{-1}(H_2))(2e^{-H_1}-1)]-\frac{1}{2}\arccos[2e^{-H_1}-1]\right]\Big\}\,.
\end{split}
\end{equation}
\end{thm}
\begin{proof}We first prove that $H(X_1+X_2|B_1B_2)$ is not smaller than the first expression in the $\max$ in (\ref{firstFGmainresult}). To begin with, note the following:
\begin{align}
I(&X_1+X_2:X_2|B_1B_2) \nonumber\\[2pt]
&=H\left(X_1+X_2|B_1B_2\right)+H(X_2|B_1B_2)-H(X_1+X_2,X_2|B_1B_2)\nonumber\\[2pt]
&=H(X_1+X_2|B_1B_2)+H(X_2|B_1B_2)-H(X_1,X_2|B_1,B_2)\nonumber\\[2pt]
&=H(X_1+X_2|B_1B_2)+H(X_2|B_2)-H(X_1|B_1)-H(X_2|B_2)\nonumber\\[2pt]
&=H(X_1+X_2|B_1B_2)-H_1,\label{writewithIcondmutualinfo}
\end{align}
where the first equality is just the definition, the second uses the fact that there is a bijective (or unitary) relation between $(X_1+X_2,X_2)$ and $(X_1,X_2)$ (namely, a CNOT gate), and the third uses (twice) that $X_1B_1$ and $X_2B_2$ are independent.

While the strong subadditivity property of the von Neumann entropy~\cite{Lieb2002, book2000mikeandike} guarantees generally that $I(X_1+X_2:X_2|B_1B_2)\geq 0$, and therefore $H(X_1+X_2|B_1B_2)-H_1$ is nonnegative, we employ the recently established breakthrough result by Fawzi and Renner \cite{FR14}, discussed in Chapter \ref{recoverability}, in order to derive our inequality (\ref{firstFGmainresult}). The result in \cite{FR14} provides a lower bound based on the so called Fidelity of Recovery defined in Equation \ref{eq:fidrec} and is stated in Equation \ref{FRlowerbound}. At the end of this section, we will also briefly comment on the potential to improve our final result by using the stronger recoverability inequalities discussed in chapter \ref{recoverability}.

To apply the inequality from Equation \ref{recIneqChain}, we introduce the quantum state $\tau_{ACB}$ with binary (classical) registers $A=X_1+X_2$ and $C=X_2$, and a quantum register $B=B_1B_2$:
\begin{align}
\tau_{ACB} &\equiv \tau_{(X_1+X_2)(X_2)(B_1B_2)} := {\textrm CNOT}_{(X_1,X_2)\mapsto(X_1+X_2,X_2)}\big(\rho^{X_1B_1}\otimes\rho^{X_2B_2}\big) \nonumber\\
&=\frac{1}{4}{\ketbra{}{0}{0}}_A\otimes{\ketbra{}{0}{0}}_C\otimes\sigma_0^{B_1}\otimes\sigma_0^{B_2}+\frac{1}{4}{\ketbra{}{1}{1}}_A\otimes{\ketbra{}{0}{0}}_C\otimes\sigma_1^{B_1}\otimes\sigma_0^{B_2}\nonumber\\
&\quad+\frac{1}{4}{\ketbra{}{1}{1}}_A\otimes{\ketbra{}{1}{1}}_C\otimes\sigma_0^{B_1}\otimes\sigma_1^{B_2}+\frac{1}{4}{\ketbra{}{0}{0}}_A\otimes{\ketbra{}{1}{1}}_C\otimes\sigma_1^{B_1}\otimes\sigma_1^{B_2}\nonumber\\
&=\frac{1}{2}{\ketbra{}{0}{0}}_C\otimes\omega_0^{AB_1}\otimes\sigma_0^{B_2}+\frac{1}{2}{\ketbra{}{1}{1}}_C\otimes\omega_1^{AB_1}\otimes\sigma_1^{B_2}, \label{InfoCombCNOT}
\end{align}
where we defined $\omega_0^{AB_1}:=\frac{1}{2}({\ketbra{}{0}{0}}_A\otimes\sigma^{B_1}_0+{\ketbra{}{1}{1}}_A\otimes\sigma^{B_1}_1)$ and $\omega_1^{AB_1}:=\frac{1}{2}({\ketbra{}{0}{0}}_A\otimes\sigma^{B_1}_1+{\ketbra{}{1}{1}}_A\otimes\sigma^{B_1}_0)$ for later convenience. The lower bound on the conditional quantum mutual information in terms of the fidelity in Equation \ref{FRlowerbound} now says that there exists a quantum channel ${\mathcal R}'_{B\to AB}$ such that the following inequality holds:
\begin{alignat}{2}
H(X_1+X_2|B_1B_2)-H_1&=I(A:&&C|B)_\tau\nonumber\\
&\geq-2 \log&& F(\tau_{ACB}, \mathcal R'_{B\rightarrow AB}(\tau_{CB}))\label{fawzi-renner-bound-in-derivation} \\
&=-2\log&&\left[\frac{1}{2}F(\omega_0^{AB_1}\otimes\sigma_0^{B_2},{\mathcal R}'_{B\to AB}(\overline{\sigma}^{B_1}\otimes\sigma_0^{B_2})) \right. \nonumber \\ 
&&&\left.+\frac{1}{2}F(\omega_1^{AB_1}\otimes\sigma_1^{B_2},{\mathcal R}'_{B\to AB}(\overline{\sigma}^{B_1}\otimes\sigma_1^{B_2}))\right]\nonumber  \\
&= -2\log&&\left[\frac{1}{2} F(\omega_0^{AB_1} \otimes \sigma_0^{B_2} , \mathcal R_{B_2\to AB}(\sigma^{B_2}_0))\right. \nonumber \\ 
&&&\left. + \frac{1}{2}F(\omega_1^{AB_1} \otimes \sigma_1^{B_2} , \mathcal R_{B_2\to AB}(\sigma^{B_2}_1)) \right]. \label{minus2logAvgF}
\end{alignat}
Here we introduced $\overline{\sigma}^{B_1}:=\frac{1}{2}(\sigma_0^{B_1}+\sigma_1^{B_1})$ and used the fact that both $\tau_{ACB}$ and ${\mathcal R}'_{B\to AB}(\tau_{CB})$ are block-diagonal on the $C$-system to partially evaluate the fidelity in the third line, and defined the quantum channel ${\mathcal R}_{B_2\to AB}(\sigma_{B_2}):={\mathcal R}'_{B\to AB}(\overline{\sigma}^{B_1}\otimes\sigma_{B_2})$ in the fourth line.

To obtain a nontrivial lower bound on $H(X_1+X_2|B_1B_2)-H_1$, we now derive a nontrivial upper bound on the expression in the square brackets in (\ref{minus2logAvgF}). Our derivation will involve a triangle inequality on the set of quantum states in order to ``join'' the two states ${\mathcal R}_{B_2\to AB}(\sigma_{0,1}^{B_2})$ occurring in this expression. There are various ways to turn the quantum fidelity $F$ into a metric (in particular, to satisfy the triangle inequality)~\cite{tomamichel2015quantum}, e.g.\ the \emph{geodesic distance} $A(\rho,\sigma):=\arccos F(\rho,\sigma)$~\cite{book2000mikeandike}, the \emph{Bures metric} $B(\rho,\sigma):=\sqrt{1-F(\rho,\sigma)}$~\cite{B69}, or the \emph{purified distance} $P(\rho,\sigma):=\sqrt{1-F(\rho,\sigma)^2}$~\cite{GLN05,R02}. The following derivation can be done analogously with either of the three, but in the end the best bound will follow via the geodesic distance $A$, which we therefore use.

Using the concavity of the $\arccos$ function on the interval $[0,1]$ in the first step and abbreviating ${\mathcal R}:={\mathcal R}_{B_2\to AB}$, we obtain:
\begin{equation}
\begin{split}\label{concavity-triangle-monotonicity}
&\!\!\!\!\!\!\!\!\!\!\!\arccos\left[\frac{1}{2} F(\omega_0^{AB_1} \otimes \sigma_0^{B_2} , \mathcal R(\sigma^{B_2}_0)) + \frac{1}{2}F(\omega_1^{AB_1} \otimes \sigma_1^{B_2} , \mathcal R(\sigma^{B_2}_1)) \right]\\[2pt]
&\geq \frac{1}{2} A(\omega_0^{AB_1} \otimes \sigma^{B_2}_0 , \mathcal R(\sigma^{B_2}_0)) +\frac{1}{2} A(\omega^{AB_1}_1 \otimes \sigma^{B_2}_1 , \mathcal R(\sigma^{B_2}_1)) \\[2pt]
&\geq \frac{1}{2} A(\omega_0^{AB_1} \otimes \sigma_0^{B_2} , \omega_1^{AB_1} \otimes \sigma_1^{B_2}) - \frac{1}{2}A(\mathcal R(\sigma_0^{B_2}), \mathcal R(\sigma_1^{B_2}))\\[2pt]
&\geq \frac{1}{2} \arccos[ F(\omega_0^{AB_1},\omega_1^{AB_1}) F(\sigma_0^{B_2},\sigma_1^{B_2}) ] - \frac{1}{2}A(\sigma_0^{B_2},\sigma_1^{B_2}) \\[2pt]
&=\frac{1}{2}\arccos[F(\sigma_0^{B_1},\sigma_1^{B_1}) F(\sigma_0^{B_2},\sigma_1^{B_2}) ] - \frac{1}{2}A(\sigma_0^{B_2},\sigma_1^{B_2}) \\[2pt]
&= \frac{1}{2} \arccos[fg]  - \frac{1}{2} \arccos g,
\end{split}
\end{equation}
where in the third line we used the triangle inequality along the path $\omega_0^{AB_1}\otimes\sigma_0^{B_2}\to{\mathcal R}(\sigma_0^{B_2})\to{\mathcal R}(\sigma_1^{B_2})\to\omega_1^{AB_1}\otimes\sigma_1^{B_2}$, in the fourth line we used the fact that the fidelity is nondecreasing under quantum channels and multiplicative on tensor product states, and in the last two lines we evaluted and abbreviated $F(\omega_0^{AB_1},\omega_1^{AB_1})=F(\sigma_0^{B_1},\sigma_1^{B_1})=:f$ and $F(\sigma_0^{B_2},\sigma_1^{B_2})=:g$. Since the $\arccos$ function is nonincreasing in $[0,1]$, the last chain of inequalities yields an upper bound on the expression in square brackets in (\ref{minus2logAvgF}), and therefore:
\begin{align}
H(X_1+X_2|B_1B_2)-H_1\geq-2\log\cos\left[\frac{1}{2}\arccos[fg]-\frac{1}{2}\arccos g\right].\label{lowerboundWithfg}
\end{align}

As the last step, it is easy to verify that the right-hand-side of the inequality (\ref{lowerboundWithfg}) is monotonically decreasing in $f\in[0,1]$ for each fixed $g\in[0,1]$, and monotonically increasing in $g$ for each fixed $f$. Therefore, in order to continue the lower bound (\ref{lowerboundWithfg}), we can replace $f$ by an upper bound on $F(\sigma_0^{B_1},\sigma_1^{B_1})$ that is consistent with the given value of $H_1=H(X_1|B_1)$; and similarly replace $g$ by a lower bound on $F(\sigma_0^{B_2},\sigma_1^{B_2})$ consistent with $H_2=H(X_2|B_2)$. Exactly such upper and lower bounds are given in Theorem \ref{tightRelationFHtheorem}, following from bounds on the concavity of the von Neumann entropy, and result in
\begin{align}
&H(X_1+X_2|B_1B_2)- H_1 \nonumber\\
&\,\geq-2\log\cos\left[\frac{1}{2}\arccos[(1-2h_2^{-1}(\log2-H_1))(e^{H_2}-1)]-\frac{1}{2}\arccos[e^{H_2}-1]\right],\label{first-asymmetric-lower-bound-eqn}
\end{align}
showing that the first expression in the $\max$ in (\ref{firstFGmainresult}) is indeed a lower bound on $H(X_1+X_2|B_1B_2)$.

The same reasoning with $\rho^{X_1B_1}$ and $\rho^{X_2B_2}$ interchanged shows that the second expression in the $\max$ in (\ref{firstFGmainresult}) is a lower bound on $H(X_1+X_2|B_1B_2)$ as well.

To show that the third expression in the $\max$ in (\ref{firstFGmainresult}) is a lower bound on $H(X_1+X_2|B_1B_2)$, we exploit the symmetries of binary input classical-quantum channels and their dual channels under the channel combination. For this, we recall from Section~\ref{duality} that 
\begin{equation}\label{eqntermafterperpequationrepeat}
H(X_1+X_2|B_1B_2)=H(\mathcal W_1\boxast \mathcal W_2)=H_1+H_2-\log2+H(\mathcal W_1^{\bot}\boxast \mathcal W_2^{\bot}).
\end{equation}
Thus, we can obtain another lower bound on $H(X_1+X_2|B_1B_2)$ by bounding the term $H(\mathcal W_1^{\bot}\boxast \mathcal W_2^{\bot})$ from below using the first expression in the $\max$ in (\ref{firstFGmainresult}). This gives the following lower bound:
\begin{align}
&H(X_1+X_2|B_1B_2) \nonumber\\
&\geq H_1+H_2-\log2+H(\mathcal W_1^{\bot})- 2\log\cos\Big[ \dots \nonumber\\
&\left.\dots\frac{1}{2}\arccos[(1-2h_2^{-1}(\log2-H(\mathcal W_1^{\bot}))(e^{H(\mathcal W_2^{\bot})}-1)]-\frac{1}{2}\arccos[e^{H(\mathcal W_2^{\bot})}-1]\right],\nonumber
\end{align}
which is exactly the third expression in the $\max$ in Equation~\eqref{firstFGmainresult} as, again by (\ref{II1}), the channels $\mathcal W_j^{\bot}$ satisfy $H(\mathcal W_j^{\bot})=\log2-H_j$. We infer that the fourth expression in the $\max$ in (\ref{firstFGmainresult}) is a lower bound on $H(X_1+X_2|B_1B_2)$ from (\ref{eqntermafterperpequationrepeat}), by bounding the term $H(\mathcal W_1^{\bot}\boxast \mathcal W_2^{\bot})$ from below using the second expression in the $\max$ in (\ref{firstFGmainresult}).
\end{proof} 

\begin{rem}\label{equality-condition-for-proof} Since the $\arccos$ function is stricly monotonically decreasing in $[0,1]$, one can see from the first expression in the $\max$ in (\ref{firstFGmainresult}) (cf.\ also (\ref{lowerboundWithfg})) that $H(X_1+X_2|B_1B_2)=H_1$ is possible only if $H_1=\log2$ or $H_2=0$. Conversely, if $H_1=\log2$ or $H_2=0$ then actually $H(X_1+X_2|B_1B_2)=H_1$ since: {\textit(a)} $H_1\leq H(X_1+X_2|B_1B_2)$ holds due to (\ref{writewithIcondmutualinfo}) along with strong subadditivity; {\textit(b)} $H(X_1+X_2|B_1B_2)\leq\log2$ holds as $X_1+X_2$ is a binary register; {\textit(c)} since the conditional entropy  $H(X_2|X_1+X_2,B_1B_2)$ of a classical system is nonnegative (similarly to Equation~\eqref{writewithIcondmutualinfo}), we have:
\begin{align*}
H(X_1+X_2|B_1B_2)&\leq H(X_1+X_2|B_1B_2)+H(X_2|X_1+X_2,B_1B_2)\\[2pt]
&=H(X_1+X_2,X_2|B_1B_2)\\[2pt]
&=H(X_1,X_2|B_1B_2)\\[2pt]
&=H(X_1|B_1)+H(X_2|B_2)=H_1+H_2.
\end{align*}
Analogously, $H(X_1+X_2|B_1B_2)=H_2$ if and only if $H_1=0$ or $H_2=\log2$. Thus, the inequality $H(X_1+X_2|B_2B_2)\geq\max\{H_1,H_2\}$ holds with equality if and only if $H_1\in\{0,\log2\}$ or $H_2\in\{0,\log2\}$. Therefore, the inequality $H(X_1+X_2|B_1B_2)\geq(H_1+H_2)/2$ holds with equality if and only if $H_1=H_2\in\{0,\log2\}$.
\end{rem}

The lower bound (\ref{firstFGmainresult}) from Theorem \ref{QuantumMrsGerberTheorem2DifferentStates} is illustrated in Figure~\ref{fig-bound-with-h1h2}.

\begin{figure}[t!]
\centering
\begin{overpic}[trim=0.2cm 0.1cm 0.2cm 0.0cm, clip, scale=0.3]{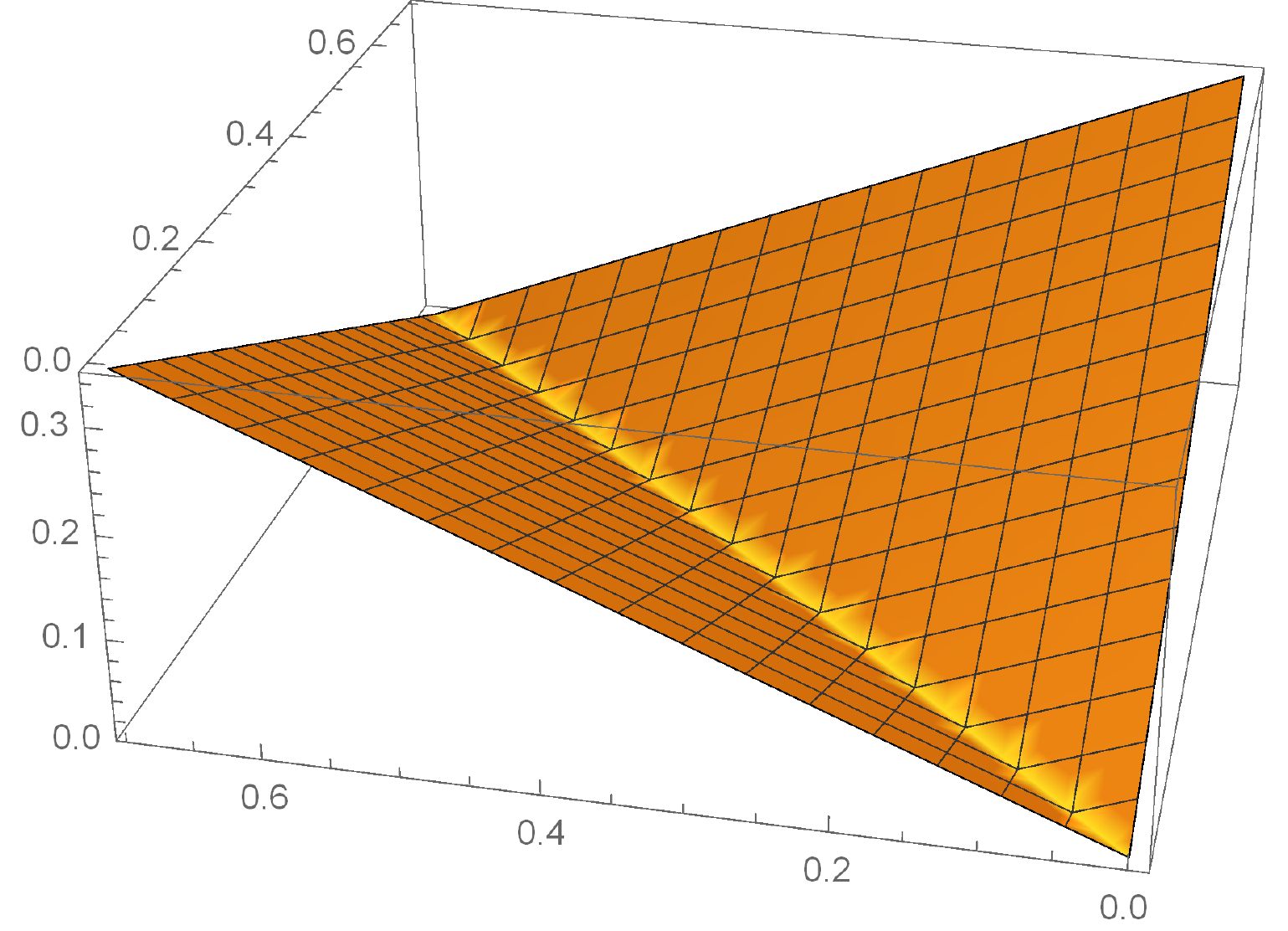}   
\put (0,62) {\small$H_1$}
\put (45,0) {\small$H_2$}
\end{overpic}\hspace{1.5cm}\begin{overpic}[trim=0.0cm 0.3cm 0.0cm 0.1cm, clip, scale=0.3]{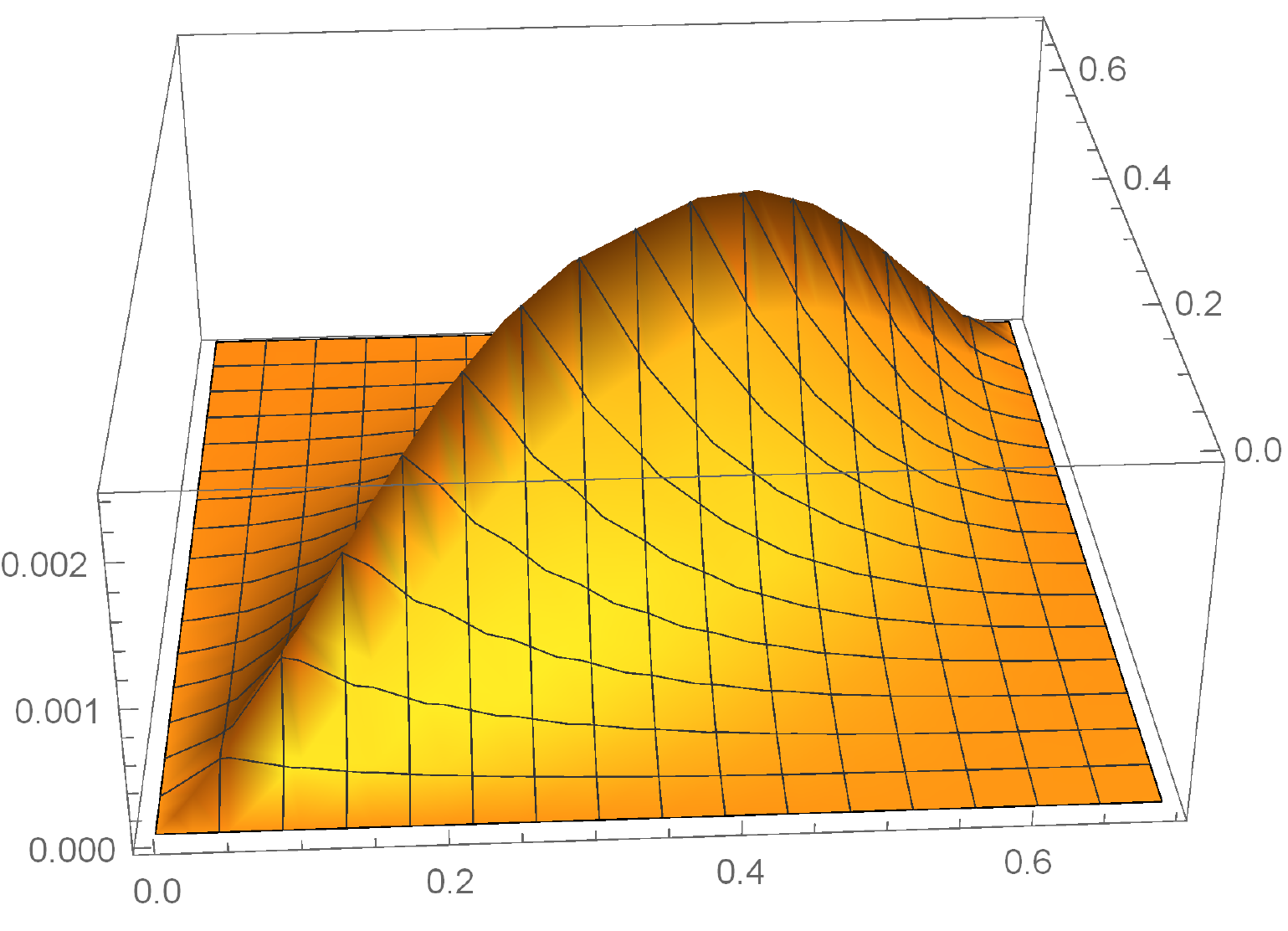}   
\put (97,55) {\small$H_1$}
\put (45,-3) {\small$H_2$}
\end{overpic}

\caption{\label{fig-bound-with-h1h2} The plot on the left shows the lower bound on $H(X_1+X_2|B_1B_2)-\frac{1}{2}(H_1+H_2)$ inferred from the bound (\ref{firstFGmainresult}) in Theorem \ref{QuantumMrsGerberTheorem2DifferentStates}, as a function of $(H_1,H_2)\in[0,\log2]\times[0,\log2]$. The plot on the right shows the lower bound on $H(X_1+X_2|B_1B_2)-\max\{H_1,H_2\}$ inferred from (\ref{firstFGmainresult}). The value of the bound along the diagonal line $H_1=H_2$ is shown again as the purple curve in Fig.\ \ref{fig-bound-with-h}.}
\end{figure}

In the important special case $\rho^{X_1B_1}=\rho^{X_2B_2}$ of Theorem \ref{QuantumMrsGerberTheorem2DifferentStates}, which will be useful in instances such as for polar codes on i.i.d.\ channels, we can use the same idea to obtain a better bound:
\begin{thm}[Mrs.\ Gerber's Lemma with quantum side information on i.i.d.\ states for uniform probabilities]\label{QuantumMrsGerberTheorem1StateTwice}
Let $\rho^{X_1B_1}=\rho^{X_2B_2}$ be \emph{identical and independent} classical-quantum states carrying uniform a priori classical probabilites on the binary variables $X_1,X_2$, i.e.\
\begin{equation}
\rho^{X_1B_1}=\rho^{X_2B_2}=\frac{1}{2}{\ketbra{}{0}{0}}\otimes\sigma_0+\frac{1}{2}{\ketbra{}{1}{1}}\otimes\sigma_1,  \label{equalStates}
\end{equation}
where $\sigma_i^{B_1}=\sigma_i^{B_2}=\sigma_i\in\setS{\cH_{d}}$ are quantum states on a $d$-dimensional Hilbert space ($i=1,2$). Denoting their conditional entropy by $H = H(X_1|B_1)=H(X_2|B_2)$, the following entropy inequality holds:
\begin{align}
H&(X_1+X_2|B_1B_2) \nonumber\\
&\geq\left\{\begin{array}{ll}H-2\log\cos\Big[\frac{1}{2}\arccos[(1-2h_2^{-1}(H))^2] \\ 
\quad\,\qquad\,\quad\,\qquad-\frac{1}{2}\arccos[1-2h_2^{-1}(H)]\Big],&\qquad H\leq\frac{1}{2}\log2
\\H-2\log\cos\Big[\frac{1}{2}\arccos[(1-2h_2^{-1}(\log2-H))^2] \\ 
\quad\,\qquad\,\quad\,\qquad-\frac{1}{2}\arccos[1-2h_2^{-1}(\log2-H)]\Big],&\qquad H>\frac{1}{2}\log2\end{array}\right.\label{FonlyMainResult}\\
&\geq\left\{\begin{array}{ll}H+0.083\cdot\frac{H}{1-\log H},&H\leq\frac{1}{2}\log2\\H+0.083\cdot\frac{\log2-H}{1-\log(\log2-H)},&H>\frac{1}{2}\log2.\end{array}\right.\label{FonlyMoreConvenientLowerBound}
\end{align}
The expressions (\ref{FonlyMoreConvenientLowerBound}) assume $\log$ to be the natural logarithm.
\end{thm}
\begin{proof} We follow the proof of Theorem \ref{QuantumMrsGerberTheorem2DifferentStates} up until Equation~\eqref{lowerboundWithfg}, which reads
\begin{align}
H(X_1+X_2|B_1B_2)-H\geq-2\log\cos\left[\frac{1}{2}\arccos[f^2]-\frac{1}{2}\arccos f\right]\label{lowerboundWithOnlyf}
\end{align}
with $f:=F(\sigma_0,\sigma_1)$. The right-hand-side of the last lower bound is monotonically increasing for $f\in[0,1/\sqrt{3}]$ and monotonically decreasing for $f\in[1/\sqrt{3},\log2]$ since these statements hold for the function $f\mapsto\frac{1}{2}\arccos[f^2]-\frac{1}{2}\arccos f$. Therefore, a lower bound based on $e^H-1\leq f\leq1-2h_2^{-1}(\log2-H)$ from Theorem \ref{tightRelationFHtheorem} can be obtained by evaluating (\ref{lowerboundWithOnlyf}) at those boundaries:
\begin{align}
\label{minof2expressionsinproofofHH}
&H(X_1+X_2|B_1B_2)-H \geq\min\Big\{\dots \nonumber\\ 
&\,-2\log\cos\Big[\frac{1}{2}\arccos[(e^H-1)^2]-\frac{1}{2}\arccos[e^H-1]\Big],\nonumber\\
&\,-2\log\cos\Big[\frac{1}{2}\arccos[(1-2h_2^{-1}(\log2-H))^2]-\frac{1}{2}\arccos[1-2h_2^{-1}(\log2-H)]\Big]\Big\}.
\end{align}
Numerically, one sees that for $H\in[\frac{1}{2}\log2,\log2]$ (and even for $H\in[0.33,\log2]$), the minimum in the last expression is attained by the second term, which gives
\begin{align}
&H(X_1+X_2|B_1B_2) - H \nonumber\\
&\,\geq -2\log\cos\Big[\frac{1}{2}\arccos[(1-2h_2^{-1}(\log2-H))^2]-\frac{1}{2}\arccos[1-2h_2^{-1}(\log2-H)]\Big]\label{proofoffirstselectorinequalHcase}
\end{align}
for $H\geq\frac{1}{2}\log2$, and shows the second selector in (\ref{FonlyMainResult}). Analytically, one can easily show this statement for $H\in[\log(1+1/\sqrt{3}),\log2]$, as this implies by Theorem \ref{tightRelationFHtheorem} that $f$ is in the range $f\in[1/\sqrt{3},1]$. Here the function $f\mapsto\frac{1}{2}\arccos[f^2]-\frac{1}{2}\arccos f$ is monotonically decreasing and we have $e^H-1\leq1-2h_2^{-1}(\log2-H)$ by Theorem \ref{tightRelationFHtheorem}. The statement is also true for $H\in[0.33,\log(1+1/\sqrt{3})]$, for the following reason: First, the statement is easily numerically certified for $H=0.33$; second, the function that maps $H$ to the first expression in the minimum in (\ref{minof2expressionsinproofofHH}) is monotonically \emph{increasing} for $H\in[0,\log(1+1/\sqrt{3})]$ since $H\mapsto e^H-1$ is increasing from $0$ to $1/\sqrt{3}$, where the right-hand-side of (\ref{lowerboundWithOnlyf}) is increasing in $f$; third, the function that maps $H$ to the second expression in the minimum in (\ref{minof2expressionsinproofofHH}) is monotonically \emph{de}creasing for $H\in[0.33,\log(1+1/\sqrt{3})]$ since the function $H\mapsto1-2h_2^{-1}(\log2-H)$ is increasing and not smaller than $1-2h_2^{-1}(\log2-0.33)\geq0.76\geq1/\sqrt{3}$, where the right-hand-side of (\ref{lowerboundWithOnlyf}) is decreasing in $f$.

To prove the first selector in (\ref{FonlyMainResult}), i.e.\ the case $H\leq\frac{1}{2}\log2$, we again use the reasoning via dual channels as in the proof of Theorem \ref{QuantumMrsGerberTheorem2DifferentStates}. Eq.\ (\ref{eqntermafterperpequation}) now reads:
\begin{align}
H(X_1+X_2|B_1B_2)=2H-\log2+H(\mathcal W^{\bot}\boxast \mathcal W^{\bot}),\label{mirrorsymmetryHHeqn}
\end{align}
where $\mathcal W$ is the channel corresponding to the state $\rho^{X_1B_1}=\rho^{X_2B_2}$ and $\mathcal W^{\bot}$ its dual. Since $H(\mathcal W^{\bot})=\log2-H\geq\frac{1}{2}\log2$ we can apply (\ref{proofoffirstselectorinequalHcase}) to the channel $\mathcal W^{\bot}$ to bound the last expression from below:
\begin{alignat}{2}\nonumber
&H(X_1+X_2|B_1B_2)&& \\ \nonumber
&\geq 2H-\log2+H(\mathcal W^{\bot})&&-2\log\cos\Big[\frac{1}{2}\arccos[(1-2h_2^{-1}(\log2-H(\mathcal W^{\bot})))^2] \\\nonumber
&&&-\frac{1}{2}\arccos[1-2h_2^{-1}(\log2-H(\mathcal W^{\bot}))]\Big].
\end{alignat}
This along with $H(\mathcal W^{\bot})=\log2-H$ gives finally the desired expression in the first selector in (\ref{FonlyMainResult}).

We show the more convenient lower bound (\ref{FonlyMoreConvenientLowerBound}) by using a few inequalities without formal proof. First we employ
\begin{align}\nonumber
\frac{1}{2}\arccos[x^2]-\frac{1}{2}\arccos x\geq\frac{\frac{1}{2}\arccos[F^2]-\frac{1}{2}\arccos F}{\sqrt{1-F}}\sqrt{1-x}\qquad\forall x\in[F,1]
\end{align}
for $F:=1-2h_2^{-1}(\frac{1}{2}\log2)$, since the function $x\mapsto(\arccos[x^2]-\arccos[x])/\sqrt{1-x}$ is monotonically increasing in $x\in[0,1)$. Using this in the first selector in (\ref{FonlyMainResult}), i.e.\ for $x=1-2h_2^{-1}(H)$, we obtain:
\begin{align}
H(X_1+X_2|B_1B_2)&\geq H-2\log\cos\left[c_1\sqrt{2h_2^{-1}(H)}\right]  \nonumber\\
&\geq H-2\log\left(1-c_2c_1^2\cdot2h_2^{-1}(H)\right)\,,\nonumber
\end{align}
for any $H\leq\frac{1}{2}\log2$, with
\begin{align}
c_1:&=\left.\frac{\frac{1}{2}\arccos[F^2]-\frac{1}{2}\arccos F}{\sqrt{1-F}}\right|_{F=1-2h_2^{-1}(\frac{1}{2}\log2)} \nonumber\\ 
c_2:&=\left.\frac{1-\cos x}{x^2}\right|_{x=c_1\sqrt{2h_2^{-1}(\frac{1}{2}\log2)}}, \nonumber
\end{align}
since the function $x\mapsto(1-\cos x)/x^2$ is monotonically decreasing in $x\in[0,\pi/2]\ni c_1\sqrt{2h_2^{-1}(\frac{1}{2}\log2)}$. From there we continue by first using the concavity of the $\log$ function:
\begin{align}
H(X_1+X_2|B_1B_2)&\geq H+4c_2c_1^2\,h_2^{-1}(H)\nonumber\\
&\geq H+4c_2c_1^2(1-e^{-1})\frac{H}{1-\log H},\nonumber
\end{align}
where in the last step we employ a convenient lower bound on $h_2^{-1}$, containing Euler's number $e$. The first selector now follows by $4c_2c_1^2(1-e^{-1})\geq0.083$, and the second selector in (\ref{FonlyMoreConvenientLowerBound}) by interchanging $H$ and $\log2-H$.
\end{proof}

The lower bounds (\ref{FonlyMainResult}) and (\ref{FonlyMoreConvenientLowerBound}) from Theorem \ref{QuantumMrsGerberTheorem1StateTwice} are shown in Fig.\ \ref{fig-bound-with-h}, where they are also compared to the bound (\ref{firstFGmainresult}) that is obtained from Theorem \ref{QuantumMrsGerberTheorem2DifferentStates} in the case $H_1=H_2=H$.

\begin{figure}[t!]
\centering
\begin{overpic}[trim=3.1cm 21.8cm 8.7cm 2.6cm, clip, scale=0.85]{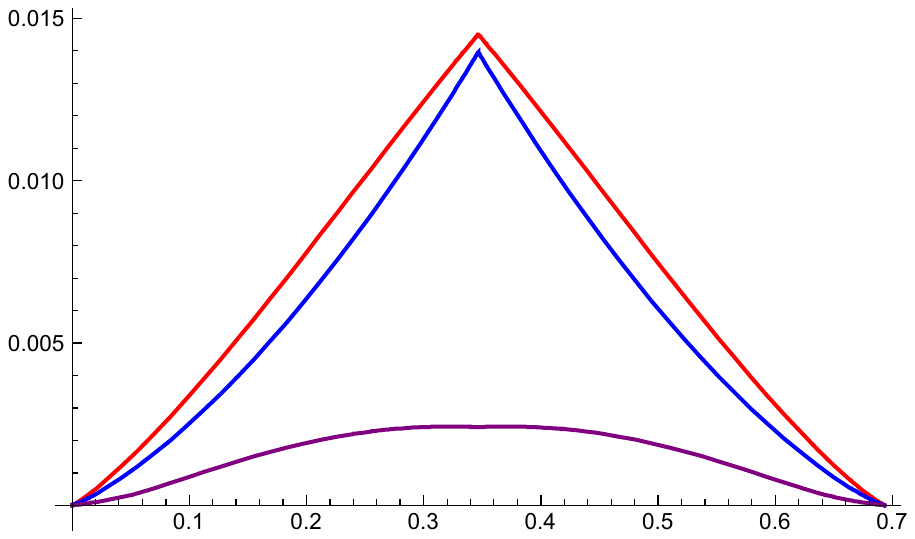}   
\put (50,-5) {\small$H$}
\end{overpic}
\vspace{0.5cm}
\caption{\label{fig-bound-with-h}The red curve shows the lower bound on $H(X_1+X_2|B_1B_2)-H$ in terms of $H\in[0,\log2]$ from Eq.\ (\ref{FonlyMainResult}), the blue curve from Eq.\ (\ref{FonlyMoreConvenientLowerBound}), and the purple curve from Eq.\ (\ref{firstFGmainresult}) in the special case $H_1=H_2=H$.}
\end{figure}

\paragraph{Potential improvements of lower bounds} 
Before ending this section we would like to give further details on the possibilities to improve the bounds given above. 
Our best lower bound (\ref{FonlyMainResult}) on $H(X_1+X_2|B_1B_2)-H$ behaves, by expanding the right-hand-side of (\ref{lowerboundWithOnlyf}) for $f\to1$, near the boundary $H\to\log2$ like
\begin{align}
\left.\left(\frac{3}{2}-\sqrt{2}\right)(1-f)+O((1-f)^2)\right|_{f=1-2h_2^{-1}(\log2-H)} \nonumber\\[1pt]
=(3-\sqrt{8})h_2^{-1}(\log2-H)+O((h_2^{-1}(\log2-H))^2). \label{behaviourofarccosf-expression}
\end{align}
Thus it vanishes faster than linearly $\propto(\log2-H)$ as $H\to\log2$ (see also Fig.\ \ref{fig-bound-with-h}; the behaviour for $H\to0$ follows by mirror symmetry $H\leftrightarrow\log2-H$ around $H=\frac{1}{2}\log2$). On the other hand, the bound vanishes at most as fast as $\Omega((\log2-H)/(-\log(\log2-H)))$ according to (\ref{FonlyMoreConvenientLowerBound}).

In contrast to this, our conjectured optimal lower bound from Conjecture \ref{QMGL} below posits that $H(X_1+X_2|B_1B_2)-H$ does \emph{not} vanish faster than the linear behaviour $(\log2-H)+o(\log2-H)$ for $H\to\log2$. When $\sigma_0,\sigma_1$ from Theorem \ref{QuantumMrsGerberTheorem1StateTwice} are pure states, then $H=\log2-h_2((1-f)/2)$ with $f=F(\sigma_0,\sigma_1)$, and one can easily compute $H(X_1+X_2|B_1B_2)-H=h_2((1-f)/2)-(1-f)\log2+O((1-f)^2\log(1-f))=(\log2-H)+o(\log2-H)$ for $H\to\log2$ (see also Section \ref{main} for our conjectured optimal states).

If one would like to prove such a linear lower bound $\Omega(\log2-H)$ on $H(X_1+X_2|B_1B_2)-H$ for $H\to\log2$ by our proof strategy, generally one would have to improve the lower bound (\ref{lowerboundWithOnlyf}) near $f\to1$ from the linear behaviour $\Omega(1-f)$ (see Equation~\eqref{behaviourofarccosf-expression}) by a logarithmic factor, e.g.\ improve it to $\Omega(-(1-f)\log(1-f))=\Omega(h_2((1-f)/2))$ (which matches the behaviour in the pure state case described in the previous paragraph). In this respect, note that the upper bound $f\leq1-2h_2^{-1}(\log2-H)$, which is also used in our derivation (by Theorem \ref{tightRelationFHtheorem}), \emph{cannot} be improved since it is tight in the pure state case.

It is unlikely that the ``missing'' logarithmic factor in the desired $\Omega(-(1-f)\log(1-f))$ bound on the right-hand-side of (\ref{lowerboundWithOnlyf}) near $f\to1$ is due to the use of concavity, triangle inequality, and monotonicity in the part (\ref{concavity-triangle-monotonicity}) of our derivation. Rather, it is the crucial Fawzi-Renner bound itself \cite{FR14} that we use in step (\ref{fawzi-renner-bound-in-derivation}) which does not seem to be strong enough. To support this statement, we evaluate the inequality (\ref{fawzi-renner-bound-in-derivation}) again in the special setting of Theorem \ref{QuantumMrsGerberTheorem1StateTwice} (i.e.\ $\sigma_0^{B_1}=\sigma_0^{B_2}=\sigma_0$ and $\sigma_1^{B_1}=\sigma_1^{B_2}=\sigma_1$) with pure states $\sigma_0,\sigma_1$ with fidelity $f=F(\sigma_0,\sigma_1)$; and even under the optimistic assumption that the so-called \emph{Petz recovery map} ${\mathcal R}'^{Petz}$ \cite{FR14} applied in a direct way would give a valid lower bound (which is not known to be true, and thus marked with `?' in the following), we would only obtain the following lower bound instead of (\ref{fawzi-renner-bound-in-derivation}):
\begin{align}
H&(X_1+X_2|B_1B_2)-H \nonumber\\
&=I(A:C|B)_\tau\nonumber\\
&\stackrel{?}{\geq}\max\left\{-2\log F(\tau_{ACB},{\mathcal R}'^{Petz}_{B\to AB}(\tau_{CB})),-2\log F(\tau_{ACB},{\mathcal R}'^{Petz}_{B\to BC}(\tau_{AB}))\right\}\nonumber\\
&=\max\left\{ \begin{array}{l} -\log\left[\frac{1}{2}\left(1+f^4+(1-f^2)\sqrt{1+f^2}\right)\right], \\[4pt] -\log\left[\frac{1}{2}\left(1+f^2+(1-f^2)^{3/2} \right)\right]\end{array}\right\}\nonumber\\
&=-\log\left[\frac{1}{2}\left(1+f^2+(1-f^2)^{3/2}\right)\right]\nonumber\\
&=(1-f)+O((1-f)^{3/2})\qquad\text{as}~f\to1. \nonumber
\end{align}
This is again linear $O(1-f)$ and thus not $\Omega(-(1-f)\log(1-f))$, even though there is only one (optimistically assumed Fawzi-Renner-type) inequality in this computation.

One may hope that the desired logarithmic factor may come into a bound $\Omega(-(1-f)\log(1-f))$ improving (\ref{lowerboundWithOnlyf}) by use of recovery results employing the \emph{measured relative entropy} $D_{\mathbb M}\geq F$ instead of the fidelity (the details of these improved bounds where introduced in Chapter~\ref{recoverability}). During such a derivation, one may need to keep more information about the involved states $\sigma_0,\sigma_1$ than their fidelity $f=F(\sigma_0,\sigma_1)$. As a first step towards this direction, we provide some numerical analysis in Figure~\ref{working}. It can be seen that using the measured relative entropy might lead to an advantage, although it would be rather small in absolute terms and it remains unclear whether this improvement could lead to the desired behavior. For comparison, the figure also includes numerics using the relative entropy of recovery, for which we know that it is generally not a valid lower bound on the conditional quantum mutual information. In this case one might still hope that it does give a valid bound in the special case where all but the conditioning system are classical, as it is the case in for our problem here. In Section~\ref{CounterExample} of the previous chapter we reviewed the known counterexamples and provided new ones that show that also in this special case the bound can still be violated. At this point, we also note that using either of the two recovery bounds would still leave us quite far from our conjectured bounds in the next section. Working out exact analytical bounds using the improved recoverability inequalities is left for future work. 

\begin{figure}[t!]
\centering
\begin{overpic}[trim=2.3cm 8cm 1cm 9cm, clip, scale=0.7]{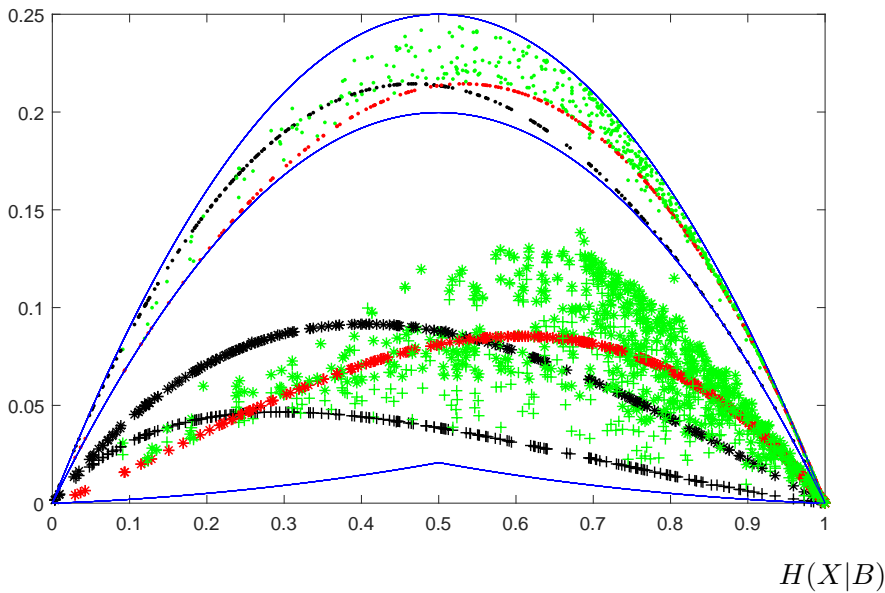}   
\put (83,0) {$H(X|B)$}
\end{overpic}
\caption{
For this figure we generated states of the form in Equation~\eqref{equalStates} and plotted several different quantities against the conditional entropy of the states. Data points marked by a dot give the conditional quantum mutual information, those with a star the relative entropy of recovery and crosses denote the measured relative entropy of recovery. The colors of the data points refer to how we drew the random states $\sigma_0$ and $\sigma_1$; for green points those are randomly drawn mixed states, red ones use classical (diagonal) states and black points use pure quantum states. For comparison we included several graphs in blue; the top line gives the (classical) upper bound on information combining from Equation~\eqref{HHupperbound}, the middle one a (classical) lower bound given in Equation~\eqref{HHlowerbound} and the bottom one our bound in Equation~\eqref{FonlyMoreConvenientLowerBound}. Note that the data points for the conditional quantum mutual information, in particular for classical and pure states, are in perfect agreement with our conjecture in Section~\ref{main}.  \label{working}}
\end{figure}

As a last note, we remark that, instead of exploiting (\ref{minof2expressionsinproofofHH}) in the regime of large $H\in[\frac{1}{2}\log2,\log2]$ and afterwards symmetrizing the bound into the regime $H\in[0,\frac{1}{2}\log2]$ via (\ref{mirrorsymmetryHHeqn}), one could instead have exploited (\ref{minof2expressionsinproofofHH}) in the regime of small $H\in[0,\frac{1}{2}\log2]$ and later symmetrized towards large $H$. Any bounds that can be obtained in this way will, however, never be better than \emph{quadratic} at the boundaries, i.e.\ they will behave like $O(H^2)$ for $H\to0$ and thus $O((\log2-H)^2)$ for $H\to\log2$, and will therefore be inferior to (\ref{FonlyMainResult}) and (\ref{FonlyMoreConvenientLowerBound}) at the boundaries. The reason for this is that: {\textit(a)} no lower bound in terms of the fidelity $f$ akin to (\ref{lowerboundWithOnlyf}) can be better than $f^2/2+O(f^4)$ near $f\to0$, because this is the behaviour of $H(X_1+X_2|B_1B_2)-H$ in the pure state case described above; {\textit(b)} no lower bound on $f$ can be larger than linear in $H$ for $H\to0$ (such as, e.g., the desired $f\geq\Omega(\sqrt{H})$), because the (mixed) states $\sigma_0={\textrm diag}(f,1-f,0)$, $\sigma_1={\textrm diag}(f,0,1-f)$ satisfy the linear relation $f=F(\sigma_0,\sigma_1)=H/\log2$.


\section{Conjectures for optimal bounds}\label{main}

In this section we will present conjectures on what the optimal bounds for information combining with quantum side information might be, i.e.\ the generalization of the inequalities in Eq.\ (\ref{minus-bounds}) to the case of quantum side information.

First we give a conjecture for a lower bound in analogy to the Mrs. Gerber's Lemma (compare to the left inequality in Eq.\ (\ref{minus-bounds})):
\begin{cj}{[Quantum Mrs. Gerber's Lemma]}\label{QMGL}
Let $\rho^{X_1B_1}$ and $\rho^{X_2B_2}$ be classical quantum states with $X_1$ and $X_2$ being binary and conditional entropy $H_1 = H(X_1|B_1)$ and $H_2 = H(X_2|B_2)$ respectively. Then the following entropy inequality holds: 
\begin{align}\label{conj:lower}
&H(X_1+X_2|B_1B_2) \nonumber\\
&\geq 
\begin{cases}
h_2(h_2^{-1}(H_1)\ast h_2^{-1}(H_2)) &\!\!\!\!\!\!\! H_1+H_2 \leq \log2 \\[3pt]
H_1 + H_2 -\log2 + h_2(h_2^{-1}(\log2 - H_1)\ast h_2^{-1}(\log2 - H_2))  \\ 
&\!\!\!\!\!\!\! H_1+H_2 \geq \log2
\end{cases}
\end{align}
\end{cj}
Additionally, we conjecture the following upper bound (compare to the second inequality in Equation\ (\ref{minus-bounds})):
\begin{cj}[Upper bound]\label{upper}
Let $\rho^{X_1B_1}$ and $\rho^{X_2B_2}$ be classical quantum states with $X_1$ and $X_2$ being binary and conditional entropy $H_1 = H(X_1|B_1)$ and $H_2 = H(X_2|B_2)$ respectively. Then the following entropy inequality holds: 
\begin{equation}
H(X_1+X_2|B_1B_2) \leq \log2 - \frac{(\log2 - H_1)(\log2 - H_2)}{\log2}.
\end{equation}
\end{cj}
In what follows, we will discuss several observations that give strong evidence in favour of our conjectures.

\begin{figure}[t!]
\centering
\begin{overpic}[clip, scale=0.55]{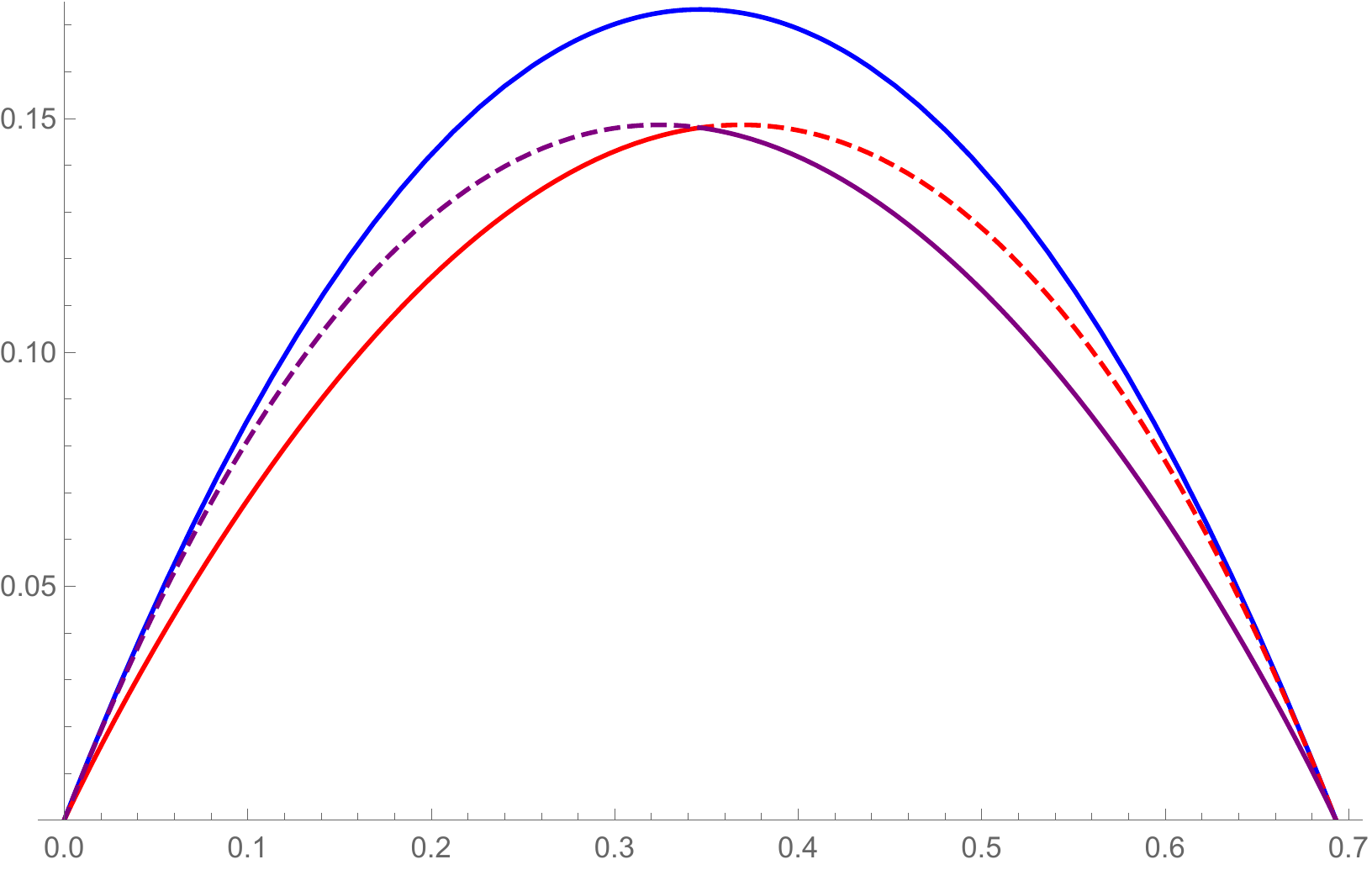}   
\put (50,-5) {\small$H$}
\end{overpic}
\vspace{0.5cm}
\caption{\label{fig:conjectures} This plot shows our conjectured bounds on $H(X_1+X_2|B_1B_2) - H$ when $H_1=H_2=H$. The blue curve is the upper bound in Conjecture~\ref{upper}, while the red curve gives the lower bound for $H\leq \frac{\log 2}{2}$ and purple for $H\geq \frac{\log 2}{2}$ in Conjecture~\ref{QMGL}. Plain lines give the actual bounds, while dashed lines are shown to illustrate the two functions in Equation~\eqref{conj:lower} and for comparison to the classical bound.}
\end{figure}

\paragraph{Quantum states that achieve equality}\label{subsec:optimal}
First we will discuss the states that achieve equality in the conjectured inequalities. 
It can easily be seen that the classical half (i.e.\ the first selector in Eq.\ (\ref{conj:lower})) of Conjecture \ref{QMGL} can be achieved by embedding a BSC into a classical quantum state as follows (with $p\in[0,1]$ chosen accordingly):
\begin{equation}
\rho = \frac{1}{2} \ketbra{}{0}{0} \otimes (p\ketbra{}{0}{0} + (1-p)\ketbra{}{1}{1}) + \frac{1}{2} \ketbra{}{1}{1} \otimes ((1-p)\ketbra{}{0}{0} + p\ketbra{}{1}{1}).
\end{equation}
The optimality of these states follows from the inequality in the classical Mrs. Gerber's Lemma (and can also be verified easily by calculating the entropy terms).
In the quantum half of Conjecture \ref{QMGL}, the optimal states represent binary classical-quantum channels with pure output states and can therefore be represented as
\begin{equation}
\rho = \frac{1}{2} \ketbra{}{0}{0} \otimes \ketbra{}{\Psi_0}{\Psi_0}  + \frac{1}{2} \ketbra{}{1}{1} \otimes \ketbra{}{\Psi_1}{\Psi_1},
\end{equation}
where $\Psi_0$ and $\Psi_1$ are pure states. Due to unitary invariance we can choose them to be $\ket{}{\Psi_0}=\left(\begin{smallmatrix} 1 \\ 0 \end{smallmatrix}\right)$ and $\ket{}{\Psi_1}=\left(\begin{smallmatrix} \cos{\alpha } \\ \sin{\alpha }\end{smallmatrix}\right)$. 
Again, this can be verified by simply calculating the involved entropies. Unfortunately, this calculation is not very insightful, therefore we choose to give an alternative proof, which might also give some intuition towards why our conjectured lower bound has the given additional symmetry. 
The alternative proof will be based on the concept of dual channels as explained in Section~\ref{duality}. 

Lets fix $\mathcal W_1$ and $\mathcal W_2$ to be channels with pure output states of the form in Equation \eqref{BSCdual} and therefore dual channels of BSCs. With the above arguments we can now show in an intuitive way that channels of this form achieve equality for the quantum side of our conjecture. 
\begin{align*}
H&( \mathcal W_1 \boxast \mathcal W_2) \\
&\! = H(\mathcal W_1) + H(\mathcal W_2) - H( \mathcal W_1 \varoast \mathcal W_2) \\
&\! =  H(\mathcal W_1) + H(\mathcal W_2) -\log2 + H(\mathcal W_1^{\bot}\boxast \mathcal W_2^{\bot}) \\ 
&\! =  H(\mathcal W_1) + H(\mathcal W_2) -\log2 + h_2(h_2^{-1}(H(\mathcal W_1^{\bot}))\ast h_2^{-1}(H(\mathcal W_2^{\bot}))) \\
&\! =  H(\mathcal W_1) + H(\mathcal W_2) -\log2 + h_2(h_2^{-1}(\log2 - H(\mathcal W_1))\ast h_2^{-1}(\log2 - H(\mathcal W_2))), 
\end{align*}
where the first equality follows from the chain rule for mutual information, the second one from Equation \eqref{Hbv}, the third from the classical Mrs. Gerber's Lemma and the final one from Equation \eqref{II1}. Note that the equality holds because, in the classical Mrs. Gerber's Lemma, binary symmetric channels achieve equality. 

\begin{rem} With an argument along the same lines one can prove immediately that our conjectured lower bound is true not only for all states that are classical channels (or embeddings of such) but also for all states that are duals of such classical channels.  
\end{rem}

Now, lets look at Conjecture \ref{upper}. From the classical upper bound it can be easily seen that equality is achieved by embeddings of binary erasure channels, which give the following class of states
\begin{equation*}
\rho = \frac{1}{2} \ketbra{}{0}{0} \otimes ((1-\epsilon)\ketbra{}{0}{0} + \epsilon\ketbra{}{e}{e}) + \frac{1}{2} \ketbra{}{1}{1} \otimes ((1-\epsilon)\ketbra{}{1}{1} + \epsilon\ketbra{}{e}{e}).
\end{equation*}
\begin{rem} 
It is interesting to note -- concerning the duality relations used before -- that  the upper bound \emph{can} coincide with the quantum bound because the dual channel of a BEC with error probability $\epsilon$ is again a channel from the same family, i.e.\ a BEC with error probability $1-\epsilon$. 
\end{rem}

\paragraph{Numerical evidence}
\begin{figure}[t!]
\centering
\begin{overpic}[clip, scale=0.9]{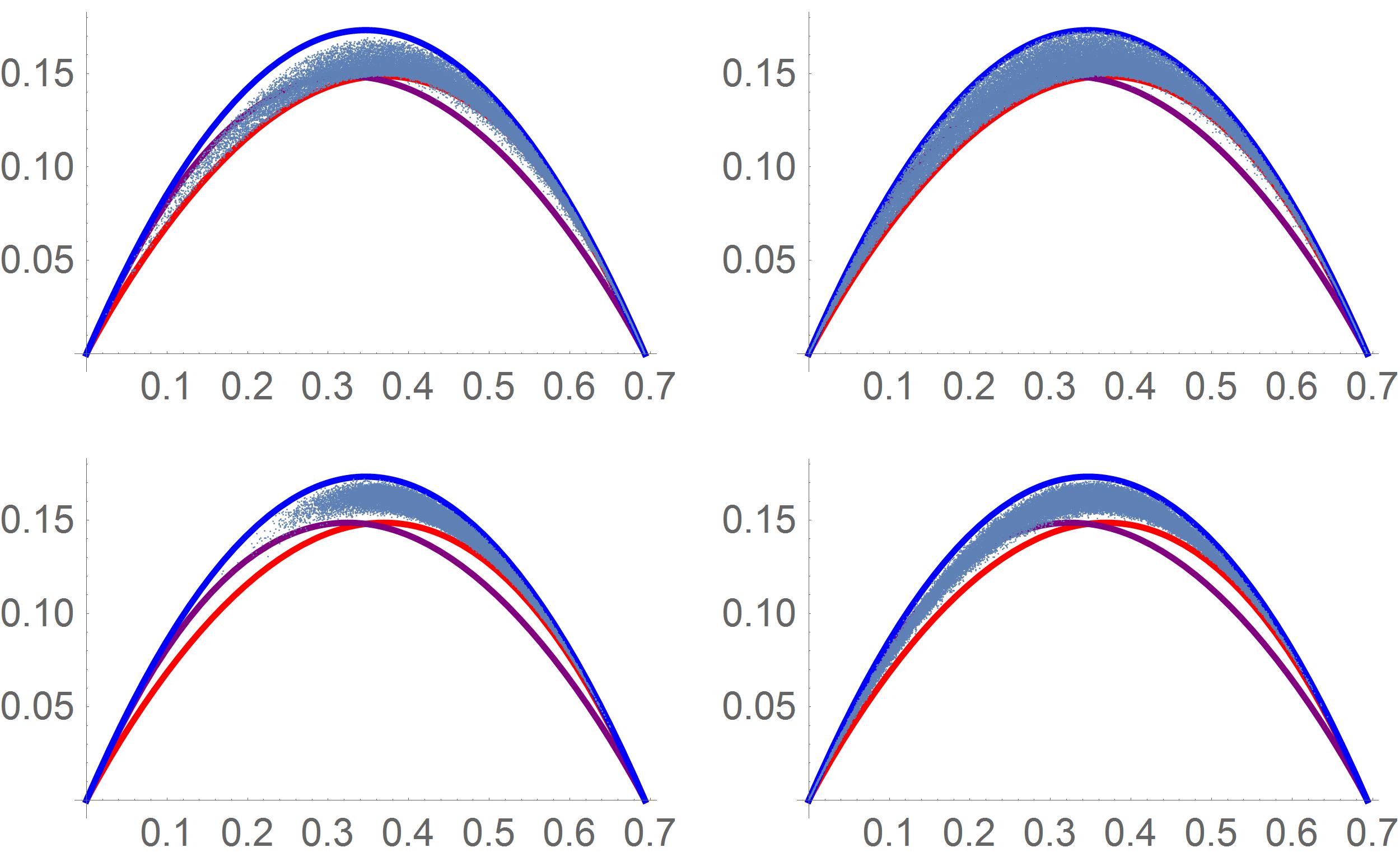}   
\put (27.5,-2) {\small$H$}
\put (80,-2) {\small$H$}
\put (27.5,30) {\small$H$}
\put (80,30) {\small$H$}
\end{overpic}
\vspace{0.2cm}
\caption{\label{fig:numerics} These pictures present some of the numerical evidence gathered to test Conjectures~\ref{QMGL} and~\ref{upper}. Plotted is $H(X_1+X_2|B_1B_2) - H$ when $H_1=H_2=H$ against $H$. For each of the examples classical-quantum states of the form in Equation~\eqref{num:test-states} where randomly generated. Those on the left with $p=\frac{1}{2}$ and those on the right with $p\in [0,1]$ at random, those on the top with quantum states of dimension $d=2$ and the bottom ones with dimension $d=4$. Each plot contains $50000$ examples.} 
\end{figure}
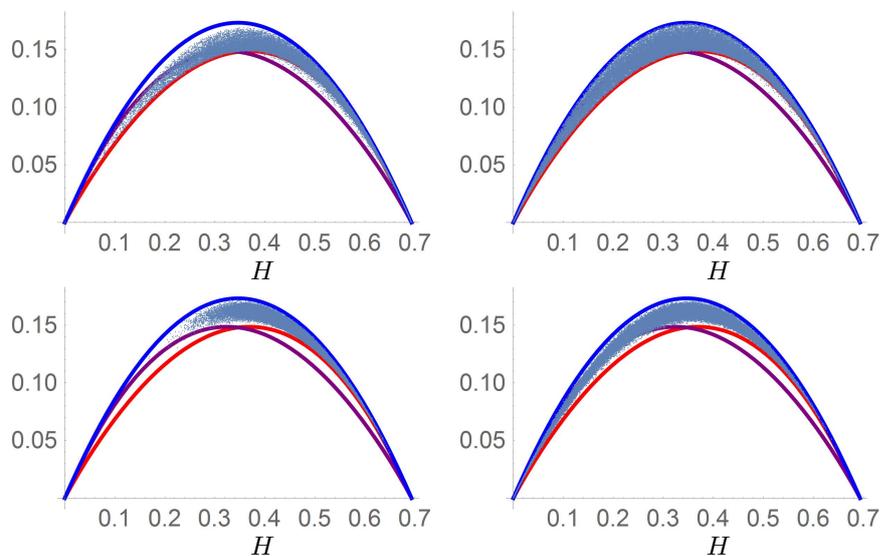

We have tested our Conjectures~\ref{QMGL} and~\ref{upper} using numerical examples in which we generated classical-quantum states of the form
\begin{equation}\label{num:test-states}
\rho^{XB}=p{\ketbra{}{0}{0}}_{X}\otimes\rho_0^{B}+(1-p){\ketbra{}{1}{1}}_{X}\otimes\rho_1^{B}.
\end{equation}
Here $\rho_0^{B}$ and $\rho_1^{B}$ are randomly chosen quantum states of dimension $d\in\{ 2,3,4,5,6 \}$, using the over-parametrized generation method (see e.g.~\cite{M15}),  and $p$ is either fixed $p=\frac{1}{2}$ or drawn at random from $p\in[0,1]$. We then used each of these states to calculate the exact value of $H(X_1+X_2|B_1B_2)$ with $H_1=H_2=H$ and compared it to our conjectured bounds. For all $10$ combinations we tested our conjectures with several $100.000$ classical-quantum states. No violations of our bounds were found. We also found that the states coming close to our conjectured bounds are close to the conjectured optimal forms stated in Section \ref{subsec:optimal}.

While none of the generated states violated our conjectured bounds, violation of the classical lower bound was easily observed. A sample of our numerics is shown in Figure~\ref{fig:numerics}. (That this violation of the classical lower bound must occur is clear from the analytical results of Section \ref{subsec:optimal}).

Furthermore, we carried out similar numerics for the case of two different classical-quantum states, i.e.\ with differing entropies $H_1\neq H_2$. Again, we found no violation of Conjectures~\ref{QMGL} and~\ref{upper}.

Additional numerical evidence can also be found in Figure~\ref{working}, also supporting the analytically found states that achieve our conjectured lower bound with equality.

\section{Application to classical-quantum polar codes}\label{PolarCodes}

In this section we apply the previously achieved results on information combining to classical-quantum polar codes. We will first introduce some technical aspects of polar codes and the underlying concept of polarization. In the following sections we will then show how our bounds can be used to translate a simple proof of polarization from the classical-classical case to the classical-quantum case. Our results also allow us to prove polarization for non-stationary channels. Finally, we will describe the impact of our quantitative bounds from Section \ref{lower-bound} on the \emph{speed} of polarization of cq-polar codes and comment on the possible speed when assuming our conjectured lower bound from Conjecture \ref{QMGL}.

But first we will start with the promised introduction. 
polar codes were introduced by Arikan as the first classical constructive capacity achieving codes with efficient encoding and decoding~\cite{A09}. This is in contrast to the \textit{random coding} technique that is usually used to prove capacity results in communication theory, but does not give explicit or efficient codes. The underlying idea of polar codes is that, by adding the input bit of a later channel onto one of an earlier channel, that earlier channel becomes harder to decode while providing side-information for decoding the later one. polar codes rely on an iteration of this scheme, which, combined with a successive cancellation decoder, eventually leads to almost perfect or almost useless channels. This process is called polarization. This decoder attempts to decode the output bit by bit, assuming at each step full knowledge of previously received bits while ignoring later outputs. Since information is sent only via channels that polarize to (almost) perfect channels while useless channels transmit so called \textit{frozen bits}, which are known to the receiver, this decoder can achieve a very low error probability. In fact, it was proven in \cite{AT09} that the block error probability scales as $O(2^{-N^\beta})$ (for any $\beta<1/2$).

Based on the classical setting, polar codes were later generalized to channels with quantum outputs~\cite{WG13}. These quantum polar codes inherit many of the desirable features like the efficient encoder and the exponentially vanishing block error probability~\cite{WG13, H14}, while especially the efficient decoder remains an open problem~\cite{WLH13}.

Since their introduction polar codes have been investigated in many ways, like adaptations to many different settings in classical~\cite{A12, A10} and quantum information theory~\cite{HMW14,CM15}.

In particular, in the classical setting, polar codes have been generalized to non-stationary channels~\cite{AT14} and it was shown that the exponentially vanishing block error rate can be achieved with just a polynomial block length \cite{GX15}. Both of these results have not so far been extended to the classical-quantum setting, and their proofs rely heavily on the classical Mrs.\ Gerber's Lemma.

Let us now look at the relationship between bounds on information combining and polar codes. The most natural quantity to track the quality of a channel during the polarization process is its conditional entropy (or equivalently, for symmetric channels, its mutual information), and the most basic element in polar coding is the application of a CNOT gate. As described in the beginning of this chapter, from such an application, we can derive one channel that is worse than either of the two original channels, and one that is better (in terms of their conditional entropy). The worse channel is usually denoted by $\langle \mathcal W_1,\mathcal W_2\rangle^-$ and the better one by $\langle \mathcal W_1,\mathcal W_2\rangle^+$, which denote exactly the channels in Equations \ref{qu:minus} and \ref{qu:plus2} respectively, where $\mathcal W_1$ and $\mathcal W_2$ are the original channels. It follows that (see Section~\ref{qcombining})
\begin{equation}
H(\langle \mathcal W_1,\mathcal W_2\rangle^-) = H(X_1+X_2 | Y_1Y_2) = H(\mathcal W_1 \boxast \mathcal W_2)
\end{equation}
and
\begin{equation}
H(\langle \mathcal W_1,\mathcal W_2\rangle^+) = H(X_2 | X_1+X_2,Y_1Y_2) = H(\mathcal W_1 \varoast \mathcal W_2).
\end{equation}
Naturally, the same is true for the corresponding quantities based on the channel's mutual information $I(\mathcal W)$, which we recall is defined by $I(\mathcal W):=\log2-H(\mathcal W)$ for the case of symmetric binary channels, which is the only case we consider here.

Therefore, it is intuitive that good bounds on information combining can be very helpful for investigating specific properties of polar codes and in particular the polarization process. This is because those bounds allow to characterize the difference in entropy between the synthesized channels $\langle \mathcal W_1,\mathcal W_2\rangle^-$ and $\langle \mathcal W_1,\mathcal W_2\rangle^+$ and the original channels $\mathcal W_1,\mathcal W_2$.

Now, we are ready to turn to the main results of this section and provide some new results on classical-quantum polar codes based on our new entropy inequalities.

\subsection{Polarization for stationary and non-stationary channels}
Polarization is one of the main features of polar codes and crucial for their ability to achieve capacity. It was first proven in the classical setting in~\cite{A09} by showing convergence of certain martingales, and a similar approach has later been used to establish polarization for classical-quantum polar codes in~\cite{WG13}. Recently a conceptually simpler proof of polarization has been found in~\cite{AT14} making use of the classical Mrs.\ Gerber's Lemma as its main tool. Besides its more intuitive approach, one of the main advantages of this new proof is that it can be extended to non-stationary channels, while the martingale approach is only known to work for stationary channels. Here we define non-stationary channels based on a set of channels $\{\mathcal W_t\}_{t=0}^\infty$, where the actual channel has the form $\mathcal W_t$ on its $t$-th application, in contrast to stationary channels which are constant throughout all applications.

In this section, we show that our results from Section \ref{lower-bound} are sufficient to extend the polarization proof from \cite{AT14} to the setting of classical-quantum channels, and also to prove polarization for non-stationary classical-quantum channels. The main observation that enables us to translate the classical proofs is the following Lemma.
\begin{lemma}{}\label{lem:lower}
Let $\mathcal W_1$ and $\mathcal W_2$ be two classical-quantum binary and symmetric channels with $I(\mathcal W_1),I(\mathcal W_2)\in [a,b]$, then the following holds
\begin{equation}
I(\langle \mathcal W_1,\mathcal W_2\rangle^+) - I(\langle \mathcal W_1,\mathcal W_2\rangle^-) \geq |I(\mathcal W_1)-I(\mathcal W_2)| + \mu(a,b), \label{ent-diff}
\end{equation}
where $\mu(a,b)>0$ whenever $0<a<b<\log2$. 
\end{lemma}
\begin{proof}
The statement follows from the results in Section~\ref{lower-bound}, in particular Remark~\ref{equality-condition-for-proof}. To see this, note that
\begin{align}
I(\langle &\mathcal W_1,\mathcal W_2\rangle^+) - I(\langle \mathcal W_1,\mathcal W_2\rangle^-)-|I(\mathcal W_1)-I(\mathcal W_2)| \nonumber\\
&=2\left(H(\langle \mathcal W_1,\mathcal W_2\rangle^-)-\max\{H(\mathcal W_1),H(\mathcal W_2)\}\right)\nonumber\\
&=2\left(H(X_1+X_1|B_1B_2)-\max\{H_1,H_2\}\right),\label{herewewantauniformlowerbound}
\end{align}
where the last line is written in the notation of Remark~\ref{equality-condition-for-proof}. Since our lower bound (\ref{firstFGmainresult}) from Theorem \ref{QuantumMrsGerberTheorem2DifferentStates} is continuous in $H_1,H_2$ and equals $0$ only on the boundary, given by the condition $H_1\in\{0,\log2\}$ or $H_2\in\{0,\log2\}$, we obtain a strictly positive uniform lower bound $\mu(a,b)>0$ on Eq.\ (\ref{herewewantauniformlowerbound}) for $H_1,H_2\in[\log2-b,\log2-a]$ with $0<a<b<\log2$ (see also Fig.\ \ref{fig-bound-with-h1h2}).
\end{proof}
In the usual setting of stationary channels it is enough to consider the two original channels $\mathcal W_1=\mathcal W_2=\mathcal W$ to be equal, in which case we can use the shorter notation $\mathcal W^\pm = \langle \mathcal W,\mathcal W\rangle^\pm$ and Equation~\eqref{ent-diff} simplifies to
\begin{equation}\label{defineDeltaW}
\Delta(\mathcal W) := I(\mathcal W^+) - I(\mathcal W^-) \geq \kappa(a,b),
\end{equation}
if $I(\mathcal W)\in[a,b]$. With this tool we are now ready to address the question of polarization for classical-quantum channel. First, we will look at stationary channels and prove polarization in the classical-quantum setting. As mentioned before this result was already achieved in~\cite{WG13}, but we will give an alternative simple proof based on~\cite{AT14}.
\begin{thm}
For any symmetric binary classical-quantum channel $\mathcal W$ and any $0<a<b<\log2,$ the following holds
\begin{align}
&\lim_{n\rightarrow\infty} \frac{1}{2^n} \#\{ s^n \in \{+,-\}^n : I(\mathcal W^{s^n})\in [0,a) \} = 1-I(\mathcal W)/\log2, \\
&\lim_{n\rightarrow\infty} \frac{1}{2^n} \#\{ s^n \in \{+,-\}^n : I(\mathcal W^{s^n})\in [a,b] \} = 0, \\
&\lim_{n\rightarrow\infty} \frac{1}{2^n} \#\{ s^n \in \{+,-\}^n : I(\mathcal W^{s^n})\in (b,\log2] \} = I(\mathcal W)/\log2.
\end{align}
\end{thm}
\begin{proof}
The proof follows essentially the one in~\cite{AT14} adjusted to the classical-quantum setting considered in our work. We will nevertheless state the important steps in the proof here. 
We start with a given classical-quantum channel $\mathcal W$ and arbitrary $0<a<b<\log2$. We define the following quantities 
\begin{align*}
\alpha_n(a) := \frac{1}{2^n} \#\{ s \in \{+,-\}^n : I(\mathcal W^{s})\in [0,a) \}, \\
\theta_n(a,b):=\frac{1}{2^n} \#\{ s \in \{+,-\}^n : I(\mathcal W^{s})\in [a,b] \}, \\
\beta_n(b):=\frac{1}{2^n} \#\{ s \in \{+,-\}^n : I(\mathcal W^{s})\in (b,\log2] \},
\end{align*}
where $s:=s^n$ to simplify the notation. Furthermore, we will need two additional quantities 
\begin{equation*}
\mu_n = \frac{1}{2^n} \sum_{ s \in \{+,-\}^n} I(\mathcal W^s)
\end{equation*}
and
\begin{equation*}
\nu_n = \frac{1}{2^n} \sum_{ s \in \{+,-\}^n} [I(\mathcal W^s)]^2.
\end{equation*}
Now, it follows directly from the chain rule (Equation~\eqref{qchain}) that 
\begin{equation*}
\mu_{n+1} = \mu_n = I(\mathcal W). 
\end{equation*}
It can also be seen that 
\begin{align*}
\nu_{n+1} &=  \frac{1}{2^{n+1}} \sum_{ s \in \{+,-\}^{n+1}} I(\mathcal W^s)^2 \\
&=  \frac{1}{2^{n}} \sum_{ t \in \{+,-\}^{n}} \frac{1}{2}[I(\mathcal W^{t+})^2 + I(\mathcal W^{t-})^2] \\
&=  \frac{1}{2^{n}} \sum_{ t \in \{+,-\}^{n}} I(\mathcal W^{t})^2 +\left(\frac{1}{2}\Delta(\mathcal W^{t})\right)^2 \\
&\geq \nu_n + \frac{1}{4}\theta_n(a,b)\kappa(a,b)^2,
\end{align*}
where $\Delta(\mathcal W)$ has been defined in (\ref{defineDeltaW}) and we take $\kappa(a,b)>0$ from Lemma \ref{lem:lower}. It follows that $\nu_n$ is monotonically increasing and since it is bounded, also converging. Particularly we can use it to bound $\theta_n(a,b)$ by
\begin{equation}
0\leq \theta_n(a,b) \leq 4\frac{\nu_{n+1} - \nu_n}{\kappa(a,b)^2}
\end{equation}
and therefore conclude that $\lim_{n\rightarrow\infty}\theta_n(a,b) = 0$. Next, we show that 
\begin{align}
I(\mathcal W) = \mu_n &\leq a \alpha_n(a) + b\theta_n(a,b) + (\log2)\beta_n(b) \\
&= a + (b-a)\theta_n(a,b) + (\log2-a)\beta_n(b),
\end{align}
thus by taking $n$ to infinity and $a$ infinitesimally small it, follows that
\begin{equation}
\liminf_{n\rightarrow\infty} \beta_n(b) \geq I(\mathcal W)/\log2. 
\end{equation}
Similarly upper bounding $1-\mu_n$ leads to
\begin{equation}
\liminf_{n\rightarrow\infty} \alpha_n(a) \geq 1-I(\mathcal W)/\log2. 
\end{equation}
Finally, the original claim follows from the fact that $\alpha_n(a) + \beta_n(b) \leq 1$.
\end{proof}
Now we will look at classical-quantum polar codes for non-stationary channels, following the treatment in \cite{AT14}. Instead of a fixed channel $\mathcal W$, we start with a collection of channels $\mathcal W_{0,t}$, where the first index numbers the coding step and the second the channel position. From here we can define the coding steps similar to the classical case recursively as
\begin{align}
\mathcal W_{n,Nm+j} &= \langle \mathcal W_{n-1,Nm+j} ,\mathcal W_{n-1,Nm+N/2+j} \rangle^- \\
\mathcal W_{n,Nm+N/2+j} &= \langle \mathcal W_{n-1,Nm+j} ,\mathcal W_{n-1,Nm+N/2+j} \rangle^+,
\end{align}
with $n\geq 1$, $N=2^n$, $0\leq j\leq N/2-1$ and $m$ numbering the multiple blocks at a given step (which get combined at later polarization steps).
With these definitions we can state the result for non-stationary channels.  
\begin{thm}
For any collection of symmetric binary classical-quantum channels $\mathcal W_{0,t}$ and any $0<a<b<\log2$, the following holds
\begin{align}
&\lim_{n\rightarrow\infty}\lim_{T\rightarrow\infty} \frac{1}{T} \#\{ 0\leq t<T : I(\mathcal W_{n,t})\in [0,a) \} = 1-\mu/\log2, \\
&\lim_{n\rightarrow\infty}\lim_{T\rightarrow\infty} \frac{1}{T} \#\{  0\leq t<T : I(\mathcal W_{n,t})\in [a,b] \} = 0, \\
&\lim_{n\rightarrow\infty}\lim_{T\rightarrow\infty} \frac{1}{T} \#\{  0\leq t<T : I(\mathcal W_{n,t})\in (b,\log2] \} = \mu/\log2,
\end{align}
with $\mu=\lim_{T\rightarrow\infty}\frac{1}{T} \sum_{t<T}I(\mathcal W_{0,t})$, under the condition that $\mu$ is well defined.
\end{thm}
\begin{proof}
Again the proof will follow very closely the one in~\cite{AT14}. For the sake of brevity we will only outline the crucial steps and refer to~\cite{AT14} for more details.
We start again by defining the fractions $\alpha_n(a)$, $\theta_n(a,b)$ and $\beta_n(b)$ as the quantities under investigation before taking the limit over $n$. Furthermore, we will, similar to the previous proof, define the quantities
\begin{equation}
\mu_n=\lim_{T\rightarrow\infty}\frac{1}{T} \sum_{t<T}I(\mathcal W_{n,t})
\end{equation}
and
\begin{equation}
\nu_n=\liminf_{T\rightarrow\infty}\frac{1}{T} \sum_{t<T}I(\mathcal W_{n,t})^2.
\end{equation}
Note that from the assumption that the limit in $\mu = \mu_0$ exists, it also follows that all $\mu_n$ are well defined, with the reasoning being the same as in the classical case (see~\cite{AT14}).
Therefore, it also follows that $\mu_n=\mu_{n+1}$ as in the previous proof.  \\
Next we are looking at the change in variance when combining two channels. From the general Lemma~\ref{lem:lower} we can also deduce the following statement
\begin{align}
\Delta^2(\mathcal W_1,\mathcal W_2) &:= \frac{1}{2} [I(\langle \mathcal W_1,\mathcal W_2\rangle^-)^2 + I(\langle \mathcal W_1,\mathcal W_2\rangle^+)^2] - \frac{1}{2} [I(\mathcal W_1)^2 + I(\mathcal W_2)^2] \nonumber\\ 
&\geq \zeta(a,b),
\end{align}
if $I(\mathcal W_1),I(\mathcal W_2)\in[a,b]$, where $\zeta(a,b)>0$ whenever $0<a<b<1$. This is sufficient to conclude that $\nu_{n+1}\geq\nu_n$; however more work is needed to relate their difference to $\theta_n$. It is easy to see that in special cases, for example when every second channel is already extremal, the combination of different channels might not lead to a positive $\zeta(a,b)$ bounding $\nu_{n+1}-\nu_n$. Nevertheless, even those seemingly ineffective coding steps deterministically permute the channels and therefore allow for progress in later coding steps. This has been made precise in~\cite{AT14} in a corollary that we will also use here. It states that if $\theta_n(a,b) > {{k}\choose{\left\lfloor k/2\right\rfloor}}/2^k:=\epsilon_k$, then
\begin{equation}
\nu_{n+k} \geq \nu_n + \delta,
\end{equation}
where $\delta>0$ is a quantity that depends only on $k$, $\theta_n$, $a$ and $b$. The proof in~\cite{AT14} is entirely algebraic and works also in our generalized setting. From this we can conclude that for every $k\in{\mathbb N}$, $\theta_n\leq\epsilon_k$ holds for sufficiently large $n$. Therefore,
\begin{equation}
\lim_{n\rightarrow\infty}\lim_{T\rightarrow\infty} \frac{1}{T} \#\{  0\leq t<T : I(\mathcal W_{n,t})\in [a,b] \} = 0,
\end{equation}
since $\lim_{k\rightarrow\infty}\epsilon_k = 0$. \\
The claims about $\alpha_n$ and $\beta_n$ now follow from the same reasoning as in the stationary case. 
\end{proof}

\subsection{Speed of polarization}
Applying our quantitative result from Theorem \ref{QuantumMrsGerberTheorem1StateTwice} to the entropy change of binary-input classical-quantum channels under the polar transform, we now prove a \emph{quantitative} result on the speed of polarization for i.i.d.\ binary-input classical-quantum channels. For our proof, we adapt the method of \cite{GX15} to the $\sim H/(-\log H)$ lower bound guaranteed by our Equation~\eqref{FonlyMoreConvenientLowerBound}, which is somewhat worse than the linear lower bound $\sim H$ for the classical-classical case in \cite[see in particular Lemma 6]{GX15}; this is the reason that our following result does not guarantee a polynomial blocklength $\sim(1/\varepsilon)^\mu$, but only a subexponential one $\sim(1/\varepsilon)^{\mu\log1/\varepsilon}$. (Here $\epsilon$ is the gap to the symmetric capacity.) However, under our Conjecture \ref{QMGL}, we can show the same polynomial blocklength result as in \cite{GX15} for classical-classical channels (as we will point out in Remark \ref{poly-blocklength-under-conjecture}). Note that we do not make any claim about efficient decoding of classical-quantum polar codes (e.g.\ with a circuit of subexponential size), which remains an open problem.

\begin{thm}[Blocklength subexponential in gap to capacity suffices for classical-quantum binary polar codes]There is an absolute constant $\mu<\infty$ such that the following holds. For any binary-input classical-quantum channel $\mathcal W$, there exists $a_{\mathcal W}<\infty$ such that for all $\varepsilon>0$ and all powers of two $N\geq a_{\mathcal W}(1/\varepsilon)^{\mu\log1/\varepsilon}$, a polar code of blocklength $N$ has rate at least $I(\mathcal W)-\varepsilon$ and block-error probability at most $2^{-N^{0.49}}$, where $I(\mathcal W)$ is the symmetric capacity of $\mathcal W$.
\end{thm}
\begin{proof}Our proof follows the proofs of \cite[Propositions 5 and 10]{GX15} (``rough'' and ``fine'' polarization). The main reason why we can guarantee only a subexponential scaling here, lies in the rough polarization step (\cite[Proposition 5]{GX15}). In the following, we outline only the main differences to the proofs in \cite{GX15} which are responsible for the altered scaling. As in \cite{GX15}, we define $T(\mathcal W):=H(\mathcal W)(1-H(\mathcal W))$. Then \cite[Lemma 8]{GX15} is modified to,
\begin{align*}
\underset{i~\text{mod}~2}{{\mathbb E}}[T(\mathcal W^{(i)}_{n+1}]\leq T(\mathcal W^{(\lfloor i/2\rfloor)}_n)-\kappa\frac{T(\mathcal W^{(\lfloor i/2\rfloor)}_n)}{-\log T(\mathcal W^{(\lfloor i/2\rfloor)}_n)}\,
\end{align*}
with some $\kappa>0$. We obtain the same relation for the full expectation values using convexity (similar to the equation in the proof of \cite[Corollary 9]{GX15}):
\begin{align*}
\underset{i}{{\mathbb E}}[T(\mathcal W^{(i)}_{n+1})]\leq \underset{i}{{\mathbb E}}[T(\mathcal W^{(i)}_n)]-\kappa\frac{\underset{i}{{\mathbb E}}[T(\mathcal W^{(i)}_n)]}{-\log\underset{i}{{\mathbb E}}[T(\mathcal W^{(i)}_n)]}.
\end{align*}
This now does not anymore guarantee that the decrease of $\underset{i}{{\mathbb E}}[T(\mathcal W^{(i)}_n)]$ is exponential in $n$, as in \cite[Corollary 9]{GX15} which was obtained from the recursion $\underset{i}{{\mathbb E}}[T(\mathcal W^{(i)}_{n+1})]\leq \underset{i}{{\mathbb E}}[T(\mathcal W^{(i)}_n)]-\kappa\underset{i}{{\mathbb E}}[T(\mathcal W^{(i)}_n)]$ (or the same recursion for $\underset{i}{{\mathbb E}}[\sqrt{T(\mathcal W^{(i)}_{n})}]$). Thus, instead of the differential equation $\frac{d}{dn}f(n)=-\kappa f(n)$, the behaviour here is goverened by the equation $\frac{d}{dn}f(n)=-\kappa\frac{f(n)}{-\log f(n)}$. This differential equation has the solution $f(n)=\exp[-\sqrt{2\kappa n+(\log f(0))^2}]$ (note, $f(n)\leq1$ for all $n$) and we therefore obtain the following bound:
\begin{align*}
\underset{i}{{\mathbb E}}[T(\mathcal W^{(i)}_n)]\leq e^{-\sqrt{2\kappa n+(\log T(\mathcal W_0^{(0)}))^2}}\leq e^{-\sqrt{2\kappa n}},
\end{align*}
guaranteeing the expectation value of $T(\mathcal W^{(i)}_n)$ to decrease at least superpolynomially with the number of polarization steps $n$.

This expectation value will thus be smaller than any $\delta>0$ if only the number of polarization steps satisfies $n\geq\frac{1}{2\kappa}\left(\log\frac{1}{\delta}\right)^2\sim\left(\log\frac{1}{\delta}\right)^2$. This expression can now be connected with the ``fine polarization step'' \cite[Proposition 10]{GX15} since for any fixed power $\delta\sim\varepsilon^p$ (with $\varepsilon$ from the statement of the theorem) we again obtain that $n\geq\widetilde{\mu}\left(\log\frac{1}{\varepsilon}\right)^2$ with some constant $\widetilde{\mu}$ suffices. Since the number $n$ of polarization steps is related to the blocklength $N$ via $N=2^n$, we find that the constructed polar code has the desired properties as soon as the blocklength satisfies $N\geq2^{\widetilde{\mu}(\log1/\varepsilon)^2}=(1/\varepsilon)^{\mu\log1/\varepsilon}$ (with $\mu=\widetilde{\mu}\log2$). The constant $a_{\mathcal W}$ from the theorem statement accounts for the fact that the above analysis is only valid for sufficiently small $\varepsilon$.

It is instructive to compare the reasoning in the previous paragraph with the blocklength result obtained in \cite{GX15}. The bound obtained from $f(n)$ in this case is $\underset{i}{{\mathbb E}}[T(\mathcal W^{(i)}_n)]\leq e^{-\kappa n}$, so that $n\geq\widetilde{\mu}\log\frac{1}{\varepsilon}$ suffices for $\underset{i}{{\mathbb E}}[T(\mathcal W^{(i)}_n)]\leq\varepsilon^p$. This shows that a blocklength $N\geq2^{\widetilde{\mu}\log1/\varepsilon}=(1/\varepsilon)^{\mu}$ is sufficient.
\end{proof}

\begin{rem}[Polynomial blocklength suffices under Conjecture \ref{QMGL}]\label{poly-blocklength-under-conjecture}
If Conjecture \ref{QMGL} holds, then one can prove the same polynomial blocklength result as \cite[Theorem 1]{GX15} for classical-quantum channels as well. The only part of the proof which has to be changed is \cite[Lemma 6]{GX15}, where the classical Mrs.\ Gerber's Lemma is to be replaced by Conjecture \ref{QMGL}. However, this change does not even affect the numerical value of $\theta$ that can be chosen in \cite[Lemma 6]{GX15}, since our conjectured optimal lower bound in the classical-quantum case is simply a symmetrization of the classical lower bound.
\end{rem}

%
\chapter{Log-determinant inequalities and recoverability in infinite dimensions}\label{Covariance}

So far, all of the chapters of this thesis discussed problems on finite dimensional systems (although many of the results hold in infinite dimensions as well). 
In this chapter, we will turn our focus to infinite dimensional systems and we try to develop some similar concepts, in particular connected to entropy inequalities and recoverability. 
Before we start with the main results in the following sections, we will introduce the important concepts and notations. 
The key of the reasoning in this chapter is to associate  an $n$-dimensional Gaussian random variable
$X\in\mathds{R}^n$ with mean $0$ and variance (aka covariance matrix)
$\operatorname{Var} X = \mathds{E}\  X X^\intercal = A$ to each positive
matrix $A\in\mathcal{M}_n(\mathds{R})$.
The density of $X$ is given by
\begin{equation}
 p_A(x) = \frac{e^{-\frac12 x^\intercal A^{-1} x}}{\sqrt{(2\pi)^n\det A}} . \label{gauss}
\end{equation}
This has the nice feature that for two independent Gaussian random
variables $X$ and $Y$ with a $0$ mean and covariance matrices $A$ and $B$ respectively,
the sum $A+B$ is the covariance matrix of $X+Y$.
Here, we consider only real matrices since they are more relevant for the applications we are interested in. However, all the results we find apply also to the Hermitian case with minor modifications.

Under the density~\eqref{gauss}, the differential entropy
$h(X) \coloneqq -\int d^n x\, p_A(x) \ln p_A(x)$ of~\eqref{gauss} takes the form
\begin{equation}
  h(X) = \frac{1}{2}\ln\det A + \frac{n}{2}\left(\ln 2\pi + 1 \right) ,
  \label{ent}
\end{equation}
while the relative entropy
$D(p_A\|p_B) \coloneqq \int d^n x\, p_A(x) \ln \frac{p_A(x)}{p_B(x)}$ is given by
\begin{equation}
  D(p_A\|p_B) = \frac{1}{2} \ln\frac{\det B}{\det A} + \frac{1}{2} \Tr(B^{-1}\! A) - \frac{n}{2} .
  \label{rel ent}
\end{equation}
Here and in the rest of the chapter $\ln$ refers to the natural logarithm. The positivity of~\eqref{rel ent} as a function of the matrices $A$ and $B$ can be
seen as an instance of Klein's inequality applied to the natural logarithm~\cite{Wehrl}.

In this picture, general inequalities involving entropies can be turned into
inequalities involving determinants thanks to~\eqref{ent} and~\eqref{rel ent}.
A prominent example of the usefulness of this approach is constituted by
\emph{strong subadditivity} (SSA), the basic ``Shannon-type'' entropy
inequality~\cite{Yeung}.
Consider a Gaussian distributed vector
$X_{ABC} = (X_A,X_B,X_C)^\intercal \in \mathds{R}^{n_A+n_B+n_C}$
with the covariance matrix $V_{ABC}$:
\begin{equation}
  V_{ABC} = \begin{pmatrix} A & X & Y \\
                            X^{\intercal} & B & Z \\
                            Y^{\intercal} & Z^{\intercal} & C \end{pmatrix}
    \geq 0 ,
  \label{global CM}
\end{equation}
The SSA inequality $I(X_{A}:X_{B}|X_{C}) \geq 0$ then reads
\begin{equation}
  \ln\det V_{AC} + \ln\det V_{BC} - \ln\det V_{ABC} - \ln\det V_C \geq 0 ,
  \label{SSA}
\end{equation}
where the local reductions $V_{AC}$, $V_{BC}$ and $V_{C}$ are the principal
submatrices of $V_{ABC}$ corresponding to the components $AC$, $BC$ and $C$,
respectively:
\begin{equation}
  V_{AC} = \begin{pmatrix} A & Y \\ Y^{\intercal} & C \end{pmatrix}, \quad
  V_{BC} = \begin{pmatrix} B & Z \\ Z^{\intercal} & C \end{pmatrix}, \quad
  V_C = C.
  \label{eq:reductions}
\end{equation}
We observe that since~\eqref{SSA} is balanced, the
contribution of the inhomogeneous second terms of~\eqref{ent} cancel out.

Inequality~\eqref{SSA} was proven for the first time in~\cite{Ando09} (see
also~\cite[Sec.~4.5]{PetzBook}).
Incidentally, the differential R\'enyi-$\alpha$ entropy of a Gaussian random variable $X$ with density
$p_A(x)$, i.e.~$H_\alpha(X) \coloneqq \frac{1}{1-\alpha}\ln \int d^n x\, p_A(x)^\alpha$,
is given by
\begin{equation*}
H_\alpha(X) = \frac{1}{2} \ln\det A + \frac{n}{2}\left( \ln 2\pi + \frac{1}{\alpha-1}\ln\alpha \right),
\end{equation*}
showing that all the differential R\'enyi entropies of Gaussian random vectors are
essentially equivalent to the differential Shannon entropy, up to a characteristic
universal additive offset.
In view of this and the above remarks, we are motivated,
given a vector valued random variable $X$ with covariance
matrix $V$, to refer from now on to the quantity
\begin{equation}
  \label{logdetent}
  M(X) \coloneqq M(V) \coloneqq \frac12 \ln \det V ,
\end{equation}
as the \textit{log-det entropy} of $V$ (or analogously $X$).
Likewise, for a bipartite covariance matrix $V_{AB} > 0$ we refer to
\begin{equation}\begin{split}
  I_{M}(A:B)_{V} &\coloneqq \frac12 \ln\frac{\det V_{A}\det V_{B}}{\det V_{AB}} \\
                   &=      M(V_{A})+M(V_{B})-M(V_{AB}) ,
  \label{logdetMI}
\end{split}\end{equation}
as the {\em log-det mutual information}, and for a tripartite covariance matrix $V_{ABC}>0$ we refer to
\begin{equation}\begin{split}
  I_{M}(A:B|C)_{V} &\coloneqq \frac12 \ln\frac{\det V_{AC}\det V_{BC}}{\det V_{C}\det V_{ABC}} \\
                   &=      M(V_{AC})+M(V_{BC}) -M(V_{ABC})-M(V_{C}) ,
  \label{I_2-matrix}
\end{split}\end{equation}
as the {\em log-det conditional mutual information}.


Every (balanced) entropic inequality thus yields a corresponding log-determinant
inequality for positive block matrices~\cite{Chan-balanced}. Thanks
to the work of Zhang and Yeung~\cite{ZhangYeung}
and followers~\cite{Dougherty,Matus}, infinitely many
independent such ``non-Shannon-type inequalities'' are known by now.
The question of what the precise
constraints on the determinants of the $2^n$ principal
submatrices of a positive matrix of size $n\times n$ are, has been raised much earlier either directly in a matrix setting~\cite{JohnsonBarrett}
or more recently in the guise of the balanced entropy inequalities of Gaussian
random variables (both real valued or vector valued)~\cite{Hassibi,Shadbakht}.
Remarkably, the latter papers show that while the entropy region of three Gaussian
real random variables is convex but not a cone, the entropy region of three Gaussian
random vectors is a convex cone and that the \emph{linear} log-det inequalities
for three Gaussian random variables (and equivalently Gaussian random vectors)
are the same as the inequalities for the differential entropy of any three
variables -- which in turn coincide with the Shannon inequalities,
cf.~\cite{Yeung,Chan-balanced}.
It is conjectured that the same identity between Gaussian vector inequalities
and general differential inequalities holds for any number of parties.

In this chapter, we will focus on a deeper investigation of the SSA inequality~\eqref{SSA}, which is nowadays widely regarded as one of the cornerstones upon which quantum information theory is built~\cite{book2000mikeandike}. Our analysis rests crucially on the connection between Gaussian random variables and positive definite matrices we have outlined here, which allows us to use tools taken from matrix analysis~\cite{Bhatia} to explore properties of the log-det conditional mutual information~\eqref{I_2-matrix}.
We find a particular strengthening of the SSA inequality in the following form as a matrix inequality:
\begin{equation}
V_{ABC} / V_{BC}\, \leq V_{AC} / V_C . \label{INEQ 1}
\end{equation}
Here, we use the powerful concept of the \emph{Schur complement} of a $2\times 2$-block matrix
$V = \lmatrix A & X \\ X^\intercal & B \rmatrix$ with respect to the
principal minor $A$, defined as
\begin{equation}
\label{Schur definition}V/A\coloneqq B - X^{\intercal}A^{-1} X .
\end{equation}
We will go into more details about the properties of the Schur complement in the next section.

Our concrete interest in~\eqref{SSA} is mostly motivated by its applications
in quantum information theory with continuous
variables~\cite{Adesso14}, as first explored in~\cite{Adesso12,Gross}.
In Section \ref{sec:Renyi-2-Gaussian-squashed} we will give a detailed introduction to the quantum setting. We will make use of the fact that every continuous variable quantum state $\rho$ of $n$ modes, subject to mild regularity
conditions, has a $2n\times 2n$-covariance matrix $V$ of the phase space variables.

The rest of this chapter is structured as follows. In Section~\ref{sec exact recov}
we derive various characterisations of the case of saturation of SSA with
equality.
Then, in Section~\ref{sec:Gaussian-recov} we turn to the case of near-saturation,
which leads to the theory of recovery maps; in Section~\ref{sec:I_M-lower-bound}
we exploit those results to derive simple and faithful lower bounds on
the log-det conditional mutual information. Up to that point, all results
hold for general covariance matrices $V>0$.
After that, in Section~\ref{sec:Renyi-2-Gaussian-squashed} we turn our
attention to quantum Gaussian states and their phase space covariance
matrices, which need to satisfy additional constraints stemming from the uncertainty principle and the canonical commutation relations. There, we introduce a measure of entanglement for quantum Gaussian states
based on the log-det conditional mutual information and prove its faithfulness and additivity. Quite remarkably, we show that the measure coincides with the \mbox{R\'enyi-$2$} Gaussian entanglement of formation introduced in~\cite{Adesso12}, equipping the latter with an interesting operational interpretation in the context of recoverability.

\section{Mathematical tools: Schur complement and geometric mean}

Two of the elementary tools we will use in the remainder of this chapter are the Schur complement and the geometric mean between positive definite matrices. In this section we will state some useful properties and observations.

Let's start with the Schur complement~\cite{Schur}. The Schur complement is an operation that takes as input a $n\times n$ matrix $M$ and one of its $k\times k$ principal submatrices~\footnote{The shorthand $X\sqsubset Y$ means $X$ is a square submatrix of $Y$} $A \sqsubset M$, and outputs a $(n-k)\times (n-k)$ matrix $M/A$. Given a $2\times 2$-block matrix
$V = \lmatrix A & X \\ X^\intercal & B \rmatrix$, the complement with respect to the
principal minor $A$ is given by $V/A$ as defined in~\eqref{Schur definition}.

Its significance relies on the (elementary) fact that $V$ as a quadratic form is congruent to
$S^\intercal V S = A \oplus V/A$, via the unideterminantal transformation
$S = \lmatrix \mathds{1}\, & -A^{-1}X \\ 0\, & \mathds{1} \rmatrix$. From this
the factorization formula
\begin{equation}
  \det V = (\det A)(\det V/A)
  \label{det factor}
\end{equation}
follows, which shows how~\eqref{INEQ 1} implies the SSA inequality~\eqref{SSA}.
A related property is that of congruence invariance: we have
\begin{equation}
\left(\begin{smallmatrix} N_1 & \\ & N_2 \end{smallmatrix}\right) V \left(\begin{smallmatrix} N_1^T & \\ & N_2^T \end{smallmatrix}\right) \Big/ N_1^T A N_1  \geq N_2 \left( V/A \right) N_2^T,
\end{equation}
for all $N_1, N_2$, with equality if $N_1$ is invertible.
 
From a point of view of linear algebra, Schur complements arise naturally when one wants to express the inverse of a block matrix in a compact form. Namely, for a matrix $V$ partitioned as above one can prove the useful formula~\cite{SCbook}
\begin{equation}
  V^{-1} = \begin{pmatrix} A^{-1} + A^{-1} X (V/A)^{-1} X^\intercal  A^{-1} & -A^{-1} X (V/A)^{-1} \\[0.7ex]
                                            -(V/A)^{-1} X^\intercal  A^{-1} & (V/A)^{-1} \end{pmatrix} .
  \label{inv}
\end{equation}
Naturally, an analogous expression holds with $A$ and $B$ interchanged.
Incidentally, many useful matrix identities can be easily derived from this latter fact.

Schur complements of positive definite matrices enjoy numerous other
useful relations. First of all, the positivity condition itself can be expressed in terms of Schur
complements as
\begin{equation}
V = \begin{pmatrix} A & X \\ X^\intercal & B \end{pmatrix} > 0\quad \Longleftrightarrow\quad A>0\ \text{and}\ V/A>0 .
\label{Schur pos}
\end{equation}
From this the variational representation
\begin{equation}
V/A = \max \big\{ \tilde{B}: V \geq 0 \oplus \tilde{B} \big\} ,
\label{Schur var}
\end{equation}
follows easily. The meaning of~\eqref{Schur var} is that the matrix set on the right hand side has a unique maximal element with respect to the L\"owner partial order (a nontrivial fact in itself) and that this maximum coincides with the left hand side, which means in particular that  $V\mapsto V/A$  is monotonically increasing and concave, while $V\mapsto (V/A)^{-1}$ is decreasing and convex.


Interestingly, it follows from the latter property that $\log\det(V/A)$ is concave in $V$ thanks to the operator concavity of the logarithm. This leads to a simple proof of the central finding of \cite{Adesso}, i.e. the inequality
\begin{equation} \label{log det ineq}  \log\det V_{AC} +\log\det V_{BC} - \log\det V_{A} - \log\det V_{B} \geq 0\, ,  \end{equation}
valid for any quantum CM $V_{ABC}$. This is obtained by noticing that (\ref{log det ineq}) is saturated for pure states and rewriting the left-hand side as $\log\det (V_{AC}/V_A) + \log\det (V_{BC}/V_B)$, which is a concave function of $V_{ABC}$. 

This suggests that the  Schur complement of CMs can define a natural notion of conditional covariance, as previously noted for classical Gaussian variables \cite{simoprl}.
Hence we will study the Schur complement  $V_{AB}/V_B$, thereby proving that many well-known properties of the standard conditional entropy $H(A|B)=H(AB)-H(B)$, where $H$ denotes respectively Shannon or von Neumann entropy for a classical or quantum system,
have a straightforward equivalent within this framework.

We start by recalling that a canonical formulation of strong subadditivity in classical and quantum information theory is $H(A|BC) \leq H(A|C)$, i.e.~partial trace on the conditioning system  {increases} the conditional entropy \cite{ArakiLieb,WehrlR,Lieb2002,NP04}. Guided by our formal analogy, our first result of this chapter is thus a generalization of the SSA inequality.

\begin{thm}[Partial trace in the denominator  {increases} Schur complement] \label{pt incr sch}
If $V_{ABC} \geq 0$ is any tripartite CM, then
\begin{equation} V_{ABC} / V_{BC}\ \leq\ V_{AC} / V_C\, . \label{INEQ 1} \end{equation}
\end{thm}

\begin{proof}
Since $V_{ABC}\geq W_A \oplus 0_{BC}$ implies $V_{AC}\geq W_A\oplus 0_C$, employing the variational representation we find
$V_{ABC} / V_{BC}\ =\ \max\big\{ W_A:\, V_{ABC}\geq W_A\oplus 0_{BC} \big\}\ \leq\ \max\big\{ W_A:\, V_{AC}\geq W_A\oplus 0_C \big\}\ =\ V_{AC} / V_C$.
\end{proof}

Clearly, taking the determinant of \eqref{INEQ 1} and applying the factorization property of the Schur complement yields the positivity of the log-det conditional mutual information  immediately.

Notice further that the invariance of $V_{AB}/V_B$ under symplectic operations on $B$ and its monotonicity under partial trace, suffice to guarantee its monotonicity under general deterministic (i.e.~trace-preserving) Gaussian channels  $\Gamma_B$ on $B$:
\begin{align*}
(\mathds{1}_A&\oplus \Gamma_B)(V_{AB}) \,\big/\, \Gamma_B (V_B)  \\
&= \big( S_{BC}\,( V_{AB}\oplus \sigma_C)\, S_{BC}^T \big)_{AB} \Big/ \big( S_{BC}\, (V_B\oplus \sigma_C)\, S_{BC}^T \big)_B \\
&\geq \big( S_{BC}\, (V_{AB}\oplus \sigma_C)\, S_{BC}^T \big) \Big/ \big( S_{BC}\, (V_B\oplus \sigma_C)\, S_{BC}^T \big) \\
&= (V_{AB}\oplus \sigma_C) \big/ (V_B\oplus \sigma_C) \\[2pt]
&= V_{AB} / V_B.
\end{align*}
But there is more: perhaps surprisingly,  the Schur complement is also  monotonically increasing under general non-deterministic \textit{classical} (i.e.~non quantum-limited) Gaussian operations on $B$. We recall that any such map acts at the level of CMs as \cite{nogo1,nogo2,nogo3}
\begin{equation}\label{CP gauss}
\Gamma_{B\rightarrow B'}:\ V_B\longmapsto \gamma_{B'} - \delta_{BB'}^T\,\left({\gamma_B+V_B}\right)^{-1}\,\delta_{BB'}\, ,
\end{equation}
where $\gamma_{BB'}=\left(\begin{smallmatrix} \gamma_B & \delta_{BB'} \\ \delta_{BB'}^T & \gamma_{B'} \end{smallmatrix}\right)>0$ is a positive matrix pertaining to a bipartite system $BB'$. If $\gamma_{BB'}$ is also a valid quantum CM (which we will define later), then \eqref{CP gauss} corresponds to a (non-deterministic) completely positive Gaussian channel, but this restricting hypothesis plays no role in stating the following general result.

\begin{thm}[Classical Gaussian maps in the denominator increase Schur complement] \label{CP incr sch}
If $\Gamma_{B\rightarrow B'}$ is a non-deterministic classical Gaussian map as in \eqref{CP gauss}, with $\left(\begin{smallmatrix} \gamma_B & \gamma_{BB'} \\ \gamma_{BB'}^T & \gamma_{B'} \end{smallmatrix}\right)>0$, then
\begin{equation*}
\Gamma_{B\rightarrow B'}(V_{AB}) \big/ \Gamma_{B\rightarrow B'}(V_B)\ \geq\ V_{AB} / V_B\, .
\end{equation*}
\end{thm}

\begin{proof}
Observing that \eqref{CP gauss} can be rewritten as $\Gamma_{B\rightarrow B'}:\ V_B\longmapsto (\gamma_{BB'}+V_B) \big/ (\gamma_B+V_B)$,
we obtain: 
\begin{align*}
\Gamma_{B\rightarrow B'}&(V_{AB}) \big/ \Gamma_{B\rightarrow B'}(V_B) \\
&=  \big((\gamma_{BB'}+V_{AB}) / (\gamma_B+V_B)\big) \Big/ \big((\gamma_{BB'}+V_B) / (\gamma_B+V_B) \big) \\
&= (\gamma_{BB'}+V_{AB})  \big/ (\gamma_{BB'}+V_B) \geq V_{AB} / V_B,
\end{align*}
where we used  the quotient property of covariance matrices together with the bound $\left.\left(\begin{smallmatrix} A & X \\ X^T & B+\sigma \end{smallmatrix}\right) \Big/ (B+\sigma)\ \geq\ \left(\begin{smallmatrix} A & X \\ X^T & B \end{smallmatrix}\right) \Big/ B\right.$.
\end{proof}

Another useful property is the additivity of ranks under Schur complements:
\begin{equation}
\rk V = \rk A + \rk (V/A) .
\label{rank add}
\end{equation}
For more details on Schur complements and applications thereof in matrix analysis
and beyond, we refer the reader to the book~\cite{SCbook}. \\


Another fundamental tool we shall take from matrix analysis is the concept of \emph{geometric mean} between two positive definite matrices $A,B>0$, usually denoted by $A\# B$~\cite{geomoriginal, Ando79}. As done in~\eqref{Schur var} for the Schur complement, also the geometric mean is most conveniently defined using a variational approach. Namely, one has
\begin{equation}
A\# B \coloneqq \max\{X=X^{\intercal}:\, A\geq XB^{-1} X\}\, . \label{geom var}
\end{equation}
From~\eqref{geom var} it is apparent, how $A\# B$ is covariant with respect to matrix congruence, i.e.
\begin{equation}
\left(SAS^\intercal \right)\# \left( SBS^\intercal \right) = S (A\# B) S^\intercal
\label{geom cov congr}
\end{equation}
for all invertible $S$. Moreover, through standard algebraic manipulations it is possible to write the explicit solution of~\eqref{geom var} as
\begin{equation}
A\# B = A^{1/2} \left( A^{-1/2} B A^{-1/2} \right)^{1/2} A^{1/2} . \label{geom expl}
\end{equation}
An excellent introduction to the theory of matrix means can be found in~\cite[Chapter 4]{Bhatia}. Here, we limit ourselves to briefly discuss an interesting interpretation of the geometric mean. We can turn the manifold of positive definite matrices into a Riemannian manifold by introducing on the tangent space the metric $ds^2\coloneqq \Tr[(A^{-1} dA)^2]$ (sometimes called ``\emph{trace metric}"). It turns out that the geodesic connecting two positive matrices $A$ and $B$ in this metric, parametrised by $t\in [0,1]$, is given by
\begin{equation}
\gamma(t) = A^{1/2} \left( A^{-1/2} B A^{-1/2} \right)^{t} A^{1/2} \eqqcolon A\#_t B,
\label{geom geod}
\end{equation}
sometimes called the \emph{weighted geometric mean}. From this we see in particular that $A\# B$ is nothing but the geodesic midpoint between $A$ and $B$. An easy consequence of the above expression is the determinantal identity
\begin{equation}
\det (A\#_t B) = (\det A)^{1-t} (\det B)^t .
\label{det geom}
\end{equation}
For more on this connection between geometric mean and Riemannian metric, see~\cite[Chapter 6]{Bhatia}.

\section{SSA saturation and exact recovery}
\label{sec exact recov}
Now we turn to studying the conditions under which the SSA inequality~\eqref{SSA}
is saturated with equality. A necessary and sufficient condition was already found
in~\cite{Ando09}~\footnote{For a comprehensive discussion, see~\cite{PetzBook}},
but here we present new proofs as well as alternative formulations which may provide new insights.

Let us start by fixing our notation concerning classical Gaussian channels,
whose actions can be described as follows. We denote the input random variable
by $X$, and consider an independent Gaussian variable $Z \sim P_K$, where
$P_{K}$ is a normal distribution with covariance matrix $K$ and zero mean.
Then the output variable $Y$ of the Gaussian channel $N$ is given by
$N(X) \coloneqq Y \coloneqq H X + Z$ for some matrix $H$ of appropriate size. At the level of covariance matrices this
translates to the description given in the last section, which we write here slightly different as
\begin{equation}
  N : V \longmapsto V' = H V H^\intercal + K ,
  \label{N}
\end{equation}
where the only constraint to be obeyed is $K\geq 0$.

The following theorem gathers some notable facts concerning log-det conditional mutual information, and provides a neat example of how useful the interplay between matrix analysis
and information theory with Gaussian random variables can be. We are going to employ these results extensively in the remainder of this chapter.
\begin{thm}
  \label{thm I cond geom}
  For all positive, tripartite matrices $V=V_{ABC}>0$, the following identities hold true:
  \begin{align}
    I_{M}(A:B|C)_{V} &= I_{M}(A:B)_{V_{ABC}/V_C} ,     \label{I cond Schur}\\
    I_{M}(A:B|C)_{V} &= I_{M}(A:B)_{V^{-1}} .                \label{I cond inv}
  \end{align}
  Furthermore, for all pairs of positive definite matrices $V_{AB},W_{AB}>0$, the log-det mutual information is convex on the geodesic connecting
  them as in~\eqref{geom geod}, i.e.
  \begin{equation}
  I_M(A:B)_{V\#_t W} \leq (1-t) I_M(A:B)_V + t I_M(A:B)_W . \label{I conv geod}
  \end{equation}
\end{thm}
\begin{proof}
Let us start by showing~\eqref{I cond Schur}. Using repeatedly the determinant factorisation property~\eqref{det factor}, we find
\begin{align*}
    & I_{M}(A:B)_{V_{ABC}/V_C} \\[0.8ex]
    &\quad = \frac12 \ln \frac{\det (V_{AB}/V_C) \det (V_{BC}/V_C)}{\det (V_{ABC}/V_C)} \\[0.8ex]
    &\quad = \frac12 \ln \frac{(\det V_{AB})(\det V_C)^{-1} (\det V_{BC})(\det V_C)^{-1}}{(\det V_{ABC})(\det V_C)^{-1}} \\[0.8ex]
    &\quad= \frac12 \ln \frac{(\det V_{AB}) (\det V_{BC} )}{(\det V_{ABC})(\det V_C)} \\[0.8ex]
    &\quad= I_M(A:B|C)_V .
\end{align*}
We now move to~\eqref{I cond inv}. The block inverse formulae~\eqref{inv} give us
\begin{align*}
  (V^{-1})_{AB} &= (V_{ABC}/V_{C})^{-1} ,  \\
  (V^{-1})_{A}  &= (V_{ABC}/V_{BC})^{-1} , \\
  (V^{-1})_{B}  &= (V_{ABC}/V_{AC})^{-1} .
\end{align*}
Putting everything together we find
\begin{align*}
    & I_{M}(A:B)_{V^{-1}} \\[0.8ex]
    &\quad = \frac12 \ln\frac{\det (V^{-1})_{A}\det (V^{-1})_{B}}{\det (V^{-1})_{AB}} \\[0.8ex]
    &\quad = \frac12 \ln\frac{\det (V_{ABC}/V_{BC})^{-1} \det (V_{ABC}/V_{AC})^{-1}}{\det (V_{ABC}/V_{C})^{-1}} \\[0.8ex]
    &\quad = \frac12 \ln\frac{\det (V_{ABC}/V_{C})}{\det (V_{ABC}/V_{BC}) \det (V_{ABC}/V_{AC})} \\[0.8ex]
    &\quad = \frac12 \ln\frac{(\det V_{ABC})(\det V_{C})^{-1}}
                        {(\det V_{ABC})(\det V_{BC})^{-1} (\det V_{ABC})(\det V_{AC})^{-1}} \\[0.8ex]
    &\quad = \frac12 \ln \frac{\det V_{AC}\det V_{BC}}{\det V_{ABC}\det V_{C}} \\[0.8ex]
    &\quad = I_{M}(A:B|C)_V ,
\end{align*}
which is what we wanted to show.

Finally, let us consider~\eqref{I conv geod}. A preliminary observation uses the monotonicity of the geometric mean under positive maps~\cite[Theorem~3]{Ando79}, written as $\Phi(V\# W)\leq \Phi(V)\# \Phi(W)$. Iterative applications of this inequality show that the same monotonicity property holds also for the weighted geometric mean~\eqref{geom geod} when $t$ is a dyadic rational, and hence (by continuity) for all $t\in [0,1]$. This standard reasoning is totally analogous to the one normally used to show that mid-point convexity and convexity are equivalent for continuous functions. Applying this to the positive map $\Phi(X)\coloneqq \Pi_A X \Pi_A^{\intercal}$, where $\Pi_A$ is the projector onto the $A$ components, yields $(V\#_t W)_A = \Pi_A (V\#_t W) \Pi_A^\intercal \leq V_A \#_t W_A$. Taking the determinant of both sides of the latter inequality and using the explicit formula~\eqref{det geom} for the right hand  side, we obtain $\det \left(V\#_t W \right)_A \leq \det\left(V_A \#_t W_A\right) = (\det V_A)^{1-t} (\det W_A)^t$. Together with the analogous inequality for the $B$ system, this gives
\begin{align*}
   & I_M(A:B)_{V\#_t W} \\[0.8ex]
   &\quad = \frac12 \ln \frac{\left(\det (V\#_t W)_A \right) \left( \det (V \#_t W)_B \right)}{\det (V\#_t W)_{AB}} \\[0.8ex]
   &\quad \leq \frac12 \ln \frac{(\det V_A)^{1-t} (\det W_A)^t (\det  V_B)^{1-t}  (\det W_B)^t}{(\det V_{AB})^{1-t} (\det W_{AB})^t} \\[0.8ex]
   &\quad = (1-t) I_M(A:B)_V + t I_M(A:B)_W ,
\end{align*}
concluding the proof.
\end{proof}

\begin{rem}
Inequality~\eqref{I conv geod} is especially notable because in general the log-det mutual information is not convex over the set of positive matrices. However, it is convex when restricted to geodesics in the trace metric, as we have just shown. Moreover, we note in passing that an inequality analogous to~\eqref{I conv geod} does not seem to hold for the log-det conditional mutual information.
\end{rem}

We now turn to the main result of this section. 

\begin{thm}
\label{thm satur}
  For an arbitrary $V_{ABC}>0$ written in block form as in~\eqref{global CM}, the following are equivalent:
  \begin{enumerate}
    \item $I_{M}(A:B|C)_{V}=0$, i.e.~\eqref{SSA} is saturated;
    \item $V_{ABC}/V_{BC}=V_{AC}/V_{C}$, i.e.~\eqref{INEQ 1} is saturated;
    \item $(V^{-1})_{AB}=(V^{-1})_{A}\oplus (V^{-1})_{B}$;
    \item $X=YC^{-1}Z^{\intercal}$ (see~\cite{Ando09} or~\cite[Thm.~4.49]{PetzBook});
    \item there is a classical Gaussian channel $N_{C\rightarrow BC}$ such
          that $(I_A\oplus N_{C\rightarrow BC})(V_{AC})=V_{ABC}$.
  \end{enumerate}
\end{thm}

\begin{proof} $\\[-3ex]$
\begin{description}
\item[$1\!\Leftrightarrow\! 2$.] Saturation of~\eqref{SSA} and~\eqref{INEQ 1} are equivalent concepts, since it is very easy to verify that if $M\geq N>0$ then $M=N$ if and only if $\det M=\det N$.

\item[$1\!\Leftrightarrow\! 3$.] It is well-known that $W_{AB}>0$ satisfies $\det W_{AB} = \det W_{A} \det W_{B}$ iff its off-diagonal block is zero, i.e. iff $W_{AB}=W_{A}\oplus W_{B}$. For instance, this can be easily seen as a consequence of~\eqref{det factor}. Thanks to Theorem~\ref{thm I cond geom}, identity~\eqref{I cond inv}, applying this observation with $W=V^{-1}$ yields the claim.

\item[$2\!\Rightarrow\! 4$.] This is known in linear algebra~\cite{Ando09}, but for
the sake of completeness we provide a different proof that fits more with the
spirit of the present work. Namely, we see that the variational representation of Schur complements~\eqref{Schur var} guarantees that~\eqref{INEQ 1} is saturated if and only if
\begin{equation}
\begin{split}
  V_{ABC} - (V_{AC}/V_C) \oplus 0_{BC}\
    &= \lmatrix A-V_{AC}/V_C & X & Y \\
                      X^\intercal  & B & Z \\
                      Y^\intercal  & Z^\intercal  & C \rmatrix \\[0.8ex]
    &= \lmatrix YC^{-1}Y^\intercal  & X & Y \\
                      X^\intercal  & B & Z \\
                      Y^\intercal  & Z^\intercal  & C \rmatrix \\[0.8ex]
    &\geq 0\, .
\end{split}
  \label{eq cond eq1}
\end{equation}
A necessary condition for~\eqref{eq cond eq1} to hold is obtained by taking suitable matrix elements:
\begin{equation*}
\begin{split}
  0 &\leq \lmatrix v \\ w \\ -C^{-1}Y^\intercal  v \rmatrix^{\intercal}
\lmatrix YC^{-1}Y^\intercal  & X & Y \\
                           X^\intercal  & B & Z \\
                           Y^\intercal  & Z^\intercal  & C \rmatrix
\lmatrix v \\ w \\ -C^{-1}Y^\intercal  v \rmatrix \\[0.8ex]
    &=    2 v^\intercal  (X-YC^{-1}Z^\intercal ) w + w^\intercal  B w .
\end{split}
\end{equation*}
This can only be true for all $v$ and $w$ if $X=YC^{-1}Z^\intercal $.
Moreover, this latter condition (together with the positivity of $V_{ABC}$)
is enough to guarantee that~\eqref{eq cond eq1} is satisfied. Indeed, we can write
\begin{align*}
  \lmatrix YC^{-1}Y^\intercal  & YC^{-1}Z^\intercal  & Y \\
                  ZC^{-1}Y^\intercal  & B & Z \\
                  Y^\intercal  & Z^\intercal  & C \rmatrix &= \lmatrix 0 & & \\ & B-ZC^{-1}Z^\intercal  & \\ & & 0 \rmatrix + \lmatrix YC^{-\frac12} \\ ZC^{-\frac12} \\ C^{\frac12} \rmatrix \lmatrix YC^{-\frac12} \\ ZC^{-\frac12} \\ C^{\frac12} \rmatrix^\intercal \\[0.8ex]
    & \geq 0 ,
\end{align*}
where $B-ZC^{-1}Z^\intercal \geq 0$ follows from
$\left(\begin{smallmatrix} B & Z \\ Z^\intercal & C \end{smallmatrix}\right) \geq 0$.

\item[$4\!\Rightarrow\! 5$.] If in~\eqref{N} we define
\begin{equation}\begin{split}
  H &= H_R \coloneqq \begin{pmatrix} \mathds{1} & 0 \\ 0 & ZC^{-1} \\ 0 & \mathds{1} \end{pmatrix} \text{ and}\\[0.8ex]
  K &= K_R \coloneqq \begin{pmatrix} 0 & & \\ & B-ZC^{-1}Z^\intercal  & \\ & & 0 \end{pmatrix} ,
  \label{gaus Petz eq2}
\end{split}\end{equation}
we directly obtain 
\begin{align*}
  (I_A\oplus N_{C\rightarrow BC}) (V_{AC}) &= H_R \lmatrix A & X \\ X^\intercal  & C \rmatrix H_R^\intercal + K_R  \\[0.8ex]
    &= \lmatrix A & X & Y \\ X^\intercal  & B & Z \\ Y^\intercal  & Z^\intercal  & C \rmatrix  \\[0.8ex]
    &= V_{ABC} ,
\end{align*}
provided that $X=YC^{-1}Z^{\intercal}$. We will see in the next section
that this map is nothing but a specialisation to the Gaussian case of a
general construction known as transpose channel, or Petz recovery map.

\item[$5\!\Rightarrow\! 2$.] From Theorem \ref{CP incr sch} it is clear that the equality in~\eqref{INEQ 1}
is a necessary condition for the existence of a Gaussian recovery map $N_{C\rightarrow BC}$.
\end{description}
\end{proof}

\section{Gaussian recoverability}
\label{sec:Gaussian-recov}
Here, we discuss the role of some well-known remainder terms for inequalities
of the form~\eqref{SSA}. In the setting of finite dimensional quantum states, we have discussed such inequalities in detail in Chapter~\ref{recoverability}.  
The much simpler classical reasoning (with a better bound) was presented in~\cite{LW14}.
We will translate these results into the Gaussian setting
in order to find an explicit expression for a remainder term to be added to~\eqref{SSA}.

For classical probability distributions $p$ and $q$ over a discrete alphabet, the following inequality , which improves on the
monotonicity of the relative entropy under channels, was shown in~\cite{LW14}:
\begin{equation}
  D(p\|q) - D(Np\|Nq) \geq D\left( p \| R N p \right) .
  \label{recov}
\end{equation}
Here, $N = (N_{ji})$ is any stochastic map (channel)
and the action of the Petz recovery map~\cite{PetzBook,BarnumKnill}
$R=R_{q,N}$ on an input distribution $r$ is uniquely
defined via the requirement that $N_{ji}q_i = R_{ij} (Nq)_j$ for all $i$ and $j$.
Explicitly,
\begin{equation}
  (R_{q,N}\, r)_i \coloneqq \sum_j \frac{q_i N_{ji}}{(Nq)_j} r_j .
  \label{Petz}
\end{equation}
Observe that $R_{q,N}$ is a bona fide channel, since
\begin{equation*}
  \sum_i (R_{q,N})_{ij} = \sum_i \frac{q_i N_{ji}}{(Nq)_j}
                        = \frac{(Nq)_j}{(Nq)_j}
                        = 1 .
\end{equation*}
In analogy to the quantum state case, we will call the right hand side of~\eqref{recov} the
\textit{relative entropy of recovery}. The proof of~\eqref{recov} is a simple
application of the concavity of the logarithm, and works as follows
\begin{align}
  D\left(p\|R_{q,N}Np\right)
     &= \sum_i p_i \Big(\ln p_i - \ln (R_{q,N}Np)_i \Big) \nonumber\\[0.8ex]
     &= \sum_i p_i \Big(\ln p_i - \ln \sum_j \frac{q_i\,N_{ji}}{(Nq)_j}\,(Np)_j \Big) \label{petz eq1} \\
     &\leq \sum_i p_i \Big(\ln p_i - \sum_j N_{ji} \ln \frac{q_i}{(Nq)_j}\,(Np)_j \Big) \label{petz eq2} \\
     &= D(p\|q) - D\left(Np\|Nq\right) . \nonumber
  \label{petz eq2}
\end{align}

Although we wrote out the proof only for random variables taking values in a discrete
alphabet, all of the above expressions make perfect sense also in more general cases,
e.g. when $i$ and $j$ are multivariate real variables.
If $N$ is a classical Gaussian channel acting as in~\eqref{N}, it
can easily be verified that the `transition probabilities' $N(x,y)$ satisfying
\begin{equation}
  (Np)(x) = \int dy\, N(x,y) p(y)
\end{equation}
take the form
\begin{equation}
  N(x,y) = \frac{e^{-\frac{1}{2} (x-Hy)^\intercal K^{-1} (x-Hy)}}{\sqrt{(2\pi)^n\det K}} .
  \label{N Gauss}
\end{equation}

Following again~\cite{LW14}, we observe that if the output of the random channel
$N$ is a deterministic function of the input, then~\eqref{recov} is always
saturated with equality.
This can be seen by noticing that in that case for all $i$ there is only one
index $j$ such that $N_{ji}\neq 0$ (and so $N_{ji}=1$). Therefore, the step
from~\eqref{petz eq1} to~\eqref{petz eq2} is an equality. There is a very special
case when this remark is useful. Consider a triplet of random variables $XYZ$,
distributed according to $p(xyz)$, a second probability distribution $q(xyz)=p(x)p(yz)$ 
and the channel $N$ consisting of discarding $Y$. Obviously, in this case the
output is a deterministic function of the input. It is easily seen that the
reconstructed global probability distribution $R_{q,N}Np$ is
\begin{equation}
  \tilde{p}(xyz) = p(xz) p(y|z) .
\end{equation}
Then the saturation of~\eqref{recov} allows us to write
\begin{equation}
  I(X:Y|Z) = D(p\|q) - D(Np\|Nq) = D(p\|\tilde{p}) .
  \label{I rel ent}
\end{equation}

\paragraph{Gaussian Petz recovery map}\label{gaus Petz}
From now on, we will consider the case in which $N$ is a {classical Gaussian channel} transforming covariance matrices according to the rule~\eqref{N}. As can be easily verified, if $q$ is also a multivariate Gaussian distribution, then $R_{q,N}$ becomes a {classical Gaussian channel} as well. We compute its action in the case we are mainly interested in, that is, when the left--hand side of~\eqref{recov} corresponds to the difference of the two sides of~\eqref{SSA}, and verify that it coincides with the recovery map introduced in Section~\ref{sec exact recov} (via the general action~\eqref{N} with the substitutions~\eqref{gaus Petz eq2}).

\begin{prop}
  \label{prop Petz}
  Let $q$ be a tripartite Gaussian probability density with zero mean and covariance matrix
  \begin{equation*}
    V_A\oplus V_{BC} = \begin{pmatrix} A & 0 & 0 \\ 0 & B & Z \\ 0 & Z^\intercal & C \end{pmatrix} ,
  \end{equation*}
  and let the channel $N$ correspond to the action of discarding the $B$ components, i.e.
  $H = \Pi_{AC}
   = \left( \begin{smallmatrix} \mathds{1} & 0 & 0 \\ 0 & 0 & \mathds{1} \end{smallmatrix}\right)$
  and $K=0$ in~\eqref{N}. Then, the action $C\rightarrow BC$ of the Petz recovery
  map~\eqref{Petz} on Gaussian variables with zero mean can be written at the level
  of covariance matrices as~\eqref{N}, where $H_R$ and $K_R$ are given by~\eqref{gaus Petz eq2}.
\end{prop}

\begin{proof}
The Petz recovery map~\eqref{Petz} is a composition of three operations: first the pointwise division by a Gaussian distribution, then the transpose of a deterministic channel, and eventually another pointwise Gaussian multiplication. It should be obvious from~\eqref{gauss} that a pointwise multiplication by a Gaussian distribution with covariance matrix $A$ is a Gaussian (non--deterministic) channel that leaves the mean vector invariant and acts on covariance matrices as $V\mapsto V'=(V^{-1}+A^{-1})^{-1}$. Furthermore, it can be proven that the transpose $N^\intercal$ of the channel $N$ in~\eqref{N} sends Gaussian variables with zero mean to other Gaussian variables with zero mean, while on the inverses of the covariance matrices it acts as
\begin{equation}
  N^\intercal : V^{-1} \longmapsto\ (V')^{-1} = H^\intercal (V+K)^{-1} H .
  \label{N^T}
\end{equation}
A way to prove the above equation is by using~\eqref{N Gauss} to directly compute the action of $N^\intercal$ on a Gaussian input distribution.

After the preceding discussion, it should be clear that under our hypotheses
the action of the Petz recovery map can be written as
\begin{equation}
  \sigma_{AC} \longmapsto \sigma'_{ABC} = \Big( V_A^{-1}\oplus V_{BC}^{-1} + (\sigma_{AC}^{-1}-V_A^{-1}\oplus V_C^{-1})\oplus 0_B \Big)^{-1} .  \label{gaus Petz eq3}
\end{equation}
The Woodbury matrix identity (see~\cite{Woodbury}, or~\cite[Equation (6.0.10)]{SCbook}),
\begin{equation}
  (S+UTV)^{-1} = S^{-1} - S^{-1}U\left(VS^{-1}U+T^{-1}\right)^{-1} V S^{-1},  \label{Wood}
\end{equation}
can be used to bring~\eqref{gaus Petz eq3} into the canonical form~\eqref{N}:
\begin{align*}
  \sigma'_{ABC} &= \left( V_A^{-1}\oplus V_{BC}^{-1} + (\sigma_{AC}^{-1}-V_A^{-1}\oplus V_C^{-1})\oplus 0_B \right)^{-1} \\[0.8ex]
                &= \big( V_A^{-1}\oplus V_{BC}^{-1} + \Pi_{AC}^\intercal  (\sigma_{AC}^{-1}-V_A^{-1}\oplus V_C^{-1}) \Pi_{AC} \big)^{-1} \\[0.8ex]
                &= V_A\oplus V_{BC} - (V_A\oplus V_{BC}) \Pi_{AC}^\intercal \cdot  \\ 
                &\quad\quad\Big( (\sigma_{AC}^{-1}\! - V_A^{-1}\oplus V_C^{-1})^{-1} + \Pi_{AC} (V_A \oplus V_{BC}) \Pi_{AC}^\intercal\Big)^{-1} \\ 
                &\quad\quad\cdot \Pi_{AC} (V_A \oplus V_{BC}) \\[0.8ex]
                &= V_A \oplus V_{BC} - (V_A \oplus V_{BC}) \Pi_{AC}^\intercal \cdot \\ 
                &\quad\quad\Big( -V_A \oplus V_C - (V_A \oplus V_C) (\sigma_{AC}-V_A \oplus V_C)^{-1} (V_A \oplus  V_C) + V_A \oplus V_C \Big)^{-1} \\
                &\quad\quad\cdot \Pi_{AC} (V_A \oplus V_{BC}) \\[0.8ex]
                &= V_A\oplus V_{BC} + (V_A\oplus V_{BC}) \Pi_{AC}^\intercal (V_A^{-1} \oplus V_C^{-1}) \\
                &\quad\quad \cdot(\sigma_{AC}-V_A\oplus V_C) \cdot (V_A^{-1}\oplus V_C^{-1}) \Pi_{AC} (V_A\oplus V_{BC}) \\[0.8ex]
                &= H_R \sigma_{AC} H_R^\intercal + K_R\, ,
\end{align*}
where we have employed the definitions
\begin{align*}
  H_R &= (V_A\oplus V_{BC}) \Pi_{AC}^\intercal (V_A^{-1}\oplus V_C^{-1})
       = \lmatrix \mathds{1} & 0 \\ 0 & ZC^{-1} \\ 0 & \mathds{1} \rmatrix \text{ and} \\[0.8ex]
  K_R &= \lmatrix 0 & & \\ & B-ZC^{-1} Z^\intercal & \\ & & 0 \rmatrix .
\end{align*}
\end{proof}

\paragraph{Gaussian relative entropy of recovery}
We are now ready to employ the classical theory of recoverability in order
to find the expression for the relative entropy of recovery in the Gaussian case.

\begin{prop}
  For all tripartite covariance matrices $V_{ABC}>0$ written in block
  form as in~\eqref{global CM}, we have
  \begin{equation}
  \begin{split}
    I_{M}(A:B|C)_{V} &= \frac12 \ln\frac{\det V_{AC}\det V_{BC}}{\det V_{ABC}\det V_C} \\
                     &= D\!\left( V_{ABC} \| \tilde{V}_{ABC}\right) ,
  \end{split}
  \label{SSA+}
  \end{equation}  where
  \begin{equation}
    \tilde{V}_{ABC} \coloneqq \begin{pmatrix} A & YC^{-1} Z^\intercal  & Y \\
                                           ZC^{-1}Y^\intercal  & B & Z \\
                                           Y^\intercal  & Z^\intercal  & C \end{pmatrix}
    \label{V tilde}
  \end{equation}
  and the relative entropy function $D(\cdot\|\cdot)$ is given by~\eqref{rel ent}.
\end{prop}

\begin{proof}
This is just an instance of~\eqref{I rel ent} applied to the continuous
Gaussian variable $(X_{A},X_{B},X_{C})$.
\end{proof}

The identity~\eqref{SSA+} is useful in deducing new constraints that will
be much less obvious coming from a purely matrix analysis perspective, by using the relation between Gaussian probability distributions and their covariance matrices.
For instance, it is well known that $D(p\|q)\geq -2\ln F(p,q)$ {(see e.g.~\cite{muller2013quantum,audenaert2012comparisons})}.
In case of Gaussian variables with the same mean, it holds
\begin{equation}
  F^2(p_A,p_B) = \frac{\det (A!B)}{\sqrt{\det A\det B}} , \label{fid gaus}
\end{equation}
where $(A!B)\coloneqq 2\left(A^{-1}+B^{-1}\right)^{-1}$ is the \emph{harmonic mean}
of $A$ and $B$.
Inserting this standard lower bound into~\eqref{SSA+} we obtain
\begin{equation}
  \frac{\det V_{AC}\det V_{BC}}{\det V_{ABC}\det V_C}
    \geq \frac{\det V_{ABC}\det \tilde{V}_{ABC}}{\left( \det (V_{ABC}!\tilde{V}_{ABC}) \right)^2},
  \label{SSA+ fid}
\end{equation}
leading to
\begin{equation}
   I_{M}(A:B|C)_{V}
    \geq \frac12 \ln\frac{\det V_{ABC}\det \tilde{V}_{ABC}}{\left( \det (V_{ABC}!\tilde{V}_{ABC}) \right)^2}.
  \label{SSA+ fid2}
\end{equation}
Using furthermore
\begin{align*}
  \det \tilde{V}_{ABC} &= \det\tilde{V}_{BC} \det (\tilde{V}_{ABC}/\tilde{V}_{BC}) \\
                       &= \det V_{BC} \det (\tilde{V}_{AC}/\tilde{V}_{C}) \\
                       &= \det V_{BC} \det (V_{AC}/V_{C}) ,
\end{align*}
we also arrive at the inequality
\begin{equation}
  \det V_{ABC} \leq \det (V_{ABC}!\tilde{V}_{ABC}) .
\end{equation}
To illustrate the power of this relation,
we note that inserting the harmonic-geometric mean inequality
for matrices~\cite[Corollary 2.1]{Ando79}
\begin{equation*}
  A!B \leq A\# B
\end{equation*}
yields again SSA~\eqref{SSA} in the form
$\det \tilde{V}_{ABC} \geq \det V_{ABC}$.

\section{A lower bound on $I_{M}(A:B|C)_{V}$}
\label{sec:I_M-lower-bound}
Throughout this section, we explore some ways of strengthening Theorem~\ref{thm satur}
by finding a suitable lower bound on the log-det conditional mutual information $I_{M}(A:B|C)_{V}$. We would like the expression to have two main features: (a) it should be easily computable in
terms of the blocks of $V_{ABC}$; and (b) the explicit saturation condition in Theorem~\ref{thm satur}(4) should be easily readable from it. This latter requirement can be accommodated, for example, if the lower bound involves some kind of distance between the off-diagonal block $X$ and its `saturation value' $YC^{-1}Z^{\intercal}$. We start with a preliminary result.

\begin{prop}
  \label{prop I(A:B)}
  For all matrices
  \begin{equation*}
    V_{AB} = \begin{pmatrix} A & X \\ X^{\intercal} & B \end{pmatrix} \geq 0,
  \end{equation*}
  we have
  \begin{equation}
    I_{M}(A:B)_{V}  \geq \frac12 \big\| A^{-1/2} X B^{-1/2} \big\|^{2}_{2} .  \label{prop I(A:B) eq}
  \end{equation}
\end{prop}

\begin{proof}
Using, the standard factorization of the determinant in terms
of the Schur complement, the identity $\ln \det V = \Tr \ln V$ (where $V>0$),
and the inequality $\ln(\mathds{1}+\Delta)\leq \Delta$ (for Hermitian
$\Delta > -\mathds{1}$), we find
\begin{align*}
  I_{M}(A:B)_{V} &=    \frac12 \ln \frac{\det V_{A}\det V_{B}}{\det V_{AB}} \\[0.8ex]
                 &=   -\frac12 \ln \det V_{A}^{-1/2}(V_{AB}/V_{B})V_{A}^{-1/2} \\[0.8ex]
                 &=   -\frac12 \ln \det (\mathds{1} - A^{-1/2}XB^{-1}X^\intercal A^{-1/2}) \\[0.8ex]
                 &=   -\frac12 \Tr \ln (\mathds{1} - A^{-1/2}XB^{-1}X^\intercal A^{-1/2})  \\[0.8ex]
                 &\geq \frac12 \Tr A^{-1/2}XB^{-1}X^\intercal A^{-1/2} \\[0.8ex]
                 &=    \frac12  \big\| A^{-1/2} X B^{-1/2} \big\|^{2}_{2}\, ,
\end{align*}
\end{proof}

This allows us to prove our main result. 

\begin{thm}
  \label{thm lower b}
  For all $V_{ABC}>0$ written in block form as in~\eqref{global CM}, we have
  the following chain of inequalities:
  \begin{align}
    I_{M}(A:B|C)_{V}
      &\geq \frac12 \Tr\Big[ (V_{AC}/V_{C})^{-1} (X-YC^{-1}Z^{\intercal}) \nonumber\\
      &\quad \cdot (V_{BC}/V_{C})^{-1} (X-YC^{-1}Z^{\intercal})^{\intercal} \Big] \\[0.8ex]
      &\geq \frac12 \left\| A^{-1/2}(X-YC^{-1}Z^{\intercal})B^{-1/2} \right\|^{2}_{2}\, .    \label{lower b}
  \end{align}
\end{thm}

\begin{proof}
We want to use the identity~\eqref{I cond inv} to lower bound $I_{M}(A:B|C)_{V}$.
In order to do so, we need to write out the $A$-$B$ off-diagonal block of the
inverse $(V_{ABC})^{-1}$. With the help of the projectors onto the $A$ and $B$
components, denoted by $\Pi_{A}$ and $\Pi_{B}$ respectively, we are seeking an
explicit expression for $\Pi_{A} (V_{ABC})^{-1} \Pi_{B}^{\intercal}$. Remember
that the block-inversion formula~\eqref{inv} gives
\begin{align}
  \Pi_{1} (W_{12})^{-1} \Pi_{1}^{\intercal} &= (W_{12}/W_{2})^{-1} , \label{inv1} \\[0.8ex]
  \Pi_{1} (W_{12})^{-1} \Pi_{2}^{\intercal} &= - W_{1}^{-1} (\Pi_{1}W_{12}\Pi_{2}^{\intercal})
                                                            (W_{12}/W_{1})^{-1} , \label{inv2}
\end{align}
for an arbitrary bipartite block matrix $W_{12}$. This allows us to write
\begin{align*}
  &\Pi_{A} (V_{ABC})^{-1} \Pi_{B}^{\intercal} \\[0.8ex]
     &= \Pi_{A} \Pi_{AB} (V_{ABC})^{-1} \Pi_{AB}^{\intercal} \Pi_{B}^{\intercal} \\[0.8ex]
     &= \Pi_{A} (V_{ABC}/V_{C})^{-1} \Pi_{B}^{\intercal} \\[0.8ex]
     &= -(V_{AC}/V_{C})^{-1} \bigl( \Pi_{A} V_{ABC}/V_{C} \Pi_{B}^{\intercal} \bigr) \bigl( (V_{ABC}/V_{C}) \big/ (V_{AC}/V_{C}) \bigr)^{-1} \\[0.8ex]
     &= -(V_{AC}/V_{C})^{-1} \bigl( X - YC^{-1} Z^{\intercal} \bigr) (V_{ABC}/V_{AC})^{-1} .
\end{align*}
Exchanging $A$ and $B$ in this latter expression and taking the transpose we arrive at
\begin{equation*}
\begin{split}
  \Pi_{A} (V_{ABC})^{-1} \Pi_{B}^{\intercal}
     &= -(V_{ABC}/V_{BC})^{-1} \big( X - YC^{-1} Z^{\intercal} \big) \cdot (V_{BC}/V_{C})^{-1} .
\end{split}
\end{equation*}
Now we are ready to invoke Proposition~\ref{prop I(A:B)} to write
\begin{align*}
  & I_{M}(A:B|C)_{V}  = I_{M}(A:B)_{V^{-1}} \\[0.8ex]
           &\quad\geq \frac12 \Tr \Big[(V^{-1})_{A}^{-1} (\Pi_{A} V^{-1} \Pi_{B}^{\intercal}) \cdot (V^{-1})_{B}^{-1} (\Pi_{B}^{\intercal} V^{-1} \Pi_{A}) \Big] \\[0.8ex]
           &\quad=    \frac12 \Tr \Bigl[ (V_{ABC}/V_{BC})  \cdot \bigl( (V_{ABC}/V_{BC})^{-1} (X-YC^{-1} Z^{\intercal}) (V_{BC}/V_{C})^{-1} \bigr) \\
           &\quad\quad \cdot (V_{ABC}/V_{AC}) \cdot \bigl( (V_{AC}/V_{C})^{-1} (X-YC^{-1} Z^{\intercal}) (V_{ABC}/V_{AC})^{-1} \big)^{\intercal} \Bigr] \\[0.8ex]
           &\quad= \frac12 \Tr \Big[ (V_{AC}/V_{C})^{-1} ( X - YC^{-1} Z^{\intercal}) \cdot (V_{BC}/V_{C})^{-1} ( X - YC^{-1} Z^{\intercal})^{\intercal} \Big]\, .
\end{align*}
Since on one hand $V_{AC}/V_{C}\leq V_{A}=A$, and on the other the
expression $\Tr R K S K^{\intercal}$ is clearly monotonic in $R,S\geq 0$,
we finally obtain
\begin{align*}
  I_{M}(A:B|C)_{V} &\geq \frac12 \Tr \bigl[ A^{-1} ( X - YC^{-1} Z^{\intercal}) \cdot B^{-1} ( X - YC^{-1} Z^{\intercal})^{\intercal} \bigr] \\[0.8ex]
                   &= \frac12 \bigl\| A^{-1/2}(X-YC^{-1}Z^{\intercal})B^{-1/2} \bigr\|^{2}_{2}\, .
\end{align*}
\end{proof}

It can easily be seen that the above result satisfies the requirements stated in the beginning of the section, i.e. it is easily computable in terms of the blocks of $V_{ABC}$ and it is faithful.

We are now ready to start the investigation of {\em quantum} covariance matrices in the next section.

\section{Strengthenings of SSA for quantum covariance matrices and correlation measures} \label{sec:Renyi-2-Gaussian-squashed}

In this section we show how to apply results on log-det conditional
mutual information to infer properties of Gaussian states in quantum optics.
Before doing so, let us provide a very brief introduction to quantum optics,
a framework of great importance for practical applications and
implementations of quantum communication protocols.

\paragraph{Gaussian states in quantum optics}

The set of $n$ electromagnetic modes that are available for transmission
of information translates to a set of $n$ pairs of canonical operators
$x_i, p_j$ ($i=1,\ldots, n$) acting on an infinite-dimensional Hilbert space
and obeying the canonical commutation relations
$[x_i, p_j] = i\delta_{ij}$ (in natural units with $\hbar=1$). These operators are the non-commutative
analogues of the classical electric and magnetic fields.
By introducing the vector notation $r\coloneqq (x_1,p_1,\ldots,x_n,p_n)^{\intercal}$
we can rewrite the canonical commutation relations in the more convenient form
\begin{equation}
[r,r^\intercal ] = i \Omega \coloneqq i {\begin{pmatrix} 0 & 1 \\ -1 & 0 \end{pmatrix}\!}^{\oplus n} = i\begin{pmatrix} 0 & \mathds{1} \\ -\mathds{1} & 0 \end{pmatrix} ,
\label{CCR}
\end{equation}
where $\Omega$ is called the \emph{standard symplectic form}. The antisymmetric, non-degenerate quadratic form identified by $\Omega$ is called \emph{standard symplectic product}, and the linear space $\mathds{R}^{2n}$ endowed with this product is a \emph{symplectic space}. In what follows, the symplectic space associated with a quantum optical system $A$ will be denoted with $\Sigma_A$. For an introduction to symplectic geometry, we refer the reader to~\cite{Gosson}.

Exactly as in the classical case, the Hamiltonian for the quantum electromagnetic fields is quadratic
in the canonical operators. Thus, not surprisingly, the states that are most
frequently produced in the laboratories are thermal states of quadratic
Hamiltonians of the form $\mathcal{H}=\frac12 r^\intercal H r$, where $H>0$
is a $2n\times 2n$ real, positive definite matrix. These states are called {\em Gaussian states}
\cite{biblioparis,weedbrook12,Adesso14}. 

For a quantum state described by a density matrix $\rho$ the
first moments are given by the expected value of the field operators, $s=\Tr [\rho r]$.
However, the information-theoretical properties of
Gaussian states can be fully understood in terms of the second-moment
correlations, encoded in the $2n\times 2n$ covariance matrix $V$
whose entries are
\begin{equation}
   V_{ij} \coloneqq \Tr \left[ \rho \left\{(r-s)_i, (r-s)_j \right\} \right] .
   \label{QCM}
\end{equation}
Here the anticommutator $\{H,K\}\coloneqq HK+KH$ is needed
in the quantum case, in order to make the above expression real, and $s\coloneqq s\cdot \text{id}$ as operators on the Hilbert space. It is customary
not to divide by $2$ when defining the covariance matrix in the quantum case.
The reason will become apparent in a moment.
Any quantum state $\rho$ of an $n$-mode electromagnetic field can be equivalently described in terms of phase space quasi-probability distributions, such as the Wigner distribution~\cite{Wigner}. Hence Gaussian states can be defined, in general, as the continuous variable states with a Gaussian Wigner distribution, given by
\begin{equation}\label{Wigner}
W_\rho(\xi) \coloneqq \frac{1}{\pi^n \sqrt{\det V}} e^{-(\xi-s)^\intercal V^{-1} (\xi-s)},
\end{equation}
in terms of the vector of first moments $s$ and the QCM $V$, with $\xi \in \mathds{R}^{2n}$ a phase space coordinate vector.

Let us have a closer look at the set of matrices arising from~\eqref{QCM}. Unlike the classical case, not every positive definite matrix $V>0$
can be the covariance matrix of a Gaussian state. In fact, Heisenberg's uncertainty
principle imposes further constraints that are quantum mechanical in nature. It turns out
\cite{simon94} that covariance matrices of quantum states (not necessarily Gaussian)
must obey the inequality
\begin{equation}
V\geq i\Omega . \label{Heisenberg}
\end{equation}
Furthermore, all $2n\times 2n$ real matrices satisfying~\eqref{Heisenberg},
collectively called \emph{quantum covariance matrices} (QCMs), can
be covariance matrices of suitably chosen Gaussian states.
Therefore, according to our convenience, we can think of Gaussian states as operators on the background Hilbert space, or we can adopt the complementary picture at the symplectic space level, and parametrise Gaussian states with their covariance matrices.

Clearly, linear transformations $r \rightarrow S r$ that preserve the commutation relations~\eqref{CCR} play a special role within this framework. Any such transformation is described by a \textit{symplectic} matrix, i.e.~a matrix $S$ with the property that $S\Omega S^\intercal=\Omega$. Symplectic matrices form a non-compact, connected Lie group that is additionally closed under transposition, and is typically denoted by $\mathrm{Sp}(2n,\mathds{R})$~\cite{pramana}. The importance of these operations arises from the fact that for any symplectic $S$ there is a unitary evolution $U_S$ on the Hilbert space such that $U_S^\dag r U_S = Sr$. When a unitary conjugation $\rho\mapsto U_S \rho U_S^\dag$ is applied to a state $\rho$, its covariance matrix transforms as $V\mapsto SVS^\intercal$. Accordingly, we observe that~\eqref{Heisenberg} is preserved under congruences by symplectic matrices. It turns out that under such congruences, positive matrices can be brought into a remarkably simple form.

\begin{lemma}[Williamson's decomposition~\cite{willy,willysim}] \label{lemma Williamson}
Let $K>0$ be a positive, $2n\times 2n$ matrix. Then there is a symplectic transformation $S$ such that $K= S \Delta S^\intercal$, where according to the block decomposition~\eqref{CCR} one has $\Delta= \lmatrix D & 0 \\ 0 & D \rmatrix$, and $D$ is a positive diagonal matrix whose nonzero entries depend (up to their order) only on $K$, and are called symplectic eigenvalues.
\end{lemma}

Thanks to Williamson's decomposition, we see that~\eqref{Heisenberg} can be cast into the simple form $D\geq \mathds{1}$, and that the minimal elements in the set of QCMs are exactly those matrices $V$ for which one of the following equivalent conditions is met: (a) $D=\mathds{1}$; (b) $\det V=1$; (c) $\rk (V \pm i\Omega) = n$ (i.e. half the maximum). These special QCMs are called ``pure'', since the corresponding Gaussian state is a rank-one projector.

When the system under examination is made of several parties (each comprising a certain number
of modes), the global QCM will have a block structure as in~\eqref{global CM}.
The symplectic form in this case is simply given by the direct sum of
the local symplectic forms, e.g. for a composite system $AB$ one has $\Omega_{AB}=\Omega_{A}\oplus \Omega_{B}$. This can be rephrased by saying that the symplectic space associated with the system $AB$ is the direct sum of the symplectic spaces associated with $A$ and $B$, expressed in formula as $\Sigma_{AB}=\Sigma_{A}\oplus \Sigma_{B}$~\cite[Equation (1.4)]{Gosson}.
Conversely, discarding a subsystem corresponds to performing an orthogonal projection of the QCM onto the corresponding symplectic subspace~\cite[Section 1.2.1]{Gosson}, in formula $V_{A}=\Pi_{A} V_{AB} \Pi_{A}^\intercal$.

Pure Gaussian states enjoy many useful properties that we will exploit multiple times throughout this section. To explore them, a clever use of the complementarity between the two pictures at the Hilbert space level and at the QCM level is of prime importance. Let us illustrate this point by presenting some lemmas that we will make use of in deriving the main results of this section.

\begin{lemma} \label{lemma pure reduction}
Let $V_{AB}$ be a QCM of bipartite system $AB$. We denote by $V_{A}=\Pi_{A} V_{AB} \Pi_{A}^\intercal$ the reduced QCM corresponding to the subsystem $A$, and analogously for $V_{B}$. If $V_{A}$ is pure, then $V_{AB} = V_{A} \oplus V_{B}$.
\end{lemma}

\begin{proof}
The statement becomes obvious at the Hilbert space level. In fact, the reduced state on $A$ of a bipartite state $\rho_{A B}$ is given by $\rho_{A} = \Tr_{B} \rho_{A B}$. Evaluating the ranks of both sides of this equation shows that if $\rho_{A}$ is pure then the global state must be factorised.
\end{proof}

Extending the system to include auxiliary degrees of freedom is a standard technique in quantum information, popularly referred to as going to the ``Church of the larger Hilbert space'', a phrase originally coined by J. Smolin (see also~\cite{DHW05RI}). Such a technique can be most notably employed in order to \emph{purify} the system under examination, as detailed in the following lemma~\cite{Gpurif}.

\begin{lemma} \label{lemma pur}
For all QCMs $V_A$ pertaining to a system $A$ there exists an extension $AE$ of $A$ and a pure QCM $\gamma_{AE}$ such that $\Pi_A \gamma_{AE} \Pi_A^\intercal = V_A$, where $\Pi_A$ is the projector onto the symplectic subspace $\Sigma_A\subset \Sigma_{AE}$.
\end{lemma}
\begin{proof}
See~\cite[Section III.D]{Gpurif}.
\end{proof}

Having all the necessary tools in place, now, we would like to obtain an operator generalization of \eqref{log det ineq} from \eqref{INEQ 1} by applying the symplectic purification trick. This will allow us to generalize the previously mentioned results on steering to Gaussian states with many modes. Note that these results are specific to bona fide quantum CMs. We first note that if a bipartite quantum CM $V_{AB}$ is symplectic, then 
\begin{equation}
V_{AB}^{-1} = \Omega_{A}^T V_{AB}^T \Omega_{AB} = \Omega_{AB}^T V_{AB} \Omega_{AB},
\end{equation}
which by comparison with (\ref{inv}) yields $V_{AB}/V_A=\Omega_B^T V_B^{-1} \Omega_B$. In conjunction with the quotient property of CMs, this implies that
\begin{equation}\label{lemma2}
V_{ABC} \mbox{ is symplectic \ \ $\Rightarrow$ \ \ } V_{AB}/V_B = \Omega_A^T (V_{AC}/V_C)^{-1} \Omega_A\,.
\end{equation}
We then get the following for any tripartite quantum system.
\begin{thm}[Schur complement of quantum CMs is monogamous] \label{sch compl mon}
If $V_{ABC}\geq i \Omega_{ABC}$ is any tripartite quantum CM, then
\begin{equation} V_{AC} / V_A\ \geq\ \Omega_C^T (V_{BC} / V_B)^{-1} \Omega_C\, . \label{INEQ 2} \end{equation}
\end{thm}

\begin{proof}
Consider a symplectic purification $V_{ABCD}$ of the system $ABC$. Applying \eqref{INEQ 1} first and then (\ref{lemma2}) yields  \eqref{INEQ 2}: $V_{AC}/V_A\ \geq\ V_{ACD}/V_{AD}\ =\ \Omega_C^T (V_{BC}/V_B)^{-1} \Omega_C$.
Alternatively, observe that the difference between right- and left-hand side of \eqref{INEQ 2} is concave in $V_{ABC}$ (as $V_{AC}/V_A$ is concave and $(V_{BC}/V_B)^{-1}$ is convex), and it vanishes on symplectic CMs by (\ref{lemma2}).
\end{proof}

We remark that the operator inequalities \eqref{INEQ 1} and \eqref{INEQ 2} are significantly stronger than the scalar ones reported in \cite{Gross,Adesso}, as the former establish algebraic limitations directly at the  level of CMs, in a similar spirit to the marginal problem \cite{Tyc2008},
for arbitrary multipartite states.
Equipped with these powerful tools, we proceed to investigate applications to quantum correlations, namely steering and entanglement.
Let us present here another useful observation.

\begin{lemma} \label{lemma fact out}
For all QCMs $V_A\geq i\Omega_A$ of a system $A$, there is a decomposition $\Sigma_A=\Sigma_{A_1} \oplus \Sigma_{A_2}$ of the global symplectic space into a direct sum of two symplectic subspaces such that
\begin{equation}
V_A = V_{A_1} \oplus \eta_{A_2} ,
\end{equation}
where $V_{A_1} > i\Omega_{A_1}$ and $\eta_{A_2}$ is a pure QCM. Furthermore, for every purification $\gamma_{AE}$ of $V_A$ (see Lemma~\ref{lemma pur}) there is a symplectic decomposition of $E$ as $\Sigma_{E}=\Sigma_{E_1} \oplus \Sigma_{E_2}$ such that: (a) $\gamma_{AE}=\gamma_{A_1E_1} \oplus \eta_{A_2} \oplus \tau_{E_2}$, with $\eta_{A_2}, \tau_{E_2}$ pure QCMs; (b) $n_{A_1} = n_{E_1}$; and (c) $\gamma_{E_1}>i\Omega_{E_1}$.
\end{lemma}

\begin{proof}
The first claim is a direct consequence of Williamson's decomposition, Lemma~\ref{lemma Williamson}. The subspace $\Sigma_{A_2}$ corresponds to those symplectic eigenvalues of $V_A$ that are equal to $1$. 

Now, let us prove the second claim. Consider an arbitrary pure QCM $\gamma_{AE}$ that satisfies $\gamma_A=V_A=V_{A_1}\oplus \eta_{A_2}$. Since in particular $\gamma_{A_2}=\eta_{A_2}$, we can apply Lemma~\ref{lemma pure reduction} and conclude that $\gamma_{AE}=\gamma_{A_1 E} \oplus \eta_{A_2}$. The first claim of the present lemma tells us that $\gamma_{E}=\gamma_{E_1}\oplus \tau_{E_2}$, with $\gamma_{E_1}> i\Omega_{E_1}$ and $\tau_{E_2}$ pure. Again, Lemma~\ref{lemma pure reduction} yields $\gamma_{AE}=\gamma_{A_1E_1}\oplus \eta_{A_2}\oplus \tau_{E_2}$, corresponding to statement (b). Hence, we have only to show that $n_{A_1}=n_{E_1}$. In order to show this, let us write
\begin{equation*}
\gamma_{A_1 E_1} = \begin{pmatrix} V_{A_1} & L \\ L^\intercal & \gamma_{E_1} \end{pmatrix}
\end{equation*}
We can invoke Equation~\ref{lemma2} to deduce the identity $V_{A_1} - L\gamma_{E_1}^{-1} L^\intercal = \Omega V_{A_1}^{-1} \Omega^\intercal$, that is, $L \gamma_{E_1}^{-1} L^\intercal = V_{A_1} - \Omega V_{A_1}^{-1} \Omega^\intercal$. Since the right hand side has maximum rank $2n_{A_1}$ due to the strict inequality $V_{A_1}>i\Omega$ (see the forthcoming Lemma~\ref{lemma gamma sharp}), we conclude that $2n_{E_1}\leq \rk \left( L \gamma_{E_1}^{-1} L^\intercal \right) = 2 n_{A_1}$, and hence $n_{E_1}\leq n_{A_1}$. But the same reasoning can be applied with $A_1$ and $E_1$ exchanged, thus giving $n_{A_1}\leq n_{E_1}$, which concludes the proof.
\end{proof}


If one wants to use Gaussian states to transmit and manipulate quantum
information, the role of measurements is of course central. 
Therefore, it is of prime importance for us to understand how
Gaussian states behave under measurements. Of course, the most natural
and easily implementable measurements are Gaussian as well, meaning that
the $X=\mathds{R}^{2n}$ and the positive operators $E(d^{2n}x)=E(x) d^{2n}x$ are positive
multiples of Gaussian states with a fixed covariance matrix $\sigma$ and
varying first moments $\Tr[E(x) r]\propto x$.
Implementing such a Gaussian measurement on a Gaussian state $\rho$
with a vector of first moments $s$ and a QCM $V$ yields
an outcome $x$ distributed according to a
Gaussian probability distribution
\begin{equation}
p(x) = \frac{2^n e^{-(x-s)^\intercal (V+\gamma)^{-1} (x-s)}}{\sqrt{\det(V+\sigma)}} .
\label{G meas}
\end{equation}
Furthermore, it can be shown that if a bipartite system $AB$ is in a Gaussian
state $\rho_{AB}$ described by a QCM $V_{AB}$ and only the second subsystem
$B$ is subjected to a Gaussian measurement described by a seed QCM $\sigma_B$, the
state of subsystem $A$ after the measurement, given by
$\rho'_A \propto \Tr_B [ \rho_{AB} \left( \text{id}_A\otimes E_B(x)\right) ]$, is again
Gaussian, and is described by first moments depending on the measurement outcome,
but by a fixed QCM which is given by the Schur complement~\cite{nogo1,nogo2,nogo3}
\begin{equation}
V'_B = (V_{AB} + 0_A\oplus\sigma_B) / (V_B +\sigma_B) .
\label{QCM after meas}
\end{equation}

Equation~\eqref{G meas} shows how quantum Gaussian states reproduce classical Gaussian probability
distributions when measured with Gaussian measurements. Thus, log-det entropies
become relevant in the quantum case as well, since they reproduce Shannon
entropies of the experimentally accessible measurement outcomes.

One could also wonder whether the log-det entropy given in~\eqref{logdetent}
can be interpreted directly at the density operator level. 
Interestingly, it can be shown that for an arbitrary Gaussian state with QCM $V$
it holds that
\begin{equation}
H_2(\rho) = \frac12 \ln \det V = M(V) = h(\xi) - n (\ln \pi  + 1).
\label{Renyi-2 G}
\end{equation}
This means that the \mbox{R\'enyi-$2$} entropy {\em coincides} with the log-det entropy
defined in~\eqref{logdetent}~\cite{Adesso12}, and  these quantities in turn coincide, up to an additive constant, with the differential entropy $h(\xi)$ of the classical Gaussian variable $\xi \in \mathds{R}^{2n}$ whose probability distribution is precisely the Wigner function $W_\rho (\xi)$ of the quantum Gaussian state $\rho$. Therefore, in the relevant case of tripartite quantum Gaussian states, the
general inequality~\eqref{SSA} for log-det entropy takes the form of a
SSA inequality for the \mbox{R\'enyi-$2$} entropy~\cite{Adesso12,Gross,Adesso}, holding in addition to the standard
one for R\'enyi-1 entropy aka von Neumann entropy, which is valid for
arbitrary (Gaussian or not) tripartite quantum states.

Note that in general it is not advisable to form entropy expressions from R\'enyi entropies, since they do not obey any nontrivial constraints in a general multi-partite system~\cite{LMW}. In information theory, this is addressed by directly defining well-behaved notions of conditional R\'enyi entropy and R\'enyi mutual information~\cite{Tomamichel-book}. Here, we evade those issues as we are restricting ourselves to Gaussian states. In fact thanks to their special structure Gaussian states satisfy also \mbox{R\'enyi-$2$} entropic inequalities.
Not surprisingly, such inequalities find several applications in continuous variable
quantum information, in particular limiting the performances of quantum
protocols with Gaussian states.
For example, as demonstrated in~\cite{Adesso,Kor}, there is no Gaussian
state of a $(n_A+n_B+n_C)$-mode system $ABC$ that is simultaneously $A\rightarrow C$ steerable 
and $B\rightarrow C$ steerable by Gaussian measurements when $n_C=1$. At the level of QCMs, this is a consequence
of the (non-balanced) inequality
\begin{equation}
M(V_{AC})+ M(V_{BC}) - M(V_A)-M(V_B)\geq 0,
\label{SSA purified}
\end{equation}
to be obeyed by all tripartite QCMs $V_{ABC}$. We stress that
\eqref{SSA purified} cannot hold for all positive definite $V$ (that is, for all
classical covariance matrices), as it can be easily seen by
rescaling it via $V\mapsto kV$, for $k>0$. However, the new matrix
$V$ becomes unphysical for sufficiently small $k$, as it violates
the uncertainty principle~\eqref{Heisenberg}.


\subsection{Gaussian \mbox{R\'enyi-$2$} entanglement of formation}

We are now ready to apply our results to strengthening the SSA inequality~\eqref{SSA} in the quantum case.
This subsection is thus devoted to finding
a sensible lower bound on the log-det conditional mutual information for
all QCMs. This bound will be given by a quantity called
\emph{\mbox{R\'enyi-$2$} Gaussian entanglement of formation}, already introduced
and studied in~\cite{Adesso12}. In general, for a bipartite quantum state $\rho_{AB}$, the \emph{R\'enyi-$\alpha$
entanglement of formation} is defined as the convex hull of the
R\'enyi-$\alpha$ entropy of entanglement defined on pure states~\cite{H42007}, i.e.
\begin{equation}
\begin{split}
  E_{F,\alpha}(A:B)_{\rho}
    &\coloneqq \inf \sum_i p_i \, H_{\alpha}\bigl(\psi_i^A\bigr) \\
           &\quad\text{ s.t. } \rho_{AB} = \sum_i p_i \psi_i^{AB} ,
\end{split}
  \label{EoF}
\end{equation}
where $\psi_{i}^{AB}$ are density matrices of pure states and $\psi_i^A = \Tr_B \psi_i^{AB}$ is the reduced state.

For quantum Gaussian states, an upper bound to this quantity can be
derived by restricting the decompositions appearing in the above infimum
to be comprised of pure Gaussian states only. One obtains what is called
{\em Gaussian R\'enyi-$\alpha$ entanglement of formation}, which is a monotone under Gaussian local operations and classical communication. In terms of
the QCM $V_{AB}$ of $\rho_{AB}$ this is given by the simpler formula~\cite{Wolf03}
\begin{equation}
\begin{split}
  E^{\text{G}}_{F,\alpha}(A:B)_{V}
    &= \inf H_{\alpha}(\gamma_{A}) \\
    &\quad \text{ s.t. } \gamma_{AB} \text{ pure QCM and } \gamma_{AB}\leq V_{AB},
\end{split}
  \label{GEoF}
\end{equation}
where with a slight abuse of notation we denoted with $H_{\alpha}(W)$
the R\'enyi-$\alpha$ entropy of a Gaussian state with QCM $W$,
and $\gamma_{AB}$ stands for the QCM of a pure Gaussian state, i.e.~with $\det\gamma_{AB}=1$.
Incidentally, it has been proven~\cite{EoFsymmetricG,Giovadd} that for some $2$-mode
Gaussian states, the formula~\eqref{GEoF} reproduces exactly
\eqref{EoF}, i.e.~Gaussian decompositions in~\eqref{EoF} are globally optimal.

The most commonly used $E_{F,\alpha}$ is the one corresponding to
the von Neumann entropy, $\alpha=1$. However, as we already saw,
\mbox{R\'enyi-$2$} quantifiers arise quite naturally in the Gaussian setting, because by virtue of~\eqref{Renyi-2 G} they reproduce Shannon entropies of measurement outcomes, cf.~(\ref{G meas}). Thus, from now on
we will focus on the case $\alpha=2$. Under this assumption,
thanks to~\eqref{Renyi-2 G} we see that~\eqref{GEoF} becomes
\begin{equation}
\begin{split}
  E^{\text{G}}_{F,2}(A:B)_{V}
    &= \inf M(\gamma_{A}) \\[0.8ex]
    &\quad \text{ s.t. } \gamma_{AB} \text{ pure QCM and } \gamma_{AB}\leq V_{AB} .
\end{split}
  \label{G R2 EoF}
\end{equation}
We will find it convenient to rewrite the above equation in a slightly different form.
Using the well-known fact that $M(\gamma_A)=M(\gamma_B)=\frac12 I_M(A:B)_\gamma$
when $\gamma_{AB}$ is the QCM of a pure state~\cite{Adesso14}, we obtain
\begin{equation}
\begin{split}
  E^{\text{G}}_{F,2}(A:B)_{V}
    &= \inf \frac12 I_M(A:B)_{\gamma} \\[0.8ex]
    &\quad \text{ s.t. } \gamma_{AB} \text{ pure QCM and } \gamma_{AB}\leq V_{AB} .
\end{split}
  \label{G R2 EoF alt}
\end{equation}

The entanglement measure~\eqref{GEoF} is known to be faithful on quantum Gaussian states,
i.e.~it becomes zero if and only if the Gaussian state with QCM $V_{AB}$ is separable. 

In \cite{LiLuo}, the inequality $I\geq 2E$ is identified as a fundamental postulate for a consistent theory of quantum versus classical correlations in bipartite systems, for an arbitrary measure of entanglement $E$ and of total correlations $I$. This follows from the fact that for pure states classical and quantum correlations are equal and add up to the total correlations \cite{Groisman2005}, while for mixed states classical correlations are intuitively expected to exceed quantum ones, which include entanglement \cite{Henderson2001,Groisman2005,LiLuo}. However, such a relation can already be violated for two-qubit states (Werner states) when $E$ is the entanglement of formation defined via the usual von Neumann entropy \cite{H42007}, and $I$ the corresponding mutual information. In larger dimensions it may even happen that $I<E$  \cite{hayden06}, undermining the interpretation of the entanglement of formation as just a fraction of total correlations.
Perhaps it is worth noticing that in \cite{RenFan} the same inequality $I\geq 2E$ is also shown to be equivalent to a monogamy relation for quantum discord in a pure tripartite state. 
Here we show that $I_M \geq 2 E_{F,2}^G$ \emph{does hold} for Gaussian states of arbitrarily many modes using the R\'enyi-2 quantifiers. 
\begin{thm}\label{I2Ev1}
Let $AB$ be in an arbitrary Gaussian quantum state. Then
\begin{equation}
\frac12 I_M(A:B) \geq E_{F,2}^G(A:B)\, .  \label{noi}
\end{equation}
If $AB$ is in a pure Gaussian state, both sides coincide with the reduced R\'enyi-2 entropy $\frac{1}{2}\log\det V_A$.
\end{thm}
\begin{proof}
The inequality admits a neat proof that makes use of the geometric mean $M\#N$ between positive matrices $M,N$. The key step is that, for any quantum CM $V_{AB}$ obeying the bona fide condition, the matrix $\gamma^{\#}_{AB} = V_{AB}\#(\Omega_{AB} V_{AB}^{-1}\Omega_{AB}^T)$ is the quantum CM of a pure Gaussian state obeying $\gamma^{\#}_{AB} \leq V_{AB}$; using it as an ansatz in Equation~\eqref{G R2 EoF} and exploiting Theorem 3 in \cite{Ando79} one shows that $E_{F,2}^G(A:B)_V  \leq \frac{1}{2} \log\det \gamma^{\#}_A \leq \frac{1}{2} I_M(A:B)_V$.
\end{proof}
This in turn allows to prove useful monogamy properties of~\eqref{G R2 EoF}, captured by the inequality
\begin{equation}
  E^{\text{G}}_{F,2}(A:B_1\ldots B_n)_V \geq  \sum_{j=1}^n E^{\text{G}}_{F,2}(A:B_j)_V,
  \label{CKW}
  \end{equation}
  for any multipartite Gaussian state with QCM $V_{AB_1 \ldots B_n}$. We delay the proof to Lemma \ref{E mono}, where we will be able to give a simple argument. 

We are now in position to apply some of the tools we have been developing
so far to prove a generalisation of the inequality~\eqref{noi} that
is of interest to us since it constitutes also a strengthening of~\eqref{SSA}. Before doing so, we provide a useful lemma. Besides being a versatile tool to be employed throughout the rest of this section, it starts to show how fruitful the application of matrix analysis tools in quantum optics can be.

\begin{lemma} \label{lemma gamma sharp}
Let $K>0$ be a positive matrix. Then $\gamma_K^\# \equiv K\#(\Omega K^{-1}\Omega^\intercal)$ is a pure QCM. Furthermore, $K> i\Omega$ if and only if $K> \Omega K^{-1} \Omega^\intercal$, if and only if $K>\gamma_K^\#$.
\end{lemma}

\begin{proof}
We apply Lemma~\ref{lemma Williamson} to decompose $K=S \Delta S^T$, where $S$ is symplectic and $\Delta$ diagonal. Then, we deduce that
\begin{align*}
\gamma_K^\# &= (S\Delta S^\intercal) \# \left( \Omega S^{-\intercal} \Delta^{-1} S^{-1} \Omega^\intercal \right) \\
&\texteq{(i)} (S\Delta S^\intercal) \# \left( S \Omega \Delta^{-1} \Omega^\intercal S^\intercal \right) \\
&\texteq{(ii)} (S\Delta S^\intercal) \# \left( S \Delta^{-1} S^\intercal \right) \\
&\texteq{(iii)} S \left( \Delta \# \Delta^{-1} \right) S^\intercal \\
&\texteq{(iv)} SS^\intercal ,
\end{align*}
where we used, in order: (i) the identity $\Omega S^\intercal = S^{-1} \Omega$, valid for all symplectic $S$; (ii) the fact that $[\Omega,\Delta]=0$, which is a consequence of Lemma~\ref{lemma Williamson}; (iii) the congruence covariance of the geometric mean,~\eqref{geom cov congr}; and (iv) the elementary observation that $\Delta \# \Delta^{-1}=\mathds{1}$, as follows from the explicit formula~\eqref{geom expl}.
Then, it is easy to observe that $\gamma_{K}^\# $ is the QCM of a pure Gaussian state. The inequality $K>i\Omega$ translates to $\Delta > \mathds{1}$, and in turn to $K=S\Delta S^\intercal > SS^\intercal = \gamma_K^\#$, or alternatively to $\Delta>\Delta^{-1}$ and thus to $K = S \Delta S^\intercal > S \Delta^{-1} S^\intercal = \Omega K^{-1} \Omega^\intercal $. This latter condition can already be found in~\cite[Lemma 1]{sep3-mode}.
\end{proof}

\begin{thm}
  \label{thm I cond G R2 EoF}
  For all tripartite QCMs $V_{ABC}\geq i\Omega_{ABC}$,
  it holds that
  \begin{equation}
    \frac12 I_{M}(A:B|C)_{V} \geq E^{\text{\emph{G}}}_{F,2}(A:B)_{V} .
    \label{ext I con}
  \end{equation}
\end{thm}

\begin{proof}
For any QCM $V_{ABC}$, using the notation of Lemma~\ref{lemma gamma sharp} define
\begin{equation}
  \gamma_{AB} \coloneqq \gamma^\#_{V_{ABC}/V_{C}} .
\end{equation}
Since $V_{ABC}/V_{C}>0$ by the positivity conditions~\eqref{Schur pos}, we see that $\gamma_{AB}$ is a pure QCM. Now we proceed to show that $\gamma_{AB}\leq V_{AB}$. On the one hand, the very definition of Schur complement implies that $V_{ABC}/V_{C}\leq V_{AB}$,
while on the other hand a special case of Theorem \ref{sch compl mon} gives us the
general inequality $V_{ABC}/V_{C}\geq \Omega V_{AB}^{-1}\Omega^{\intercal}$,
i.e.~$\Omega (V_{ABC}/V_{C})^{-1} \Omega^{\intercal}\leq V_{AB}$.
Since the geometric mean is well-known to be monotonic~\cite{Ando79},
we obtain $\gamma_{AB}\leq V_{AB}$. This shows that $\gamma_{AB}$ can be used as an ansatz in~\eqref{G R2 EoF alt}.
We can write
\begin{align*}
     E^{\text{G}}_{F,2}(A:B)_{V} &\leq \frac12 I_M(A:B)_\gamma \\[0.8ex]
     &= \frac12 I_M(A:B)_{(V_{ABC}/V_{C})\# (\Omega (V_{ABC}/V_{C})^{-1}\Omega^\intercal ) } \\[0.8ex]
     &\textleq{(i)} \frac14 I_M(A:B)_{V_{ABC}/V_{C}} + \frac14 I_M(A:B)_{\Omega (V_{ABC}/V_{C})^{-1}\Omega^\intercal} \\[0.8ex]
     &\texteq{(ii)} \frac14 I_M(A:B)_{V_{ABC}/V_{C}} + \frac14 I_M(A:B)_{(V_{ABC}/V_{C})^{-1}} \\[0.8ex]
     &\texteq{(iii)} \frac14 I_M(A:B|C)_V + \frac14 I_M(A:B|C)_V \\[0.8ex]
     &= \frac12 I_M(A:B|C)_V ,
\end{align*}
where we employed, in order: (i) the convexity of log-det mutual information on the trace metric geodesics~\eqref{I conv geod}, (ii) the fact that since $\Omega_{AB}=\Omega_A \oplus \Omega_B$, the equality $I_M(A:B)_{\Omega W \Omega^\intercal }=I_M(A:B)_W$ holds true; and (iii) the identity~\eqref{I cond Schur} for the first term and~\eqref{I cond inv} followed again by~\eqref{I cond Schur} for the second.
\end{proof}


\subsection{Gaussian steerability and its monogamy.}
We call a bipartite state steerable, if one party can remotely steer the other system into different ensembles by making different measurements on their part. 
Here we show how to use our results and techniques to prove general properties of the quantitative measure of steerability by Gaussian measurements proposed in \cite{steerability}. 
Consider a $n$-mode continuous variable quantum system, and denote by $\nu_i(A)$ the $i$--th smallest symplectic eigenvalue of a positive definite CM $0<A=A^T\in \mathcal{M}_{2n}(\mathds{R})$. We define the two functions
\begin{equation} g_\pm (A) =  {\sum}_{i=1}^n\, \max\,\big\{\pm \log \nu_i(A),\, 0\big\}\, . \label{func g} \end{equation}

The function $g_-$ finds many applications in continuous variable quantum information. For instance, the logarithmic negativity \cite{VidalWerner,plenioprl} of a bipartite state $\rho_{AB}$, defined as  $E_N(\rho_{AB}) =  \log \|\rho_{AB}^{\text{\reflectbox{$\Gamma$}}} \|_1$ (where \reflectbox{$\Gamma$} denotes partial transposition), takes the form $E_N(\rho_{AB})=g_-(\tilde{V}_{AB})$ if $\rho_{AB}$ is a Gaussian state with quantum CM $V_{AB}$; here, the partial transpose of the CM is given by $\tilde{V}_{AB}=\Theta V_{AB} \Theta$, with $\Theta=\left(\begin{smallmatrix} \mathds{1} & \\ & -\mathds{1} \end{smallmatrix}\right)_A\oplus \mathds{1}_B$. Furthermore, a quantitative measure of Gaussian steerability (i.e., steerability by Gaussian measurements) has been recently introduced for any state $\rho_{AB}$ with quantum CM $V_{AB}$ \cite{steerability}, that takes the form
\begin{equation}
\mathcal{G}(A\rangle B)_V = g_-(V_{AB}/V_A)\,,
\label{G steer}
\end{equation}
 {in the case of party $A$ steering party $B$. Notice that $\mathcal{G}(A\rangle B)_V>0$ is necessary and sufficient for ``$A$ to $B$'' steerability of a Gaussian state with quantum CM $V_{AB}$ by means of Gaussian measurements on $A$ \cite{steerability,Wiseman}, but is only sufficient if either the state \cite{JOSAB} or the measurements \cite{NoGauss1,NoGauss2} are non-Gaussian.}

The functions $g_\pm$ have useful properties (see \cite{LHAW16} for details):
$g_\pm(A)=g_\pm(SAS^T)$ for all symplectic $S$, $g_\pm (A^{-1}) = g_\mp (A)$, $g_+(A)-g_-(A)=\frac{1}{2}\,\log\det A$, $g_\pm (A\oplus B)=g_\pm (A) + g_\pm (B)$, $g_-(A)$ is monotonically decreasing and convex in $A$, while $g_+(A)$ is monotonically increasing but neither convex nor concave in $A$, and finally  $g_-$ is superadditive in the subsystems,
\begin{equation}\label{dec red g- eq}
g_-(V_{AB}) \geq g_-(V_A) + g_-(V_B)\,.
\end{equation}

Based on these facts, for which the proofs rely on recent advances in the study of symplectic eigenvalues \cite{SympIneq}, we can prove fully general properties of the steerability measure \eqref{G steer}, extending the results of \cite{steerability} where these properties were only proven in the special case of one-mode steered subsystem ($n_B=1$).

\begin{thm}[Properties of Gaussian steerability] \label{G prop}
 The steerability measure \eqref{G steer} enjoys the following properties.
 \begin{enumerate}
\item  ${ \mathcal G}(A\rangle B)_V$ is convex and decreasing in the CM $V_{AB}$;
\item ${ \mathcal G}(A\rangle B)$ is additive under tensor products, i.e.~under direct sums of CMs,
${ \mathcal G}(A_1 A_2\rangle B_1 B_2)_{V_{A_1B_1}\oplus W_{A_2 B_2}} =  \mathcal{G}(A_1\rangle B_1)_{V_{A_1B_1}} +  \mathcal{G}(A_2\rangle B_2)_{W_{A_2 B_2}}$;
\item for arbitrary states, ${ \mathcal G}(A\rangle B)$ is decreasing under general, non-deterministic Gaussian maps on the steering party $A$;
\item for Gaussian states, ${ \mathcal G}(A\rangle B)$ is decreasing under general, non-deterministic Gaussian maps on the steered party $B$;
\item for any quantum CM $V_{ABC}$,  it holds ${ \mathcal G}(A\rangle C)_V  \leq  g_+(V_{BC}/V_B)$. 
\end{enumerate}
\end{thm}
\begin{proof} See Appendix \ref{appGauss} for detailed proofs. \end{proof}

 {Theorem~\ref{G prop} establishes ${ \mathcal G}(A \rangle B)_V$ as a convex monotone for arbitrary Gaussian states with quantum CM $V_{AB}$ under arbitrary local Gaussian operations on either the steering or the steered parties, hence fully validating it within the Gaussian subtheory of the recently formulated resource theory of steering \cite{resource}.} Moreover, our framework allows us to address the general problem of the monogamy of ${ \mathcal G}(A \rangle B)$ for {\textit arbitrary} (Gaussian or not) multimode states. For a state with quantum CM $V_{AB_1\ldots B_k}$, consider the following inequalities
\begin{eqnarray}
\mathcal{G}(A\rangle B_1\ldots B_k) &\geq& {\sum}_{j=1}^k \mathcal{G}(A\rangle B_j) \label{mon steer 1}\,,\\
\mathcal{G}(B_1\ldots B_k \rangle A) &\geq& {\sum}_{j=1}^k \mathcal{G}(B_j\rangle A) \label{mon steer 2}\,.
\end{eqnarray}
In a very recent study \cite{MonSteer}, both inequalities were proven in the special case of a  $(k+1)$-mode system with one single mode per party, i.e., $n_A=n_{B_j}=1$ ($j=1,\ldots,k$). We now show that  only one of these constraints holds in full generality.

\begin{thm}[Monogamy of Gaussian steerability] \label{G mono}
(a) Ineq.~(\ref{mon steer 1}) holds for any multimode quantum CM $V_{A B_1 \ldots B_k}$. (b) Ineq.~(\ref{mon steer 2}) holds for any multimode quantum CM $V_{A B_1 \ldots B_k}$ such that either $A$ comprises a single mode ($n_A=1$), or $V_{A B_1 \ldots B_k}$ belongs to a pure state, but can be violated otherwise.
\end{thm}
\begin{proof}  See Appendix \ref{appGauss} for detailed proofs. \end{proof}

The Gaussian steerability is thus not monogamous with respect to a common steered party $A$ when the latter is made of two or more modes, with violations of (\ref{mon steer 2}) existing already in a tripartite setting ($k=2$) with $n_{B_1}=n_{B_2}=1$ and $n_A=2$; a counterexample is reported in the detailed proof of Theorem \ref{G mono} in Appendix \ref{appGauss}. 
What is truly monogamous is the log-determinant of the Schur complement, which only happens to coincide with the function $g_-$ when $n_A=1$.

At the end of this section, we would like to reconnect to the results of the past section with the following theorem. 
\begin{thm}[Gaussian R\'enyi-2 correlations hierarchy]\label{I2E}
Let $AB$ be in an arbitrary Gaussian quantum state. Then
\begin{equation}
\mbox{$\frac12 I_M(A:B) \geq E_{F,2}^G(A:B) \geq { \mathcal G}(A \rangle B)$}\, . \label{I>2E}
\end{equation}
If $AB$ is in a pure Gaussian state, all the above three quantities coincide with the reduced R\'enyi-2 entropy $\frac{1}{2}\log\det V_A$.
\end{thm}
\begin{proof}
The inequality on the left is simply Theorem \ref{I2Ev1}. The inequality on the right is a corollary of Theorem~\ref{G prop}. 
\end{proof}
Remarkably, this proves that the involved measures quantitatively capture the general hierarchy of correlations \cite{ABC} in arbitrary Gaussian states \cite{Adesso14}: the Gaussian steerability is generally smaller than the entanglement degree, which accounts for a portion of quantum correlations up to half the total ones.

\subsection{Gaussian \mbox{R\'enyi-$2$} squashed entanglement} \label{subsec Gauss sq}

In finite-dimensional quantum mechanics, the positivity of conditional
mutual information allows to construct a powerful entanglement measure
called \emph{squashed entanglement}, defined for a bipartite state
$\rho_{AB}$ by~\cite{CW04}
\begin{equation}
  E_{\text{sq}}(A:B)_{\rho} \coloneqq \inf_{\rho_{ABC}} \frac12 I(A:B|C)_{\rho} ,
\end{equation}
where the infimum ranges over all possible ancillary quantum systems $C$
and over all the possible states $\rho_{ABC}$ having marginal $\rho_{AB}$.
We are now in a position to discuss a similar quantity tailored to Gaussian states.
First, we can restrict the infimum by considering only Gaussian extensions,
which corresponds to the step leading from~\eqref{EoF} to~\eqref{GEoF}.
Secondly, as it was done to arrive at~\eqref{G R2 EoF}, we can substitute
von Neumann entropies with \mbox{R\'enyi-$2$} entropies. The result is
\begin{equation}
  E^{\text{G}}_{\text{sq},2}(A:B)_{V} \coloneqq \inf_{V_{ABC}} \frac12 I_{M}(A:B|C)_{V},
  \label{Gauss sq}
\end{equation}
where the infimum is on all extended QCMs $V_{ABC}$ satisfying the condition
$\Pi_{AB} V_{ABC}\Pi_{AB}^\intercal = V_{AB}$ on the $AB$ marginal (and~\eqref{Heisenberg}).
We dub the quantity in~\eqref{Gauss sq} {\em Gaussian \mbox{R\'enyi-$2$} squashed
entanglement}, stressing that it is a quantifier specifically tailored to Gaussian states
and different from the R\'enyi squashed entanglement defined
in~\cite{SBW14} for general states, for which an alternative expression
for the conditional R\'enyi-$\alpha$ mutual information is adopted instead.

Despite the complicated appearance of the expression~\eqref{Gauss sq}, it turns out that \emph{the Gaussian \mbox{R\'enyi-$2$} squashed entanglement coincides with the Gaussian \mbox{R\'enyi-$2$} entanglement of formation for all bipartite QCMs}. This unexpected fact shows once more that \mbox{R\'enyi-$2$} quantifiers are particularly well behaved when employed to analyse Gaussian states, while at the same time it provides us with a novel, alternative expression of $E^{\text{G}}_{F,2}$ that can be used to understand its basic properties in a different and sometimes more intuitive way. Before stating the main result of this subsection, we need some preliminary results.

\begin{lemma} \label{lemma follia 0}
Let $\gamma_{AB}$ be a pure QCM of a bipartite system $AB$ such that $n_A=n_B=n$ and $\gamma_A>i\Omega_A$. Then
\begin{equation*}
\left(\gamma_{AB}+ i \Omega_{AB}\right) \big/ \left( \gamma_A + i\Omega_A \right) = 0_B .
\end{equation*}
\end{lemma}

\begin{proof}
From Williamson's decomposition, Lemma~\ref{lemma Williamson}, we see that whenever $\gamma_{AB}$ is pure, one has $\rk(\gamma_{AB}+i\Omega_{AB}) = n_A + n_B = 2n$ (i.e. half the maximum). Since already $\rk(\gamma_A + i\Omega_A) = 2n$, the additivity of ranks under Schur complements~\eqref{rank add} tells us that $\rk \left( \left(\gamma_{AB}+ i \Omega_{AB}\right) \big/ \left( \gamma_A + i\Omega_A \right) \right) = 0$, concluding the proof.
\end{proof}

\begin{prop} \label{follia prop}
Let $V_{AB}$ be a QCM of a bipartite system, and let $\gamma_{ABC}$ be a fixed purification of $V_{AB}$ (see Lemma~\ref{lemma pur}). Then, for all pure QCMs $\tau_{AB} \leq V_{AB}$ there exists a one-parameter family of pure QCMs $\sigma_C (t)$ (where $0<t\leq 1$) on $C$ such that
\begin{equation}
\gamma'_{AB}(t)\coloneqq \left(\gamma_{ABC} + 0_{AB} \oplus \sigma_C(t) \right) \big/ \left( \gamma_C +\sigma_C(t) \right) .
\label{follia prop eq}
\end{equation}
is a pure QCM for all $t>0$, and $\lim_{t\rightarrow 0^+} \gamma'_{AB}(t) = \tau_{AB}$. Equivalently, there is a sequence of Gaussian measurements on $C$, identified by pure seeds $\sigma_C(t)$, such that the QCM of the post-measurement state on $AB$ is pure and tends to $\tau_{AB}$ (see~\eqref{QCM after meas}). 
\end{prop}
\begin{proof} See Appendix \ref{appGauss} for detailed proofs. \end{proof}

Now, we are ready to state the main result of this section.

\begin{thm}
  \label{thm GSq=GEoF}
  For all bipartite QCMs $V_{AB}\geq i\Omega_{AB}$, the Gaussian \mbox{R\'enyi-$2$} squashed
  entanglement coincides with the Gaussian \mbox{R\'enyi-$2$} entanglement of formation, i.e.
  \begin{equation}
  E^{\text{\emph{G}}}_{\text{\emph{sq}},2}(A:B)_{V} = E_{F,2}^{\text{\emph{G}}}(A:B)_V .
  \label{GSq=GEoF}
  \end{equation}
\end{thm}

\begin{proof}
The inequality $E^{\text{G}}_{\text{sq},2}(A:B)_V \geq E_{F,2}^{\text{G}}(A:B)_V$ is an easy consequence
of~\eqref{ext I con} together with~\eqref{Gauss sq}. To show the converse, we employ the expression~\eqref{G R2 EoF alt} for the Gaussian R\'enyi-2 entanglement of formation. Consider an arbitrary purification $\gamma_{ABC}$ of $V_{AB}$, and pick a pure state $\tau_{AB}\leq V_{AB}$. By construction, we have $\gamma_{AB}=V_{AB}$. Now, thanks to Proposition~\ref{follia prop} one can construct a sequence of measurements identified by $\sigma_C(t)$ such that~\eqref{follia prop eq2} holds. Then, we have
\begin{align*}
&\frac12 I_M(A:B)_\tau \\[0.8ex]
&\quad= \frac12 I_M(A:B)_{\lim_{t\rightarrow 0^+} \left(\gamma_{ABC} + 0_{AB} \oplus \sigma_C(t) \right) / \left( \gamma_C +\sigma_C(t) \right) } \\
&\quad\texteq{(i)} \lim_{t\rightarrow 0^+} \frac12 I_M(A:B)_{\left(\gamma_{ABC} + 0_{AB} \oplus \sigma_C(t) \right) / \left( \gamma_C +\sigma_C(t) \right) } \\
&\quad\texteq{(ii)} \lim_{t\rightarrow 0^+} \frac12 I_M(A:B|C)_{\gamma_{ABC} + 0_{AB} \oplus \sigma_C(t)} \\
&\quad\textgeq{(iii)} E^{\text{G}}_{\text{sq},2}(A:B)_V ,
\end{align*}
where we used, in order: (i) the continuity of the log-det mutual information; (ii) the identity~\eqref{I cond Schur}; and (iii) the fact that the QCMs $\gamma_{ABC} + 0_{AB} \oplus \sigma_C(t)$ constitute valid extensions of $V_{AB}$, thus being legitimate ansatzes in~\eqref{Gauss sq}.
\end{proof}

\begin{rem}
A by-product of the above proof of Theorem~\ref{thm GSq=GEoF} is that in~\eqref{Gauss sq} we can restrict ourselves to systems of bounded size $n_C \leq n_{AB} = n_A + n_B$. Moreover, the extension can be taken of the form $\gamma_{ABC} + 0_{AB} \oplus \sigma_C$ up to certain limits, where $\gamma_{ABC}$ is a fixed purification of $V_{AB}$ and $\sigma_C$ is a pure QCM.
\end{rem}

This surprising identity between two seemingly very different entanglement measures, even though tailored to Gaussian states, is remarkable. On the one hand, it provides an interesting operational interpretation for the Gaussian \mbox{R\'enyi-$2$} entanglement of formation in terms of log-det conditional mutual information, via the recoverability framework. On the other hand, it simplifies the notoriously difficult evaluation of the squashed entanglement, in this case restricted to Gaussian extensions and log-det entropy, because it recasts it as an optimisation of the form~\eqref{G R2 EoF} which thus involves matrices of bounded instead of unbounded size (more precisely, of the same size as the mixed QCM whose entanglement is being computed).
In general, Theorem~\ref{thm GSq=GEoF} allows us to export useful properties between the two frameworks it connects. For instance, it follows from the identity~\eqref{GSq=GEoF} that the Gaussian \mbox{R\'enyi-$2$} squashed entanglement is faithful on Gaussian states and a monotone under Gaussian local operations and classical communication; in contrast, proving the property of faithfulness for the standard squashed entanglement was a very difficult step to perform~\cite{BCY11}. On the other hand, the arguments establishing many basic properties of the standard squashed entanglement can be imported from~\cite{CW04} and applied to~\eqref{Gauss sq}, providing new proofs of the same properties for the Gaussian \mbox{R\'enyi-$2$} entanglement of formation. Let us give an example of how effective the interplay between the two frameworks is by providing an alternative, one-line proof of the following result.

\begin{lemma}\label{E mono}
The Gaussian \mbox{R\'enyi-$2$} entanglement of formation is monogamous on arbitrary Gaussian states, i.e.
\begin{equation}
E_{F,2}^{\text{\emph{G}}}(A:BC) \geq E_{F,2}^{\text{\emph{G}}}(A:B) + E_{F,2}^{\text{\emph{G}}}(A:C) ,
\label{monogamy GEoF}
\end{equation}
and analogously for more than three parties.
\end{lemma}

\begin{proof}
Thanks to Theorem~\ref{thm GSq=GEoF}, we can prove the monogamy relation~\eqref{GEoF} for the Gaussian \mbox{R\'enyi-$2$} squashed entanglement. We use basically the same argument as in~\cite[Proposition 4]{CW04}. Namely, call $V_{ABC}$ the QCM of the system $ABC$. Then for all extensions $V_{ABCE}$ of $V_{ABC}$ one has
\begin{align*}
I_M(A:BC|E)_V &= I_M(A:B|E)_V + I_M (A:C|BE) \\[0.8ex]
&\geq 2  E^{\text{G}}_{\text{sq},2}(A:B)_V + 2 E^{\text{G}}_{\text{sq},2}(A:C)_V ,
\end{align*}
where we applied the chain rule for the conditional mutual information together with the obvious facts that $V_{ABE}$ is a valid extension of $V_{AB}$ and $V_{ABCE}$ a valid extension of $V_{AC}$.
\end{proof}

Lemma~\ref{E mono} yields the {\textit most general} result to date regarding quantitative monogamy of continuous variable entanglement \cite{ourreview,Adesso14}, as all previous proofs (for the R\'enyi-2 measure \cite{Adesso12} or other quantifiers \cite{hiroshima_2007,strongmono}) were restricted to the special case of one mode per party. 
A monogamy inequality is a powerful tool in dealing with entanglement measures. For instance, when combined with monotonicity under local operations, it leads to the additivity of the measure under examination.

\begin{cor}
The Gaussian entanglement measure $E_{F,2}^{\text{\emph{G}}} = E^{\text{G}}_{\text{sq},2}$ is additive under tensor products (equivalently, direct sum of covariance matrices). In formulae,
\begin{equation}
\begin{split}
E_{F,2}^{\text{\emph{G}}} (A_1 A_2: B_1 B_2)_{V_{A_1 B_1} \oplus W_{A_2 B_2}} &= E_{F,2}^{\text{\emph{G}}} (A_1: B_1)_{V} \\
&\quad + E_{F,2}^{\text{\emph{G}}} (A_2: B_2)_{W} .
\end{split}
\label{additivity GEoF}
\end{equation}
\end{cor}

\begin{proof}
Applying first~\eqref{monogamy GEoF} and then the monotonicity of $E_{F,2}^{\text{\emph{G}}}$ under the operation of discarding some local subsystems, we obtain
\begin{align*}
&E_{F,2}^{\text{\emph{G}}} (A_1 A_2: B_1 B_2)_{V_{A_1 B_1} \oplus W_{A_2 B_2}} \\[0.8ex]
&\quad \geq E_{F,2}^{\text{\emph{G}}} (A_1 A_2 : B_1)_{V_{A_1 B_1}\oplus W_{A_2}} \\
&\quad\quad + E_{F,2}^{\text{\emph{G}}} (A_1 A_2: B_2)_{V_{A_1} \oplus W_{A_2 B_2}} \\[0.8ex]
&\quad \geq E_{F,2}^{\text{\emph{G}}} (A_1: B_1)_{V} + E_{F,2}^{\text{\emph{G}}} (A_2: B_2)_{W} .
\end{align*}
The opposite inequality follows by inserting factorised ansatzes $\gamma_{A_1 B_1}\oplus \tau_{A_2 B_2}$ into~\eqref{GEoF}.
\end{proof}

As established in this section, the Gaussian \mbox{R\'enyi-$2$} entanglement of formation alias Gaussian \mbox{R\'enyi-$2$} squashed entanglement also emerges as a rare example of an additive entanglement monotone (within the Gaussian framework) which satisfies the general monogamy inequality \eqref{CKW}. We remark that the conventional (R\'enyi-$1$) entanglement of formation cannot fundamentally be monogamous~\cite{lancien_2016}, while the standard squashed entanglement is monogamous on arbitrary multipartite systems~\cite{koashi_2004}.

\chapter{Final thoughts} \label{theEnd}

It is now time to recap the results presented in the previous chapters and in particular review the connections we found between the seemingly different areas of quantum information theory. 

The first part of this thesis was devoted to hypothesis testing. Although this is a well known topic in the field, in particular in terms of quantum state discrimination, many questions were left unanswered for a long time. A prime example is that of composite quantum state discrimination, where many special cases had been investigated, but a general solution was not yet known. Here, we provided a solution for very general convex null and alternative hypotheses, closing this gap. While proving this composite quantum Stein's Lemma, several tools were developed including the asymptotic equivalence of the relative entropy and its measured version for permutation invariant quantum states. One rather unfortunate feature of our bound is that it is given by a regularized function, which makes it in general hard to compute. It is however remarkable, that it becomes clear that this regularization is indeed necessary due to the problems connected with the second part of the thesis, namely recoverability. 

While quantum state discrimination was also previously investigated closely, the reminder of the first part turned to a rather unexplored topic. That is, the quite natural question of what happens when we optimize over the states to discriminate rather than the measurements in quantum hypothesis testing. The optimal rate at which a fixed measurement can discriminate between states was labeled the discrimination power and optimal rates in several settings were given. Here, the crucial observation is that, while the input states can be arbitrarily chosen, entangled or even adaptively chosen states do not help in the asymptotic setting. This is in contrast to the state discrimination scenario, where no such results are known and it is even conjectured that collective measurements are necessary to achieve the optimal asymptotic rate. 

In the second part of the thesis, the topic of entropy inequalities and recoverability was considered. Here, the close connection to the first part became evident especially in the first chapter where recoverability inequalities were investigated. First, using the aforementioned result on asymptotic relative entropies, a novel lower bound on the conditional quantum mutual information was found in terms of a regularized relative entropy featuring an explicit and universal recovery map. Next, it was shown that the most commonly used recoverability quantities, in particular the regularized relative entropy of recovery, have an operational interpretation given by discrimination scenarios as they were developed in the first part of the thesis. No other  operational interpretation was previously known. Finally, this chapter also allowed us to prove that the regularization in our composite Stein's Lemma is indeed needed, again exploring the connection between hypothesis testing and entropy inequalities. 

In the next chapter we turned to a different type of entropy inequalities, called bounds on information combining. While these bring their very own difficulties, most importantly handling conditioning on quantum systems, it turns out that the results from the last chapter are of great help in investigating these bounds. Using the lower bound on the conditional quantum mutual information given by the fidelity of recovery, we manage to prove non-trivial lower bounds. Interestingly, there is also a connection in the other direction, which is that the states which arise naturally in this setting allowed us to find counterexamples on a conjecture concerning recoverability bounds in the previous chapter, also disproving a certain special case where several of the involved systems are classical. 
Furthermore, conjectures of the optimal lower and upper bounds where given as well as applications to finite blocklength and non-stationary behavior of polar codes. 

Finally, in the last chapter the focus was changed to looking at infinite dimensional system, with a particular focus on Gaussian quantum states. Using a remarkable connection between the \mbox{R\'enyi-$2$} entropy of a Gaussian state and the log-determinant of its covariance matrix allows us to use tools from matrix analysis to investigate entropy inequalities with a particular focus on correlation measures. The main results of this chapter are, first that an operator strengthening of the strong subadditivity inequality allows for novel monogamy results for a quantitative measure of Gaussian steerability, and second that of gaining insights into two measures of entanglement, namely the \mbox{R\'enyi-$2$} Gaussian entanglement of formation and the \mbox{R\'enyi-$2$} Gaussian squashed entanglement, which we ultimately prove to be equal, which is surprising considering that no such result is known in the von Neumann case. 

Now, in the last section of this thesis we will briefly discuss some open problems to point the reader towards some interesting future research ideas.


\section{Some open problems}

As is often the case, answering the questions addressed in this thesis leaves us with a bunch of new problems that might be the basis of additional research. In the following section we would like to point out some questions which we find particularly interesting and hope to enthuse the reader in these. 

\paragraph{Quantum channel discrimination:}
In Chapter~\ref{CompHypo} we discussed state discrimination with composite hypothesis. A closely related problem is that of channel discrimination. In this scenario, $n$ copies of two channels are given and one optimizes over all possible input states. It is apparent that the setting bears a certain similarity with composite discrimination, namely discriminating between the sets of possible output states. Nevertheless, the problem turns out to be significantly more complicated. Aside from assigning different priorities to the types of error, one can discuss many different settings in the channel case, such as product inputs, entangled inputs or even adaptively chosen inputs. Also one can envision the states to be chosen to aid the discrimination but also in an adversary setting. To find the ultimate rate in either case, one would have to take into account all possible strategies. In the classical case we know that \iid inputs are optimal and that even adaptive strategies bring no advantage~\cite{H08}. In the quantum case the picture is less clear, e.g. it is known that there exist quantum channels that cannot be perfectly discriminated with a finite number of independent product input states, but it is possible if the states can be chosen adaptively~\cite{HHLW09}. 
However, despite the apparent differences, one might expect that the tools developed in Chapter~\ref{CompHypo} help in the investigation of quantum channel discrimination as well. 

In recent work~\cite{berta2018amortized}, we give general converse bounds on quantum channel discrimination allowing for the most general (adaptive) strategies and show that for classical-quantum channels \iid inputs are optimal in many settings, extending the results in~\cite{H08}.

\paragraph{Characterization of optimal states for discrimination power:}
In Chapter~\ref{discPower} optimal rates for the discrimination power of a quantum measurement have been given in several scenarios. Those rates have a relatively simple form which needs  to be evaluated only on a single copy of the system. Nevertheless, they still include an optimization over all possible pairs of input states. Restricting the sets over which one needs to optimize would significantly simplify the computation of the rates. As conjectured in Section~\ref{optRates}, it seems natural to assume that the states in the optimal pair are orthogonal to each other.  

Related problems arise when one slightly alters the setting in the original question. One could, for example, consider the scenario where the input states are restricted in their total energy, which is of particular importance in infinite dimensional systems. Another possibility is that the states might have been generated at a different location and undergo a change of reference frame or more generally a noisy channel, before being measured. It is clear that the solution is given by restricting the optimization to all possible output states of the given channel. Nevertheless, depending on the channel, one might be able to bring the optimization for practically relevant cases into a significantly simpler form. 

\paragraph{Recoverability:}
As we discussed in Chapter \ref{recoverability}, it is desirable to give the lower bounds on the conditional quantum mutual information a simple form. Much progress has already been made, by finding explicit and universal recovery maps, avoiding the need for optimization. Still, the best known bounds include an integration over rotated Petz recovery maps. In particular, one might hope that it is sufficient to use the non-rotated Petz recovery map in the fidelity lower bound, as suggested in Equation \ref{conjFidelity}. 

\paragraph{Optimal bounds on information combining:}
This might be the most obvious open problem in this section, to find a proof for our conjectured bounds in Section~\ref{main}. This comes along with several other open questions, such as a better understanding of conditioning on a quantum system and duality in quantum information theory as well as new bounds on strong subadditivity. 
Also, our given lower bound as well as the conjectured ones can be seen as special cases of the Mrs. Gerber's Lemma by Wyner and Ziv, which in their version not only applies to single copies of the channel but $n$ copies. Since its discovery, the Mrs. Gerber's Lemma has been generalized to many settings~\cite{W74, AK77, JA12, OS15, C14}, all of which pose natural open problems in the quantum setting. While the $n$-copy case could be useful in Shannon theory, generalization to non-binary inputs would have applications to coding such as polar codes for arbitrary classical-quantum channels (see e.g~\cite{GV14,GB15,NR17}). 
Natural starting points for investigation could be the convex programming formulations of the measured relative entropy of recovery in~\cite{BFT15} or extending the problem to different entropies like the $\alpha$-Renyi entropy. 

\paragraph{Equivalence of entanglement measures:}
Finally, within the context of continuous variable quantum information with Gaussian states, it could be interesting to establish whether the equivalence between the Gaussian \mbox{R\'enyi-$2$} squashed entanglement and the Gaussian \mbox{R\'enyi-$2$} entanglement of formation proven in Section~\ref{subsec Gauss sq} further extends to a third measure of entanglement, namely the recently introduced Gaussian intrinsic entanglement~\cite{GIE}. 
It could also be worth exploring whether for states where Gaussian decompositions attain the global convex roof optimization for the entanglement of formation (such as symmetric $2$-mode Gaussian states), one could extend our techniques to show that even the standard squashed entanglement defined in terms of von Neumann conditional mutual information~\cite{CW04} may be optimised by Gaussian extensions and perhaps be shown to coincide with the conventional entanglement of formation; this would constitute a unique instance of computable squashed entanglement on states which find applications in quantum optics. 

\appendix
\part{ Appendix}

\chapter{Some helpful Lemmas}

Here we present several lemmas that are used in the main part. We start with Sion's minimax theorem.

\begin{lemma}\label{Sion}\cite{sion58}
Let $X$ be a compact convex subset of a linear topological space and $Y$ a convex subset of a linear topological space. If a real-valued function on $X\times Y$ is such that
\begin{itemize}
\item $f(x,\cdot)$ is upper semi-continuous and quasi-concave on $Y$ for every $x\in X$
\item $f(\cdot,y)$ is lower semi-continuous and quasi-convex on $X$ for every $y\in Y$\,,
\end{itemize}
then we have
\begin{align}
\min_{x\in X} \sup_{y\in Y} f(x,y) = \sup_{y\in Y} \min_{x\in X} f(x,y)\,.
\end{align}
\end{lemma}

If the measured relative entropy is optimized over closed, convex sets then Sion's minimax theorem can be applied.

\begin{lemma}\cite[Lem.~20]{BHLP14}\label{applySion}
Let $\cS,\cT\subseteq S(\cH)$ be closed, convex sets. Then, we have
\begin{align}
\min_{\substack{\rho\in\cS\\ \sigma\in\cT}}D_{\mathcal{M}}(\rho\|\sigma)=\sup_{(\mathcal{X},M)}\min_{\substack{\rho\in\cS\\ \sigma\in\cT}}D\left(\sum_{x\in\mathcal{X}}\tr\left[M_x\rho\right]|x\rangle\langle x|\middle\|\sum_{x\in\mathcal{X}}\tr\left[M_x\sigma\right]|x\rangle\langle x|\right)\,.
\end{align}
\end{lemma}


We have the following discretization result.

\begin{lemma}\label{carat}
For every measure $\mu$ over a subset $\cS\subseteq S(\cH)$ with the dimension of $\cH$ given by $d$, there exists a probability distribution $\{ p_i\}_i^N$ with $N\leq (n+1)^{2d^2}$ and $\rho_i\in\cS$ such that
\begin{align}
\int \rho^{\otimes n}\;\mathrm{d}\mu(\rho)=\sum_{i=1}^N p_i\rho_i^{\otimes n}\,.
\end{align}
\end{lemma}

\begin{proof}
We use Carath\'eodory theorem together with the smallness of the symmetric subspace. For pure states the proof from~\cite[Cor.~D.6]{BCR11} applies and the general case follows immediately by considering purifications and taking the partial trace over the purifying system.
\end{proof}

The von Neumann entropy has the following quasi-convexity property (besides its well-known concavity).

\begin{lemma}\label{caratentropy}
Let $\rho_i\in S(\cH)$ for $i=1,\ldots,N$ and $\{p_i\}$ be a probability distribution. Then, we have
\begin{align}
H\left(\sum_{i=1}^N p_i\rho_i\right)\leq \sum_{i=1}^N p_i H(\rho_i) + \log N\,.
\end{align}
\end{lemma}

\begin{proof}
This follows from elementary entropy inequalities:
\begin{align}
H\left(\sum_{i=1}^N p_i\rho_i\right)\leq\sum_{i=1}^N p_i H\left(\rho_i\right)+H(p_i)\leq \sum_{i=1}^N p_i H(\rho_i)+\log N\,.
\end{align}
\end{proof}


The following is a property of the quantum relative entropy.

\begin{lemma}\label{relEntropyGasym}~\cite[Thm.~3]{GMS09}
Let $\mathcal{N}$ be a trace-preserving, completely positive map with $\mathcal{N}(1)=1$ (unital) and $\mathcal{N}^2=\mathcal{N}$ (idempotent). Then, the minimum relative entropy distance between $\rho\in S(\cH)$ and $\sigma\in S(\cH)$ in the image of $\mathcal{N}$ satisfies 
\begin{align}
\inf_{\sigma\in\mathrm{Im}(\mathcal{N})}D(\rho\|\sigma)=H(\mathcal{N}(\rho))-H(\rho)=D(\rho\|\mathcal{N}(\rho))\,.
\end{align}
In particular, we have for the relative entropy of coherence that $D_{\cC}(\rho)=D(\rho\|\rho_{\mathrm{diag}})$, where $\rho_{\mathrm{diag}}$ denotes the state obtained from $\rho$ by deleting all off-diagonal elements.
\end{lemma}

Audenaert's matrix inequality originally used to derive the quantum Chernoff bound can be stated as follows.

\begin{lemma}\cite[Thm.~1]{ACMBMAV07}\label{lem:audenaert}
Let $X,Y\gg0$ and $s\in(0,1)$. Then, we have
\begin{align}
\tr\left[X^sY^{1-s}\right]\geq\tr\left[X\left(1-\left\{X-Y\right\}_+\right)\right]+\tr\left[Y\left\{X-Y\right\}_+\right]\,.
\end{align}
\end{lemma}


\chapter{Chernoff bound for the qubit covariant measurement}\label{appendixCBcm}

In this appendix we will give the missing details of the calculation in Section~\ref{DPexamplesCovMeas}. The goal is to determine the dual of the Chernoff bound for a qubit covariant measurement. 

We use the same definitions as in Section~\ref{DPexamplesCovMeas} and give the details of investigating Equation~\ref{covCa}, which we will repeat here for accessibility: 
\begin{equation}
C=\int_0^\pi \sin\theta \int_0^{2\pi}{d\phi\over4\pi} \left(1+\cos\theta \right)^{1-s}
\left(1+\cos\theta\cos\alpha+\cos\phi\sin\theta\sin\alpha \right)^{s}. 
\end{equation}
We start by using the relation
\begin{equation}
a^s={\sin(s\pi)\over\pi}\int_0^\infty
dx{a x^{s-1}\over a+x} 
\end{equation}
in the second factor and making the usual change of variables, $\theta\to u:=\cos\theta$, which gives
\begin{align}
C=\int_{-1}^{1} du &\left(1+u\right)^{1-s}{\sin(s\pi)\over4\pi^2}\int_0^\infty dx \dots \nonumber\\
& \dots\int_0^{2\pi}{d\phi\over4\pi}
{(1+u\cos\alpha+\sqrt{1-u^2}\cos\phi\sin\alpha)\,x^{s-1}\over
1+x+u\cos\alpha+\sqrt{1-u^2}\cos\phi\sin\alpha} .
\end{align}
The $\phi$-integration can be carried out using the residue theorem, with the result
\begin{align}
C=\int_{-1}^{1}  &  du \left(1+u\right)^{1-s}{\sin(s\pi)\over2\pi}\dots \nonumber\\
&\dots\int_0^\infty dx\, x^{s-1}
\left[
1 - {x\over\sqrt{(1 + x + u\cos\alpha)^2 - (1 - u^2)\sin^2\alpha}}
\right].
\label{after x}
\end{align}
Minimization over $\vec m_1$ and $\vec m_2$ is equivalent to minimization over $0\le\alpha\le \pi$. So,
let us take the derivative with respect to $\alpha$. One has
\begin{alignat}{2}
{dC\over d\alpha}  &=  -\sin\alpha{\sin(s\pi)\over2\pi}&&\int_0^\infty dx\, x^s\int_{-1}^{1}   du \dots\nonumber\\
&&&\dots\left(1+u\right)^{1-s}{(1+x)u+\cos\alpha\over\left[(1 + x + u\cos\alpha)^2 - (1 - u^2)\sin^2\alpha\right]^{3/2}}\nonumber\\
&:= -\sin\alpha{\sin(s\pi)\over2\pi}&&\int_0^\infty dx\, x^s f(\alpha,x),
\label{def f}
\end{alignat}
which shows that $\alpha=0,\,\pi$, are extreme points of $C$. To prove that these are the only extremes,
we next show that $f(\alpha,x)>0$ if $0<\alpha< \pi$. We first notice that the integrant in the first line of Equation~\eqref{def f} is
$$
-{1\over \sin^2\alpha}\; {d\over d u} {1+x+u \cos\alpha\over \sqrt{(1 + x + u\cos\alpha)^2 - (1 - u^2)\sin^2\alpha}}.
$$
Thus, integrating by parts one has
\begin{equation}
f(\alpha,x)={1 - s\over\sin^2 \alpha} \int_{-1}^1  \, {dx\,(1+u)^{-s}(1+x+u\cos\alpha)\over \sqrt{(1 + x + u\cos\alpha)^2 - (1 - u^2)\sin^2\alpha}}
-{2^{1-s}\over\sin^2\alpha}.
\end{equation}
We further note that if $0<\alpha<\pi$,
\begin{equation}
{1+x+u\cos\alpha\over \sqrt{(1 + x + u\cos\alpha)^2 - (1 - u^2)\sin^2\alpha}}>1.
\end{equation}
Hence
\begin{equation}
f(\alpha,x)>{1 - s\over\sin^2 \alpha}\int_{-1}^1dx\, (1+u)^{-s}-{2^{1-s}\over\sin^2\alpha}=0.
\end{equation}
Therefore, $dC/d\alpha\le 0$, and it only vanishes at $\alpha=0,\pi$. It follows that $C$ has a maximum at $\alpha=0$ and a minimum at $\alpha=\pi$.

The remainder of the calculation can be found in Section~\ref{DPexamplesCovMeas}. 


\chapter{Bounds on the concavity of the von Neumann entropy}\label{BoundsOnConc}
In Chapter \ref{infoCombining} we needed to relate the fidelity characteristic $f=F(\rho_0,\rho_1)$ of a binary-input classical-quantum channel with output states $\rho_0$ and $\rho_1$ back to its symmetric capacity $\log2-H$, therefore we need a lower bound on the concavity of the von Neumann entropy. This will be a special case of the following new bounds (see also Remark \ref{remarkOnRenesBound}):

\begin{thm}[Lower bounds on concavity of von Neumann entropy]\label{vNconcavityimprovement}
Let $\rho_i\in{\mathcal B}(\mathbb C^d)$ be quantum states for $i=1,\ldots,n$ and $\{p_i\}_{i=1}^n$ be a probability distribution. Then:
\begin{align}
H\left(\sum_{i=1}^np_i\rho_i\right)&-\sum_{i=1}^np_iH(\rho_i)\nonumber\\
&=H(\{p_i\})-D\Big(\sum_{i,j=1}^n\sqrt{p_ip_j}|i\rangle\langle j|\otimes\sqrt{\rho_i}\sqrt{\rho_j}\Big\|\sum_{i=1}^np_i|i\rangle\langle i|\otimes\rho_i\Big)\label{equalityinconcavity}\\
&\geq H(\{p_i\})-\log\Big(1+2\sum_{1\leq i<j\leq n}\sqrt{p_ip_j}{\textrm tr}[\sqrt{\rho_i}\sqrt{\rho_j}]\Big)\label{concavityLBwithSQRT}\\
&\geq H(\{p_i\})-\log\Big(1+2\sum_{1\leq i<j\leq n}\sqrt{p_ip_j}F(\rho_i,\rho_j)\Big).\label{concavityLBwithF}
\end{align}
\end{thm}
\begin{proof} We will obtain the equality \eqref{equalityinconcavity} by keeping track of the gap term in the proof of the upper bound on the concavity in \cite[Theorem 11.10]{book2000mikeandike}, and the further inequalities by bounding the relative entropy from above. For the proof, define $\rho:=\sum_{i=1}^np_i\rho_i$.

Denote by $|\Omega\rangle_{AC}:=\sum_{i=1}^d|i\rangle_A\otimes|i\rangle_C$ the (unnormalized) maximally entangled state between two systems $A$ and $C$ of dimension $d$. Then $|\phi_i\rangle_{AC}:=(\mathbbm1\otimes\sqrt{\rho_i})|\Omega\rangle_{AC}$ are purifications of the $\rho_i$ in the sense that ${\textrm tr}_A[|\phi_i\rangle\langle\phi_i|_{AC}]=\rho_i$. We also have ${\textrm tr}_C[|\phi_i\rangle\langle\phi_i|_{AC}]=\rho_i^T$, where $^T$ denotes the transposition w.r.t.\ the basis $\{|i\rangle\}_A$. For a system $B$ of dimension $n$ with orthonormal basis $\{|i\rangle_B\}_{i=1}^n$, the state
\begin{align*}
|\psi\rangle_{ABC}:=\sum_{i=1}^n\sqrt{p_i}|i\rangle_B\otimes|\phi_i\rangle_{AC}
\end{align*}
is therefore a purification of $\rho^T$ in the sense that $\rho^T=\psi_A:={\textrm tr}_{BC}[\psi_{ABC}]$, where we have defined $\psi_{ABC}:=|\psi\rangle\langle\psi|_{ABC}$. Since the transposition leaves the spectrum invariant, we have $H(\rho)=H(\rho^T)=H(\psi_A)=H(\psi_{BC})$, where
\begin{align*}
\psi_{BC}:={\textrm tr}_A[\psi_{ABC}]=\sum_{i,j=1}^n\sqrt{p_ip_j}|i\rangle\langle j|_B\otimes(\sqrt{\rho_i}\sqrt{\rho_j})_C.
\end{align*}

Consider now the map $P_B(X):=\sum_{i=1}^n|i\rangle\langle i|_BX|i\rangle\langle i|_B$ acting on subsystem $B$, such that $P_B(\psi_{BC})=\sum_{i=1}^np_i|i\rangle\langle i|\otimes\rho_i$. Note that $P_B=P_B^*$ represents a projective measurement on $B$ and is selfadjoint w.r.t.\ the Hilbert-Schmidt inner product. We can therefore write:
\begin{align*}
D(\psi_{BC}\|P_B(\psi_{BC}))&=-H(\psi_{BC})-{\textrm tr}[\psi_{BC}\log P_B(\psi_{BC})]\\
&=-H(\psi_{BC})-{\textrm tr}[P_B(\psi_{BC})\log P_B(\psi_{BC})]\\
&=-H(\rho)+H(P_B(\psi_{BC}))\\
&=-H(\rho)+H(\{p_i\})+\sum_{i=1}^np_iS(\rho_i),
\end{align*}
which proves the equality \eqref{equalityinconcavity}.

To obtain the lower bound \eqref{concavityLBwithSQRT}, we bound the relative entropy from above by the \emph{sandwiched Renyi-$\alpha$ divergence} of order $\alpha=2$ \cite{WWY14,muller2013quantum,tomamichel2015quantum}:
\begin{align*}
D(\psi_{BC}\|P_B(\psi_{BC}))&\leq D_2(\psi_{BC}\|P_B(\psi_{BC})) \\
&=\log{\textrm tr}[(P_B(\psi_{BC}))^{-1/2}\psi_{BC}(P_B(\psi_{BC}))^{-1/2}\psi_{BC}].
\end{align*}
Note that the sandwiched Renyi divergences are the \emph{minimal} quantum generalizations of the classical Renyi-$\alpha$ divergences \cite{tomamichel2015quantum}, which will be advantageous to obtain a good lower bound. We can continue by using the explicit forms of $\psi_{BC}$ and $P_B(\psi_{BC})$ from above:
\begin{align*}
D(\psi_{BC}\|P_B(\psi_{BC}))&\leq\log{\textrm tr}\Big[\Big(\sum_{i,j=1}^n|i\rangle\langle j|_B\otimes\mathbbm1_C\Big)\Big(\sum_{k,l=1}^n\sqrt{p_kp_l}|k\rangle\langle l|_B\otimes\sqrt{\rho_k}\sqrt{\rho_l}\Big)\Big]\\
&=\log\Big(\sum_{i,j=1}^n\sqrt{p_ip_j}{\textrm tr}[\sqrt{\rho_i}\sqrt{\rho_j}]\Big),
\end{align*}
which agrees with \eqref{concavityLBwithSQRT} since the terms with $i=j$ sum to $\sum_{i=1}^np_i{\textrm tr}[\rho_i]=1$. The final bound \eqref{concavityLBwithF} is obtained by noting that $F(\rho_i,\rho_j)=\|\sqrt{\rho_i}\sqrt{\rho_j}\|_1\geq{\textrm tr}[\sqrt{\rho_i}\sqrt{\rho_j}]$ holds for any quantum states \cite{book2000mikeandike,audenaert2012comparisons}.
\end{proof}

\begin{rem}[Upper bounds on concavity of von Neumann entropy]
The equality \eqref{equalityinconcavity} in Theorem \ref{vNconcavityimprovement} can also be used to obtain \emph{upper} bounds on the concavity of von Neumann entropy: As opposed to the proof of Theorem \ref{vNconcavityimprovement}, where we used the upper bound $D\leq D_2$ involving the sandwiched Renyi-$2$ divergence, one could bound the relative entropy $D$ from below, e.g.\ using the Pinsker inequality~\cite{W13} or using a smaller divergence measure such as one of the various Renyi-$\alpha$ divergences with parameter $\alpha\in[0,1)$.
\end{rem}

We will later need the special case $n=2$ of Theorem \ref{vNconcavityimprovement} with uniform probabilities $\{p_i\}$ together with a bound from \cite{PhysRevLett.105.040505}, in order to obtain a bound on the fidelity parameter $f$ in terms of the channel entropy $H$ for binary-input classical-quantum channels.
\begin{thm}[Relation between fidelity parameter and channel entropy]\label{tightRelationFHtheorem}
Let $\sigma_0,\sigma_1$ be quantum states, and define $f:=F(\sigma_0,\sigma_1)$ and $H=\log2-H((\sigma_0+\sigma_1)/2)+(H(\sigma_0)+H(\sigma_1))/2=H(X|B)$, where $H(X|B)$ is evaluated on the state $\frac{1}{2}|0\rangle\langle0|_X\otimes(\sigma_0)_B+\frac{1}{2}|1\rangle\langle1|_X\otimes(\sigma_1)_B$. Then the following bound holds:
\begin{align}\label{tightrelationHf}
e^H-1\leq f\leq1-2h_2^{-1}(\log2-H),
\end{align}
where $h_2^{-1}:[0,\log2]\to[0,1/2]$ is the the inverse of the binary entropy function.
\end{thm}
\begin{proof}The lower bound follows immediately from Theorem \ref{vNconcavityimprovement} in the special case of $n=2$  states $\sigma_0,\sigma_1$ with equal probabilities $p_0=p_1=1/2$:
\begin{align*}
\log2-H &=H\Big(\frac{\sigma_0+\sigma_1}{2}\Big)-\frac{H(\sigma_0)+H(\sigma_1)}{2} \\
&\geq\log2-\log(1+F(\sigma_0,\sigma_1))=\log2-\log(1+f).
\end{align*}
For the other direction, we need the following bound from \cite{PhysRevLett.105.040505}:
\begin{align*}
\log2-H &=H\Big(\frac{\sigma_0+\sigma_1}{2}\Big)-\frac{H(\sigma_0)+H(\sigma_1)}{2} \\ 
&\leq h_2\Big(\frac{1-F(\sigma_0,\sigma_1)}{2}\Big)=h_2\Big(\frac{1-f}{2}\Big),
\end{align*}
where $h_2$ is the binary entropy function. The upper bound in \eqref{tightrelationHf} follows now by noting that the inverse function $h_2^{-1}:[0,\log2 ]\to[0,1/2]$ is monotonically increasing.
\end{proof}

\begin{rem}\label{remarkonotherconcavitybounds}The main feature of the bound \eqref{tightrelationHf} for our purposes is that it is tight on \emph{both} ends of the interval $H\in[0,\log2]$. Namely, the bound implies $H=0\Leftrightarrow f=0$ as well as $H=\log2\Leftrightarrow f=1$, see also Fig.\ \ref{figboundHf}. In particular, we are not aware of any previous bound showing that small fidelity $f\approx0$ implies $H$ to be close to $0$. Such a statement, however, is needed for our proofs of Theorems \ref{QuantumMrsGerberTheorem2DifferentStates} and \ref{QuantumMrsGerberTheorem1StateTwice} (see Eqs.\ (\ref{first-asymmetric-lower-bound-eqn}) and (\ref{minof2expressionsinproofofHH})).

In particular, the bound $\log2-H\geq\frac{1}{2}\left(\frac{1}{2}\|\sigma_0-\sigma_1\|_1\right)^2$, which is the main result of \cite{kim2014bounds}, can never yield any non-trivial information for $H\in[0,(\log2)-1/2)$ (i.e.\ near $f\approx0$), since its right-hand side will never exceed $\frac{1}{2}$. Using the Fuchs-van de Graaf inequality $\frac{1}{2}\|\sigma_0-\sigma_1\|_1\geq1-f$ \cite{book2000mikeandike}, we would only obtain the bound $f\geq1-\sqrt{2(\log2-H)}$, which is also shown in Fig.\ \ref{figboundHf}.

Our lower bounds (\ref{concavityLBwithSQRT}) and (\ref{concavityLBwithF}) are generally good when the states $\rho_i$ are close to pairwise orthogonal: If $\max_{i\neq j}F(\rho_i,\rho_j)$ becomes close to $0$ then these lower bounds approach the value $H(\{p_i\})$, which is the value of the left-hand-side of the inequality for exactly pairwise orthogonal states $\rho_i$. Note however that the lower bounds (\ref{concavityLBwithSQRT}) and (\ref{concavityLBwithF}) can become negative and therefore trivial, e.g.\ when all states $\rho_i$ coincide (or have high pairwise fidelity) and the probability distribution $\{p_i\}$ is not uniform on its support. For a uniform probability distribution $\{p_i=1/n\}$ and any states $\rho_i$, the bounds (\ref{concavityLBwithSQRT}) and (\ref{concavityLBwithF}) are however always nonnegative (this case also covers Theorem \ref{tightRelationFHtheorem}).
\end{rem}
\begin{rem}\label{remarkOnAlexDaniel}
In \cite{MFW16} a different lower bound on the concavity of the von Neumann entropy was found which was shown to outperform the bound in \cite{kim2014bounds} in some cases. The bound is given in terms of the relative entropy, which can be easily bounded by $D(\rho || \sigma) \geq -2 \log{F(\rho, \sigma)}$, see \cite{muller2013quantum}. Nevertheless this bound can not be used in our general scenario since it becomes trivial whenever the involved states are pure. Note that this is not the case for our bound presented above. 
\end{rem}
\begin{rem}\label{remarkOnRenesBound}
Shortly before the initial submission of our paper we discovered that the bound from Theorem \ref{tightRelationFHtheorem} (the case of uniform input distribution) has recently been given in \cite{NR17}, also in the context of polar codes. Our Theorem \ref{vNconcavityimprovement} is however more general, it constitutes an \emph{equality} form of the concavity of the von Neumann entropy which allows for convenient relaxations, and is valid for \emph{non-uniform} distributions. Furthermore a weaker bound can already be found in~\cite{SRDR13}.
\end{rem}

\begin{figure}[t!]
\centering
\begin{overpic}[trim=3.1cm 21.3cm 8.7cm 2.6cm, clip, scale=0.8]{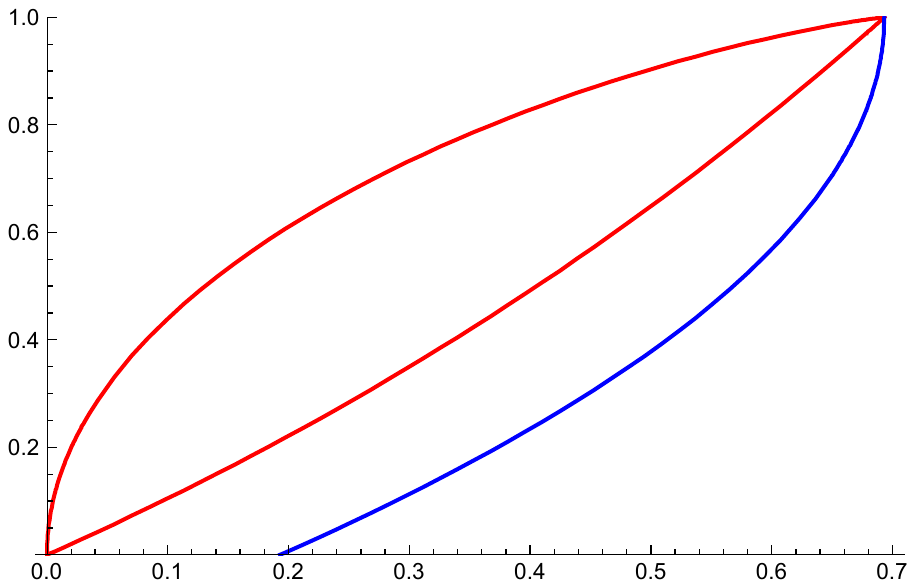}   
\put (-9,62) {$f$}
\put (106,0) {$H$}
\end{overpic}
\caption{\label{figboundHf} 
The red curves show the upper and lower bounds from (\ref{tightrelationHf}), the blue curve the lower bound via \cite{kim2014bounds} (see Remark \ref{remarkonotherconcavitybounds}).}
\end{figure}

\chapter{Properties of gaussian correlation measures}\label{appGauss}

We now prove the physically fundamental properties of the Gaussian steerability measure \eqref{G steer} as stated in Theorem \ref{G prop}. In \cite{steerability}, these facts were stated and proven only in the particular case in which the steered system is made of one mode.

\begin{proof}[Proof of Theorem \ref{G prop}] $ \\[-3ex] $
\begin{itemize}
\item[(1)] \textit{ $\mathcal{G}(A\rangle B)_V$ is convex and decreasing as a function of the CM $V_{AB}>0$.} \\
Both properties follow straightforwardly by combining concavity and monotonicity of the Schur complement with the convexity of $g_-$. Let us prove convexity for instance. Since the Schur complement is concave, for any $V_{AB},W_{AB}>0$ and $0\leq p\leq 1$ we obtain
\begin{equation*}
(pV_{AB}+(1-p)W_{AB}) \big/ (pV_A+(1-p) W_A)\ \geq\ p\, V_{AB}/V_A + (1-p)\, W_{AB}/W_A\, .
\end{equation*}
Applying the fact that $g_-$ is decreasing and convex gives
\begin{align}
&\mathcal{G}(A\rangle B)_{pV_{AB}+(1-p)W_{AB}}\ \nonumber\\
&=\ g_- \big( (pV_{AB}+(1-p)W_{AB}) \big/ (pV_A+(1-p) W_A) \big)\ \nonumber\\ 
&\leq\ g_- \big( p V_{AB}/V_A + (1-p) W_{AB}/W_A\big)\ \nonumber\\
&\leq\ p\, g_-(V_{AB}/V_A)\, +\, (1-p)\, g_- (W_{AB}/W_A)\ \nonumber\\ 
&=\ p\, \mathcal{G}(A\rangle B)_V\, +\, (1-p)\, \mathcal{G}(A\rangle B)_W \, . \nonumber
\end{align}

\item[(2)] \textit{ $\mathcal{G}(A\rangle B)$ is monotonically decreasing under general, non--deterministic Gaussian maps on the steering system $A$.} \\
Using the monotonicity of the Schur complement under general Gaussian maps, as given in Theorem \ref{CP incr sch} of the main text, one gets
\begin{equation*}
\Gamma_{A\rightarrow A'} (V_{AB}) \big/ \Gamma_{A\rightarrow A'}(V_A)\ \geq\ V_{AB} / V_A\, .
\end{equation*}
Applying $g_-$ to both sides yields exactly
\begin{equation}
\mathcal{G}(A'\rangle B)_{\Gamma_{A\rightarrow A'}(V_{AB})}\ \leq\ \mathcal{G}(A\rangle B)_{V_{AB}} \, .
\end{equation}

\item[(3)] \textit{ ${\mathcal G}(A\rangle B)$ is monotonically decreasing under deterministic, quantum Gaussian channels on the steered party $B$.} \\
Using the tripartite version of \eqref{mon steer 1} together with the positivity of the steerability measure, one easily deduces $\mathcal{G}(C\rangle AB)\geq \mathcal{G}(C \rangle A)$. Together with the invariance under symplectic transformations, this proves that $\mathcal{G}$ is decreasing under deterministic quantum Gaussian channels on the steered system.

\item[(4)] \textit{ ${\mathcal G}(A\rangle B)$ is additive under tensor products.}
We show 
\begin{equation*}
{\mathcal G}(A_1 A_2\rangle B_1 B_2)_{V_{A_1B_1}\oplus W_{A_2 B_2}} =  \mathcal{G}(A_1\rangle B_1)_{V_{A_1B_1}} +  \mathcal{G}(A_2\rangle B_2)_{W_{A_2 B_2}}.
\end{equation*}
This is straight forward, since
\begin{align}
&\mathcal{G}(A_1A_2\rangle B_1B_2)_{V_{A_1B_1}\oplus W_{A_2 B_2}}\ \nonumber\\
&=\ g_-\left( (V_{A_1B_1}\oplus W_{A_2 B_2}) \big/ (V_{A_1}\oplus W_{A_2}) \right)\, \nonumber\\ 
&=\, g_-\left( V_{A_1B_1} / V_{A_1} \oplus W_{A_2 B_2}/ W_{A_2} \right)\, \nonumber\\
&=\, g_-\left( V_{A_1B_1} / V_{A_1}\right)\, +\, g_-\left(W_{A_2 B_2} / W_{A_2} \right)\ \nonumber\\ 
&=\, \mathcal{G}(A_1\rangle B_1)_{V_{A_1B_1}}\, +\, \mathcal{G}(A_2\rangle B_2)_{W_{A_2 B_2}}\, . \nonumber
\end{align}

\item[(5)] \textit{ The upper bound ${\mathcal G}(A\rangle C)_V  \leq  g_+(V_{BC}/V_B)$ holds for any quantum CM $V_{ABC}$.} \\
Taking \eqref{INEQ 2}, applying $g_-$ and using its elementary properties yields exactly
\begin{equation}
\mathcal{G}(A\rangle C)\ =\ g_-(V_{AC}/V_A)\ \leq\ g_+(V_{BC}/V_B)\, .
\end{equation}
\end{itemize}
\end{proof}

\begin{rem}
We have just seen that the steerability measure $\mathcal{G}(A\rangle B)$ is decreasing under quantum Gaussian channels on the steered system $B$ (point 4 of Theorem \ref{G prop}). One could wonder, whether this is true also for more general non-deterministic quantum Gaussian maps. The answer is no, as some numerical counterexamples can easily show.
\end{rem}

\begin{rem}
As a corollary of Theorem \ref{G prop}, one sees easily that the R\'{e}nyi--2 measure of entanglement
\begin{equation}
E^{\text{G}}_{F,2}(A:B)\, =\, \inf_{\gamma_{AB}\, \text{pure}:\ \gamma_{AB}\,\leq\, V_{AB}}\, \frac{1}{2}\, \log\det \gamma_A\, ,
\end{equation}
is an upper bound on the steerabilities $\mathcal{G}(A\rangle B)$ and $\mathcal{G}(B\rangle A)$. In fact, consider the optimal pure $\gamma_{AB}\leq V_{AB}$ in the above equation and write
\begin{equation*}
E^{\text{G}}_{F,2}(A:B)_V\, =\, \frac{1}{2}\log\det \gamma_A\, =\, g_-(\gamma_{AB}/\gamma_A)\, \geq\, g_-(V_{AB}/V_A)\, =\, \mathcal{G}(A\rangle B)_V\, ,
\end{equation*}
where we used first the expression of the steerability in terms of local determinant for pure states and then the fact that $\mathcal{G}({A \rangle B})$ is monotonically decreasing as a function of the CM.
\end{rem}

We continue by providing the proof for Theorem \ref{G mono}.

\begin{proof}[Proof of Theorem \ref{G mono}] 
(a) It suffices to prove the inequality $\mathcal{G}(A\rangle BC)\, \geq\, \mathcal{G}(A\rangle B) + \mathcal{G}(A\rangle C)$ for a tripartite quantum CM $V_{ABC}$, as (\ref{mon steer 1}) would follow by iteration. Observe that $V_{AB}/V_A$ and $V_{AC}/V_A$ form the diagonal blocks of the bipartite matrix $V_{ABC}/V_A$. Applying \eqref{dec red g- eq} one thus obtains $\mathcal{G}(A\rangle BC)_V = g_-(V_{ABC}/V_A) \geq g_-(V_{AB}/V_A) + g_-(V_{AC}/V_A) = \mathcal{G}(A\rangle B)_V + \mathcal{G}(A\rangle BC)_V$, concluding the proof.

(b) For the case $n_A=1$ with $n_{B_j}$ arbitrary, one exploits the fact that only one term $\mathcal{G}(B_j\rangle A)$ in the right-hand side of (\ref{mon steer 2}) can be nonzero, due to the impossibility of jointly steering a single mode by Gaussian measurements as implied by (\ref{log det ineq}) \cite{Adesso}, combined with the monotonicity of $\mathcal{G}(B_1\ldots B_k \rangle A)$ under partial traces on the steering party as implied by Theorem~\ref{G prop}. 

The validity of \eqref{mon steer 2} for pure Gaussian states can be easily inferred by putting together inequality \eqref{E mono} and Theorem \ref{I2E}:
\begin{align*}
&\mathcal{G}(B_1\ldots B_k \rangle A)_V\, =\, E^{\text{G}}_{F,2}(B_1\ldots B_k:A)_V\,  \\
&\geq\, \sum_{j=1}^k E^{\text{G}}_{F,2}(B_j:A)_V \, \geq\, \sum_{j=1}^k \mathcal{G}(B_j\rangle A)_V\, ,
\end{align*}
where the first equality holds specifically for pure states.

On the contrary, already in the simplest case $k=2$, $n_A=2,\, n_{B_1}=n_{B_2}=1$, there exist mixed states violating inequality~\eqref{mon steer 2}. A counterexample is as follows:
\begin{equation} V_{AB_1B_2}\ =\ \begin{pmatrix}
 1.2 & -0.3 & 0.4 & -2.7 & 1.8 & -1.9 & 0.4 & -0.1 \\
 -0.3 & 0.9 & -1.2 & 0.4 & -1.2 & 0.5 & -0.4 & 0.1 \\
 0.4 & -1.2 & 4.5 & 1.6 & -1.4 & 1.8 & -0.1 & -0.3 \\
 -2.7 & 0.4 & 1.6 & 12. & -9.5 & 10.1 & -1.4 & -0.3 \\
 1.8 & -1.2 & -1.4 & -9.5 & 11.9 & -11.5 & 1.6 & 0.8 \\
 -1.9 & 0.5 & 1.8 & 10.1 & -11.5 & 11.9 & -1. & -1.4 \\
 0.4 & -0.4 & -0.1 & -1.4 & 1.6 & -1. & 2.4 & -2. \\
 -0.1 & 0.1 & -0.3 & -0.3 & 0.8 & -1.4 & -2. & 2.8
\end{pmatrix}\, .
\end{equation}
Here, the first four rows and columns pertain to $A$, the fifth and sixth to $B_1$, the last two to $B_2$. It can be easily verified that the minimum symplectic eigenvalue of the above matrix with respect to the symplectic form $\Omega_A\oplus \Omega_{B_1}\oplus\Omega_{B_2}$ is $\nu_{\min}(V_{B_1B_2A})=1.01359$, so that $V_{B_1B_2 A}$ is a legitimate quantum CM. However,
\begin{equation}
\mathcal{G}(B_1B_2\rangle A)_V\, -\, \mathcal{G}(B_1\rangle A)_V\, -\, \mathcal{G}(B_2\rangle A)_V\, =\, -0.816863\, .
\end{equation}
\end{proof}

Finally, we prove Proposition \ref{follia prop}. 

\begin{proof}[Proof of Proposition \ref{follia prop}]
Let us start by applying Lemma~\ref{lemma fact out} to decompose the symplectic space of $AB$ as $\Sigma_{AB}=\Sigma_R \oplus \Sigma_S$ in such a way that $V_{AB}=V_R \oplus \eta_S$, where $V_R > i\Omega_R$ and $\eta_S$ is a pure QCM. According to Lemma~\ref{lemma fact out}, the purification $\gamma_{ABC}$ can be taken to be of the form $\gamma_{ABC}=\gamma_{RC_{1}} \oplus \eta_S\oplus \delta_{C_{2}}$, with $\gamma_{C_{1}}>i\Omega_{C_{1}}$, $n_{C_{1}}=n_R$, and $\delta_{C_{2}}$ pure. If $\tau\leq V$ is a pure QCM, a projection onto $\Sigma_S$ reveals that $\tau_S = \Pi_S \tau \Pi_S^\intercal \leq \eta_S$. Since $\tau_S$ must be a legitimate QCM, and pure states are minimal within the set of QCMs, we deduce that $\tau_S=\eta_S$. Then, an application of Lemma~\ref{lemma pure reduction} allows us to conclude that $\tau = \tau_R \oplus \eta_S$, and accordingly $\tau_R \leq V_R$.

We claim that for all pure $\tau_R < V_R$ there is a pure QCM $\sigma_{C_{1}}$ such that
\begin{equation}
\left(\gamma_{RC_{1}} + 0_{R} \oplus \sigma_{C_{1}} \right) \big/ \left( \gamma_{C_{1}} +\sigma_{C_{1}} \right) = \tau_R .
\label{follia prop eq2}
\end{equation}
Constructing the extension $\sigma_{C}\coloneqq \sigma_{C_{1}}\oplus \tilde{\sigma}_{C_{2}}$, where $\tilde{\sigma}_{C_{2}}$ is an arbitrary pure QCM, we see that~\eqref{follia prop eq2} can be rewritten as 
\begin{equation}
\left(\gamma_{ABC} + 0_{AB} \oplus \sigma_{C} \right) \big/ \left( \gamma_{C} +\sigma_{C} \right) = \tau_R \oplus \eta_{S} .
\label{follia prop eq2bis}
\end{equation}
In fact, adding the ancillary system $C_{2}$ does not produce any effect on the Schur complement, since there are no off-diagonal block linking $C_{2}$ with any other subsystem. Analogously, the $S$ component of the $AB$ system can be brought out of the Schur complement because it is in direct sum with the rest.

In light of~\eqref{follia prop eq2bis}, we know that once~\eqref{follia prop eq2} has been established, in~\eqref{follia prop eq} we can achieve all QCMs $\gamma'$ that can be written as $\tau_{R}\oplus \eta_{S}$, with $\tau_{R}<V_{R}$. It is not difficult to see that this would allow us to conclude. Before proving~\eqref{follia prop eq2}, let us see why. The main point here is that every pure QCM $\tau_R\leq V_R$ can be thought of as the limit of a sequence of pure QCMs $\tau_R(t)< V_R$. An explicit formula for such a sequence reads $\tau_R(t) = \tau_R \#_t \gamma_{V_R}^\#$, where $\gamma_{V_R}^\#$ is the pure QCM defined in Lemma~\ref{lemma gamma sharp}, and $\#_t$ denotes the weighted geometric mean~\eqref{geom geod}. Observe that: (i) $\tau_R(t)$ is a QCM since it is known that the set of QCMs is closed under weighted geometric mean~\cite[Corollary 8]{SympIneq}; (ii) $\tau_R(t)$ is in fact a pure QCM, because according to~\eqref{det geom} its determinant satisfies $\det \tau_R(t) = \left(\det \tau_R\right)^{1-t} \big(\det \gamma_{V_R}^\# \big)^t = 1$; (iii) $\lim_{t\rightarrow 0^+} \tau_R(t)=\tau_R$ as can be seen easily from~\eqref{geom geod}; and (iv) $\tau_R(t)< V_R$ for all $t>0$. This latter fact can be justified as follows. Since $V_R > i\Omega_R$, from Lemma~\ref{lemma gamma sharp} we deduce $\gamma_{V_R}^\# < V_R$. Taking into account that $\tau_R \leq V_R$, the claim follows from the strict monotonicity of the weighted geometric mean, in turn an easy consequence of~\eqref{geom geod}.

Now, let us prove~\eqref{follia prop eq2}. We start by writing
\begin{equation*}
\gamma_{RC_{1}} = \begin{pmatrix} V_R & L \\ L^\intercal & \gamma_{C_{1}} \end{pmatrix} ,
\end{equation*}
where $V_R>i\Omega_R$, $\gamma_{C_{1}}> i\Omega_{C_{1}}$, and the off-diagonal block $L$ is square. As a matter of fact, more is true, namely that $L$ is also invertible. The simplest way to see this involves two ingredients: (a) the identity $\Omega V_{R}^{-1} \Omega^\intercal = \gamma_{RC_{1}} / \gamma_{C_{1}} = V_{R} - L \gamma_{C_{1}}^{-1} L^\intercal$, easily seen to be a special case of Equation~\ref{lemma2}; and (b) the fact that $V_R> \Omega V_R^{-1} \Omega^\intercal$ because of Lemma~\ref{lemma gamma sharp}. Combining these two ingredients we see that
\begin{equation*}
V_{R} > \Omega V_{R}^{-1} \Omega^\intercal = V_{R} - L \gamma_{C_{1}}^{-1} L^\intercal ,
\end{equation*}
which implies $L \gamma_{C_{1}}^{-1} L^\intercal>0$ and in turn the invertibility of $L$. Now, for a pure QCM $\tau_R< V_R$, take $\sigma_{C_{1}} = L^\intercal (V_R-\tau_R)^{-1} L - \gamma_{C_{1}}$. On the one hand, 
\begin{align*}
\big(\gamma_{RC_{1}} \!+ 0_R\!\oplus\! \sigma_{C_{1}}\big) \big/ \big(\gamma_{C_{1}} \!+ \sigma_{C_{1}}\big) &= V_R - L \left(\gamma_{C_{1}}\! + \sigma_{C_{1}}\right)^{-1}\! L^\intercal \\[0.8ex]
&= \tau_R
\end{align*}
by construction. On the other hand, write
\begin{align*}
\sigma_{C_{1}} \!- i\Omega_{C_{1}} &= L^\intercal (V_R - \tau_R)^{-1} L - (\gamma_{C_{1}} + i\Omega_{C_{1}}) \\[0.8ex]
&= L^\intercal (V_R - \tau_R)^{-1} L - L^\intercal (V_R + i \Omega_R)^{-1} L \\[0.8ex]
&= L^\intercal \left( (V_R - \tau_R)^{-1} - (V_R + i \Omega_R)^{-1} \right) L \\[0.8ex]
&= L^\intercal (V_R - \tau_R)^{-1} \times \left( (V_R + i \Omega_R) - (V_R - \tau_R) \right)\times (V_R + i \Omega_R)^{-1} L \\[0.8ex]
&= L^\intercal (V_R - \tau_R)^{-1} \left( \tau_R + i \Omega_R \right) (V_R + i \Omega_R)^{-1} L ,
\end{align*}
where we employed Lemma~\ref{lemma follia 0} in the form $\gamma_{C_{1}} + i \Omega_{C_{1}} = L^\intercal (V_R + i \Omega_R)^{-1} L$ and performed some elementary algebraic manipulations. Now, from the third line of the above calculation it is clear that $\sigma_{C_{1}} - i \Omega_{C_{1}}\geq 0$, since from $V_R - i \Omega_R \geq V_R - \tau_R > 0$ we immediately deduce $(V_R - \tau_R)^{-1} \geq (V_R + i \Omega_R)^{-1}$. This shows that $\sigma_{C_{1}}$ is a valid QCM. Moreover, observe that
\begin{align*}
\rk \left(\sigma_{C_{1}} - i \Omega_{C_{1}}\right) &= \rk \left( L^\intercal (V_R - \tau_R)^{-1} \left( \tau_R + i \Omega_R \right) (V_R + i \Omega_R)^{-1} L \right) \\[0.8ex]
&= \rk \left( \tau_R + i \Omega_R \right) \\
&= n_R \\[0.8ex]
&= n_{C_{1}} ,
\end{align*}
which tells us that $\sigma_{C_{1}}$ is also a pure QCM.
\end{proof}
\bibliographystyle{alpha}
\bibliography{biblio}

\end{document}